\documentclass
  [
    DIV=14,
    fontsize=12,
    abstract=true,
    bibliography=totoc
  ]
  {scrartcl}

\newif\ifJournalVersion
\JournalVersiontrue 

\newcommand{\notinjournal}[1]{%
  \ifJournalVersion\else#1\fi%
}
\newcommand{\onlyinjournal}[1]{%
  \ifJournalVersion#1\else\fi%
}

\setkomafont{date}{\small\rmfamily}
\setkomafont{author}{\large\scshape}

\usepackage{authblk}

\newcommand
{\NameEmail}[2]
{%
  \stackunder%
  {#1}%
  {%
   \clap{%
   \normalfont%
   \footnotesize%
   \texttt{#2}%
  }}%
  \hspace{-3.5pt}%
}

\newcommand\keywords[1]{%
  \begingroup
  \renewcommand\thefootnote{}
  \footnote{{keywords:} #1}%
  \addtocounter{footnote}{-1}
  \endgroup
}
\newcommand\MSCcodes[1]{%
  \begingroup
  \renewcommand\thefootnote{}
  \footnote{{MSC codes:} #1}%
  \addtocounter{footnote}{-1}
  \endgroup
}
\newcommand\funding[1]{%
  \begingroup
  \renewcommand\thefootnote{}
  \footnote{{funding:} #1}%
  \addtocounter{footnote}{-1}
  \endgroup
}

\usepackage{scrlayer-scrpage}
\clearpairofpagestyles 
\ihead{Higher Gauge Theory} 
\ohead{\headmark}
\automark{section}

\NewCommandCopy{\oldparagraph}{\paragraph}
\RenewDocumentCommand{\paragraph}{s O{#3} m}{%
  \IfBooleanTF{#1}
    {\oldparagraph*{#3.}}       
    {\oldparagraph[#2]{#3.}}    
}

\usepackage{caption}
\usepackage
  [leftcaption]
  {sidecap}

\usepackage{enumitem}
\setlist[enumerate,1]{label=\textbf{\textup{(\roman*)}}}
\setlist[enumerate,2]{label=\textbf{\textup{(\alph*)}}}
\setlist{
  topsep=1pt,
  itemsep=2pt,
  leftmargin=1cm,
  before=\leavevmode
}

\usepackage[
  backend=biber,
  style=alphabetic,
  giveninits=true, 
  backref=true
]{biblatex}
\usepackage{lmodern}
\usepackage{anyfontsize}
    
\usepackage{multirow}
\usepackage{hhline}

\usepackage{tabularray}
\UseTblrLibrary{booktabs}

\usepackage{adjustbox}
\usepackage{tcolorbox}

\newtcolorbox{standout}{
  colback=gray!15,
  boxrule=0pt,
  left=.3cm,
  right=.3cm,
  top=.18cm,
  bottom=.18cm,
  boxsep=0pt
}

\usepackage{amsmath}
\allowdisplaybreaks
\usepackage{mathtools}
\usepackage{amssymb}
\usepackage{amsthm}
\usepackage{xfrac}
\usepackage{stmaryrd} 
\usepackage{scalerel} 
\usepackage{stackengine} 
\usepackage[safe]{tipa} 
\usepackage[
  cal=euler,
  scr=dutchcal
]
 {mathalpha} 

 \newcommand{\bracket}[3]{%
  \stretchleftright
    {#1}
    {%
      \ensurestackMath{\addstackgap[1pt]{#2}}%
      \vrule width 0pt depth 2pt height 0pt
    }
    {#3}%
} 
\newcommand{\scaledbracket}[3]{%
  \ThisStyle{%
    \stretchleftright
      {#1}
      {
        \ensurestackMath{\addstackgap[1pt]{\SavedStyle #2}}%
        \vrule width 0pt depth 1.5pt height 0pt
      }
      {#3}%
  }%
}
\newcommand{\bracketmid}[4]{%
  \stretchleftright{#1}{%
    \ensurestackMath{%
      \addstackgap[1pt]{#2}%
      \,\stretchrel*{|}{\addstackgap[1pt]{#2#3}}\,%
      \addstackgap[1pt]{#3}%
    }%
  }{#4}%
}

\theoremstyle{plain}
\newtheorem{theorem}{Theorem}[section]
\newtheorem{lemma}[theorem]{Lemma}
\newtheorem{proposition}[theorem]{Proposition}

\theoremstyle{definition}
\newtheorem{definition}[theorem]{Definition}
\newtheorem{example}[theorem]{Example}

\theoremstyle{remark}
\newtheorem{remark}[theorem]{Remark}
\usepackage[
  colorlinks=true,
  linkcolor=darkgreen,
  citecolor=darkgreen,
  urlcolor=darkgreen
]{hyperref}
\usepackage{xurl} 

\DeclareSourcemap{
  \maps[datatype=bibtex]{
    \map{
      \step[fieldsource=url, final=true]
      \step[fieldset=verba, origfieldval, final=true]
      \step[fieldsource=verba, match=\regexp{\Ahttps?://}, replace={}]
    }
  }
}
\DeclareFieldFormat{url}{%
  \mkbibacro{URL}\addcolon\space
  \href{#1}{\nolinkurl{\thefield{verba}}}%
}

\usepackage{cleveref}
\crefname{equation}{}{}
\crefname{section}{\S}{\S\S}
\crefname{subsection}{\S}{\S\S}
\crefname{subsubsection}{\S}{\S\S}
\crefname{definition}{Def.}{Defs.}
\crefname{theorem}{Thm.}{Thms.}
\crefname{corollary}{Cor.}{Cors.}
\crefname{lemma}{Lem.}{Lems.}
\crefname{proposition}{Prop.}{Props.}
\crefname{remark}{Rem.}{Rems.}
\crefname{notation}{Ntn.}{Ntns.}
\crefname{fact}{Fact}{Facts}
\crefname{example}{Ex.}{Exs.}
\crefname{figure}{Fig.}{Figs.}
\crefname{table}{Tab.}{Tabs.}
\crefname{footnote}{ftn.}{ftns.}
\Crefname{footnote}{Ftn.}{Ftns.}

\usepackage{tikz}
\usetikzlibrary{
  cd,
  calc,
  arrows.meta,
  backgrounds,
  decorations,
  decorations.pathmorphing, 
}

\tikzcdset{
  arrow style=tikz, 
  diagrams={
    trim left=
    {([xshift=5pt]current bounding box.west)},
    trim right=
    {([xshift=-5pt]current bounding box.east)},
    baseline=
    {([yshift=-2pt]current bounding box.center)},
    >={
      Computer Modern Rightarrow[
        length=4pt, width=4pt
      ]
    },
    hook/.style={
      {Hooks[right, length=2pt, width=8pt]}->
    },
    hook'/.style={
      {Hooks[left, length=2pt, width=8pt]}->
    },
    shorten=0pt,
  }
}

\definecolor{darkblue}{rgb}{0.05,0.25,0.65}
\definecolor{darkgreen}{RGB}{20,140,10}
\definecolor{lightgray}{rgb}{0.9,0.9,0.9}
\definecolor{darkorange}{RGB}{200,100,5}
\definecolor{darkyellow}{rgb}{.91,.91,0}
\definecolor{lightolive}{RGB}{225, 220, 185}

\let\originalsslash\sslash
\renewcommand{\sslash}{\mathord{\originalsslash}}

\makeatletter
\newcommand{\cpt}{\mathpalette\cpt@inner\relax}
\newcommand{\cpt@inner}[2]{%
  \scalebox{0.5}[0.9]{$#1\cup$}
  #1\{\infty\}
}
\makeatother

\makeatletter
\newcommand{\plus}{\mathpalette\sqcpt@inner\relax}
\newcommand{\sqcpt@inner}[2]{%
  \scalebox{0.5}[0.9]{$#1\sqcup$}
  #1\{\infty\}
}
\makeatother

\tikzset{
  snake left/.style={
    rounded corners,
    to path={
      let \p1 = (\tikztostart.east),
          \p2 = (\tikztotarget.west),
          \p3 = ($(\p1)!0.5!(\p2)$),
          \n1 = {8pt} 
      in
      (\p1)
      -- (\x1 + \n1, \y1)
      -- (\x1 + \n1, \y3)
      -- (\x2 - \n1, \y3) \tikztonodes
      -- (\x2 - \n1, \y2)
      -- (\p2)
    }
  }
}

\tikzset{
  uphordown/.style={
    rounded corners,
    to path={
      let \p1 = (\tikztostart.north),
          \p2 = (\tikztotarget.north),
          \n1 = {max(\y1,\y2) + 8pt}
      in
      (\p1)
      -- (\x1, \n1)
      -- (\x2, \n1) \tikztonodes 
      -- (\p2)
    }
  }
}

\tikzset{
  downhorup/.style={
    rounded corners,
    to path={
      let \p1 = (\tikztostart.south),
          \p2 = (\tikztotarget.south),
          \n1 = {min(\y1,\y2) - 8pt}
      in
      (\p1)
      -- (\x1, \n1)
      -- (\x2, \n1) \tikztonodes 
      -- (\p2)
    }
  }
}

\tikzset{
  rightvertleft/.style={
    rounded corners,
    to path={
      let \p1 = (\tikztostart.east),
          \p2 = (\tikztotarget.east),
          \n1 = {max(\x1,\x2) + 8pt}
      in
      (\p1)
      -- (\n1, \y1)
      -- (\n1, \y2) \tikztonodes 
      -- (\p2)
    }
  }
}

\tikzset{
  leftvertright/.style={
    rounded corners,
    to path={
      let \p1 = (\tikztostart.west),
          \p2 = (\tikztotarget.west),
          \n1 = {min(\x1,\x2) - 8pt}
      in
      (\p1)
      -- (\n1, \y1)
      -- (\n1, \y2) \tikztonodes 
      -- (\p2)
    }
  }
}

\newcommand{\inlinetikzcd}[1]{\begin{tikzcd}[sep=small, ampersand replacement=\&]#1\end{tikzcd}}

\renewcommand{\setminus}{-}

\newcommand{\defneq}{\mathrel{\equiv}}

\newcommand{\shape}{%
  {\mathord{\scalerel*{\raisebox{0.1ex}{\textesh}}{f}}%
  \mkern 2mu}
}

\newcommand{\truncation}[2]
  {\bracket[{#2}]_{#1}}

\begin{document}

\setlength{\abovedisplayskip}{3.5pt}
\setlength{\belowdisplayskip}{3.5pt}
\setlength{\abovedisplayshortskip}{-5pt}
\setlength{\belowdisplayshortskip}{3pt}

\titlehead{
  \href{https://ncatlab.org/nlab/show/Center+for+Quantum+and+Topological+Systems}{CQTS} Lecture Notes
}

\subject{
  Introduction to
}
\title{Higher Gauge Theory}
\subtitle{
via Differential Nonabelian Cohomology
\\
in Cohesive Homotopy Theory
\\
\phantom{A}
\\
\phantom{A}
}

\author[1,2]
{\NameEmail
  {Hisham Sati}
  {hsati@nyu.edu}
}

\author[1]
{\NameEmail
  {Urs Schreiber}
  {us13@nyu.edu}
}

\affil[1]{Mathematics Program and Center for Quantum and Topological Systems (CQTS), \newline New York University Abu Dhabi, UAE}
\affil[2]{The Courant Institute of Mathematical Sciences, \newline New York University, New York, USA}

\maketitle

\vspace{.5cm} 
\begin{abstract}
   This is a streamlined introduction to the global (infrared) completion of Maxwell-type higher gauge fields (as in the higher gauge sectors of higher-dimensional supergravity and its brane probes) by electromagnetic flux quantization in differential nonabelian cohomology, using cohesive homotopy theory. Applications include D/NS-brane charge in (unstable) K-theory, M-brane charge in unstable Cohomotopy, and geometric engineering of topological quantum order on probe M5-branes.
\end{abstract}


\MSCcodes{
  Primary:
  53C08,
  81T13,
  18N60,
  55P62,
  55N20;	
  Secondary:
  83E50,
  81T30,
  81V70	
}

\funding{
  by \textit{Tamkeen UAE} under the \textit{Abu Dhabi Research Institute Grant} \texttt{CG008}
}

\newpage

\tableofcontents

\newpage

\section
{Introduction}

\paragraph
{The Problem}
It is an open secret \notinjournal{(cf. \cite{Schreiber2014})} that traditional methods of mathematical physics, sophisticated as they may be (cf. \cite{HenneauxTeitelboim1992}), are insufficient for capturing the global (topological, deep infrared) structure of higher gauge theories, on which, however, the nature of the topological monopole/brane charges notably depends: 

For ordinary gauge theory, the standard assumption (already more ``advanced'' than what most textbooks get to, cf. \cite[\S 2]{GiachettaMangiarottiSardanashvily2009}) that field configurations are sections of an ordinary fiber bundle over spacetime, the \emph{field bundle} (cf. \parencites[\S 1]{Zuckerman1987}[p. ix]{Anderson1989}[\S 2]{Anderson1992}[\S 4.5.3]{Claessens2005}[(1)]{GS25-FieldsI}[p. 3]{Giotopoulos2025}), means restricting attention to a single topological sector (to be contrasted with the global situation, cf.  \cite{BeniniSchenkelSchreiber2018}). For higher gauge fields (cf. \parencites{Freed2002}{Szabo2013}[\S II]{MooreSaxena2025} and \parencites{Alfonsi2025HigherGeometry}{BorstenEtAl2025}) the assumption of an ordinary field bundle tacitly restricts to just the \emph{trivial} topological sector. 

Similarly, the usual Lagrangian densities (cf. \parencites[p. 58--59]{Anderson1992}[\S 6.2.1]{Basdevant2007}) for higher gauge fields (cf. \parencites{HenneauxTeitelboim1986}{BunsterHenneaux2011}[\S 16.4]{BlumenhagenLuestTheisen2013}[\S\S21--23]{Ortin2015}), being functions of globally defined gauge potential forms, tacitly assume that the given theory is restricted to the trivial topological sector or to a single chart of spacetime --- because only there are gauge potentials globally defined. 

\paragraph
{Flux Quantization}
The global definition of higher gauge fields is given by \emph{flux quantization} \cite{SS25-Flux}, but this has traditionally been relegated, at best, to an afterthought when faced with inconsistencies emerging in a naive approach. For instance: 
\begin{enumerate}
\item The action functional of the electron is found to be inconsistent (\emph{anomalous}) unless the ambient Maxwell field is flux-quantized according to \cite{Dirac1931} (cf. \parencites{WuYang1976}{Alvarez1985}[\S 5]{WenWitten1985}[\S 7.1]{Brylinski1993}[\S 16.4e]{Frankel2011}). 

\item The action functional of the M2-brane is anomalous unless the ambient C-field is shifted flux-quantized according to \cite{Witten1997Flux,Witten1997Fivebrane}.
\item The global definition of the action functional of the M5-brane similarly requires a flux quantization condition on the ambient C-field \cite{FSS21-Hopf}.
\end{enumerate}
In the first case, the mathematical structure of the flux quantization has become classical; in the second case, there are proposals in \parencites[\S 2.7]{HopkinsSinger2005} and \cite[\S 3.4]{FSS20-H}. The last case appears to have been addressed explicitly only in \cite{FSS21-Hopf}.

In general, the problem of IR-completion 
\cite{GSS26-SuGra}
of (higher gauge sectors of) higher-dimensional supergravity theories (cf. \parencites
{vanNieuwenhuizen1981}
{Cremmer1982}
{DuffNilssonPope1986}
{SalamSezgin1989}
{CDF1991}
{WessBagger1992}
{BuchbinderKuzenko1995}
{FreedmanVanProeyen2012}
[\S6--9]{Fre2013}
{Ortin2015}
{TraubenbergEtAl2020}
{Sezgin2023}%
{DomingoTraubenberg2023}
) by flux quantization has only recently received attention (cf. \cite{LazaroiuShahbazi2022a,LazaroiuShahbazi2022b,GSS24-SuGra,GagliardoEtAl2025}). And yet, this is paramount for many questions: It may be premature to speculate about the UV-completion of 11D Supergravity (\cite{CremmerJuliaScherk1978}, cf. \parencites{MiemiecSchnakenburg2006}{GSS24-SuGra})---working title: ``M-theory'' (\parencites[p. 1]{Horava1996}[p. 1]{Witten1996}{Duff1996}{Townsend1996}{Witten1998}{Duff1999}{Duff1999World}, originally short for \emph{Membrane} theory, following \parencites{DuffHoweInamiStelle1987})---before investigating its IR-completion by proper flux quantization (cf. \cite{GSS26-SuGra}).

\paragraph
{Cohesive Approach}
To settle such questions in examples, one needs a formalism that sets the stage in generality. In particular, one needs to have basic notions, such as a generalization of the notion of \emph{phase space}, to the case of higher gauge fields with possibly nontrivial topological sectors. 

The aim of this note is to lay out steps in this direction, following \cite{SS24-Phase,SS25-Flux} but further streamlining and expanding on some details.

Our methods are necessarily \emph{higher geometric} (cf. \cite{CattaneoGiaquintoXu2011,FSS15-Stacky,JurcoEtAl2019,FSS2019,Alfonsi2025HigherGeometry}), hence of \emph{geometric homotopy theory} (\parencites{Brown1973}{Jardine2015}{ToenVezzosi2005}{Lurie2009}, specifically cohesive $\infty$-topos theory, \cite{SS26-Orb,SS26-Bun,Sc13-dcct}, exposition in \cite{Schreiber2025}), which we briefly recall in \cref{OnSomeCohesiveHomotopyTheory}. In lowest degree this means generalizing smooth manifolds to \emph{smooth sets} \parencites[Def. 2.1]{KhavkineSchreiber2026} (for discussion in field theory see \parencites{GS25-FieldsI}{GS25-FieldsII}, further exposition is in \parencites{Giotopoulos2025}{IbortMas2025}) while in general degree this concerns generalizing \emph{Lie groupoids} (cf. \parencites{Weinstein1996Groupoids}{BursztyndelHoyo2025}{Moerdijk2003}) first to Lie $n$-groupoids (Kan simplicial smooth manifolds, cf. \parencites[Def. 1.2]{Henriques2008}[Def. 1.2]{Zhu2006}[Def. 1.3]{Zhu2009}{Zhu2009b}[p. 17]{Pridham2013}[\S 4.2]{NSS2015b}{Rogers2024}{Cabrera2026}) and then further to \emph{smooth $\infty$-groupoids} (smooth $\infty$-stacks, for concise definitions see \cite[\S 1]{FSS23-Char}, more exposition is in \cite{FSS15-Stacky}).

\paragraph
{The Solution}
Using these methods, the problem of global (topological, infrared) completion of higher gauge fields turns out to have an elegant solution; in outline:
\begin{enumerate}

\item  The smooth moduli set of higher flux densities solving their Maxwell-type equations of motion is identified with that of closed (meaning: flat, Maurer-Cartan) differential forms, 
\begin{equation}
  \big\{
  \text{On-shell flux densities}
  \big\}
  \,\simeq\,
  \mathbf{\Omega}^1_{\mathrm{cl}}\bracket({X^d; \mathfrak{a}})
\end{equation}
on a Cauchy surface $X^d$ with coefficients in a characteristic $L_\infty$-algebra $\mathfrak{a}$ encoding the duality symmetric Bianchi identities and the electromagnetic higher Gauss laws (\cref{GaussLawAsLInfinityClosure}).

\item The admissible electromagnetic flux quantization laws are given by classifying spaces $\mathcal{A}$ whose real Whitehead-bracket $L_\infty$-algebra reproduces the Gauss laws (\cref{OnRationalHomotopy}): 
\begin{equation}
  \mathfrak{l}\mathcal{A} 
    \simeq 
  \mathfrak{a}
  \mathrlap{\,.}
\end{equation}

\item 
A \emph{nonabelian character map} (generalizing the Chern character on K-theory) 
  \begin{equation}
    \begin{tikzcd}
      \mathbf{Map}\bracket({
        X^d,
        \mathcal{A}
      })
      \ar[
        r,
        "{
          \mathbf{ch}^{\smash{\mathcal{A}}}
        }"
      ]
      &
      \shape 
      \mathbf{\Omega}^1_{\mathrm{cl}}
      \bracket({
        X^d;
        \mathfrak{a}
      })
    \end{tikzcd}
  \end{equation}
sends (monopole/brane) \emph{charges} and their higher global symmetries in the nonabelian $\Omega \mathcal{A}$-cohomology of $X^d$ to the flux densities and their deformations sourced by these charges (\cref{OnNonabelianCharacter}).

\item The on-shell higher gauge field configurations globally completed by flux quantization in $\Omega\mathcal{A}$-cohomology are triples of (a) on-shell fluxes $\vec B$, (b) compatible charges $\chi$, (c) gauge potentials $\hat A$ witnessing that these fluxes are sourced by these charges. That is, we have a homotopy cone of smooth $\infty$-groupoids of this form (\cref{OnThePhaseSpace}):
\begin{equation}
  \begin{tikzcd}
    &&
    \mathbf{Map}\bracket({
      X^d,
      \mathcal{A}
    })
    \ar[
      dd,
      "{
        \mathbf{ch}^{\mathcal{A}}
      }"
    ]
    \\
    \ast
    \ar[
      dr,
      dashed,
      "{
        \vec B
      }"{name=fluxes},
      "{
        \text{fluxes}
      }"{
        swap, 
        sloped,
        color=darkblue
      }
    ]
    \ar[
      urr,
      bend left=25,
      dashed,
      "{
        \chi
      }"{
        swap,
        name=charges,
        pos=.4,
      },
      "{
        \text{charges}
      }"{
        sloped,
        pos=.4,
        color=darkblue
      }
    ]
    \ar[
      from=charges,
      to=fluxes,
      Rightarrow,
      dashed,
      "{ 
        \hat A 
      }"{
        pos=.5,
        swap
      },
      "{
        \text{potentials}
      }"{
        sloped, 
        swap, 
        pos=.4,
        yshift=-2pt,
        color=darkgreen
      }
    ]
    \\
    & 
    \mathbf{\Omega}^1_{\mathrm{cl}}
    \bracket({
      X^d;
      \mathfrak{a}
    })
    \ar[
      r,
      "{
        \eta^{\shape}
      }"
    ]
    &
    \shape
    \mathbf{\Omega}^1_{\mathrm{cl}}
    \bracket({
      X^d;
      \mathfrak{l}
      \mathcal{A}
    })
    \mathrlap{\,,}
  \end{tikzcd}
\end{equation}
whereby the completed phase space stack is the homotopy fiber product
\begin{equation}
  \mathbf{Phs}\bracket({
    X^d;
    \mathcal{A}
  })
  =
  \mathbf{\Omega}^1_{\mathrm{cl}}\bracket({
    X^d;
    \mathfrak{a}
  })
  \underset
    {
      \shape 
      \mathbf{\Omega}^1_{\mathrm{cl}}
      \scaledbracket({
        X^d;
        \mathfrak{a}
      })
    }
    {\times}
  \mathbf{Map}\bracket({
    X^d,
    \mathcal{A}
  })
  \mathrlap{\,.}
\end{equation}

\item For linear Gauss laws this reproduces familiar global completions by flux quantization such as in ordinary differential cohomology (line bundles, bundle gerbes, etc.) and in differential K-theory (though there is ample room for alternative choices not traditionally appreciated); but it crucially generalizes also to the non-linear electric Gauss laws that are famously seen (but traditionally ignored) in higher-dimensional supergravity, whose global infrared completions thereby become accessible (\cref{On11D10DSupergravity}).

\item A twisted relative generalization of this procedure applies also to higher gauge fields on probe branes, notably to the self-dual flux on M5-branes (\cref{OnM5BraneProbes}).

In particular, there exists an admissible completion (``Hypothesis H'') of 11D SuGra with M5-probes at $A_n$-singularities, which serves to geometrically engineer the topological quantum order of fractional quantum Hall systems in fine detail, making experimentally accessible predictions, cf. \cref{SuperconductingIslands} in \cref{OnM5BraneProbes} below. 
\end{enumerate}

\paragraph
{Comparison to the Literature}

The higher gauge theory via differential nonabelian cohomology that we present (\cref{OnThePhaseSpace}) is \emph{different} (cf. \cref{TwoNotionsOfHigherGaugeTheory,TwoApproachesToHigherGaugeTheory}) from that via \emph{higher principal connections} (\parencites[\S 11-12]{Schreiber2005}{BaezSchreiber2007}{FSSt12-DiffClasses}{SSS12}{SchreiberWaldorf2013}{BorstenEtAl2025}[\S 2]{Alfonsi2025HigherGeometry}{RistSaemannWolf2026}{Waldorf2026}{BunkEtAl2026}): 

The basis of our formulation is exactly the electromagnetic Bianchi identities / Gauss laws that characterize higher Maxwell-type physical equations of motion (\cref{OnEquationsOfMotion}) as they appear notably in type I/II/M supergravity (\cref{ExamplesAndApplications}). Just choosing flux quantization laws compatible with these (\cref{OnFluxQuantization}) already determines the full higher gauge potential/field structure globally (\cref{OnGlobalCompletion}), without need or justification for a further principality constraint. For instance, differential K-theory is not a higher principal connection theory, but is a special case of differential nonabelian cohomology (\cref{ExamplesOfDifferentialCohomology}).

While our discussion thus pertains in particular to the realm of ``M-theory'' (\parencites[p. 1]{Horava1996}[p. 1]{Witten1996}{Duff1996}{Townsend1996}{Witten1998}, cf. \parencites{Duff1999}{Duff1999World}{Sati2010}{JurcoEtAl2019}{FSS2019}) --- geometrically engineering topological quantum materials (\cref{OnM5BraneProbes}) in a realistic $N=1$ alternative to traditional large-$N$ holographic condensed matter theory (\cite{Herzog2007,GauntlettSonnerWiseman2009,GubserPufuRocha2010,GauntlettetalSonnerWiseman2010,DonosetalEtAl2013,DonosGauntlettPantelidou2013}, cf. \cite[\S 4.3]{Gaggioli2019}) --- at no point do we rely on string folklore. All of the following constructions have definitions and all claims have proofs.

\paragraph{Acknowledgements}
We thank the organizers of the meetings \emph{\href{https://www.wiko-greifswald.de/en/registrations/higher-differential-geometry}{Higher Differential Geometry}} (WIKO Greifswald, May 2026) and \emph{\href{http://thphys.irb.hr/phygeo2026/}{6th Mini Symposium on Physics and Geometry}} (RBI Zagreb, May 2026) where parts of these lectures were first presented.

\section
{Cohesive Homotopy Theory}
\label
{OnSomeCohesiveHomotopyTheory}

The completed phase spaces of higher gauge fields in \cref{OnTheGlobalPhaseSpace} combine the differential geometry of flux densities with the homotopy theory of (``topological'') charge sectors, cf. \cref{GeomHomotopySchematics}.
\begin{SCfigure}[.73][htb]
\caption{
\label{GeomHomotopySchematics}
The geometric homotopy theory of $\infty$-topoi is the conceptual unification of geometry with homotopy theory, hence the theory of \emph{higher geometry}. 
In our context, \emph{fluxes} are (differential) geometric, \emph{charges} are homotopical, and their unification happens through the construction of the \emph{phase space} (\cref{OnThePhaseSpaceAsSuch}) in the ambient cohesive $\infty$-topos (\cref{GeometricHomotopy}).
}
\centering
\adjustbox{scale=0.87, 
  rndfbox=4pt
}{
$
  \begin{tikzcd}[
    column sep=-25pt
  ]
    \adjustbox{
      rndfbox=4pt
    }{
      Flux densities
    }
    \ar[
      d,
      phantom,
      "{
        \text{described by}
      }"
    ]
    &&
    \adjustbox{
      rndfbox=4pt
    }{
      Charges
    }
    \ar[
      d,
      phantom,
      "{
        \text{described by}
      }"
    ]
    \\
    \adjustbox{
      rndfbox=4pt
    }{
      geometry
    }
    \ar[
      dr,
      phantom,
      "{
        \text{subsumed by}
      }"
    ]
    &&
    \adjustbox{
      rndfbox=4pt
    }{
      homotopy theory
    }
    \ar[
      dl,
      phantom,
      "{
        \text{subsumed by}
      }"
    ]
    \\
    &
    \adjustbox{
      rndfbox=4pt
    }{
      \begin{tabular}{c}
        geometric homotopy
        \\
        ($\infty$-topos theory)
      \end{tabular}
    }
  \end{tikzcd}
$
}
\end{SCfigure}

Such unification of geometry with homotopy theory happens in \emph{geometric homotopy theory} (in \emph{$\infty$-topoi} of simplicial sheaves \cite{Brown1973,Jardine1987,ToenVezzosi2005,Lurie2009}) and for \emph{differential} geometry specifically in \emph{cohesive homotopy theory} (in \emph{cohesive $\infty$-topoi} \cite{Sc13-dcct,SS26-Orb,SS26-Bun}\notinjournal{, lecture notes in \cite{Sc18-ToposLectures}}).

We aim to provide an exposition of the basics, but we assume the reader is familiar with the basic concepts and facts of ordinary category theory, as presented in \cite{Awodey2010,Sc18-ToposLectures}. For further background reading cf. \cite{Richter2020}.

\subsection
{Geometry}

Consider the ordinary category of smooth manifolds with smooth maps between them (cf. \cite{Tu2011,Lee2012}), and its full subcategory of Cartesian spaces ($\mathbb{R}^n$s for $n \in \mathbb{N}$):
\begin{equation}
\label
{CartesianSpacesAmongManifolds}
  \begin{tikzcd}
    \mathrm{CartSp}
    \ar[r, hook]
    &
    \mathrm{SmthMfd}
    \mathrlap{\,.}
  \end{tikzcd}
\end{equation}
We will denote by $\ast := \mathbb{R}^0$ the \emph{point manifold}.

One reason why nonperturbative quantum field theory appears technically hard is that spaces of field configurations that ought to be objects in differential geometry tend to fall outside the realm of standard textbook constructions. For example, for $X, Y$ smooth manifolds (such as $X \defneq \text{spacetime}$  and $Y = \text{scalar moduli space}$), the mapping space (field space)
\begin{equation}
  \label{AskingForTheInternalHom}
  \mathbf{Map}\bracket({
    X,
    Y
  })
  \in
  \;
  \text{??}
\end{equation}
is subtle to understand as an infinite-dimensional manifold: When $X$ is non-compact, as is common for spacetime, standard constructions of manifold charts generally fail to recognize as smooth some families of maps/fields that do depend smoothly on both spacetime and their parameters  (cf. \parencites[Cor. 11.9]{Michor1980}[Lem. 42.5 \& Thm. 42.14]{KrieglMichor1997}[\S 5.3]{Michor2020}).

At the same time, it is elementary to say what should count as a smooth map from a manifold like $\mathbb{R}^n$ into this mapping space \cref{AskingForTheInternalHom}, namely this should correspond to a smooth map of ordinary smooth manifolds to $Y$ from the product of $X$ with $\mathbb{R}^n$:
\begin{equation}
  \begin{aligned}
  &
  \big\{
    \begin{tikzcd}
      \mathbb{R}^n
      \ar[r]
      &
      \mathbf{Map}\bracket({
        X,Y
      })
    \end{tikzcd}
  \big\}
  \\
  \simeq
  \;
  &
  \big\{
    \begin{tikzcd}
      \mathbb{R}^n
      \times X
      \ar[r]
      &
      Y
    \end{tikzcd}
  \big\}
  \mathrlap{\,.}
  \end{aligned}
\end{equation}
But since knowing all such \emph{probes} of a generalized smooth space is going to be as good as knowing the space ``itself'', we take this next as our definition, which will resolve the above issue in \cref{WaldorfEmbedding} below. 

\subsubsection
{Smooth Sets}
\label
{OnSmoothSets}

A \emph{smooth set} $\mathbf{X}$ is (equivalently encoded by) a presheaf $\mathrm{Plt}\bracket({-,\mathbf{X}})$ of sets on the category of Cartesian spaces \cref{CartesianSpacesAmongManifolds}, hence a functor
\begin{equation}
\label
{FunctorOfPlots}
  \begin{tikzcd}[row sep=-2pt, 
    column sep=0pt
  ]
    \mathrm{CartSp}^{\mathrm{op}}
    \ar[rr]
    &&
    \mathrm{Set}
    \\
    \mathbb{R}^n
    &\longmapsto&
    \mathrm{Plt}\bracket({
      \mathbb{R}^n, \mathbf{X}
    })
  \end{tikzcd}
\end{equation}
which we think of as assigning:
\begin{itemize}
    \item  to $\mathbb{R}^n$ (for $n \in \mathbb{N}$) the set of smooth maps from this $\mathbb{R}^n$ into $\mathbf{X}$ (the \emph{plots} of $\mathbf{X}$ by the \emph{probe} $\mathbb{R}^n$), and 

\item to smooth maps $\inlinetikzcd{ \mathbb{R}^{n'} \! \ar[r, "{ \phi }"] \& \mathbb{R}^n }$ their precomposition operation $\phi^\ast \!\!: \!\inlinetikzcd{ \mathrm{Plt}\bracket({\mathbb{R}^n, \mathbf{X}}) \!  \ar[r] \& \!\mathrm{Plt}\bracket({\mathbb{R}^{n'}, \mathbf{X}})}$ with such plots, defined thereby.
\end{itemize}
In this vein of characterizing spaces via the system of their \emph{plots} by \emph{probe spaces}, a \emph{smooth map} $\inlinetikzcd{\mathbf{X} \ar[r, "{ f }"] \& \mathbf{Y}}$ between such smooth sets must, by postcomposition, be a rule that sends $\mathbb{R}^n$-plots $\phi$ of $\mathbf{X}$ to $\mathbb{R}^n$-plots $f \circ \phi$ of $\mathbf{Y}$, compatible with the precomposition operation of plots, hence must be a natural transformation
$f_\ast : \inlinetikzcd{ \mathrm{Plt}\bracket({-,\mathbf{X}}) \ar[r, Rightarrow] \& \mathrm{Plt}\bracket({-,\mathbf{Y}}) }$ between the functors of plots \cref{FunctorOfPlots}.

So far, this appears to say that the category of smooth sets should be the category of presheaves over $\mathrm{CartSp}$. 
However,  some of this information ought in fact to be redundant: For probing the (smooth) structure of a smooth set $\mathbf{X}$ it ought to already be sufficient to know just the restriction of these $\mathbb{R}^n$-plots to arbitrarily small open balls $\mathbb{D}^n_{r > 0} \simeq \mathbb{R}^n \subset \mathbb{R}^n$ around any point of $\mathbb{R}^n$ (which we may fix to be the origin), hence to know the \emph{germs} of plots, to be denoted
\begin{equation}
\label
{GermsOfPlots}
  \mathrm{Plt}\bracket({
    \mathbb{G}^n,
    \mathbf{X}
  })
  :=
  \varinjlim_{r \in \mathbb{R}_{> 0}}
  \mathrm{Plt}\bracket({
    \mathbb{D}^n_r,
    \mathbf{X}
  })
  \mathrlap{\,.}
\end{equation}
Any natural transformation $f_\ast$ between presheaves of plots restricts to a transformation of such germs of plots, and if that restriction is a natural \emph{iso}morphism, then one says that $f_\ast$ is a \emph{local isomorphism}. We denote the class of local isomorphisms by 
\begin{equation}
  W_{\mathrm{l}} 
    \subset 
  \mathrm{Mor}\bracket({
  \mathrm{Func}\bracket({
    \mathrm{CartSp}^{\mathrm{op}}, \mathrm{Set}
  })
}) 
\mathrlap{\,.}
\end{equation}
Hence we should regard local isomorphisms as \emph{actual} isomorphisms between (systems of plots of) smooth sets, a process known as \emph{localization} $L^{W_{\mathrm{l}}}(-)$ at the class of local isomorphisms. Doing so changes their category to the \emph{sheaf topos} over $\mathrm{CartSp}$, which we thereby recognize as the correct \emph{category of smooth sets} (cf. \cite{GS25-FieldsI,GS25-FieldsII,IbortMas2025}):
\begin{equation}
  \label{CategoryOfSmoothSets}
  \mathrm{SmthSet}
  :=
  L^{W_{\mathrm{l}}}
  \mathrm{Func}\bracket({
    \mathrm{CartSp}^{\mathrm{op}},
    \mathrm{Set}
  })
  \simeq
  \mathrm{Sh}\bracket({
    \mathrm{CartSp}
  })
  \mathrlap{\,.}
\end{equation}
True to this effective identification of smooth sets $\mathbf{X}$ with their systems of sets of plots, we will often abbreviate the latter as
\begin{equation}
\label
{ShorthandForSetOfPlots}
  \mathbf{X}\bracket({
    \mathbb{R}^n
  })
  :=
  \mathrm{Plt}\bracket({
    \mathbb{R}^n, 
    \mathbf{X}
  })
  \mathrlap{\,.}
\end{equation}

For example, an ordinary smooth manifold becomes a smooth set by taking its probes to be the ordinary smooth functions $C^\infty(-,-)$ between smooth manifolds, a fully faithful construction known, in generality, as the \emph{Yoneda embedding}:
\begin{equation}
\label
{YonedaEmbedding}
  \begin{tikzcd}[row sep=-2pt, 
    column sep=0pt
  ]
    \mathrm{SmthMfd}
    \ar[
      rr,
      hook
    ]
    &&
    \mathrm{SmthSet}
    \\
    X 
      &\longmapsto&
    C^\infty\bracket({-,X})
    \mathrlap{\,.}
  \end{tikzcd}
\end{equation}

Interestingly, this raises a potential for inconsistency: While we started out \emph{declaring} the sets $\mathrm{Plt}\bracket({\mathbb{R}^n, \mathbf{X}})$ of smooth maps from $\mathbb{R}^n$ to a smooth set $\mathbf{X}$ defined thereby, we now see, with \cref{YonedaEmbedding}, that there are also the actual sets of such smooth maps $\inlinetikzcd{\mathbb{R}^n \ar[r] \& \mathbf{X}}$ between smooth sets. Remarkably, it is the full \emph{Yoneda lemma} which ensures that these naturally coincide, after the fact:
\begin{equation}
  \mathrm{Plt}\bracket({
    \mathbb{R}^n,
    \mathbf{X}
  })
  \simeq
  \mathrm{Hom}\bracket({
    \mathbb{R}^n, 
    \mathbf{X}
  })
  \mathrlap{\,,}
\end{equation}
where $\mathrm{Hom}(-,-)$ on the right denotes the hom-sets of $\mathrm{SmthSet}$.

However, smooth sets may be considerably more general than smooth manifolds. For example, \emph{moduli spaces} $\mathbf{\Omega}^n_{\mathrm{dR}}(\ast)$ of differential $n$-forms ($n \in \mathbb{N}$) exist as smooth sets (more on this example in \cref{OnUniversalDeRhamComplex}):
\begin{equation}
\label
{SmoothModuliSetOfDifferentialNForms}
  \begin{aligned}
  \mathbf{\Omega}^n
    _{\mathrm{dR}}(\ast)
  & \in
  \mathrm{SmthSet}
  \;\;
  \substack{
  \text{
    smooth moduli set of differential $n$-forms
  }
  }
  \\
  \mathrm{Plt}\bracket({
    \mathbb{R}^n,
    \mathbf{\Omega}^n_{\mathrm{dR}}
    (\ast)
  })
  &
  :=
  \Omega^n_{\mathrm{dR}}\bracket({
    \mathbb{R}^n
  })
  \;\;
  \substack{
    \text{ordinary set of smooth differential $n$-forms on $\mathbb{R}^n$}
  }
  \,,
  \end{aligned}
\end{equation}
where precomposition of plots by smooth functions is given by pullback of differential forms.

Moreover, it is straightforward to construct further smooth sets from given ones, following the logic of their plots. For instance, given a pair $\mathbf{X}, \mathbf{Y} \in \mathrm{SmthSet}$, their Cartesian product $\mathbf{X} \times \mathbf{Y} \in \mathrm{SmthSet}$ must---as plots---have pairs of plots into the factors, which is already the definition:
\begin{equation}
  \label{ProductOfSmoothSets}
  \mathrm{Plt}\bracket({
    \mathbb{R}^n,
    \mathbf{X} 
     \times
    \mathbf{Y}
  })
  :=
  \mathrm{Plt}\bracket({
    \mathbb{R}^n,
    \mathbf{X}
  })
  \times
  \mathrm{Plt}\bracket({
    \mathbb{R}^n,
    \mathbf{Y}
  })
  \mathrlap{\,.}
\end{equation}

Finally, in the category of smooth sets, the motivating problem \cref{AskingForTheInternalHom} is naturally resolved: For any pair $\mathbf{X}, \mathbf{Y} \in \mathrm{SmthSet}$ we have that the set of smooth maps between them becomes itself a smooth set, with plots
\begin{equation}
\label
{InternalHomOfSmoothSets}
  \begin{aligned}
  \mathbf{Map}\bracket({
    \mathbf{X},
    \mathbf{Y}
  })
  & 
  \in
  \mathrm{SmthSet}
  \\
  \mathrm{Plt}\bracket({
    \mathbb{R}^n,
    \mathbf{Map}\bracket({
      \mathbf{X},
      \mathbf{Y}
    })
  })
    & 
    :=
    \mathrm{SmthSet}\bracket({
      \mathbb{R}^n 
        \times
      \mathbf{X},
      \mathbf{Y}
    })
    \mathrlap{\,.}
  \end{aligned}
\end{equation}
Standard results show that this evident definition already implies natural isomorphisms
\begin{equation}
\label
{CartesianClosureOfSmoothSets}
  \mathbf{Map}\bracket({
    \mathbf{X}
    \times
    \mathbf{Y},
    \mathbf{Z}
  })
  \simeq
  \mathbf{Map}\bracket({
    \mathbf{X},
    \mathbf{Map}\bracket({
      \mathbf{Y},
      \mathbf{Z}
    })
  })  
  \mathrlap{\,.}
\end{equation}
This exhibits the category $\mathrm{SmthSet}$ as being \emph{Cartesian closed}, one of the key properties of what is called a \emph{convenient category of spaces} (going back to \cite{Steenrod1967}, cf. \cite{BaezHoffnung2011}).
In fact, $\mathrm{SmthSet}$ is ``particularly convenient'' in that it is \cite{Sc13-dcct} a \emph{cohesive topos} in the sense of \cite{Lawvere2007}. We begin describing this in \cref{OnConcreteSmoothSets} and will exhibit the full structure in \cref{OnCohesion}.

\subsubsection
{Concreteness}
\label
{OnConcreteSmoothSets}

We discuss the \emph{concrete} objects among the smooth sets from \cref{OnSmoothSets}, identified with \emph{diffeological spaces} (the latter due to \parencites{Souriau1984}, cf. \parencites{BaezHoffnung2011}{IglesiasZemmour2013}).

Traditionally, mathematical notions of \emph{space} tend to be conceived as sets of points equipped with further structure (topology, smooth structure, etc.). But smooth sets $\mathbf{X}$ \cref{CategoryOfSmoothSets} exist whose underlying set of points, $\mathrm{Plt}\bracket({\mathbb{R}^0,\mathbf{X}})$, does not support all their higher-dimensional plots. 

For example, with $n \geq 1$ the smooth moduli set $\mathbf{\Omega}^n_{\mathrm{dR}}(\ast)$ \cref{SmoothModuliSetOfDifferentialNForms} has a single point (the zero $n$-form on $\mathbb{R}^0$) but an infinite set of higher $k$-dimensional plots! 

But we may neatly characterize those smooth sets that \emph{are} underlying sets of points with extra structure, as the \emph{concrete} smooth sets, commonly known as \emph{diffeological spaces}. This proceeds as follows.

\paragraph
{The Flat Modality}
First, observe that the assignment of underlying sets of points is a functor
\begin{equation}
  \begin{tikzcd}[row sep=-2pt, 
    column sep=0pt
  ]
    \mathrm{SmthSet}
    \ar[
      rr, 
      "{ \mathrm{Pnt} }"
    ]
    &&
    \mathrm{Set}
    \\
    \mathbf{X}
    &\longmapsto&
    \mathrm{Plt}\bracket({
      \ast,
      \mathbf{X}
    })
    \mathrlap{\,,}
  \end{tikzcd}
\end{equation}
which has a fully faithful left adjoint: 
\begin{equation}
  \begin{tikzcd}
    \mathrm{SmthSet}
    \ar[
      rr,
      phantom,
      shift left=8pt,
      "{ \bot }"{scale=.7}
    ]
    \ar[
      rr,
      "{
        \mathrm{Pnt}
      }"{description}
    ]
    &&
    \mathrm{Set}
    \mathrlap{\,.}
    \ar[
      ll,
      hook',
      shift right=16pt,
      "{
        \mathrm{Dsc}
      }"{description}
    ]
  \end{tikzcd}
\end{equation}
Namely,  
$\mathrm{Dsc}\bracket({S})$ is the \emph{discrete smooth structure} on $S \in \mathrm{Set}$, whose only smooth plots are constant maps:
\begin{equation}
    \mathrm{Plt}\bracket({
      \mathbb{R}^n,
      \mathrm{Dsc}\bracket(S)
    })
    =
    S
    \mathrlap{\,.}
\end{equation}
The composite functor
\begin{equation}
\label
{FlatModalityOnSmoothSets}
  \flat
  :
  \begin{tikzcd}
    \mathrm{SmthSet}
    \ar[
      r,
      "{ \mathrm{Pnt} }"
    ]
    &
    \mathrm{Set}
    \ar[
      r,
      hook,
      "{ 
        \mathrm{Dsc}
      }"
    ]
    &
    \mathrm{SmthSet}
  \end{tikzcd}
\end{equation}
sends a smooth set to its underlying set of points, regarded as a discrete smooth set:
\begin{equation}
  \mathrm{Plt}\bracket({
    \mathbb{R}^n,
    \flat\mathbf{X}
  })
  \simeq
  \mathrm{Hom}\bracket({
    \ast, 
    \mathbf{X}
  })
  \mathrlap{\,.}
\end{equation}
By the adjunction property, this comes with an evident natural transformation:
\begin{equation}
  \begin{tikzcd}[row sep=-2pt, 
    column sep=10pt
  ]
    \flat
    \mathbf{X}
    \ar[
      rr,
      "{
        \epsilon^\flat_{\mathbf{X}}
      }"
    ]
    &&
    \mathbf{X}
    \\
    \mathrm{Hom}\bracket({
      \ast,
      \mathbf{X}
    })
    \ar[
      rr
    ]
    &&
    \mathrm{Hom}\bracket({
      \mathbb{R}^n,
      \mathbf{X}
    })
    \\
    \bracket({
      \ast 
      \overset{x}{\to}
      \mathbf{X}
    })
    &\longmapsto&
    \bracket({
      \mathbb{R}^n
      \to
      \ast 
      \overset{x}{\to}
      \mathbf{X}
    })
    \mathrlap{\,,}
  \end{tikzcd}
\end{equation}
which has the universal property that it uniquely factors smooth maps out of any $\flat \mathbf{Y}$:
\begin{equation}
  \begin{tikzcd}
    \flat \mathbf{Y}
    \ar[
      r,
      dashed,
      "{ \exists! }"
    ]
    \ar[
      rr,
      uphordown,
      "{ 
        \forall 
      }"{description}
    ]
    &
    \flat \mathbf{X}
    \ar[
      r,
      "{
        \epsilon
          ^\flat
          _{\mathbf{X}}
      }"
    ]
    &
    \mathbf{X}
    \mathrlap{\,.}
  \end{tikzcd}
\end{equation}
In words, this just says that a map to $\mathbf{X}$ from a discrete set of points factors through the underlying set of points of $\mathbf{X}$.

\paragraph
{The Sharp Modality}

But the functor $\mathrm{Pnt}$, furthermore, has a fully faithful right adjoint:
\begin{equation}
\label
{DscPntChtAdjointTriple}
  \begin{tikzcd}
    \mathrm{SmthSet}
    \ar[
      rr,
      phantom,
      shift left=8pt,
      "{ \bot }"{scale=.7}
    ]
    \ar[
      rr,
      "{
        \mathrm{Pnt}
      }"{description}
    ]
    \ar[
      rr,
      phantom,
      shift right=8pt,
      "{ \bot }"{scale=.7}
    ]
    &&
    \mathrm{Set}
    \mathrlap{\,.}
    \ar[
      ll,
      hook',
      shift right=16pt,
      "{
        \mathrm{Dsc}
      }"{description}
    ]
    \ar[
      ll,
      hook',
      shift left=16pt,
      "{
        \mathrm{Cht}
      }"{description}
    ]
  \end{tikzcd}
\end{equation}
Here 
$\mathrm{Cht}\bracket({S})$ is the \emph{chaotic smooth structure} on $S \in \mathrm{Set}$, for which \emph{all} maps from the discrete underlying set of a probe space count as smooth plots:
  \begin{equation}
  \begin{aligned}
    \mathrm{Plt}\bracket({
      \mathbb{R}^n,
      \mathrm{Cht}\bracket(S)
    })
    & =
    \mathrm{Hom}\bracket({
      \flat \mathbb{R}^n,
      \mathrm{Dsc}(S)
    })
    \mathrlap{\,.}
    \end{aligned}
  \end{equation}
This chaotic smooth structure is typically not of interest in itself, but it serves to characterize concrete smooth sets, as follows.

The composite functor
\begin{equation}
  \sharp
  :
  \begin{tikzcd}
    \mathrm{SmthSet}
    \ar[r, "{\mathrm{Pnt}}"]
    &
    \mathrm{Set}
    \ar[
      r,
      hook,
      "{ \mathrm{Cht} }"
    ]
    &
    \mathrm{SmthSet}
  \end{tikzcd}
\end{equation}
sends a smooth set to the chaotic smooth structure on its underlying set of points,
\begin{equation}
  \mathrm{Plt}\bracket({
    \mathbb{R}^n,
    \sharp \mathbf{X}
  })
  \simeq
  \mathrm{Hom}\bracket({
    \flat \mathbb{R}^n,
    \mathbf{X}
  })
  \mathrlap{\,.}
\end{equation}
By the adjunction property, this receives a natural transformation from the identity, given by restricting $\mathbb{R}^n$-plots to all points in $\mathbb{R}^n$: 
\begin{equation}
\label
{TheSharpUnit}
  \begin{tikzcd}[row sep=-2pt, 
    column sep=10pt
  ]
    \mathbf{X}
    \ar[
      rr,
      "{ 
        \eta^{\sharp}_{\mathbf{X}} 
       }"
    ]
    &&
    \sharp \mathbf{X}
    \\
    \mathrm{Hom}\bracket({
      \mathbb{R}^n, 
      \mathbf{X}
    })
    \ar[
      rr,
      "{ 
        (
          \epsilon
            ^\flat
            _{\mathbb{R}^n}
        )^\ast 
      }"
    ]
    &&
    \mathrm{Hom}\bracket({
      \flat \mathbb{R}^n, 
      \mathbf{X}
    })
    \\
    \big({
    \mathbb{R}^n 
    \overset{\phi}{\to}
    \mathbf{X}
    }\big)
    &\longmapsto&
    \big({\,
    \flat \mathbb{R}^n
    \overset
      {\epsilon^\flat_{\mathbb{R}^n}}
      {\longrightarrow}
    \mathbb{R}^n 
    \overset{\phi}{\to}
    \mathbf{X}
    \,}\big)
    \mathrlap{\,,}
  \end{tikzcd}
\end{equation}
which has the universal property that it uniquely factors smooth maps into any $\sharp \mathbf{Y}$:
\begin{equation}
\label
{UniversalPropertyOfSharpUnit}
\begin{tikzcd}
  \mathbf{X}
  \ar[
    r,
    "{
      \eta^\sharp_{\mathbf{X}}
    }"
  ]
  \ar[
    rr,
    uphordown,
    "{
      \forall
    }"{description}
  ]
  &
  \sharp \mathbf{X}
  \ar[
    r,
    dashed,
    "{ \exists ! }"
  ]
  &
  \sharp \mathbf{Y}
\end{tikzcd}
\end{equation}

\paragraph
{Concretification}
\label
{OnConcretification}
With this, we may now say that $\mathbf{X} \in \mathrm{SmthSet}$ is \emph{concrete} iff  $\eta^\sharp_{\mathbf{X}}$ \cref{TheSharpUnit} is a monomorphism:
\begin{equation}
\label
{Concreteness}
  \mathbf{X} 
  \;
  \text{concrete}
  \;\;\;
  \Leftrightarrow
  \;\;\;
  \begin{tikzcd}
    \mathbf{X}
    \ar[
      r, 
      hook,
      "{ 
        \eta^\sharp
          _{\mathbf{X}}
      }"
    ]
    &
    \sharp \mathbf{X}
    \,,
  \end{tikzcd}
\end{equation}
because this means that the smooth probe maps $\inlinetikzcd{\mathbb{R}^n \ar[r] \& \mathbf{X}}$ form a subset of the maps of underlying points from $\flat \mathbb{R}^n$ to $\flat \mathbf{X}$. Such concrete smooth sets are equivalently known as \emph{diffeological spaces} (for which cf. \cite{IglesiasZemmour2013}):
\begin{equation}
  \bracketmid\{{
    \mathbf{X} \in 
    \mathrm{SmthSet}
  }{
    \text{$\mathbf{X}$ concrete}
  }
  \}
  \simeq
  \text{diffeological spaces}
  \mathrlap{\,.}
\end{equation}

Better yet, every smooth set $\mathbf{X}$ has a universal approximation by a concrete one, its \emph{concretification} $\sharp_{1} \mathbf{X}$, given by the epi/mono factorization (cf. \cite[\S 4.4]{Borceux1994-I}) of $\eta^\sharp_{\mathbf{X}}$ through its image:
\begin{equation}
\label
{Concretification}
  \begin{tikzcd}
    \mathbf{X}
    \ar[
      r,
      ->>,
      "{
        \eta^{\sharp_1}
          _{\mathbf{X}}
      }"
    ]
    \ar[
      rr,
      uphordown,
      "{
        \eta^\sharp
          _{\mathbf{X}}
      }"
    ]
    &
    \sharp_{1}
    \mathbf{X}
    \ar[
      r,
      hook
    ]
    &
    \sharp\mathbf{X}
    \mathrlap{\,.}
  \end{tikzcd}
\end{equation}
In words, the plots of $\sharp_1 \mathbf{X}$ are those of $\mathbf{X}$ after forgetting all information that disappears when restricting to points of probe spaces.

A key example of concretification is that of smooth moduli sets of differential forms, see \cref{ConcretifiedSmoothDeRhamComplex}.

Noticing that $\mathrm{Pnt}\bracket({\sharp_1 \mathbf{X}}) \simeq \mathrm{Pnt}\bracket({\mathbf{X}})$, naturally, we have natural isomorphisms
\begin{equation}
  \begin{aligned}
    \sharp \,
    \sharp_1 \mathbf{X}
    & \simeq
    \sharp \mathbf{X}
    \mathrlap{\,,}
    \\
    \flat 
    \, \sharp_1 \mathbf{X}
    & \simeq
    \flat \mathbf{X}
    \mathrlap{\,.}
  \end{aligned}
\end{equation}

In variation of \cref{UniversalPropertyOfSharpUnit},  a smooth map to a concrete smooth set \cref{Concreteness} factors uniquely  through the concretification \cref{Concretification}:
\footnote{
  The universal property \cref{UniversalPropertyOfConcretification} follows from the orthogonal epi/mono factorization in $\mathrm{SmthSet}$, which implies a unique dashed lift in the following diagram:
  \[
    \begin{tikzcd}[
      ampersand replacement=\&
    ]
      \mathbf{X}
      \ar[
        r,
        "{ f }"
      ]
      \ar[
        d,
        ->>,
        "{
          \eta^{\sharp_1}
            _{\mathbf{X}}
        }"
      ]
      \&
      \sharp_1
      \mathbf{Y}
      \ar[
        d,
        equals,
      ]
      \\
      \sharp_1 \mathbf{X}
      \ar[
        d,
        hook
      ]
      \ar[
        r,
        dashed
      ]
      \&
      \sharp_1 \mathbf{Y}
      \ar[
        d,
        hook
      ]
      \\[-5pt]
      \sharp \mathbf{X}
      \ar[
        r,
        "{
          \sharp f
        }"
      ]
      \&
      \sharp \mathbf{Y}
      \mathrlap{\,.}
    \end{tikzcd}
  \]
}
\begin{equation}
\label
{UniversalPropertyOfConcretification}
\begin{tikzcd}
  \mathbf{X}
  \ar[
    r,
    "{
      \eta^{\sharp_1}
        _{\mathbf{X}}
    }"
  ]
  \ar[
    rr,
    uphordown,
    "{
      \forall
    }"{description}
  ]
  &
  \sharp_1 \mathbf{X}
  \ar[
    r,
    dashed,
    "{
      \exists!
    }"
  ]
  &
  \sharp_1 \mathbf{Y}
\end{tikzcd}
\end{equation}

We highlight that the points of a concretification are the original points:
\begin{equation}
\label
{PointsOfConcretification}
  \flat \sharp_1
  \mathbf{X}
  \simeq
  \flat \mathbf{X}
\end{equation}
This is clear on plots. An abstract way to see it is that the image factorization in a topos is the colimit of a limit (the coequalizer of the kernel pair) and both are preserved by the left and right adjoint functor $\flat$.

\paragraph
{Field Spaces for Physics}
With this in hand, one may resolve the motivating question for the differential geometry of (field spaces that arise as) mapping spaces \cref{AskingForTheInternalHom}:

Possibly infinite-dimensional smooth manifolds (in the sense of \emph{Fr{\'e}chet manifolds} \parencites[\S I.4]{Hamilton1982}) also embed faithfully into concrete smooth sets (hence into diffeological spaces, \parencites[Thm. 3.1.1]{Losik1994}{Losik1992}), via their evident systems of plots just as in \cref{YonedaEmbedding}:
\begin{equation}
  \begin{tikzcd}[sep=0pt]
    \mathrm{FrSmthMfd}
    \ar[
      rr,
      hook,
      "{ \iota }"
    ]
    &&
    \mathrm{SmthSet}
    \\
    X &\longmapsto&
    C^\infty(-,X)
  \end{tikzcd}
\end{equation}
Under this embedding, the simple prescription for the mapping smooth set \cref{InternalHomOfSmoothSets} coincides with the infinite-dimensional manifold structure on the mapping space, for compact domains \cite[Lem. A.1.7]{Waldorf2012}:
\begin{equation}
  \label
  {WaldorfEmbedding}
  \left.
  \substack{
    X, Y \in \mathrm{SmthMfd}
    \\
    \text{$X$ compact}
  }
  \right\}
  \;\;\;\Rightarrow\;\;\;
  \iota
  \underbrace{
  C^\infty\bracket({
    X,Y
  })
  }_{\mathclap{
    \substack{
      \text{Fr{\'e}chet manifold}
      \\
      \text{of smooth maps}
    }
  }}
  \;\simeq\;
  \underbrace{
  \mathbf{Map}\bracket({
    \iota X
    ,
    \iota Y
  })}_{\mathclap{
    \substack{
      \text{concrete smooth set}
      \\
      \text{of smooth maps}
    }
  }}
  \mathrlap{\,.}
\end{equation}

\paragraph
{Cohesion}
Finally, there is a further left adjoint functor to the triple \cref{DscPntChtAdjointTriple}, which may be understood as taking smooth sets to their sets of \emph{connected components}. The resulting adjoint quadruple \cite{Sc13-dcct} exhibits $\mathrm{SmthSet}$ as a \emph{cohesive topos} in the sense of \cite{Lawvere2007}:
\begin{equation}
\label
{AdjointQuadruple}
  \begin{tikzcd}
    \mathrm{SmthSet}
    \ar[
      rr,
      shift left=30pt,
      "{
        \mathclap{\times}
      }"{description, pos=0}
    ]
    \ar[
      rr,
      phantom,
      shift left=24pt,
      "{ \bot }"{scale=.7}
    ]
    \ar[
      rr,
      phantom,
      shift left=8pt,
      "{ \bot }"{scale=.7}
    ]
    \ar[
      rr,
      "{
        \mathrm{Pnt}
      }"{description}
    ]
    \ar[
      rr,
      phantom,
      shift right=8pt,
      "{ \bot }"{scale=.7}
    ]
    &&
    \mathrm{Set}
    \mathrlap{\,.}
    \ar[
      ll,
      hook',
      shift right=16pt,
      "{
        \mathrm{Dsc}
      }"{description}
    ]
    \ar[
      ll,
      hook',
      shift left=16pt,
      "{
        \mathrm{Cht}
      }"{description}
    ]
  \end{tikzcd}
\end{equation}

But this further left adjoint functor comes into its own only after generalization from the cohesive 1-topos of smooth sets to the cohesive $\infty$-topos of smooth $\infty$-groupoids. This is what we turn to next in \cref{OnHomotopy}.

\subsection
{Homotopy}
\label
{OnHomotopy}

\subsubsection
{Gauged Sets}
\label
{OnGaugedSets}

Where the term \emph{group} is a historically evolved shorthand for \emph{group of symmetries} of a single object, the term \emph{groupoid} (cf. \parencites{Weinstein1996Groupoids}{IbortRodriguez2021}) is to be understood as the \emph{consistent collection of symmetry transformations} between possibly several objects. In a context specifically of \emph{gauge symmetries} we have, as a slogan, that \emph{groupoids are sets with gauge equivalences} and \emph{higher groupoids are sets with higher gauge-of-gauge equivalences}.  Or, for short: \emph{Groupoids are gauged sets}. 

For example, where the collection of scalar field configurations on Minkowski spacetime $\mathbb{R}^{1,d}$ forms the plain set
\begin{equation}
  \bracketmid\{
    { \Phi }
    {
      \Phi \,\in\, 
        C^\infty({
          \mathbb{R}^{1,d},
          \mathbb{C}
        })
      }
  \}
  \,,
\end{equation}
the field configurations of a gauge field (cf. \parencites{Rubakov2002}{Makeenko2002}[\S 10.5]{Nakahara2018}) with gauge group $G$ (and Lie algebra $\mathfrak{g}$) form the groupoid which we may schematically but suggestively denote as
\begin{equation}
  \left\{
    \begin{tikzcd}[
      column sep=-5pt
    ]
      A
      \ar[
        r,
        bend left=10,
        "{ g }"
      ]
      &[18pt]
      A'
      &
      =
      g (\mathrm{d} + A) g^{-1}
    \end{tikzcd}
  \,
  \middle\vert
  \,
    \substack{
      A,A' \,\in\, 
      \Omega^1_{\mathrm{dR}}
      \scaledbracket({
        \mathbb{R}^{1,d}
        ;
        \mathfrak{g}
      })
      \\
      g \,\in\, 
      C^\infty
      \scaledbracket({
        \mathbb{R}^{1,d},
        G
      })
    }
  \right\}
  \,,
\end{equation}
while the configurations of the B-field (cf. \parencites{Gawedzki1988}{BonoraFerrariSavelli2008}{FerrariRuffino2012}) form the 2-groupoid (cf. \parencites{Noohi2007}), which in this schematic but suggestive vein we may denote as:
\begin{equation}
  \left\{
  \begin{tikzcd}[
    column sep=-5pt
  ]
    & 
    B' \!=\! B + \mathrm{d}a
    \ar[
      dr,
      "{ a' }"
    ]
    \ar[
      d,
      Rightarrow,
      shorten=6pt,
      "{ \lambda }"
    ]
    \\
    B
    \ar[
      ur,
      "{
        a
      }"
    ]
    \ar[
      rr,
      "{ 
        a + a'
        + 
        \lambda \mathrm{d}
        \lambda^{-1}
      }"{swap}
    ]
    &{}&
    B''
    \!=\!
    B' + \mathrm{d}a'
  \end{tikzcd}
  \,\middle\vert\,
  \substack{
    B, B', B'' \,\in\,
    \Omega^2_{\mathrm{dR}}\scaledbracket({
      \mathbb{R}^{1,d}
    })
    \\
    a, a' \,\in\,
    \Omega^1_{\mathrm{dR}}\scaledbracket({
      \mathbb{R}^{1,d}
    })
    \\
    \lambda 
    \,\in\,
    C^\infty\scaledbracket({
      \mathbb{R}^{1,d},
      \mathrm{U}(1)
    })
  }
  \right\}
  \mathrlap{\,.}
\end{equation}
And so on.

To make this precise, it is expedient to stick to the strategy of defining generalized spaces by the sets of \emph{plots} by basic \emph{probe spaces} which they admit. In the present case, as the above examples indicate, the elementary \emph{$n$-fold gauge transformation} should be thought of as the $n$-dimensional generalization of the cellular triangle, known as the cellular \emph{$n$-simplex} $\Delta^n$. The cellular maps between these cellular $n$-simplices (mapping vertices to vertices, edges to edges, etc.) form a category accordingly known as the \emph{simplex category} $\Delta$.

Hence, a higher gauged set $\mathcal{X}$ should be characterized by the system of plots $\mathrm{Plt}\bracket({ \Delta^n, \mathcal{X} })$ of the shape of the standard simplices, forming a functor
\begin{equation}
  \begin{tikzcd}[sep=0pt]
    \Delta^{\mathrm{op}}
    \ar[
      rr
    ]
    &&
    \mathrm{Set}
    \\
    \Delta^n
    &\longmapsto&
    \mathrm{Plt}\bracket({
      \Delta^n,
      \mathcal{X}
    })
    \mathrlap{\,.}
  \end{tikzcd}
\end{equation}
This suggests that (the systems of plots of) higher gauged sets form the functor category
\begin{equation}
  \mathrm{SSet}
  :=
  \mathrm{Func}\bracket({
    \Delta^{\mathrm{op}},
    \mathrm{Set}
  })
\end{equation}
known as the category of \emph{simplicial sets} (cf. \parencites{Friedman2012}{GoerssJardine2009}).

But in more detail, we should still require that $n$-fold gauge transformations may consistently be composed and have inverses, up to higher gauge transformations. To this end, we recognize the boundary $\partial \Delta^n$ of the $n$-simplex with its $k$th $(n-1)$-face $\mathrm{d}_k\bracket({\Delta^n})$ removed, denoted
\begin{equation}
  \Lambda^n_k 
    :=
  \partial \Delta^n 
    \setminus
  \mathrm{d}_k \Delta^n
  \in
  \mathrm{SSet}
  \,,
\end{equation}
as an abstract situation of $n$ composable $(n-1)$-transformations, some of them inverted. Then the requirement of composites and inverses in a simplicial set $\mathcal{X}$ means that every image $\inlinetikzcd{\Lambda^n_k \ar[r] \& \mathcal{X}}$ of this situation in $\mathcal{X}$ may be completed to a full $n$-simplex $\inlinetikzcd{\Delta^{n} \ar[r] \& \mathcal{X}}$, hence that $\mathcal{X}$ admits all dashed lifts in diagrams of the following form (\parencites{Kan1956III}{Kan1957css}, cf. \parencites[\S I.3]{GoerssJardine2009}[\S 7]{Friedman2012}):
\begin{equation}
\label
{KanFibrancyCondition}
  \text{$\mathcal{X}$ Kan complex}
  \;\;\;\;
  \Leftrightarrow
  \;\;\;\;
  \forall_{
    \substack{
      n \geq 1,
      \\
      k \leq n
    }
  }
  \;\;
  \begin{tikzcd}[row sep=15pt, 
    column sep=25pt
  ]
    \Lambda^n_k
    \ar[
      r, 
      "{ \forall }"{pos=.4}
    ]
    \ar[
      d,
      hook
    ]
    &
    \mathcal{X}  \mathrlap{\,.}
    \\
    \Delta^{n}
    \ar[
      ur,
      dashed,
      "{ \exists }"{swap}
    ]
  \end{tikzcd}
\end{equation}
We write
\begin{equation}
  \mathrm{KanSSet}
  \subset
  \mathrm{SSet}
\end{equation}
for the full subcategory of simplicial sets on those that satisfy this \emph{Kan fibrancy condition} \cref{KanFibrancyCondition}.

By an \emph{$\infty$-category} (\parencites{Joyal2008}{Lurie2009}, cf. \parencites{Bergner2010}{AntolinCamarena2016}{RiehlVerity2022}) we may understand a category whose hom-sets are promoted to (hence which is \emph{enriched in}, cf. \parencites{Kelly1982}[\S II.1]{Quillen1967}) such Kan simplicial sets. In such an $\infty$-category we have a notion of morphisms being \emph{homotopy equivalences} if they have inverses up to gauge transformation.

There is a fairly evident notion of \emph{homotopy groups} of such \emph{Kan fibrant simplicial sets}. Let
\begin{equation}
  W_{\mathrm{wh}} 
  \subset 
  \mathrm{Mor}\bracket({
    \mathrm{KanSSet}
  })
\end{equation}
be the class of maps that induce isomorphisms on all these homotopy groups: the \emph{simplicial weak homotopy equivalences}. 
Then there is an $\infty$-category which is the universal solution to turning weak homotopy equivalences into actual homotopy equivalences, denoted
\begin{equation}
  \mathrm{Grpd}_\infty
  :=
  L^{W_{\mathrm{wh}}} \mathrm{KanSSet}
  \mathrlap{\,.}
\end{equation}
This is the \emph{$\infty$-category of $\infty$-groupoids}, making precise the idea of the \emph{gauged category of gauge sets}.

Ordinary (ungauged) sets are faithfully included here as the \emph{0-truncated} $\infty$-groupoids ($\mathrm{Set} \simeq \mathrm{Grpd}_0$), with the inclusion having a left adjoint (\emph{0-truncation}):
\begin{equation}
\label
{0TruncationAdjunction}
  \begin{tikzcd}
    \mathrm{Grpd}_\infty
    \ar[
      r, 
      shift left=12pt,
      "{ \tau_0 }"
    ]
    \ar[
      r,
      phantom,
      shift left=6pt,
      "{ \bot }"{scale=.7}
    ]
    &
    \mathrm{Set}
    \mathrlap{\,.}
    \ar[
      l,
      hook'
    ]
  \end{tikzcd}
\end{equation}
By adjunction, the composite $\infty$-functor
\begin{equation}
\label
{ZeroTruncationModality}
  [-]_0
  :
  \begin{tikzcd}
    \mathrm{Grpd}_\infty
    \ar[
      r,
      "{ \tau_0 }"
    ]
    &
    \mathrm{Set}
    \ar[r, hook]
    &
    \mathrm{Grpd}_\infty
  \end{tikzcd}
\end{equation}
comes with a natural transformation 
\begin{equation}
\label
{0truncationUnit}
  \begin{tikzcd}
    \mathllap{
    \substack{
      \text{gauged}
      \\
      \text{set}
    }}
    \;\;
    \mathcal{X}
    \ar[
      rr, 
      "{ 
        \eta^{[-]_0}
          _{\mathcal{X}} 
      }"
    ]
    &&
    \bracket[{\mathcal{X}}]_0
    \;\;
    \mathrlap{
    \substack{
      \text{set of gauge}
      \\
      \text{equivalence classes}
    }}
  \end{tikzcd}
\end{equation}
with the universal property that it uniquely factors maps to 0-truncated $\infty$-groupoids:
\begin{equation}
\label
{UniversalPropertyOfZeroTruncation}
  \begin{tikzcd}[
    column sep=large
  ]
    \mathcal{X}
    \ar[
      r,
      "{
        \eta^{[-]_0}
          _{\mathcal{X}}
      }"
    ]
    \ar[
      rr,
      uphordown,
      "{ \forall }"{description}
    ]
    &
    \bracket[{\mathcal{X}}]_0
    \ar[
      r,
      dashed,
      "{ \exists! }"
    ]
    &
    \bracket[{\mathcal{Y}}]_0
    \mathrlap{\,.}
  \end{tikzcd}
\end{equation}

As indicated in \cref{0truncationUnit}, with  $\mathcal{X} \in \mathrm{Grpd}_\infty$ thought of as a gauged set, its 0-truncation $\bracket[{\mathcal{X}}]_0 \in \mathrm{Set}$ is its set of \emph{gauge equivalence classes}.

\subsubsection
{Coherence}

The grand insight of $\infty$-category theory is that it proceeds essentially along the lines of ordinary category theory (a point originally made in \cite{Joyal2008}), with equalities generalized to homotopies and uniqueness generalized to contractible $\infty$-groupoids of choices.

For instance:  
\begin{itemize}
\item 
A \emph{Cartesian/pullback square} over a pair of coincident maps of (smooth) sets
\begin{equation}
  \begin{tikzcd}[row sep=small]
    & 
    \mathbf{Y}
    \ar[
      d
    ]
    \\
    \mathbf{X}
    \ar[r]
    &
    \mathbf{B}
  \end{tikzcd}
\end{equation}
is a commuting square (meaning that the diagonal composite maps are \emph{equal})
\begin{equation}
  \begin{tikzcd}[row sep=small]
    \mathbf{X} 
      \!\times_{_{\mathbf{B}}}\!\!
    \mathbf{Y}
    \ar[r]
    \ar[d]
    \ar[
      dr,
      phantom,
      "{ \lrcorner }"{pos=.1}
    ]
    & 
    \mathbf{Y}
    \ar[
      d
    ]
    \\
    \mathbf{X}
    \ar[r]
    &
    \mathbf{B}
    \mathrlap{\,,}
  \end{tikzcd}
\end{equation}
which is universal with this property, in that for any other such square there is a \emph{unique} dashed comparison map
\begin{equation}
  \begin{tikzcd}[row sep=14pt]
    \mathbf{Q}
    \ar[
      drr,
      bend left=15
    ]
    \ar[
      ddr,
      bend right=15
    ]
    \ar[
      dr,
      dashed
    ]
    &[-20pt]
    \\[-10pt]
    &
    \mathbf{X} 
      \!\times_{_{\mathbf{B}}}\!\!
    \mathbf{Y}
    \ar[r]
    \ar[d]
    \ar[
      dr,
      phantom,
      "{ \lrcorner }"{pos=.1}
    ]
    & 
    \mathbf{Y}
    \ar[
      d
    ]
    \\
    &
    \mathbf{X}
    \ar[r]
    &
    \mathbf{B}
  \end{tikzcd}
\end{equation}
so that the two resulting triangles commute.

\item
A \emph{homotopy Cartesian/pullback square} (cf. \parencites[\S 7.3]{Strom2011}[\S 6]{Arkowitz2011}[\S 2.2]{MayPonto2012}) over a pair of coincident maps of (smooth, cf. below) $\infty$-groupoids
\begin{equation}
  \begin{tikzcd}
    & 
    \boldsymbol{\mathcal{Y}}
    \ar[
      d
    ]
    \\
    \boldsymbol{\mathcal{X}}
    \ar[
      r
    ]
    &
    \boldsymbol{\mathcal{B}}
  \end{tikzcd}
\end{equation}
is a square that commutes up to a specified \emph{homotopy} (often notationally suppressed!)
\begin{equation}
\label
{HomotopyPullback}
  \begin{tikzcd}
    \boldsymbol{\mathcal{X}}
    \!\times_{_{\boldsymbol{\mathcal{B}}}}\!\!
    \boldsymbol{\mathcal{Y}}
    \ar[r, "{\ }"{swap,name=s, pos=.8}]
    \ar[d, "{\ }"{name=t, pos=.8}]
    \ar[
      from=s,
      to=t,
      Rightarrow
    ]
    \ar[
      dr,
      phantom,
      "{ \lrcorner }"{pos=.0}
    ]
    & 
    \boldsymbol{\mathcal{Y}}
    \ar[
      d
    ]
    \\
    \boldsymbol{\mathcal{X}}
    \ar[
      r
    ]
    &
    \boldsymbol{\mathcal{B}}
    \mathrlap{\,,}
  \end{tikzcd}
\end{equation}
which is universal with this property in that for any other such square the $\infty$-groupoid of compatible dashed comparison data 
\begin{equation}
\label
{UniversalPropertyOfHomotopyPullback}
  \begin{tikzcd}
    \boldsymbol{\mathcal{Q}}
    \ar[
      drr,
      bend left=15,
      ""{swap, name=s1, pos=.4}
    ]
    \ar[
      ddr,
      bend right=25,
      ""{name=t2, pos=.45}
    ]
    \ar[
      dr,
      dashed,
      ""{name=t1, pos=.7},
      ""{swap, name=s2, pos=.7}
    ]
    \ar[
      from=s1,
      to=t1,
      Rightarrow,
      dashed
    ]
    \ar[
      from=s2,
      to=t2,
      Rightarrow,
      dashed
    ]
    &[-15pt]
    \\[-5pt]
    &
    \boldsymbol{\mathcal{X}}
    \!\times_{_{\boldsymbol{\mathcal{B}}}}\!\!
    \boldsymbol{\mathcal{Y}}
    \ar[r, "{\ }"{swap,name=s}, pos=.8]
    \ar[d, "{\ }"{name=t, pos=.8}]
    \ar[
      dr,
      phantom,
      "{ \lrcorner }"{pos=.0}
    ]
    \ar[
      from=s,
      to=t,
      Rightarrow
    ]
    & 
    \boldsymbol{\mathcal{Y}}
    \ar[
      d
    ]
    \\
    &
    \boldsymbol{\mathcal{X}}
    \ar[
      r
    ]
    &
    \boldsymbol{\mathcal{B}}
    \mathrlap{\,,}
  \end{tikzcd}
\end{equation}
is \emph{contractible}.

\end{itemize}

Several ingenious ways have been found to handle such \emph{homotopy coherence}. A generally useful one is \emph{model category} theory (\cite{Quillen1967}, cf. \parencites{DwyerSpalinski1995}{Hovey1999}{GoerssSchemmerhorn2007}), which is a toolbox to cleverly model such situations in (for our case) the ordinary category of simplicial presheaves $\mathrm{Fun}\bracket({\mathrm{CartSp}^{\mathrm{op}}, \mathrm{SSet}})$ such that the homotopy coherent universal construction (say a homotopy pullback) is modeled by the corresponding ordinary universal construction (say an ordinary pullback). A concise introduction to these methods, aimed at our context, is in \cite[\S 1]{FSS23-Char}.

\subsection
{Geometric Homotopy}
\label
{GeometricHomotopy}

\subsubsection
{Smooth Gauged Sets}
\label
{OnSmoothInfinityGroupoids}

The unification, as advertised in \cref{GeomHomotopySchematics}, of the differential geometry of smooth sets (\cref{OnSmoothSets}) with the homotopy theory of gauged sets ($\infty$-groupoids, \cref{OnGaugedSets}) is now immediate:

A \emph{smooth gauged set} (\emph{smooth $\infty$-groupoid} or \emph{smooth $\infty$-stack}) should be a generalized space whose structure is detected by plotting out probe spaces which are formal products $\mathbb{R}^n \times \Delta^k$ of a Cartesian space (for probing the smooth structure) with a cellular simplex (for probing the higher gauge structure). Such formal products are the objects in the product site $\mathrm{CartSp} \times \Delta$ and hence (the system of plots of) a smooth gauged set $\boldsymbol{\mathcal{X}}$ is encoded by a functor of the form
\begin{equation}
  \begin{tikzcd}[
   row sep=-3pt, column sep=0pt
  ]
    \bracket({
    \mathrm{CartSp} 
      \times
    \Delta
    })^{\mathrm{op}}
    \ar[rr]
    &&
    \mathrm{Set}
    \\
    \mathbb{R}^n 
      \times
    \Delta^k
    &\longmapsto&
    \mathrm{Plt}\bracket({
      \mathbb{R}^n \times \Delta^k,
      \boldsymbol{\mathcal{X}}
    })
    \mathrlap{\,,}
  \end{tikzcd}
\end{equation}
or equivalently of the form
\begin{equation}
  \begin{tikzcd}[
   row sep=-3pt, column sep=0pt
  ]
    \mathrm{CartSp}^{\mathrm{op}}
    \ar[rr]
    &&
    \mathrm{SSet}
    \\
    \mathbb{R}^n
    &\longmapsto&
    \bracket({
      \Delta^k 
      \mapsto
      \mathrm{Plt}\bracket({
        \mathbb{R}^n 
        \times
        \Delta^k,
        \boldsymbol{\mathcal{X}}
      })
    })
    \mathrlap{\,.}
  \end{tikzcd}
\end{equation}
As such, this is known as a \emph{simplicial presheaf} on the site $\mathrm{CartSp}$. In more detail, the simplicial sets of plots of a smooth gauged set should be Kan complexes \cref{KanFibrancyCondition}, hence of the form
\begin{equation}
  \begin{tikzcd}
    \mathrm{CartSp}^{\mathrm{op}}
    \ar[
      rr
    ]
    &&
    \mathrm{KanSSet}
    \ar[
      r,
      hook
    ]
    &
    \mathrm{SSet}
    \mathrlap{\,.}
  \end{tikzcd}
\end{equation}
As before, there is redundancy in this description: For detecting the smooth structure, it ought to be sufficient to consider germs of probes, and for properly recognizing the gauged structure we need to regard weak homotopy equivalences of simplicial sets of plots as actual homotopy equivalences. Thus, denoting by
\begin{equation}
\label
{ClassOfLocalWeakHomotopyEquivalences}
  W_{\mathrm{lwh}}
  \subset
  \mathrm{Mor}\bracket({
    \mathrm{Func}\bracket({
      \mathrm{CartSp}^{\mathrm{op}},
      \mathrm{KanSSet}
    })
  })
\end{equation}
the class of such \emph{local} (germ-wise) \emph{weak homotopy equivalences}, the $\infty$-category of \emph{smooth gauged sets} is the simplicial localization
\begin{equation}
\label
{InfinityCategoryOfSmoothInfinityGroupoids}
  \mathrm{SmthGrpd}_\infty
  :=
  L^{W_{\mathrm{lwh}}}
  \mathrm{Func}\bracket({
    \mathrm{CartSp}^{\mathrm{op}},
    \mathrm{KanSSet}
  })
  \mathrlap{\,.}
\end{equation}

As a fundamental example, for a surjective submersion $p : \inlinetikzcd{ U \ar[r, ->>] \& X }$ of smooth manifolds (such as an open cover), its {\v C}ech nerve is
\begin{equation}
  \begin{aligned}
  \boldsymbol{\mathcal{U}} 
  & \in \mathrm{SmthGrpd}_\infty
  \\
  \mathrm{Plt}\bracket({
    \mathbb{R}^n \times \Delta^k,
    \boldsymbol{\mathcal{U}}
  })
  &
  :=
  C^\infty\bracket({
    \mathbb{R}^n,
    U^{\times_X^{k+1}}
  })  
  \end{aligned}
\end{equation}
and the comparison map
\[
  \begin{tikzcd}[
    row sep=-3pt, column sep=small
  ]
    \boldsymbol{\mathcal{U}}
    \ar[rr]
    &&
    X
    \\
    C^\infty\bracket({
      \mathbb{R}^n,
      U^{\times_X^{k+1}}
    })  
    \ar[rr, "{ p_\ast }"]
    &&
    C^\infty\bracket({
      \mathbb{R}^n,
      X
    })  
  \end{tikzcd}
\]
is a local weak homotopy equivalence.

In evident generalization of \cref{0TruncationAdjunction}, we have that smooth sets form a reflective sub-$\infty$-category of smooth $\infty$-groupoids, which we denote by the same symbols:
\begin{equation}
\label
{Smooth0TruncationAdjunction}
  \begin{tikzcd}
    \mathrm{SmthGrpd}_\infty
    \ar[
      r, 
      shift left=12pt,
      "{ \tau_0 }"
    ]
    \ar[
      r,
      phantom,
      shift left=6pt,
      "{ \bot }"{scale=.7}
    ]
    &
    \mathrm{SmthSet}
    \mathrlap{\,.}
    \ar[
      l,
      hook'
    ]
  \end{tikzcd}
\end{equation}
with the corresponding 0-truncation modality, generalizing \cref{ZeroTruncationModality}
\begin{equation}
\label
{SmoothZeroTruncationModality}
  [-]_0
  :
  \begin{tikzcd}
    \mathrm{SmthGrpd}_\infty
    \ar[
      r,
      "{ \tau_0 }"
    ]
    &
    \mathrm{SmthSet}
    \ar[r, hook]
    &
    \mathrm{SmthGrpd}_\infty
    \mathrlap{\,.}
  \end{tikzcd}
\end{equation}
This is just given by plotwise 0-truncation (``followed by sheafification'', which however we have absorbed in the localization \cref{InfinityCategoryOfSmoothInfinityGroupoids}):
\begin{equation}
\label
{PlotsOfSMooth0Truncation}
  \mathrm{Plt}\bracket({
    \mathbb{R}^n,
    \bracket[{
      \boldsymbol{\mathcal{X}}
    }]_0
  })
  \simeq
  \bracket[{
    \mathrm{Plt}\bracket({
      \mathbb{R}^n,
      \boldsymbol{\mathcal{X}}
    })
  }]_0
  \mathrlap{\,.}
\end{equation}

This $\infty$-category \cref{InfinityCategoryOfSmoothInfinityGroupoids} is an $\infty$-topos. In particular, this means that it is \emph{Cartesian closed}, in generalization of \cref{InternalHomOfSmoothSets,CartesianClosureOfSmoothSets}, in that for all $\boldsymbol{\mathcal{X}}, \boldsymbol{\mathcal{Y}} \in \mathrm{SmthGrpd}_\infty$ we have the smooth $\infty$-groupoid of smooth maps between them (the \emph{mapping stack})
\begin{equation}
\label
{TheMappingStack}
  \begin{aligned}
    \mathbf{Map}\bracket({
      \boldsymbol{\mathcal{X}},
      \boldsymbol{\mathcal{Y}}
    })
    &
    \in
    \mathrm{SmthGrpd}_\infty
    \\
    \mathrm{Plt}\bracket({
    \mathbb{R}^n,
    \mathbf{Map}\bracket({
      \boldsymbol{\mathcal{X}},
      \boldsymbol{\mathcal{Y}}
    })    
    })
    &
    =
    \mathrm{Hom}\bracket({
      \mathbb{R}^n
      \times
      \boldsymbol{\mathcal{X}},
      \boldsymbol{\mathcal{Y}}
    })
    \,,
  \end{aligned}
\end{equation}
where now $\mathrm{Hom}(-,-)$ denotes the plain hom $\infty$-groupoid between smooth $\infty$-groupoids.

But as an $\infty$-topos, $\mathrm{SmthGrpd}_\infty$ is special: it is \emph{cohesive}. This is what we discuss next.

\subsubsection
{Cohesion}
\label
{OnCohesion}

Above, we motivated the construction of $\mathrm{SmthGrpd}_\infty$ \cref{InfinityCategoryOfSmoothInfinityGroupoids} by imagining probing these higher smooth spaces by plotting out little smooth balls inside them. The intuition is that these balls probe how the spatial constituents smoothly \emph{hang together}, hence how there is a smooth ``cohesion'' among them. This is the \emph{differential geometric} aspect of smooth gauged sets. 

Remarkably, beyond motivating the definition, this cohesion is visible \emph{axiomatically} (\cite{Lawvere2007}, following \parencites{Lawvere1991}{Lawvere1994}[\S C.1]{LawvereRosebrugh2003}) in the abstract properties of the $\infty$-topos, namely in that it carries an adjoint quadruple of the form \cref{AdjointQuadruple}, but now promoted to $\infty$-functors over $\infty$-groupoids (\parencites[Prop. 4.4.8]{Sc13-dcct}[\S 9.1.1]{SS26-Orb}, cf. \cite{Dugger1998}):
\begin{equation}
\label
{AdjointQuadrupleOfInftyFunctors}
  \begin{tikzcd}
    \mathrm{SmthGrpd}_\infty
    \ar[
      rr,
      shift left=32pt,
      "{
        \mathclap{\times}
      }"{description, pos=0},
      "{
        \mathrm{Shp}
      }"{description}
    ]
    \ar[
      rr,
      phantom,
      shift left=24pt,
      "{ \bot }"{scale=.7}
    ]
    \ar[
      rr,
      phantom,
      shift left=8pt,
      "{ \bot }"{scale=.7}
    ]
    \ar[
      rr,
      "{
        \mathrm{Pnt}
      }"{description}
    ]
    \ar[
      rr,
      phantom,
      shift right=8pt,
      "{ \bot }"{scale=.7}
    ]
    &&
    \mathrm{Grpd}_\infty
    \mathrlap{\,.}
    \ar[
      ll,
      hook',
      shift right=16pt,
      "{
        \mathrm{Dsc}
      }"{description}
    ]
    \ar[
      ll,
      hook',
      shift left=16pt,
      "{
        \mathrm{Cht}
      }"{description}
    ]
  \end{tikzcd}
\end{equation}
Here, the bottom three functors are just as in \cref{DscPntChtAdjointTriple} in \cref{OnConcreteSmoothSets} and have the analogous interpretation: They witness the underlying ($\infty$-groupoids of) points and the discrete and chaotic smooth structures on these, generalized straightforwardly from smooth sets to smooth $\infty$-groupoids.

\paragraph
{The Shape Modality}
More profound is the generalization embodied by the further left adjoint functor $\mathrm{Shp}$ in \cref{AdjointQuadrupleOfInftyFunctors}, to be pronounced the \emph{shape} operation (as in \emph{shape theory} \parencites{Borsuk1975}{BorsukDydak1980}{DydakSegal1978}[\S 7.1.6]{Lurie2009}, where for the well-behaved topological spaces of interest here (homotopy equivalent to CW-complexes) their \emph{shape} is equivalently their ordinary \emph{homotopy type}, cf. \parencites[\S I.2]{MardesicSegal1982}[\S II]{Mardesic2000}). The shape operation may be understood as contracting all ``cohesive blobs'' of points (the $\mathbb{R}^k$) to points, while retaining their connectivity structure among each other.

The corresponding composite functor we denote by
\begin{equation}
\label{TheShapeModality}
  \shape
  :
  \begin{tikzcd}
    \mathrm{SmthGrpd}_\infty
    \ar[
      r,
      "{ \mathrm{Shp} }"
    ]
    &
    \mathrm{Grpd}_\infty
    \ar[
      r,
      hook,
      "{ \mathrm{Dsc} }"
    ]
    &
    \mathrm{SmthGrpd}_\infty
  \end{tikzcd}
\end{equation}
where ``$\shape$'' is (not the integral symbol but) the \emph{esh} symbol from the phonetic alphabet: The pronunciation of ``\emph{shape}'' is denoted: \textipa{/SeIp/}.

We may equivalently think of $\shape \boldsymbol{\mathcal{X}}$ as being the \emph{smooth path $\infty$-groupoid} of $\boldsymbol{\mathcal{X}}$, as follows:

Consider the cosimplicial object $\boldsymbol{\Delta}$ of \emph{extended smooth simplices} 
\footnote{
 The usual bounded (smooth) simplices are as in \cref{ExtendedSmoothSimplices} except for the additional condition $\forall_i : x^i \geq 0$, which makes them manifolds with boundary. The extended simplices  \cref{ExtendedSmoothSimplices} are a ``collared'' version of these ordinary simplices, which improves their technical properties without changing their homotopy theoretic content.
}
\begin{equation}
\label
{ExtendedSmoothSimplices}
  \begin{tikzcd}[
    sep=0pt
  ]
    \Delta
    \ar[rr]
    &&
    \mathrm{SmthMfd}
    \subset
    \mathrm{SmthGrpd}_\infty
    \\
    \Delta[k]
    &\longmapsto&
    \boldsymbol{\Delta}^k
    :=
    \bracketmid\{{
      \vec x \in \mathbb{R}^{k+1}
    }{
      \sum_i x^i = 1
    }\}
    \simeq
    \mathbb{R}^k
  \end{tikzcd}
\end{equation}
which on injections $\inlinetikzcd{\Delta[k] \ar[r, hook] \& \Delta[k+1]}$ is given by inserting a zero coordinate, and on surjections $\inlinetikzcd{ \Delta[k+1] \ar[r, ->>] \& \Delta[k] }$ is given by adding consecutive coordinates.
Then \parencites[Prop. 1.3]{PavlovEtAl2024}[Prop. 4.14]{Bunk2022}[\S 5]{AmabelEtAl2021}[Prop. 8.36]{BunkShahbazi2023}[Prop. 12.10]{Pavlov2024}:
\begin{equation}
\label
{SimplicialFormulaForShape}
  \shape \boldsymbol{\mathcal{X}}
  \simeq
  \varinjlim_{k}
  \,
  \mathbf{Map}\bracket({
    \boldsymbol{\Delta}^k,
    \boldsymbol{\mathcal{X}}
  })
  \simeq
  \mathrm{Dsc}\bracket({
    \varinjlim_{k}
    \,
    \boldsymbol{\mathcal{X}}\bracket({
      \Delta^k
    })
  })
  \mathrlap{\,.}
\end{equation}
For instance, when $\boldsymbol{\mathcal{X}} \defneq \mathbf{X}$ is 0-truncated (a smooth set), then the homotopy colimit in \cref{SimplicialFormulaForShape} is equivalently the \emph{smooth singular simplicial set}
\begin{equation}
\label
{SmoothSingularSet}
  \left.
  \substack{
    \mathbf{X}
    \,\in\,
    \mathrm{SmthSet}
  }
  \right\}
  \;\;\;
  \vdash
  \;\;\;
  \left\{
  \begin{aligned}
    \mathrm{Sing}\bracket({
      \mathbf{X}
    })
    &
    \in
    \mathrm{SSet}
    \to
    \mathrm{Grpd}_\infty
    \\
    \mathrm{Plt}\bracket({
      \Delta[k],
      \mathrm{Sing}\bracket({
        \mathbf{X}
      })
    })
    & =
    \mathrm{Hom}\bracket({
      \boldsymbol{\Delta}^k,
      \mathbf{X}
    })
    \phantom{
    \left.
    \substack{
      \mathbf{X}
      \,\in\,
      \mathrm{SmthSet}
    }
  \right\}
  \;\;\;
  \vdash    
    }
  \end{aligned}
  \right.
\end{equation}
in that
\begin{equation}
\label
{ShapeEquivalentToSingularComplex}
  \left.
  \substack{
    \mathbf{X}
    \,\in\,
    \mathrm{SmthSet}
  }
  \right\}
  \;\;\;
  \Rightarrow
  \;\;\;
  \shape \mathbf{X}
  \simeq
  \mathrm{Dsc}
  \,
  \mathrm{Sing}\bracket({
    \mathbf{X}
  })
  \mathrlap{\,.}
  \phantom{
  \left.
  \substack{
    \mathbf{X}
    \,\in\,
    \mathrm{SmthSet}
  }
  \right\}
  \;\;\;
  \Rightarrow
  }
\end{equation}
This way, one may generally understand $\shape(-)$ as a generalization of the traditional $\mathrm{Sing}(-)$ to smooth $\infty$-groupoids (cf. \cite[pp. 155]{SS26-Bun}).

The archetypical example here is the shape of Cartesian spaces, which is trivial:
\begin{equation}
\label
{CartesianSpaceHasTrivialShape}
  \shape \mathbb{R}^n 
    \simeq 
  \ast
  \mathrlap{\,.}
\end{equation}
In fact, by the left-adjointness of $\mathrm{Shp}$ in \cref{AdjointQuadrupleOfInftyFunctors} this  property \cref{CartesianSpaceHasTrivialShape} already characterizes the shape operation; and it implies that the mapping stack construction \cref{TheMappingStack} satisfies:
\begin{equation}
  \label
  {MappingIntoGeometricallyDiscrete}
  \mathbf{Map}\bracket({
    \boldsymbol{\mathcal{X}}
    ,\,
    \shape
    \boldsymbol{\mathcal{Y}}
  })
  \simeq
  \mathbf{Map}\bracket({
    \shape \boldsymbol{\mathcal{X}}
    ,\,
    \shape
    \boldsymbol{\mathcal{Y}}
  })
  \simeq
  \shape
  \mathbf{Map}\bracket({
    \shape \boldsymbol{\mathcal{X}}
    ,\,
    \shape
    \boldsymbol{\mathcal{Y}}
  })
  \simeq
  \shape
  \mathbf{Map}\bracket({
    \boldsymbol{\mathcal{X}}
    ,\,
    \shape
    \boldsymbol{\mathcal{Y}}
  })
  \mathrlap{\,.}
\end{equation}

Moreover, by adjunction, there is a natural transformation (the \emph{shape unit})
\begin{equation}
\label
{ShapeUnit}
  \begin{tikzcd}
    \boldsymbol{\mathcal{X}}
    \ar[
      rr,
      "{
        \eta^{\shape}
          _{\boldsymbol{\mathcal{X}}}
      }"
    ]
    &&
    \shape
    \boldsymbol{\mathcal{X}}
  \end{tikzcd}
\end{equation}
which, under the equivalence \cref{SimplicialFormulaForShape}, is the inclusion of constant paths. Under the shape operation, the shape unit is an equivalence:
\begin{equation}
\label
{IdempotencyViaShapeUnit}
  \begin{tikzcd}
    \shape
    \boldsymbol{\mathcal{X}}
    \ar[
      rr,
      "{
        \shape
        \eta^{\shape}
          _{\boldsymbol{\mathcal{X}}}
        \,\simeq\,
        \eta^{\shape}
          _{\shape \boldsymbol{\mathcal{X}}}
      }",
      "{
        \sim
      }"{swap}
    ]
    &&
    \shape
    \shape
    \boldsymbol{\mathcal{X}}
    \mathrlap{\,.}
  \end{tikzcd}
\end{equation}

Further important properties of the shape operation are:
\begin{enumerate}
\item
\textbf{(Partial exactness).}
Shape preserves homotopy fiber products \cref{HomotopyPullback} at least:
\begin{enumerate}
\item over geometrically discrete objects \parencites[Prop. 4.3.8]{SS26-Bun}[Thm. 3.8.19]{Sc13-dcct}:
\begin{equation}
\label
{ShapePreservesHomotopyFibersOverDiscrete}
  \left.
  \substack{
    \boldsymbol{\mathcal{X}},
    \boldsymbol{\mathcal{Y}}
    \,\in\,
    \mathrm{SmthGrpd}_\infty
    \\
    \mathcal{B} \,\in\,
    \mathrm{Dsc}\scaledbracket({
      \mathrm{Grpd}_\infty
    })
  }
  \right\}
  \;\;
  \Rightarrow
  \;\;\
  \shape\bracket({
    \boldsymbol{\mathcal{X}}
    \times_{\mathcal{B}}
    \boldsymbol{\mathcal{Y}}
  })
  \simeq
  \bracket({
    \shape
    \boldsymbol{\mathcal{X}}
  })
  \times_{\mathcal{B}}
  \bracket({
    \shape
    \boldsymbol{\mathcal{Y}}
  })  
  \mathrlap{\,,}
  \phantom{
  \left.
  \substack{
    \boldsymbol{\mathcal{X}},
    \boldsymbol{\mathcal{Y}}
    \,\in\,
    \mathrm{SmthGrpd}_\infty
    \\
    \mathcal{B} \,\in\,
    \mathrm{Dsc}\scaledbracket({
      \mathrm{Grpd}_\infty
    })
  }
  \right\}
  \;\;
  \Rightarrow  
  }
\end{equation}
\item of 0-truncated objects if one of the maps becomes a Kan fibration under the smooth singular simplicial complex $\mathrm{Sing}$ \cref{SmoothSingularSet}:
\footnote{
 The implication \cref{ShapePreservesFiberProductsWithSingFibrantMaps} follows because $\mathrm{Sing}$ \cref{SmoothSingularSet} evidently preserves 1-categorical fiber products and since (cf. \cite[Prop. A.2.4.4]{Lurie2009} with \cite[\S II.8.6]{GoerssJardine2009}) in $\mathrm{SSet}$ these present homotopy fiber products as soon as one of the maps is a Kan fibration.
}
\begin{equation}
\label
{ShapePreservesFiberProductsWithSingFibrantMaps}
  \left.
  \substack{
  \adjustbox{scale=.9}{
  \begin{tikzcd}[
    ampersand replacement=\&,
    row sep=-5pt,
    column sep=10pt
  ]
    \mathbf{X}
    \ar[
      dr,
      shorten=-2pt,
      "{f}"{pos=.4}
    ]
    \&\&
    \mathbf{Y}
    \ar[
      dl,
      shorten=-2pt,
      "{ g }"{pos=.4,swap}
    ]
    \\
    \&
    \mathbf{B}
  \end{tikzcd}
  }
  \in\, 
  \mathrm{SmthSet}
  \\
  \mathrm{Sing}(f)
  \text{ is Kan fibration}
  }
  \right\}
  \;\Rightarrow\;
  \shape\bracket({
    \mathbf{X} 
    \times_{\mathbf{B}}
    \mathbf{Y}
  })
  \simeq
  \bracket({
    \shape \mathbf{X}
  })
  \underset{
    \shape \mathbf{B}
  }{\times}
  \bracket({
    \shape \mathbf{Y}
  })\,.
\end{equation}
\end{enumerate}
\item
\textbf{(Smooth Oka principle).}
Shape preserves mappings \cref{TheMappingStack} out of smooth manifolds \parencites[Prop. 1.4]{PavlovEtAl2024}[Thm. 4.3.54]{SS26-Bun}[Prop. 7.2.27]{Clough2023}:
\begin{equation}
\label
{SmoothOkaPrinciple}
  \hspace{-.7cm}
  \left.
  \substack{
    X \,\in\, \mathrm{SmthMfd}
    \\
    \boldsymbol{\mathcal{Y}}
    \,\in\,
    \mathrm{SmthGrpd}_\infty
  }
  \right\}
  \;\;\;
  \Rightarrow
  \;\;\;
  \shape\, 
  \mathbf{Map}\bracket({
    X,
    \boldsymbol{\mathcal{Y}}
  })
  \simeq
  \mathbf{Map}\bracket({
    X,
    \shape \boldsymbol{\mathcal{Y}}
  })
  \simeq
  \mathbf{Map}\bracket({
    \shape X,
    \shape \boldsymbol{\mathcal{Y}}
  })
  \mathrlap{\,.}
\end{equation}
\end{enumerate}

\paragraph
{The Flat Modality}

The understanding of the flat modality \cref{FlatModalityOnSmoothSets} on gauged smooth sets
\begin{equation}
\label
{FlatModalityOnGaugedSmoothSets}
  \begin{tikzcd}
    \mathrm{SmthGrpd}_\infty
    \ar[
      rr,
      uphordown,
      "{
        \flat
      }"{description}
    ]
    \ar[
      r,
      "{ \mathrm{Pnt} }"
    ]
    &
    \mathrm{Grpd}_\infty
    \ar[
      r,
      hook,
      "{ 
        \mathrm{Dsc}
      }"
    ]
    &
    \mathrm{SmthGrpd}_\infty
  \end{tikzcd}
\end{equation}
is more straightforward. We highlight the following immediate but important point:
By \cref{PlotsOfSMooth0Truncation}, the flat modality commutes with the 0-truncation modality \cref{SmoothZeroTruncationModality}:
\begin{equation}
\label
{FlatCommutesWith0Truncation}
  \bracket[{
  \flat 
  \boldsymbol{\mathcal{X}}
  }]_0
  \simeq
  \flat
  \bracket[{
    \boldsymbol{\mathcal{X}}
  }]_0
  \mathrlap{\,.}
\end{equation}

\section
{Higher Maxwell-type Equations of Motion}
\label
{OnHigherMaxwellTypeEquationsOfMotion}

We make precise the notion of higher Maxwell-type equations of motion (EoM) and how they are equivalently solved by closed differential forms with coefficients in a characteristic $L_\infty$-algebra.

\subsection
{Flux Densities}

\subsubsection
{Differential Form Moduli}
\label
{OnUniversalDeRhamComplex}

\paragraph
{Universal de Rham Complex}
The smooth moduli sets of differential $n$-forms $
  \mathbf{\Omega}^n_{\mathrm{dR}}(\ast)
  \in
  \mathrm{SmthSet}
$
\cref{SmoothModuliSetOfDifferentialNForms}
canonically carry the further structure of smooth $\mathbb{R}$-vector spaces (by plotwise addition and multiplication):
\begin{equation}
  \mathbf{\Omega}^n_{\mathrm{dR}}(\ast)
  \in
  \mathrm{Vec}_{{}_{\mathbb{R}}}\bracket({
    \mathrm{SmthSet}
  })
  \mathrlap{\,.}
\end{equation}

Moreover, equipped with the universal wedge products (given plot-wise by the ordinary wedge product):
\begin{equation}
  \begin{tikzcd}
    \mathbf{\Omega}^{n_1}_{\mathrm{dR}}(\ast)
    \times
    \mathbf{\Omega}^{n_2}_{\mathrm{dR}}(\ast)
    \ar[
      rr,
      "{ (-) \wedge (-) }"
    ]
    &&
    \mathbf{\Omega}^{n_1 + n_2}_{\mathrm{dR}}(\ast)
  \end{tikzcd}
\end{equation}
and universal de Rham differential (again plotwise given by the ordinary de Rham differential),
\begin{equation}
  \begin{tikzcd}
    \mathbf{\Omega}^{n}
      _{\mathrm{dR}}(\ast)
    \ar[
      r,
      "{ \mathrm{d} }"
    ]
    &
    \mathbf{\Omega}^{n+1}_{\mathrm{dR}}(\ast)
    \mathrlap{\,,}
  \end{tikzcd}
\end{equation}
these smooth vector spaces form the \emph{universal de Rham complex}:
\begin{equation}
  \mathbf{\Omega}^\bullet_{\mathrm{dR}}(\ast)
  \in
  \mathrm{dgcAlg}_{_{\mathbb{R}}}\bracket({
    \mathrm{SmthSet}
  })
  \mathrlap{\,.}
\end{equation}

\paragraph
{De Rham moduli over a Smooth Set}
More generally, given $\mathbf{X} \in \mathrm{SmthSet}$, we say that its plain \emph{de Rham complex} is
\begin{equation}
\label
{DeRhamComplexOfSmoothSet}
  \begin{aligned}
  \Omega^\bullet_{\mathrm{dR}}\bracket({
    \mathbf{X}
  })
  & :=
  \flat \mathbf{Map}\bracket({
    \mathbf{X},
    \mathbf{\Omega}^\bullet_{\mathrm{dR}}(\ast)
  })
  \\
  & =
  \mathrm{Hom}\bracket({
    \mathbf{X},
    \mathbf{\Omega}^\bullet_{\mathrm{dR}}(\ast)
  })  
  \in
  \mathrm{dgcAlg}_{_{{\mathbb{R}}}}\bracket({
    \mathrm{Set}
  })
  \mathrlap{\,,}
  \end{aligned}
\end{equation}
and its \emph{smooth de Rham complex} is
\begin{equation}
\label
{SmoothDeRhamComplex}
  \mathbf{\Omega}^\bullet_{\mathrm{dR}}
  \bracket({
    \mathbf{X}
  })
  :=
  \mathbf{Map}\bracket({
    \mathbf{X},
    \mathbf{\Omega}^\bullet_{\mathrm{dR}}(\ast)
  })
  \in
  \mathrm{dgcAlg}_{_{\mathbb{R}}}\bracket({
    \mathrm{SmthSet}
  })
  \mathrlap{\,.}
\end{equation}

We note that, by the formula \cref{InternalHomOfSmoothSets} for $\mathbf{Map}(-,-)$, we have:
\begin{equation}
\label
{PlotsOfSmoothDeRhamComplexOfSmoothSet}
  \begin{aligned}
  \mathrm{Plt}\bracket({
    \mathbb{R}^n,
    \mathbf{\Omega}^\bullet_{\mathrm{dR}}
    \bracket({
      \mathbf{X}
    })
  })
  & 
  \defneq
  \mathrm{Plt}\bracket({
    \mathbb{R}^n,
    \mathbf{Map}\bracket({
      \mathbf{X},
      \mathbf{\Omega}^\bullet_{\mathrm{dR}}(\ast)
    })
  })
  \\
  &
  =
  \mathrm{Hom}\bracket({
    \mathbb{R}^n \times \mathbf{X},
    \mathbf{\Omega}^\bullet_{\mathrm{dR}}(\ast)
  })
  \\
  & =
  \Omega^\bullet_{\mathrm{dR}}\bracket({
    \mathbb{R}^n \times \mathbf{X}    
  })
  \mathrlap{\,.}
  \end{aligned}
\end{equation}

\paragraph
{Concrete Smooth de Rham complex}
The smooth de Rham complex \cref{SmoothDeRhamComplex} is far from concrete \cref{Concreteness}.
It follows from \cref{PlotsOfSmoothDeRhamComplexOfSmoothSet} that its \emph{concretification} \cref{Concretification} has as $\mathbb{R}^n$-plots 
the vertical differential forms relative to $\mathbb{R}^n$, hence 
the smoothly $\mathbb{R}^n$-parameterized differential forms on $\mathbf{X}$ (\cite[Prop. 1.266]{Sc13-dcct}):
\begin{equation}
\label
{ConcretifiedSmoothDeRhamComplex}
 \begin{aligned}
    \sharp_{1}
    \mathbf{\Omega}^\bullet
      _{\mathrm{dR}}
    \bracket({
      \mathbf{X}
    })
    & 
    \in
    \mathrm{SmthSet}
    \\
    \mathrm{Plt}\bracket({
      \mathbb{R}^n,
      \sharp_{1}
      \mathbf{\Omega}^\bullet_{\mathrm{dR}}
      \bracket({
        \mathbf{X}
      })      
    })
    &
    \simeq
    \Omega^\bullet_{\mathrm{vert}/\mathbb{R}^n}
    \bracket({
      \mathbb{R}^n 
      \times
      \mathbf{X}
    })
    \\
    & \simeq
    C^\infty\bracket({
      \mathbb{R}^n,
      \Omega^\bullet_{\mathrm{dR}}
      \bracket({
        \mathbf{X}
      })
    })
    \mathrlap{\,.}
  \end{aligned}
\end{equation}

\paragraph
{Compactly Supported de Rham complex}
For $X \in \mathrm{SmthMfd}$, write $\mathrm{Cpt}(X)$ for the category whose objects are compact submanifolds with boundary and whose morphisms are inclusions. 
The smooth de Rham complex of differential forms on $X$ which are \emph{compactly supported} (hence in particular \emph{vanishing at infinity}) is the colimit over such $K$ of the moduli of differential forms whose pullback $\iota^\ast_K$ to the complement $X \setminus K^{\circ}$ (of the interior of $K$) vanishes: 
\begin{equation}
\label
{DifferentialFormModuliVanishingAtInfinity}
  \begin{aligned}
  \mathbf{\Omega}^\bullet_{%
    \mathrm{dR},
    \mathrm{c}
  }\bracket({
    X
  })
  &
  :=
  \varinjlim_{K \in \mathrm{Cpt}(X)}
  \,
  \mathbf{\Omega}^\bullet_{\mathrm{dR}}\bracket({
    X, X \setminus K^\circ
  })
  \\[-5pt]
  & 
  :=
  \varinjlim_{K \in \mathrm{Cpt}(X)}
  \,
  \mathrm{fib}
  \Big({
    \inlinetikzcd{
      \mathbf{\Omega}^\bullet_{\mathrm{dR}}
      \bracket({
        X
      })
      \ar[
        r,
        "{ 
          \iota^\ast_K
        }"
      ]
      \&
      \mathbf{\Omega}^\bullet_{\mathrm{dR}}
      \bracket({
        X \setminus K^\circ
      })
    }
  }\Big)
  \mathrlap{\,.}
  \end{aligned}
\end{equation}

\subsubsection
{Flux Density Moduli}

Consider a finite set $I$ and an $I$-indexed set of degrees:
\begin{equation}
\label{degrees}
 \mathrm{deg}_{(-)}
 :
 \begin{tikzcd}
   I
   \ar[r]
   &
   \mathbb{N}
   \mathrlap{\,.}
 \end{tikzcd}
\end{equation}
Given $i \in I$, write
\begin{equation}
  \begin{tikzcd}
  b^{(i)}
  \in 
  \Omega_{\mathrm{dR}}
    ^{\mathrm{deg}_i}
    \bracket({
      \mathbf{\Omega}_{\mathrm{dR}}
        ^{\mathrm{deg}_i}(\ast)
    })
    \ar[r, hook]
    &
  \Omega_{\mathrm{dR}}
    ^{\bullet}
    \bracket({
      \mathbf{\Omega}_{\mathrm{dR}}
        ^{\mathrm{deg}_i}(\ast)
    })
  \end{tikzcd}
\end{equation}
for the universal $\mathrm{deg}_i$-form, the one given by the identity map
\begin{equation}
\label
{UniversalDifferentialForm}
  b^{(i)}
  \defneq 
  \mathrm{id}_{
    \mathbf{\Omega}
      ^{\mathrm{deg}_i}
      _{\mathrm{dR}}(\ast)
  }
  \mathrlap{\,.}
\end{equation}
Via pullback along the product projection maps
\begin{equation}
  \begin{tikzcd}
    \prod_j 
    \mathbf{\Omega}
      ^{\mathrm{deg}_j}
      _{\mathrm{dR}}
    (\ast)
    \ar[
      r,
      "{
        \mathrm{pr}_i
      }"
    ]
    &
    \mathbf{\Omega}
      ^{\mathrm{deg}_i}
      _{\mathrm{dR}}
    (\ast)\,,
  \end{tikzcd}
\end{equation}
these induce differential forms
\begin{equation}
\label
{UniversalDiffFormsOnProduct}
  \mathrm{pr}_i^\ast
  b^{(i)}
  \in
  \Omega^{\mathrm{deg}_i}_{\mathrm{dR}}
  \bracket({
    \prod_j
    \mathbf{\Omega}
      ^{\mathrm{deg}_j}
      _{\mathrm{dR}}
    (\ast)
  })
\end{equation}
on their joint product moduli set. 
This way, for $\mathbf{X} \in \mathrm{SmthSet}$, a map
\begin{equation}
\label
{FluxDensities}
  \begin{tikzcd}
    \mathbf{X}
    \ar[
      r, 
      "{ \vec B }"
    ]
    &
    \prod_i
    \mathbf{\Omega}
      ^{\mathrm{deg}_i}
      _{\mathrm{dR}}(\ast)
  \end{tikzcd}
\end{equation}
is equivalently an $i \in I$-indexed set of $\mathrm{deg}_i$-forms on $\mathbf{X}$:
\begin{equation}
  \vec B
  =
  \bracket({
    B^{(i)} 
      := 
    \vec B^\ast
    \mathrm{pr}_i^\ast
    b^{(i)}
  })_{i \in I}
  \in 
  \prod_i
  \Omega^{\mathrm{deg}_i}_{\mathrm{dR}}\bracket({
    \mathbf{X}
  })
  \mathrlap{\,.}
\end{equation}

In this manner, 
\begin{equation}
  \prod_j \mathbf{\Omega}^{\mathrm{deg}_j}_{\mathrm{dR}}(\ast)
  \in
  \mathrm{SmthSet}
\end{equation}
is the \emph{smooth moduli set} for \emph{flux densities} in degrees $\mathrm{deg}_\bullet$ on $\mathbf{X}$.

Next, we impose equations of motion on such flux densities.

\subsection
{Equations of Motion}
\label
{OnEquationsOfMotion}

In \cref{HigherMaxwellTypeEquations} we formulate higher Maxwell-type equations of motion (following \cite{SS24-Phase}) in the form that: 
\begin{enumerate}
\item
in electromagnetism is known as the \emph{pre-metric form} 
(\parencites
 {Kottler1922}
 [\S 80]{Cartan1924}
 {vanDantzig1934}
 {Post1972}
 {HehlObukhov2003}
 {BunsterHenneaux2011}
 {Delphenich2015}
 {HehlItinObukhov2016}
 [Rem. 2.3]{BeckerBeniniSchenkelSzabo2017}
 [Def. 1.16]{LazaroiuShahbazi2022a}%
),
\item
in 10/11D supergravity (when promoted to an action principle) as the \emph{duality-symmetric form} or \emph{democratic form} 
(\parencites
  {BandosBerkovitsSorokin1997}
  {DallAgataLechnerTonin1998}
  [\S 2--3]{CremmerJuliaLuPope1998}
  {BergshoeffKalloshOrtinRoestVanProeyen2001}
  {BandosNurmagambetovSorokin2003}
  {MkrtchyanValach2022}%
),
\item 
in the general higher abelian case (\emph{$p$-form electromagnetism}) was considered in \parencites
  [Ex. 3.8]{Freed2002}
  [\S 31.3]{MooreSaxena2025}.
\end{enumerate}

Here  field equations are stated as Bianchi identities \cref{HigherBianchiIdentities} on \emph{a priori} independent pairs of magnetic and electric fluxes, which are then forced to be suitable Hodge dual pairs by a further equation, the \emph{constitutive relation} \cref{ConstitutiveRelations}. 

This way of writing the equations of motion is trivially equivalent to the more common asymmetric way, where Bianchi identities are imposed on magnetic flux densities and equations of motion on their Hodge duals. However, the duality-symmetric formulation serves to naturally bring out, in \cref{OnCharacteristicLInfinityAlgebra}, the \emph{characteristic $L_\infty$-algebra} underlying such equations, which is the key to IR-completion in \cref{OnTheGlobalPhaseSpace}.

\subsubsection
{Higher Maxwell-type EoM}
\label
{HigherMaxwellTypeEquations}

Let $X^{1,d}$ be an oriented Lorentzian spacetime manifold of dimension $D := 1+d$, hence in particular a smooth manifold equipped with a pseudo-Riemannian metric with respect to which we have the Hodge star operator
\begin{equation}
  \star 
  :
  \inlinetikzcd{
    \Omega^\bullet_{\mathrm{dR}}
    \bracket({
      X^{1,d}
    })
    \ar[r, "{ \sim }"]
    \&
    \Omega^{1+d -\bullet}_{\mathrm{dR}}
    \bracket({
      X^{1,d}
    })
   \mathrlap{\,.}
  }
\end{equation}
For $I$-indexed sets of flux densities in positive degrees $\mathrm{deg}_i \in \mathbb{N}_+$ \cref{degrees}
\begin{equation}
\label
{FluxDensitiesOnX}
  \text{flux densities:}
  \;\;\;
  \vec F
  :=
  \bracket\{{
    F^{(i)}
    \in
    \Omega^{\mathrm{deg}_i}
      _{\mathrm{dR}}
    (X)
  }\}_{i \in I}\,,
\end{equation}
we say \cite[\S 2.1]{SS24-Phase} that \emph{higher Maxwell-type equations of motion} are
systems of equations of the form%
\footnote{
  Generally by a (differential) \emph{Bianchi identity} one understands differential identities satisfied by curvature tensors (cf. \parencites[\S 22.2]{Tu2017}[\S 10.3.5]{Nakahara2018}). Here in the duality-symmetric formulation, both the magnetic and electric flux densities play the role of curvature components of higher gauge fields, and as such their differential equations \cref{HigherBianchiIdentities}, before imposing the constitutive relation \cref{ConstitutiveRelations}, both play the role of Bianchi identities. The Bianchi identities on the electric fluxes become \emph{equations of motion} (only) once the constitutive relation \cref{ConstitutiveRelations} is imposed.
}
\begin{subequations}
\label
{HigherMaxwellTypeEoM}
\begin{alignat}{2}
  \label{HigherBianchiIdentities}
  \text{Bianchi identities:}
  &\phantom{--}&
  \forall_i 
  \;\;\;
  \mathrm{d}\, F^{(i)}
  &=
  P^{(i)}\bracket({
    \vec F
  })
  \mathrlap{\,,}
  \\
  \label{ConstitutiveRelations}
  \text{constitutive relations:}
  &&
  \forall_i 
  \;\;
  \star\, F^{(i)}
  &=
  \mu^{(i)}\bracket({
    \vec F
  })
  \mathrlap{\,,}
\end{alignat}
\end{subequations}
where:
\begin{enumerate}
\item
  the $P^{(i)}$ are graded-symmetric polynomials of degree $\mathrm{deg}_i + 1$ in graded variables $j \in I$ of degrees $\mathrm{deg}_j$,
\item
  $\mu$ is a linear automorphism on the linear span of $I$ compatible with $\star$, in that:
  \begin{enumerate}
  \item
  $\mathrm{deg} \circ \mu(-) = D-\mathrm{deg}(-)$,

  \item
  $\mu^2\bracket({F^{(i)}}) = \pm F^{(i)}$ in accordance with $\star^2 F^{(i)} = \pm F^{(i)}$.
  \end{enumerate}
\end{enumerate}

Basic examples are the following (where we indicate the degrees of differential forms in their subscripts):
\begin{subequations}
\label
{ExamplesOfHigherMaxwellEquations}
\def\arraycolsep{20pt}
\begin{alignat}{3}
\label
{EoMOf4DVacuumElectromagnetism}
\adjustbox{
      rndfbox=4pt, scale=0.7
    }{\begin{tabular}{c} 
    \text{4D Vacuum Maxwell}
    \\
    \text{(gauge sector of 4D SuGra, cf.}
    \\
    \text{
     \parencites
     [\S 80]{Cartan1924}
     [\S 3.4 \& 4.3]{MisnerThorneWheeler1973}
     {Frankel1974}
    )}
    \end{tabular} 
  }
  &&\phantom{--}&
    \begin{aligned}
      \mathrm{d}\, F_2 & = 0
      \\
      \mathrm{d}\, G_2 & = 0
      \mathrlap{\,,}
    \end{aligned}
    &\phantom{--}&
    \star\, F_2 = G_2
    \mathrlap{\,,}
  \\[+7pt]
  \label{ChiralFluxEoM}
 \adjustbox{
      rndfbox=4pt, scale=0.7
    }{\begin{tabular}{c}  
    \text{Chiral Flux in $D= 4k\!+\!2$}
    \\
    \text{(cf. \parencites
      {HenneauxTeitelboim1988}
      {PastiSorokinTonin1997}
      {BelovMoore2006})}
  \end{tabular}
  }
  &&&
    \mathrm{d}\, F_{2k+1} = 0
    \,,
    &&
    \star\, F_{2k+1} \!=\! F_{2k+1}
    \mathrlap{,}
  \\[+7pt]
  \label{EoMOf5DMCSTheory}
  \adjustbox{
      rndfbox=4pt, scale=0.7
    }{\begin{tabular}{c} 
    \text{5D Maxwell-Chern-Simons} 
    \\
    \text{(gauge sector of 5D SuGra,}
    \\
    \text{cf. \parencites
      {LikoBooth2008}
      [\S I.2.4]{Compere2007}
    }
    \\
    \text{
      \&
      \parencites
      {Cremmer1980}
      [\S 9]{Ortin2015})
    }
    \end{tabular} 
  }
  &&&
    \begin{aligned}
      \mathrm{d}\, F_2 & = 0
      \\
      \mathrm{d}\, G_3 & = 
      \tfrac{1}{2}
      F_2 \wedge F_2
      \mathrlap{\,,}
    \end{aligned}
    &&
    \star\, F_2 = G_3
    \mathrlap{\,,}
  \\[+7pt]
  \label{EoMOfIIA}
 \adjustbox{
      rndfbox=4pt, scale=0.7
    }{\begin{tabular}{c} 
    \text{10D IIA NS/RR-Fluxes} 
    \\
    \text{(gauge sector of 10D IIA SuGra,}
    \\
    \text{cf. 
    \parencites
      {GianiPernici1984}
      {CampbellWest1984}
      {HuqNamazie1985}
      \&
    }
    \\
    \text{
      \parencites
        [\S 12.1]{Polchinski1998II}
        [\S 22.1.3]{Ortin2015}
        {MkrtchyanValach2022}
      )}
 \end{tabular}
  }
  &&&
    \begin{aligned}
      \mathrm{d}\, F_2 
      & 
      = 0
      \\
      \mathrm{d}\, F_{4}
      & 
      =
      H_3 \wedge F_{2}
      \\
      \mathrm{d}\, F_{6}
      & 
      =
      H_3 \wedge F_{4}
      \\
      \mathrm{d}\, F_{8}
      & 
      =
      H_3 \wedge F_{6}
      \\
      \mathrm{d}\, H_3 & = 0
      \\
      \mathrm{d}\, H_7 & = 
      \tfrac{1}{2}
      F_4 \wedge F_4
      -
      F_2 \wedge F_6
      \mathrlap{\,,}
    \end{aligned}
    &&
    \begin{aligned}
      \mathclap{
         \forall_{
           p \in \{2,4\}
        }
      }
      &
      \\
      \star\, 
      F_{p} 
      & = F_{10-p}
      \\
      \star\, H_3 & = H_7
      \mathrlap{\,,}
    \end{aligned}
  \\[+7pt]
  \label{11DSuGraEoM}
  \adjustbox{
      rndfbox=4pt, scale=0.7
    }{\begin{tabular}{c} 
   \text{11D Maxwell-Chern-Simons} 
   \\
   \text{(gauge sector of 11D SuGra,}
   \\
   \text{
     cf. 
     \parencites
      {CremmerJuliaScherk1978}
      [\S 3.1.3]{MiemiecSchnakenburg2006})
   }
   \end{tabular}
  }
  &&&
    \begin{aligned}
      \mathrm{d}\, G_4 & = 0
      \\
      \mathrm{d}\, G_7 & = 
      \tfrac{1}{2} G_4 \wedge G_4
      \mathrlap{\,,}
    \end{aligned}
    &&
    \star\, G_4 = G_7
    \mathrlap{\,.}
\end{alignat}
\end{subequations}

\emph{Not quite} an example is the flux data in 10D type IIA inside \cref{EoMOfIIA} that is directly obtained by dimensional reduction of the 11D fluxes \cref{11DSuGraEoM}, since that is missing the flux $F_8$ Hodge dual to $F_2$, which has no plain 11D lift (the solution to this issue is discussed in \cite{GSS26-KfromM}):
\begin{equation}
\label{RestrictedEoMOfIIA}
\begin{alignedat}{3}
  \adjustbox{
      rndfbox=4pt, scale=0.7
    }{\begin{tabular}{c} 
    \text{10D IIA NS/RR-Fluxes from 11D} 
    \\
    \text{(gauge sector of 10D IIA SuGra,}
    \\
    \text{cf. \cite[\S 4.2]{MathaiSati04-E8}}
    \end{tabular} 
  }
  &&\phantom{--}&
    \begin{aligned}
      \mathrm{d}\, F_{2}
      & 
      =0
      \\
      \mathrm{d}\, F_4 
      &= 
      H_3 \wedge F_2
      \\
      \mathrm{d}\, F_6 
      & =
      H_3 \wedge F_4
      \\
      \mathrm{d}\, H_3 & = 0
      \\
      \mathrm{d}\, H_7 & = 
      \tfrac{1}{2}
      F_4 \wedge F_4
      -
      F_2 \wedge F_6
      \mathrlap{\,.}
    \end{aligned}
    &\phantom{--}&
    \begin{aligned}
      \star\, F_{4} 
      & = F_{6}
      \\
      \star\, H_3 & = H_7
      \mathrlap{\,,}
    \end{aligned}
\end{alignedat}
\end{equation}

\emph{Not} an example are the equations of motion of nonabelian Yang-Mills theory (\cite{WuYang1975}) with nontrivial structure constants $(f^i_{j k})$, since their Bianchi identities (cf. \parencites[(10.83a)]{Nakahara2018})
\begin{equation}
  \mathrm{d}
  \, F^i
  =
  \sum_{j,k}
  \,
  f^i_{j k}
  \,
  A^j \wedge F^k
\end{equation}
involve \emph{gauge potential} forms $A^i$ that are not in general globally defined as required in \cref{FluxDensitiesOnX}. 

Due to this restriction, we speak of \emph{Maxwell-type} (hence: abelian Yang-Mills type) equations of motion. But beware that, even so, the nonlinear Bianchi identities that we admit, through the polynomials $P^{(i)}$ in \cref{HigherBianchiIdentities}, make such Maxwell-type fluxes in general have quantization laws in \emph{nonabelian cohomology} --- this is the topic of \cref{OnTheGlobalPhaseSpace}.

Also \emph{not} immediately an example of \cref{HigherMaxwellTypeEoM} are Chern-Simons type equations of motion that require flux densities to vanish, $F^{(i)} = 0$. For such gauge fields, their entire non-trivial content is in the flux-quantized gauge potentials, which we come to in \cref{OnTheGlobalPhaseSpace}. But this case is subsumed by a natural generalization which we discuss in \cref{OnChernSimonsType}.

\paragraph
{The Solution Set}

Now consider the special case that spacetime is globally hyperbolic, so that (cf. \parencites[\S 3.11]{MinguzziSanchez2008}[Thm. 1.1]{BernalSanchez2005})
\begin{equation}
  X^{1,d} 
    \simeq 
  \mathbb{R}^{1,0} \times X^d
  \mathrlap{\,,}
\end{equation}
consider a choice of Cauchy surface
\begin{equation}
\label
{CauchySurfaceInclusion}
  \begin{tikzcd} 
    \{0\} \times X^d 
    \ar[r, hook, "{ \iota }"] 
    &
    X^{1,d}  
  \end{tikzcd}
  \mathrlap{\,,}
\end{equation}
and consider its open neighborhoods of the form
\begin{equation}
\label
{TheOpenneighborhoods}
  U_\epsilon 
    :=
  \bracketmid\{{
    (t,x)
    \in 
    \mathbb{R}^{1,0}
    \times
    X^d
  }{
    \epsilon(x)
      > 
    \vert t \vert 
  }\}
  \;\;\;
  \text{for}
  \;\;\;
  \epsilon
  \in
    C^\infty\bracket({
      X^d, \mathbb{R}_+
    })
  \mathrlap{\,.}
\end{equation}

We say that the 
\emph{solution space} of higher Maxwell-type EoM, as in \cref{HigherMaxwellTypeEoM}, 
is the set of \emph{temporal germs} of solutions of these equations around the Cauchy surface:
\begin{equation}
\label
{SolutionSetOfHigherMaxwellEquations}
  \begin{aligned}
  &
  \mathrm{Sol}\bracket({X^d})
  :=
  \\
  &
  \bracketmid\{{
    \epsilon \in 
    C^\infty\bracket({
      X^d, \mathbb{R}_+
    })
    ,\,
    \bracket({
    F^{(i)} \in 
    \Omega^{\mathrm{deg}_i}_{\mathrm{dR}}
    \bracket({
      U_\epsilon
    })
    })_{i \in I}
  }{
  \smash{
    \overbrace{
    \mathrm{d}\, \vec F
    =
    \vec P\scaledbracket({\vec F})
    }^{
      \substack{
       \text{Bianchi identities}
      }
    }
    ,
    \;
    \overbrace{
    \star \vec F
    =
    \vec\mu\scaledbracket({
      \vec F
    })
    }^{
      \substack{
        \text{dualities}
      }
    }
  }
  }
  \}
  \Big/\!\!\sim%
  \mathrlap{\,,}
  \end{aligned}
\end{equation}
where $\bracket({\epsilon,\vec F}) \sim \bracket({\epsilon',\vec F'})$ iff there exists $U_{\epsilon''} \subset U_\epsilon \cap U_{\epsilon'}$
such that $\vec F\vert_{U_{\epsilon''}} = \vec F'\vert_{U_{\epsilon''}}$. 

Observe that restriction (pullback) along the Cauchy surface inclusion \cref{CauchySurfaceInclusion}
gives a map from solutions of the spacetime equations of motion to solutions of \emph{Gauss laws} on the Cauchy surface,
\begin{equation}
\label
{MapToCauchyData}
  \begin{tikzcd}
    \mathrm{Sol}\bracket({X^d})
    \ar[
      r, 
      "{ \iota^\ast }"
    ]
    &
    \bracketmid\{{
      B^{(i)}
      \in
      \Omega^{\mathrm{deg}_i}_{\mathrm{dR}}
      \bracket({
        X^d
      })
    }{%
    \smash{
      \overbrace{
      \mathrm{d}\, \vec B
      =
      \vec P\bracket({ \vec B })
      }^{
        \substack{\text{
          Gauss laws
        }}
      }
    }
    }
    \}%
    \mathrlap{\,.}
  \end{tikzcd}
\end{equation}
where we write
\begin{equation}
  \vec B :=
  \iota^\ast \vec F
  \,,
  \;\;
  \text{hence}
  \;\;
  \bracket\{{
    B^{(i)}
      := 
    \iota^\ast F^{(i)}
  }\}_{i \in I}
\end{equation}
for the restriction of the flux densities to the Cauchy surface. 

\medskip

Next, in \cref{OnCharacteristicLInfinityAlgebra}, we discuss a natural sufficient integrability condition, which ensures that \cref{MapToCauchyData} is a \emph{bijection}, hence that spatial flux data subject to Gauss laws provides unique initial value data for higher Maxwell-type equations of motion. This condition turns out to characterize higher-Maxwell EoMs by an $L_\infty$-algebra, and this characterization is the key to IR-completion of the field theory by flux quantization, in \cref{OnTheGlobalPhaseSpace}.

\subsubsection{
\texorpdfstring
{Characteristic $L_\infty$-Algebra}
{Characteristic L-infinity Algebra}
}
\label
{OnCharacteristicLInfinityAlgebra}

We proceed to further identify this Cauchy data \cref{MapToCauchyData} with closed differential forms with coefficients in a characteristic $L_\infty$-algebra (\cref{OnCharacteristicLInfinityAlgebra}). For this purpose we now require (in \cref{OnUniversalExteriorDifferentialSystems}) that the higher Maxwell-type EoM are suitably universal (which is the case for all the examples \cref{ExamplesOfHigherMaxwellEquations} of interest, but deserves the following formalization).

\paragraph
{Universal Exterior Differential Systems}
\label
{OnUniversalExteriorDifferentialSystems}

Recall that on $X \in \mathrm{SmthMfd}$, an \emph{exterior differential system} (cf. \parencites{BryantEtAl1991}{BryantGriffithsGrossman2003}{McKay2019}) is a dgc-ideal in the de Rham complex $\mathscr{I} \subset \Omega^\bullet_{\mathrm{dR}}(X)$. Its \emph{integral solutions} are submanifold embeddings along which this ideal pulls back to zero, forming the solution set
\begin{equation}
  \mathrm{Sol}\bracket({
    X,\mathscr{I}
  })
  :=
  \bracketmid\{{
    \inlinetikzcd{
      Y \ar[r, hook, "{ \phi }"]
      \&
      X
    }
  }{
    \phi^\ast \mathscr{I} = 0
  }\}
  \mathrlap{\,.}
\end{equation}
This is a way of encoding (solutions to) differential equations with variables in $X$.

Our goal is to formulate \emph{universal} such differential equations which do not depend on the choice of domain $X$ but only on an abstract set $i \in I$ of $\mathrm{deg}_i$-forms \cref{degrees} on any domain. But with the universal de Rham complex (\cref{OnUniversalDeRhamComplex}) in hand, this is immediate: We take such \emph{universal exterior differential systems} to be dgc-ideals $\mathscr{I}$ in the de Rham complex \cref{DeRhamComplexOfSmoothSet} on the smooth moduli set \cref{FluxDensities} of $\mathrm{deg}_i$-forms, $i \in I$:
\begin{equation}
\label
{UniversalExteriorDifferentialSystem}
  \mathscr{I}
  \subset
  \Omega^\bullet_{\mathrm{dR}}\bracket({
    \prod_i
    \mathbf{\Omega}
      ^{\mathrm{deg}_i}
      _{\mathrm{dR}}
    (\ast)
  })
  \mathrlap{\,,}
\end{equation}
and say that the solution set on $\mathbf{X} \in \mathrm{SmthSet}$ is
\begin{equation}
\label
{SolutionSetForUniversalExteriorDiffSystem}
  \mathrm{Sol}\bracket({
    \mathbf{X},
    \mathscr{I}
  })
  :=
  \bracketmid\{{
    \inlinetikzcd{
      \mathbf{X}
      \ar[r, "{ \phi }"]
      \&
      \prod_i
      \mathbf{\Omega}
        ^{\mathrm{deg}_i}
        _{\mathrm{dR}}
      (\ast)      
    }
  }{
    \phi^\ast \mathscr{I} = 0
  }\}
  \mathrlap{\,.}
\end{equation}

Specifically, we say that:
\begin{enumerate}
\item
\emph{Integrable higher Gauss laws} in degrees $\mathrm{deg}_i$, $i \in I$ \cref{degrees} are universal exterior differential systems \cref{UniversalExteriorDifferentialSystem} which are \emph{algebraically} generated%
\footnote{
  To amplify the logic: We \emph{define} \cref{TheDifferentialIdeal} as an algebraic ideal, but then \emph{demand} that this algebraic ideal is already also a differential ideal. This means that the differential of its algebraic generators must vanish modulo the ideal, and this condition is identified as the integrability condition on the $\vec P$, in the proof of \cref{ExistenceOfDifferentialForCharacteristicLAlg} below.
}
as
\begin{equation}
\label
{TheDifferentialIdeal}
  \mathscr{I}
  =
  \bracket({
    \mathrm{d}
    \,
    \mathrm{pr}_i^\ast
    b^{(i)}
    -
    P^{(i)}\bracket({
      \bracket({
        \mathrm{pr}_j^\ast 
        b^{(j)}
      })_{j \in I}
    })
  })_{i \in I}
  \subset
  \Omega^\bullet_{\mathrm{dR}}\bracket({
    \prod_i 
    \mathbf{\Omega}_{\mathrm{dR}}
      ^{\mathrm{deg}_i}(\ast)
  })
  \mathrlap{\,,}
\end{equation}
for graded-symmetric (wedge product) polynomials $P^{(i)}$ in the set $I$ of $\mathrm{deg}$-graded variables, evaluated here on the universal differential forms 
\cref{UniversalDiffFormsOnProduct}.

\item \emph{Flux densities} $\vec B$ \cref{FluxDensities} on $\mathbf{X}$ \emph{solve} the Gauss law \cref{TheDifferentialIdeal} if they are in the corresponding solution set \cref{SolutionSetForUniversalExteriorDiffSystem}, hence if
\begin{equation}
  \vec B^\ast \mathscr{I}
  =
  0
  \in
  \mathbf{\Omega}^\bullet_{\mathrm{dR}}\bracket({
    \mathbf{X}
  })
  \mathrlap{\,,}
\end{equation}
which means equivalently that
\begin{equation}
  \forall_i
  \;:\;
  \mathrm{d}\, B^{(i)}
  =
  P^{(i)}\bracket({
    \bracket\{{
      B^{(j)}
    }\}_{j \in I}
  })
  \mathrlap{\,.}
\end{equation}
\end{enumerate}

\paragraph
{Characteristic $L_\infty$-Algebra}
For given degrees \cref{degrees} $\mathrm{deg}_i \in \mathbb{N}$, let
\begin{equation}
\label{TheAbstractAlgebraGenerators}
  \wedge^1 \mathfrak{a}^\vee
  :=
  \mathbb{R}
  \bracket\langle{
    \beta^{(i)}
  }\rangle
  \in
  \mathrm{Vec}_{_{\mathbb{R}}}^{\mathbb{N}}
\end{equation}
denote  the $\mathbb{N}$-graded vector space spanned by elements $\beta^{(i)}$ in degree $\mathrm{deg}_{i}$, and denote by
\begin{equation}
\label
{CharacteristicGradedGrassmannAlgebra}
  \wedge^\bullet \mathfrak{a}^\vee
  :=
  \mathrm{Sym}\bracket({
   \wedge^1 \mathfrak{a}^\vee
  })
\end{equation}
its graded-symmetric algebra, cf. \cref{GrassmannAlgebraOnDualLieAlgebra}.

\begin{lemma}
\label[lemma]
{ExistenceOfDifferentialForCharacteristicLAlg}
  Given integrable higher Gauss laws \cref{TheDifferentialIdeal}, in degrees $\mathrm{deg}$ \cref{degrees}, a differential $\mathrm{d}$ \textup{(graded derivation of degree +1 that squares to 0)} on $\wedge^\bullet \mathfrak{a}^\vee$ \cref{CharacteristicGradedGrassmannAlgebra} is given on generators by the formulas
  \begin{equation}
    \label{TheEoMsOnFreeGenerators}
    \forall_i 
    \;:\;
    \mathrm{d}\,
    \beta^{(i)}
    =
    P^{(i)}\bracket({
      \bracket\{{
        \beta^{(j)}
      }\}_{j \in I}
    })
    \mathrlap{\,,}
  \end{equation}
  in that it follows that $\mathrm{d}^2 = 0$.
\end{lemma}
\begin{proof}
  By the graded derivation property of $\mathrm{d}$, we are equivalently claiming that the expressions
  \begin{equation}
    \label{JacobiExpressionOnFreeGenerators}
    \mathrm{d}
    \bracket({
    P^{(i)}\bracket({
      \bracket\{{
        \beta^{(j)}
      }\}_{j \in I}
    })
    })
    =
    \sum_{k \in I}
    P^{(k)}
    \frac{\partial}{\partial \beta^{(k)}}
    P^{(i)}\bracket({
      \bracket\{{
        \beta^{(j)}
      }\}_{j \in I}
    })    
  \end{equation}
  vanish for all $i \in I$ (where on the right we have the sum over graded partial derivatives). Since the $\beta^{(j)}$ are free generators, this is the case iff the polynomials on the right vanish identically as functions (for all possible arguments), hence iff
  \begin{equation}
    \label
    {VanishingDifferentialsOfPolynomials}
    \forall_i
    \;\;\;\;
    \sum_{k \in I}
    P^{(k)}
    \frac{\partial}{\partial \beta^{(k)}}
    P^{(i)}
    =
    0
    \mathrlap{\,.}
  \end{equation}
  This is to be shown.
  
  Now, the fact that $\mathscr{I}$ \cref{TheDifferentialIdeal} is a differential ideal implies analogously that, for all $i \in I$: 
  \begin{subequations}
    \begin{align}
    \mathrm{d}
    \bracket({
    P^{(i)}\bracket({
      \bracket\{{
        \mathrm{pr}_j^\ast
        b^{(j)}
      }\}_{j \in I}
    })
    })
    &
    =
    0
    \mod
    \mathscr{I}
    \\
    \Leftrightarrow
    \;\;\;
    \textstyle{\sum_{k \in I}}
    P^{(k)}
    \frac{\partial}{
      \partial
      \mathrm{pr}_k^\ast
      b^{(k)}
    }
    P^{(i)}\bracket({
      \bracket\{{
        \mathrm{pr}_j^\ast
        b^{(j)}
      }\}_{j \in I}
    })    
    &
    =
    0
    \mod
    \mathscr{I}
    \\
    \label
    {PolynomialInUndifferentiatedForms}
    \Leftrightarrow
    \;\;\;
    \textstyle{\sum_{k \in I}}
    P^{(k)}
    \frac{\partial}{
      \partial
      \mathrm{pr}_k^\ast
      b^{(k)}
    }
    P^{(i)}\bracket({
      \bracket\{{
        \mathrm{pr}_j^\ast
        b^{(j)}
      }\}_{j \in I}
    })
    &
    =
    0
    \mathrlap{\,,}
    \end{align}
  \end{subequations}
  where the last step follows since the expression on the left is among the polynomials in the undifferentiated universal differential forms, which have vanishing intersection with the differential ideal $\mathscr{I}$.
  
  Since \cref{PolynomialInUndifferentiatedForms} is an expression in $\Omega^{\mathrm{deg}_i + 2}_{\mathrm{dR}}\bracket({ \prod_j \mathbf{\Omega}^{\mathrm{deg}_j}_{\mathrm{dR}}(\ast) })$, its vanishing implies that for all $U \in \mathrm{SmthMfd}$ and all $\bracket({ B^{(j)} \in \Omega_{\mathrm{dR}}^{\mathrm{deg}_j}(U) })_{j \in I}$ we have
  \begin{equation}
    \label
    {JacobiIdentityOnGivenFluxDensities}
    \sum_{k \in I}
    P^{(k)}
    \frac{\partial}{\partial B^{(k)}}
    P^{(i)}\bracket({
      \bracket\{{
        B^{(j)}
      }\}_{j \in I}
    })    
    =
    0
    \mathrlap{\,.}
  \end{equation}
  In particular, we may take
  $
    U 
      :=
    \mathbb{R}^{N}
  $
  with $N$ so large that there exist differential forms $B^{(j)} \in \Omega^{\mathrm{deg}_j}_{\mathrm{dR}}\bracket(U)$, $j \in I$, with the property that the graded algebra homomorphism
  \begin{equation}
  \label
  {QuasiFreeDifferentialForms}
    \begin{tikzcd}[
      sep=0pt
    ]
      \wedge^\bullet
      \mathfrak{a}^\vee
      \ar[
        rr
      ]
      &&
      \Omega^\bullet_{\mathrm{dR}}
      \bracket({
        \mathbb{R}^N
      })
      \\
      \beta^{(i)}
      &\mapsto&
      B^{(i)}
    \end{tikzcd}
  \end{equation}
  is injective in degrees $\leq \mathrm{max}_i\bracket({ \mathrm{deg}_{i} + 2 })$.
  Such differential forms always exist: 
  \begin{enumerate}
  \item
  For odd $\mathrm{deg}_{j}$, take $B^{(j)}$ to be the canonical volume form on a factor of $\mathbb{R}^{\mathrm{deg}_j}$ inside $\mathbb{R}^N$, 
  \item
  for even $\mathrm{deg}_j$, take $B^{(j)}$ to be the sum of the volume forms of $\lfloor \mathrm{max}_i\bracket({ \mathrm{deg}_i + 2 })/\mathrm{deg}_j \rfloor$ factors of $\mathbb{R}^{\mathrm{deg}_j}$,
  \end{enumerate}
  using distinct factors in each case, with $\mathbb{R}^N$ the product of all the factors needed.

  Specialized to these differential forms \cref{QuasiFreeDifferentialForms},   
  equation
  \cref{JacobiIdentityOnGivenFluxDensities} implies \cref{VanishingDifferentialsOfPolynomials}, thereby the vanishing of \cref{JacobiExpressionOnFreeGenerators} and thus the claim \cref{TheEoMsOnFreeGenerators}.
\end{proof}

With the terminology recalled in \cref{OnLInfinityAlgebras}, we have:
\begin{definition}
\label[definition]
{CharacteristicLInfinityAlgebra}
  The dgc-algebra implied by \cref{ExistenceOfDifferentialForCharacteristicLAlg}, to be denoted
  \begin{equation}
  \label
  {CEAlgebraOfCharacteristicLInftyAlgebra}
    \mathrm{CE}(\mathfrak{a})
    :=
    \bracket({
      \wedge^\bullet\mathfrak{a}^\vee,
      \mathrm{d}
    })
    \in
    \mathrm{dgcAlg}_{_{\mathbb{R}}}\bracket({
      \mathrm{Set}
    })
    \mathrlap{\,,}
  \end{equation}
  is the Chevalley-Eilenberg algebra of an $L_\infty$-algebra
  \begin{equation}
  \label
  {TheCharacteristicLInfinityAlgebra}
    \mathfrak{a}
    \in 
    L_\infty\mathrm{Alg}_{_{\mathbb{R}}}
    \bracket({
      \mathrm{Set}
    })
    \mathrlap{\,,}
  \end{equation}
  which we call the \emph{characteristic $L_\infty$-algebra} of the given pregeometric higher Maxwell-type EoMs.
\end{definition}

\begin{remark}
In summary, this says: 
\begin{standout}
  If the Gauss laws \cref{MapToCauchyData} of a higher Maxwell-type theory \cref{HigherBianchiIdentities} are integrable \cref{TheDifferentialIdeal}, then the restrictions of on-shell flux densities to a Cauchy surface \cref{CauchySurfaceInclusion} form exactly the set of closed differential forms with coefficients in the theory's characteristic $L_\infty$-algebra $\mathfrak{a}$ \cref{TheCharacteristicLInfinityAlgebra}:
  \begin{equation}
  \label
  {GaussLawSolutionsAsClosedDiffForms}
    \bracketmid\{{
      B^{(i)}
      \in
      \Omega^{\mathrm{deg}_i}_{\mathrm{dR}}
      \bracket({
        X^d
      })
    }{%
    \smash{
      \overbrace{
      \mathrm{d}\, \vec B
      =
      \vec P\bracket({ \vec B })
      }^{
        \substack{
          \text{integrable}
          \\
          \text{Gauss laws}
        }
      }
    }
    }
    \}%
    \rule{0pt}{30pt}
    \;\simeq\;
    \Omega^1_{\mathrm{cl}}\bracket({
      X^d;
      \mathfrak{a}
    })
    \mathrlap{\,.}
  \end{equation}
\end{standout}
\end{remark}

The grand conclusion of this section is now the following
\begin{theorem}[{\cite[\S 2.1]{SS24-Phase}}]
\label[theorem]
{GaussLawAsLInfinityClosure}
  The solution set \cref{SolutionSetOfHigherMaxwellEquations} of higher Maxwell-type equations \cref{HigherMaxwellTypeEoM} that are integrable in the above sense \cref{TheDifferentialIdeal}
  is isomorphic to the set of closed differential forms \cref{ClosedLInfinityValuedDifferentialForms} on the Cauchy surface with coefficients in the characteristic $L_\infty$-algebra \cref{TheCharacteristicLInfinityAlgebra}:
  \begin{equation}
  \label
  {TheInitialValueBijection}
  \begin{tikzcd}
    \mathrm{Sol}\bracket({
      X^d
    })
    \ar[
      r, 
      "{ \sim }",
      "{
        \text{\cref{MapToCauchyData}}
      }"{swap}
    ]
    &
    \bracketmid\{{
      B^{(i)}
      \in
      \Omega^{\mathrm{deg}_i}_{\mathrm{dR}}
      \bracket({
        X^d
      })
    }{
      \mathrm{d}\, \vec B
      =
      \vec P\bracket({ \vec B })
    }\}
    \ar[
      r,
      "{ \sim }",
      "{
        \text{\cref{GaussLawSolutionsAsClosedDiffForms}}
      }"{swap}
    ]
    &
    \Omega^1_{\mathrm{cl}}\bracket({
      X^d;
      \mathfrak{a}
    })
    \mathrlap{\,.}
  \end{tikzcd}
  \end{equation}
\end{theorem}
\begin{proof}
  The claim is that the left map \cref{MapToCauchyData} in \cref{TheInitialValueBijection} is indeed a bijection.
  This means: Given Cauchy data $\vec B \in \Omega^1_{\mathrm{cl}}\bracket({ X^d;\, \mathfrak{a} })$, we need to show that it extends uniquely to a germ of a solution $\vec F$ of the higher Maxwell EoMs.

  Choose an orthogonal temporal splitting of $X^{1,d}$ near $X^d$. This makes flux densities decompose uniquely as
  \begin{equation}
    F^{(i)}
    =
    B^{(i)}
    +
    \mathrm{d}t
    \wedge
    E^{(i)}    ,\,
    \;\;
    \text{with}
    \;\;\;
    \left\{
    \substack{
      \iota_{\partial_t} B^{(i)}
       = 0
      \\
      \iota_{\partial_t} 
      E^{(i)}
      = 
      0
    }
    \right.
  \end{equation}
  Imposing on these the constitutive Hodge duality relation \cref{ConstitutiveRelations} gives algebraic relations which make the $\vec E$-fluxes functions of the $\vec B$-fluxes:
  \begin{equation}
  \label{EFluxViaBFlux}
    E^{(i)}
    =
    E^{(i)}
    \bracket({
        \vec B
    })
    \mathrlap{\,.}
  \end{equation}
  This way, the temporal component of the higher Bianchi identities \cref{HigherBianchiIdentities},
  \begin{equation}
    \label
    {TemporalEvolutionEquation}
    \iota_{\partial_t}
    \bracket({
      \mathrm{d} \vec F
        -
      \vec P\bracket({
        \vec F
      })
    })
    =
    0
  \end{equation}
  is a first-order semilinear evolution PDE for the $B^{(i)}$ as functions of $t$.

  Along the lines of \parencites[Prop. 3.1]{DragoGinouxMurro2024}, one sees that this system of differential equations \cref{TemporalEvolutionEquation} is symmetric hyperbolic. This implies (cf. \parencites[p. 51]{Geroch1996}[\S 4.3]{Khavkine2012}) that our initial data $\vec B$ extends to a unique solution of \cref{TemporalEvolutionEquation} on some neighborhood $U_\epsilon$ \cref{TheOpenneighborhoods}. 
  
  From now on, let $\vec F$ denote this unique solution of \cref{EFluxViaBFlux,TemporalEvolutionEquation}, extending our initial value data $\vec B = \iota^\ast \vec F$.
  To  conclude, we need to show that this $\vec F$ already solves the remaining spatial components of the Bianchi identities \cref{HigherBianchiIdentities}, hence that the Gauss law constraint at $t = 0$ is preserved by the time evolution of $\vec F$.

  To this end, abbreviate
  \begin{subequations}
    \begin{align}
    \label
    {TheConstraints}
    C^{(i)}\bracket({\vec F})
    & :=
    \mathrm{d}\, F^{(i)}
      -
    P^{(i)}\bracket({
      \vec F
    })
    \\
    \label
    {dFInTermsOfC}
    \Leftrightarrow
    \;\;\;
    \mathrm{d}\, 
    F^{(i)}
    & =
    \bracket({
      P^{(i)}
        + 
      C^{(i)}
    })\bracket({
      \vec F
    })
    \mathrlap{\,,}
    \end{align}
  \end{subequations}
  hence let $C^{(i)}$ denote the potential remaining failure of $F^{(i)}$ to solve its full equations of motion. In particular its initial value is 
  \begin{equation}
    \label
    {InitialGaussLawAsInitialVanishingOfConstraint}
    \left.
    \vec C\bracket({
      \vec F
    })
    \right\vert_{t = 0}
    =
    0
    \;\;\;\;
    \Leftrightarrow
    \;\;\;\;
    \text{initial Gauss law constraint}
  \end{equation}
  and we need to show that this implies that $\vec C\bracket({\vec F})$ vanishes generally.
  
  Now, the two summands in \cref{dFInTermsOfC} give the following two summands in the differential of $C^{(i)}$:
  \begin{equation}
  \label
  {ComputingDifferentialOfC}
    \begin{aligned}
      -\mathrm{d}\, 
      C^{(i)}
      \bracket({
        \vec F
      })
      & =
      \mathrm{d}\, 
      P^{(i)}\bracket({
        \vec F
      })
      \\
      & =
      \underbrace{
      \sum_{j \in I}
      P^{(j)}\bracket({
        \vec F
      })
      \frac{\partial}{\partial F^{(j)}}
      P^{(i)}
      \bracket({
        \vec F
      })
      }_{
      \left.
      \bracket({
        \mathrm{d}^2 \beta^{(i)}
      })
      \right\vert_{ \vec \beta = \vec F }      
      =
      0
      }
      \,+\,
      \sum_{j \in I}
      C^{(j)}\bracket({
        \vec F
      })
      \frac{\partial}{\partial F^{(j)}}
      P^{(i)}
      \bracket({
        \vec F
      })
    \mathllap{\,,}
    \end{aligned}
  \end{equation}
  where we understand the derivatives as graded derivations.
  Under the brace in \cref{ComputingDifferentialOfC} we recognize the expression \cref{VanishingDifferentialsOfPolynomials},
  shown in terms of the generators \cref{TheAbstractAlgebraGenerators} of $\mathrm{CE}(\mathfrak{a})$ \cref{CEAlgebraOfCharacteristicLInftyAlgebra}, which vanishes by the assumption of integrability (\cref{ExistenceOfDifferentialForCharacteristicLAlg}). 

  It follows that the temporal derivative $\mathcal{L}_{\partial_t}$ of the spatial constraints \cref{TheConstraints} is proportional to these spatial constraints themselves,
  \begin{equation}
    \begin{alignedat}{2}
      \mathcal{L}_{\partial_t}
      C^{(i)}\bracket({
        \vec F
      })
      &
      =
      \mathrm{d}
      \underbrace{
      \iota_{\partial_t}
      C^{(i)}\bracket({
        \vec F
      })
      }_{ 
        = 0 
        \;
        \text{\cref{TemporalEvolutionEquation}}
      }
      +
      \iota_{\partial_t}
      \mathrm{d}\,
      C^{(i)}\bracket({
        \vec F
      })
      &\phantom{--}&
      \substack{
        \text{by Cartan's formula}
      }
      \\
      & =
      -
      \iota_{\partial_t}
      \textstyle{\sum_{j \in I}}
      C^{(j)}\bracket({
        \vec F
      })
      \frac{\partial}{\partial F^{(j)}}
      P^{(i)}
      \bracket({
        \vec F
      })
      &&
      \substack{
        \text{by \cref{ComputingDifferentialOfC}}
      }
      \\
      & =
      \textstyle{\sum_{j \in I}}
      C^{(j)}\bracket({
        \vec F
      })
      (-1)^{\mathrm{deg}_j}
      \bracket({
      \iota_{\partial_t}
      \frac{\partial}{\partial F^{(j)}}
      P^{(i)}
      \bracket({
        \vec F
      })
      })
      &&
      \substack{
        \text{by \cref{TemporalEvolutionEquation}}
      }
    \end{alignedat}
  \end{equation}
  This finally shows that if the initial value data satisfies the Gauss laws, cf.  \cref{InitialGaussLawAsInitialVanishingOfConstraint},
  then its unique temporal evolution $\vec F$, under \cref{EFluxViaBFlux,TemporalEvolutionEquation}, is everywhere already a full solution satisfying $\vec C\bracket({\vec F}) = 0$ \cref{TheConstraints}.
\end{proof}

\begin{remark}
[Two notions of higher gauge theory]
\label[remark]
{TwoNotionsOfHigherGaugeTheory}
Beware (cf. \cref{TwoApproachesToHigherGaugeTheory}) that the role of $L_\infty$-algebras here is different from the ``higher principal connection'' approach to higher gauge theory (\parencites[\S 11-12]{Schreiber2005}{BaezSchreiber2007}{FSSt12-DiffClasses}{SSS12}{SchreiberWaldorf2013}{BorstenEtAl2025}[\S 2]{Alfonsi2025HigherGeometry}{RistSaemannWolf2026}{Waldorf2026}{BunkEtAl2026}): There $L_\infty$-algebras serve as the coefficients of gauge potential forms which are generically not flat/closed --- while here they serve as coefficients of the flux densities and their closure is no less than the electromagnetic Gauss law constraint putting the theory on shell.
\end{remark}

\begin{SCfigure}[.7][htb]
\caption{\label
{TwoApproachesToHigherGaugeTheory} 
There are two approaches to globally modeling ``higher gauge fields'': Either as higher generalization of connections on principal/fiber bundles, or as nonabelian generalization of ordinary differential cohomology. Here we are concerned with the latter (from \cite[\S 9]{FSS23-Char}).
}
\adjustbox{
  rndfbox=4pt, scale=0.9
}{
\hspace{-6pt}
\begin{tblr}{
  colspec={c||c|c},
  row{odd} = {gray!30},
  colsep=3pt
}
  &
  \textbf{Principal}
  &
  \textbf{Cohomological}
  \\
  \hline
  \def\arraystretch{.85}
  \begin{tabular}{@{}c@{}}
   Specified
    \\
   $L_\infty$-algebra
  \end{tabular}
  &
  \def\arraystretch{.85}
  \begin{tabular}{@{}c@{}}
    coefficient of
    \\
    gauge potentials
    \\
    (connections)
  \end{tabular}
  &
  \def\arraystretch{.85}
  \begin{tabular}{@{}c@{}}
    coefficient of
    \\
    flux densities
    \\
    (curvatures)
  \end{tabular}
  \\
  \def\arraystretch{.85}
  \begin{tabular}{@{}c@{}}
    Flatness /
    \\
    closure
  \end{tabular}
  &
  non-generic
  &
  \def\arraystretch{.85}
  \begin{tabular}{@{}c@{}}
    Gauss law /
    \\
    Bianchi identity
  \end{tabular}
  \\
  Examples
  &
  \def\arraystretch{.85}
  \begin{tabular}{@{}c@{}}
    heterotic SuGra /
    \\
    B-field on T-folds
  \end{tabular}
  &
  \def\arraystretch{.85}
  \begin{tabular}{@{}c@{}}
    type I/II SuGra 
    \\
    / 11D SuGra
  \end{tabular}
\end{tblr}
}

\end{SCfigure}

But with the electromagnetic Gauss law constraint on initial Cauchy data of flux densities understood as an $L_\infty$-closure condition, by \cref{GaussLawAsLInfinityClosure}, we have opened the door to understanding the ``quantization'' (discretization) of flux densities, in generality. This is the topic of \cref{OnFluxQuantization}.

\section
{Global Phase Space of Higher Gauge Fields}
\label
{OnTheGlobalPhaseSpace}

\subsection
{Flux Quantization}
\label
{OnFluxQuantization}

We have characterized (\cref{GaussLawAsLInfinityClosure}) the solution space of higher flux densities. This generalizes the historical situation of \cite{Maxwell1865}, when Maxwell's equations for electromagnetic flux were formulated. Later, \cite{Dirac1931} noticed that the electromagnetic \emph{gauge potential} exhibits a global ``quantization'' (discretization) of the admissible total flux, and our task now is to discuss the corresponding generalization of that \emph{flux quantization} via \emph{gauge potentials}.

\subsubsection
{Rational Homotopy}
\label
{OnRationalHomotopy}

\paragraph
{Smooth Moduli Set of On-Shell Fluxes}
In \cref{GaussLawAsLInfinityClosure} we identified the solution set of Maxwell-type equations around a Cauchy surface $X^d$ as the set of closed $\mathfrak{a}$-valued differential forms for a characteristic $L_\infty$-algebra $\mathfrak{a}$:
\begin{equation}
  \mathrm{Sol}\bracket({
    X^d
  })
  \simeq
  \Omega^1_{\mathrm{cl}}\bracket({
    X^d; \mathfrak{a}
  })
  \mathrlap{\,.}
\end{equation}
This means that there is actually a \emph{smooth moduli set} of solutions/closed forms, inside the smooth de Rham complex \cref{SmoothDeRhamComplex} from \cref{OnUniversalDeRhamComplex}:
\begin{equation}
\label
{SmoothModuliSetOfOnShellFLuxDensities}
  \begin{aligned}
    \mathbf{\Omega}^1_{\mathrm{cl}}
    \bracket({
      X^d; \mathfrak{a}
    })
    & 
    \in
    \mathrm{SmthSet}
    \\
    \mathrm{Plt}\bracket({
      \mathbb{R}^n,
      \mathbf{\Omega}^1_{\mathrm{cl}}
      \bracket({
        X^d; 
        \mathfrak{a}
      })
    })
    & 
    :=
    \Omega^1_{\mathrm{cl}}\bracket({
      \mathbb{R}^n
      \times
      X^d;
      \mathfrak{a}
    })
    \mathrlap{\,,}
  \end{aligned}
\end{equation}
hence, by \cref{InternalHomOfSmoothSets}, with:
\smallskip 
\begin{equation}
\label
{ModuliSetOfClosedFormsOnManifold}
  \mathbf{\Omega}^1_{\mathrm{cl}}
  \bracket({
    X^d;
    \mathfrak{a}
  })
  \simeq
  \mathbf{Map}\bracket({
    X^d,
    \mathbf{\Omega}^1_{\mathrm{cl}}
    (\ast; \mathfrak{a})
  })
  \mathrlap{\,.}
\end{equation}

\paragraph
{Shape of On-Shell Flux Moduli}
The shape \cref{TheShapeModality} of this smooth moduli set is, by \cref{ModuliSetOfClosedFormsOnManifold,SmoothOkaPrinciple,ShapeEquivalentToSingularComplex}:
\begin{equation}
\label
{ShapeOfModuliOfClosedForms}
  \begin{aligned}
  \shape
  \mathbf{\Omega}^1_{\mathrm{cl}}
  \bracket({X^d; \mathfrak{a}})
  & \simeq
  \shape
  \mathbf{Map}\bracket({
    X^d,
    \mathbf{\Omega}^1_{\mathrm{cl}}
    (\ast,\mathfrak{a})
  })
  \\
  & \simeq
  \mathbf{Map}\bracket({
    \shape
    X^d,
    \shape
    \mathbf{\Omega}^1_{\mathrm{cl}}
    (\ast,\mathfrak{a})
  })
  \\
  & \simeq
  \mathrm{Dsc}
  \,
  \mathrm{Map}\bracket({
    \mathrm{Sing}\, X^d,
    \mathrm{Sing}\,
      \mathbf{\Omega}^1_{\mathrm{cl}}
      (\ast;\mathfrak{a})
  })
  \mathrlap{\,,}
  \end{aligned}
\end{equation}
with
\smallskip 
\begin{equation}
  \begin{aligned}
  \mathrm{Sing}\bracket({
    \mathbf{\Omega}^1_{\mathrm{cl}}(\ast)
  })
  & 
  \in
  \mathrm{SSet} \to 
  \mathrm{Grpd}_\infty
  \\
  \mathrm{Plt}\bracket({
  \Delta[k],
  \mathrm{Sing}\bracket({
    \mathbf{\Omega}^1_{\mathrm{cl}}(\ast)
  })
  })
  &=
  \Omega^1_{\mathrm{cl}}\bracket({
    \boldsymbol{\Delta}^k;
    \mathfrak{a}
  })
  \mathrlap{\,.}
  \end{aligned}
\end{equation}

\paragraph
{Rationalizing Classifying Spaces}
This \emph{simplicial set of closed $\mathfrak{a}$-valued differential forms} is a famous construction in \emph{rational homotopy theory} (RHT, for which cf.  \parencites{FHT2000}{Hess2007}[\S 5]{FSS23-Char}):

For a simply connected $\mathcal{A} \in \mathrm{Grpd}_\infty$, its \emph{$\mathbb{R}$-rationalization} is $L^{\mathbb{R}} \mathcal{A} \in \mathrm{Grpd}_{\infty}$ whose homotopy groups are the $\mathbb{R}$-rationalizations of those of $\mathcal{A}$:
\begin{equation}
  \forall_{n \geq 2}
  \;:\;
  \pi_n\bracket({
    L^{\mathbb{R}}
    \mathcal{A}
  })
  \simeq
  \pi_n\bracket({
    \mathcal{A}
  })
  \otimes_{_{\mathbb{Z}}}
  \mathbb{R}
\end{equation}
and which is equipped with a map
\begin{equation}
\label
{RRationalizationUnit}
  \begin{tikzcd}
    \mathcal{A}
    \ar[
      r,
      "{
        \eta^{\mathbb{R}}
          _{\mathcal{A}}
      }"
    ]
    &
    L^{\mathbb{R}}
    \mathcal{A}
  \end{tikzcd}
\end{equation}
that induces the canonical homomorphisms, 
\begin{equation}
  \begin{tikzcd}
    \pi_n\bracket({
      \mathcal{A}
    })
    \ar[
      r,
      "{
        \scaledbracket({
          \eta^{\mathbb{R}}
          _{\mathcal{A}}
        })_\ast
      }"
    ]
    &
    \pi_n\bracket({
      \mathcal{A}
    })
    \otimes_{_{\mathbb{Z}}}
    \mathbb{R}
  \end{tikzcd}
\end{equation}
and is suitably universal with that property. 

Part of the \emph{fundamental theorem of dg-algebraic RHT} says (cf. \cite[Prop. 5.13]{FSS23-Char}) that for a simply connected $\mathcal{A}$ whose rational cohomology/homotopy groups are degreewise finite-dimensional, there exists an essentially unique minimal $L_\infty$-algebra 
\begin{equation}
\label
{RealWhiteheadBracketLInfinityAlgebra}
  \begin{aligned}
    \mathfrak{l}
    \mathcal{A}
    &
    \in
    L_\infty\mathrm{Alg}
      _{_{\mathbb{R}}}
    \\
    \bracket({
      \mathfrak{l}
      \mathcal{A}
    })_\bullet
    & 
    \simeq
    \pi_\bullet\bracket({
      \Omega \mathcal{A}
    })
    \otimes_{_{\mathbb{Z}}}
    \mathbb{R}
    \mathrlap{\,,}
  \end{aligned}
\end{equation}
whose brackets are the $\mathbb{R}$-rationalized \emph{higher Whitehead brackets} of $\mathcal{A}$ (and whose CE-algebra is called the \emph{minimal Sullivan model} of $\mathcal{A}$)
such that:%
\footnote{
  The equivalence \cref{FundamentalTheoremOfDGAlgRHT} is 
  traditionally discussed for polynomial differential forms on bounded simplices; but it holds just as well for the smooth differential forms on extended simplices used here, since these still satisfy (cf. \cite[Cor. 9.9]{GriffithsMorgan2013}) the relevant \emph{extension lemma} (cf. \parencites[Prop. 5.10]{FSS23-Char}{nLab:ExtensionLemmaForDifferentialForms}).
}
\begin{equation}
\label
{FundamentalTheoremOfDGAlgRHT}
  L^{\mathbb{R}}
  \mathcal{A}
  \simeq
  \mathrm{Sing}\,
  \mathbf{\Omega}^1_{\mathrm{cl}}
  (\ast; \mathfrak{l}\mathcal{A})
  \mathrlap{\,.}
\end{equation}
Regarding this under the tacit faithful embedding $\mathrm{Dsc} : \inlinetikzcd{ \mathrm{Grpd}_\infty \ar[r, hook] \& \mathrm{SmthGrpd}_\infty }$, it reads equivalently:
\begin{equation}
  L^\mathbb{R}
  \mathcal{A}
  \simeq
  \shape
  \mathbf{\Omega}^1_{\mathrm{cl}}
  \bracket({
    \ast; 
    \mathfrak{l}\mathcal{A}
  })
  \mathrlap{\,.}
\end{equation}
In this form it immediately generalizes, by the smooth Oka principle \cref{SmoothOkaPrinciple} and using \cref{ModuliSetOfClosedFormsOnManifold}, to:%
\footnote{
   Our one-line proof, via cohesion, of
   \cref{ParameterizedFundamentalRHTTheorem} 
   reproduces the result of \cite[Thm. 1.4]{Berglund2015} (cf. also
   \parencites[Thm. 8.1]{Lazarev2013}[Thm. 3.1]{BuijsFelixMurillo2011})
   for the special case when the domain $X$ is a smooth manifold and for rationalization over the real numbers.
}
\begin{equation}
\label
{ParameterizedFundamentalRHTTheorem}
  \left.
  \substack{
    X \,\in\, \mathrm{SmthMfd}
  }
  \right\}
  \;\;\;
  \Rightarrow
  \;\;\;
  \mathbf{Map}\bracket({
    X
    ,
    L^{\mathbb{R}}
    \mathcal{A}
  })
  \simeq
  \shape
  \mathbf{\Omega}^1_{\mathrm{cl}}
  \bracket({
    X;
    \mathfrak{l}\mathcal{A}
  })
  \mathrlap{\,.}
  \hphantom{
  \left.
  \substack{
    X \,\in\, \mathrm{SmthMfd}
  }
  \right\}
  \;\;\;
  \Rightarrow  
  }
\end{equation}

\paragraph
{Restriction to Solitonic Flux Densities}
An important variant of \cref{ParameterizedFundamentalRHTTheorem} is  \cref{OnShapeOfSolitonicFluxesIsPointedMaps} below.
The \emph{solitonic} on-shell flux densities are those on-shell flux densities \cref{SmoothModuliSetOfOnShellFLuxDensities} which on a Cauchy surface $X$ are \emph{compactly supported} (in particular \emph{vanishing at infinity}) according to \cref{DifferentialFormModuliVanishingAtInfinity}: 
\begin{equation}
\label
{SolitonicOnShellFluxDensities}
  \mathbf{\Omega}^1_{\mathrm{cl},\mathrm{c}}
  \bracket({
    X
    ;\,
    \mathfrak{a}
  })
  :=
  \varinjlim_{K \in \mathrm{Cpt}(X)}
  \,
  \mathrm{fib}
  \Big(
    \inlinetikzcd{
      \mathbf{\Omega}^1_{\mathrm{cl}}
      \bracket({
        X; \mathfrak{a}
      })
      \ar[
        r,
        "{ \iota_K^\ast }"
      ]
      \&
      \mathbf{\Omega}^1_{\mathrm{cl}}
      \bracket({
        X \setminus K^\circ; \mathfrak{a}
      })      
    }
  \Big)
  \mathrlap{\,.}
\end{equation}

We assume now that the Cauchy surface manifold $X$ is non-pathological at infinity. One way of saying this is that it is \emph{finitary} in that it is the interior of a compact manifold with boundary.

In this case the compactly supported maps out of $X$ with coefficients in a pointed space $\mathcal{Y}$,
\begin{equation}
\label
{CompactlySupportedMappingSpace}
  \mathbf{Map}_c\bracket({
    X,
    \mathcal{Y}
  })
  :=
  \varinjlim_{K \in \mathrm{Cpt}(X)}
  \mathrm{fib}\Big(
    \inlinetikzcd{
      \mathbf{Map}\bracket({
        X, \mathcal{Y}
      })
      \ar[
        r,
        "{ \iota^\ast_K }"
      ]
      \&
      \mathbf{Map}\bracket({
        X - K^\circ, \mathcal{Y}
      })
    }
  \Big)
  \mathrlap{\,,}
\end{equation}
are equivalently (cf. \parencites[p. 6]{AyalaFrancis2025}) the \emph{pointed} maps out of the \emph{one-point compactification} $X_{\cpt}$ (cf. \parencites[\S 3]{James1984}[\S 2.2]{SS25-Flux}):
\begin{equation}
\label
{CompactlySupportedMapsEquivalentToPointedMaps}
  \left.
  \substack{
    \text{finitary } X \in \mathrm{SmthMfd},
    \\
    \mathcal{Y} \in \mathrm{Grpd}^\ast_\infty
  }
  \right\}
  \;\;\Rightarrow\;\;
  \inlinetikzcd{
    \mathbf{Map}_c\bracket({
      X,
      \mathcal{Y}
    })
    \ar[
      r,
      "{ \sim }"
    ]
    \&
    \mathbf{Map}^\ast\bracket({
      X_{\cpt},
      \mathcal{Y}
    })
    \mathrlap{\,.}
  }
\end{equation}

\begin{proposition}
\label[proposition]
{OnShapeOfSolitonicFluxesIsPointedMaps}
  The shape of the moduli of solitonic on-shell flux densities \cref{SolitonicOnShellFluxDensities} on a finitary Cauchy surface $X$ is the \emph{pointed} mapping space:%
  \footnote{
    Our quick proof, via cohesion, of \cref{ShapeOfSolitonicFluxesIsPointedMaps} reproduces essentially the result of \cite[Thm. 1.2]{BuijsFelixMurillo2011} for the special case that the domain $X$ is a smooth manifold and for rationalization over the real numbers.
  }
  \begin{equation}
    \label
    {ShapeOfSolitonicFluxesIsPointedMaps}
    \mathbf{Map}^\ast\bracket({
      X_{\cpt},
      L^{\mathbb{R}}
      \mathcal{A}
    })
    \simeq
    \shape
    \mathbf{\Omega}^1_{\mathrm{cl},\mathrm{c}}
    \bracket({
      X
      ;\,
      \mathfrak{l}\mathcal{A}
    })
    \mathrlap{\,.}
  \end{equation}
\end{proposition}
\begin{proof}
We have the following sequence of natural equivalences:
\stepcounter{footnote}
\begin{alignat*}{2}
    &
    \shape 
    \mathbf{\Omega}^1_{\mathrm{cl},\mathrm{c}}
    \bracket({
      X^d;
      \mathfrak{l}
      \mathcal{A}
    })
    &\;\;&
    \\
    &
    \defneq
    {\mathcolor{purple}{\shape}}
    \,
      \varinjlim_{K}
      \,
      \mathrm{fib}
      \Big({
        \inlinetikzcd{
          \mathbf{\Omega}^\bullet_{\mathrm{cl}}
          \bracket({
            X; 
            \mathfrak{l}\mathcal{A}
          })
          \ar[
            r,
            "{ 
              \iota^\ast_K
            }"
          ]
          \&
          \mathbf{\Omega}^\bullet_{\mathrm{cl}}
          \bracket({
            X \setminus K^\circ;
            \mathfrak{l}\mathcal{A}
          })
        }
      }\Big)
      &&
      \substack{
       \text{ 
         by \cref{SolitonicOnShellFluxDensities}
      }
      }
     \\
     & \simeq
      \,
      \varinjlim_{K}
      \,
      {\mathcolor{purple}{\shape}}
      \mathrm{fib}
      \Big({
        \inlinetikzcd{
          \mathbf{\Omega}^\bullet_{\mathrm{cl}}
          \bracket({
            X;
            \mathfrak{l}\mathcal{A}
          })
          \ar[
            r,
            "{ 
              \iota^\ast_K
            }"
          ]
          \&
          \mathbf{\Omega}^\bullet_{\mathrm{cl}}
          \bracket({
            X \setminus K^\circ;
            \mathfrak{l}\mathcal{A}
          })
        }
      }\Big)
      &&
      \begin{subarray}{l}
        \text{since $\shape$}
        \\
        \text{preserves} 
        \\
        \text{colimits}
      \end{subarray}
     \\
     & \simeq
      \,
      \varinjlim_{K}
      \,
      \mathrm{fib}
      \Big({
        \inlinetikzcd{
          {\mathcolor{purple}{\shape}}
          \mathbf{\Omega}^\bullet_{\mathrm{cl}}
          \bracket({
            X;
            \mathfrak{l}\mathcal{A}
          })
          \ar[
            r,
            "{ 
              \iota^\ast_K
            }"
          ]
          \&
          {\mathcolor{purple}{\shape}}
          \mathbf{\Omega}^\bullet_{\mathrm{cl}}
          \bracket({
            X \setminus K^\circ;
            \mathfrak{l}\mathcal{A}
          })
        }
      }\Big)
      &&
      \substack{
      \text{ 
        using \cref{ShapePreservesFiberProductsWithSingFibrantMaps}%
      \footnotemark[\value{footnote}]
      }}
      \\
      & \simeq
      \,
      \varinjlim_{K}
      \,
      \mathrm{fib}
      \Big({
        \inlinetikzcd{
          \mathbf{Map}\bracket({
            X^d,
            L^{\mathbb{R}}\mathcal{A}
          })
          \ar[
            r,
            "{ 
              \iota^\ast_K
            }"
          ]
          \&
          \mathbf{Map}\bracket({
            X^d \setminus K^\circ,
            L^{\mathbb{R}}\mathcal{A}
          })
        }
      }\Big)
      &&
      \substack{
      \text{
        by \cref{ParameterizedFundamentalRHTTheorem}
      }}
      \\
      & 
      \simeq
      \mathbf{Map}^\ast\bracket({
        X^d_{\cpt},
        L^{\mathbb{R}}\mathcal{A}
      })
      &&
      \substack{
      \text{
        by \cref{CompactlySupportedMapsEquivalentToPointedMaps}.
      }}
      \;\;\;\;\;\;\;\;\;\;\;\;\;\;\qedhere
\end{alignat*}
\end{proof}

\footnotetext[\value{footnote}]{
  \label{AcyclicDGCAFibrationArgument}
  That
  $
    \mathrm{Sing}
    \big({
    \inlinetikzcd{
      \iota^\ast_K 
      :
      \mathbf{\Omega}^1_{\mathrm{cl}}
      \bracket({
        X;
        \mathfrak{l}\mathcal{A}
      })
      \ar[
        r
      ]
      \&
      \mathbf{\Omega}^1_{\mathrm{cl}}
      \bracket({
        X \setminus K^\circ;
        \mathfrak{l}\mathcal{A}
      })
    }
    }\big)
  $
  is a Kan fibration---as needed for applying \cref{ShapePreservesFiberProductsWithSingFibrantMaps} in the proof of \cref{OnShapeOfSolitonicFluxesIsPointedMaps}---may be seen by a relative form of the extension lemma (\parencites[Lem. 9.4 \& proof of Cor. 9.9]{GriffithsMorgan2013}{nLab:ExtensionLemmaForDifferentialForms}): This shows that for $n \in \mathbb{N}_+$ and $0 \leq k \leq n$, the restriction 
  $
    \inlinetikzcd{
    \Omega^\bullet_{\mathrm{dR}}\bracket({
      X
      \times
      \boldsymbol{\Delta}^n
    })
    \ar[r]
    \&
    \Omega^\bullet_{\mathrm{dR}}\bracket({
      X
      \times
      \boldsymbol{\Lambda}^n_k
    })
    \times_{
      \Omega^\bullet_{\mathrm{dR}}
      \scaledbracket({
        \scaledbracket({
          X - K^\circ
        })
        \times
        \boldsymbol{\Lambda}^n_k
      })    
    }
    \Omega^\bullet_{\mathrm{dR}}\bracket({
      \scaledbracket({
        X - K^\circ
      })
      \times
      \boldsymbol{\Delta}^n
    })
    }
  $
  (of de Rham complexes to the $\iota_K$-relative $k$-horns of our extended smooth $n$-simplices \cref{ExtendedSmoothSimplices})
  is degreewise surjective. It is also a quasi-isomorphism, as exhibited by the de Rham homotopy operator for contracting $\boldsymbol{\Delta}^n$ to its $k$-th vertex. Therefore this restriction is an acyclic fibration in the Bousfield-Gugenheim model structure on dgc-algebras, with the right lifting property against the cofibrant algebra $\mathrm{CE}\bracket({\mathfrak{l}\mathcal{A}})$. This lifting provides the required horn fillers.
}

\paragraph
{Admissible Classifying Spaces}
We hence ask for classifying spaces $\mathcal{A}$ whose real Whitehead $L_\infty$-algebra $\mathfrak{l}\mathcal{A}$ \cref{RealWhiteheadBracketLInfinityAlgebra} coincides up to isomorphism
(cf. \cref{UnderstandingTheAdmissibilityCondition} below)
with the characteristic $L_\infty$-algebra $\mathfrak{a}$ (\cref{CharacteristicLInfinityAlgebra}) of a given Maxwell-type higher gauge theory. 
For example:
\begin{enumerate}
\item
\colorbox{lightgray}{{$\mathfrak{l}
  B^2 \mathbb{Z}^2
$}}, cf. \cref{TheEMSpaces},
is the characteristic $L_\infty$-algebra of 4D vacuum electromagnetism \cref{EoMOf4DVacuumElectromagnetism}:
\begin{equation}
\label
{CharacteristicOfVacuumElectromagnetism}
  \mathrm{CE}\bracket({
    \mathfrak{l}B^2\mathbb{Z}^2
  })
  \simeq
  \mathbb{R}_{_{\mathrm{d}}}
  \left[
  \begin{aligned}
    & f_2
    \\
    & g_2
  \end{aligned}
  \right]
  \Big/
  \left(
  \begin{aligned}
    \mathrm{d}\, 
    f_2 & = 0
    \\
    \mathrm{d}\,
    g_2 & = 0
  \end{aligned}
  \right)
  \mathrlap{.}
\end{equation}

\item
\colorbox{lightgray}{{$\mathfrak{l} B^{2k+1} \mathbb{Z}$}}, cf. \cref{TheEMSpaces}, 
is the characteristic $L_\infty$-algebra of the chiral flux in $D = 4k+2$ \cref{ChiralFluxEoM}:
\begin{equation}
  \mathrm{CE}\bracket({
    \mathfrak{l}
    B^{2k+1}\mathbb{Z}
  })
  \simeq
  \mathbb{R}_{_{\mathrm{d}}}
  \big[
    c_{2k+1}
  \big]
  \big/
  \big(
    \mathrm{d}\,
    c_{2k+1}
    = 0
  \big)
  \mathrlap{.}
\end{equation}

\item
\colorbox{lightgray}{{$\mathfrak{l}S^2$}}, cf. \cref{Cohomotopy}, is the characteristic $L_\infty$-algebra of 5D Maxwell-Chern-Simons theory \cref{EoMOf5DMCSTheory}:
\begin{equation}
  \mathrm{CE}\bracket({
    \mathfrak{l}S^2
  })
  \simeq
  \mathbb{R}_{_{\mathrm{d}}}
  \left[
  \begin{aligned}
    & f_2
    \\
    & g_3
  \end{aligned}
  \right]
  \Big/
  \left(
  \begin{aligned}
    \mathrm{d}\, f_2
    & =
    0
    \\
    \mathrm{d}\, g_3
    & =
    \tfrac{1}{2} f_2^2
  \end{aligned}
  \right)
  \mathrlap{.}
\end{equation}
\item
\colorbox{lightgray}{{$\mathfrak{l}\bracket({\truncation{9}{B \mathrm{U}} \sslash B\mathrm{U}(1)})$}}, cf. \cref{ClassifyingFibrationForTwistedKTheory,9TruncationOfBU}, is almost the characteristic $L_\infty$-algebra of type IIA supergravity \cref{EoMOfIIA}, but is missing $H_7$ and its nonlinear Bianchi identity:
\begin{equation}
\label
{CharacteristicOfIIA}
  \mathrm{CE}\bracket({
    \mathfrak{l}\bracket({
      \truncation{9}{B \mathrm{U}} 
        \sslash 
      B\mathrm{U}(1)})
  })
  \simeq
  \mathbb{R}_{_{\mathrm{d}}}
  \left[
  \begin{aligned}
    &
    h_3
    \\
    & 
    f_{2}
    \\
    & 
    f_{4}
    \\
    & 
    f_{6}
    \\
    &
    f_{8}
  \end{aligned}
  \right]
  \Big/
  \left(
  \begin{aligned}
    \mathrm{d}\,
    h_3 
    & =
    0
    \\
    \mathrm{d}\, f_2 
    &
    = 0
    \\
    \mathrm{d}\, f_{4}
    &=
    h_3 f_{2}
    \\
    \mathrm{d}\, f_{6}
    &=
    h_3 f_{4}
    \\
    \mathrm{d}\, f_{8}
    &=
    h_3 f_{6}
  \end{aligned}
  \right)
  \mathrlap{.}
\end{equation}

\item
\colorbox{lightgray}{{$\mathfrak{l}
  \mathrm{Cyc}\bracket({
    S^4
  })
  $}}, cf. \cref{Cyclification}, is the characteristic $L_\infty$-algebra of the IIA flux sector obtained from 11D \cref{RestrictedEoMOfIIA}, including $H_7$ and its Bianchi identity:
\begin{equation}
  \mathrm{CE}\bracket({
    \mathfrak{l}
    \mathrm{Cyc}\bracket({
      S^4
    })
  })
  \simeq
  \mathbb{R}_{_{\mathrm{d}}}
  \left[
  \begin{aligned}
    f_2
    \\
    f_4
    \\
    f_6
    \\
    h_3
    \\
    h_7
  \end{aligned}
  \right]
  \Big/
  \left(
  \begin{aligned}
    \mathrm{d}\, f_2 & = 0
    \\
    \mathrm{d}\, f_4 & = 
    h_3 f_2
    \\
    \mathrm{d}\, f_6 & =
    h_3 f_4
    \\
    \mathrm{d}\, h_3 & = 0
    \\
    \mathrm{d}\, h_7 & =
    \tfrac{1}{2}
    f_4 f_4 - 
    f_2 f_6
  \end{aligned}
  \right)
  \mathrlap{.}
\end{equation}
\item
\colorbox{lightgray}{{$\mathfrak{l}S^4$}}, cf. \cref{Cohomotopy}, is the characteristic $L_\infty$-algebra of 11D supergravity \cref{11DSuGraEoM}:
\begin{equation}
\label
{CEAlgebraOflS4}
  \mathrm{CE}\bracket({
    \mathfrak{l}
    S^4
  })
  \simeq
  \mathbb{R}_{_{\mathrm{d}}}
  \left[
  \begin{aligned}
    & g_4
    \\
    & g_7
  \end{aligned}
  \right]
  \Big/
  \left(
  \begin{aligned}
    \mathrm{d}\, 
    g_4 & = 0
    \\
    \mathrm{d}\,
    g_7 & =
    \tfrac{1}{2}
    g_4^2
  \end{aligned}
  \right)
  \mathrlap{.}
\end{equation}
\end{enumerate}

But beware that $\mathfrak{l}(-)$ forgets a lot of structure, namely all the ``torsion information''. Accordingly, there is always an infinitude of inequivalent classifying spaces $\mathcal{A}$ with the same rationalization $\mathfrak{l}\mathcal{A}$. 
For example: 
\begin{equation}
  \mathfrak{l}\bracket({
    \mathcal{A}
    \times 
    B^2 \bracket({
      \mathbb{Z}_{/n}
    })
  })
  \simeq
  \mathfrak{l}\mathcal{A}
  \mathrlap{\,.}
\end{equation}

Accordingly, when we come to \emph{flux quantization} in \cref{OnNonabelianCharacter}, the actual choice of classifying space $\mathcal{A}$, which is \emph{admissible}, in that there is an isomorphism
\begin{equation}
\label
{AdmissibleFluxQuantizationCondition}
  \mathfrak{l}\mathcal{A}
  \simeq
  \mathfrak{a}
\end{equation}
for the given characteristic Gauss law $L_\infty$-algebra $\mathfrak{a}$, is a substantial datum that needs to be provided/chosen in order to globally complete the higher gauge field theory (cf. also \cref{UnderstandingTheAdmissibilityCondition} below).

\subsubsection
{Nonabelian Character}
\label
{OnNonabelianCharacter}

By combining the facts from \cref{OnRationalHomotopy}, we obtain the natural \emph{differential character maps}
\begin{equation}
\label
{DifferentialCharacterMap}
  \left.
  \substack{
    X \,\in\, \mathrm{SmthMfd}
    \\
    \mathcal{A} \,\in\,
    \mathrm{Grpd}
      ^{\mathrm{ft}_{\smash{\mathbb{Q}}}}
      _{[2,\infty]}
  }
  \right\}
  \;\;
  \vdash
  \;\;
  \begin{tikzcd}[column sep=40pt]
    \mathbf{Map}\bracket({
      X,
      \mathcal{A}
    })
    \ar[
      r,
      "{
        \scaledbracket({
          \eta^{\mathbb{R}}_{\mathcal{A}}
        })_\ast
      }",
      "{
        \text{\cref{RRationalizationUnit}}
      }"{swap}
    ]
    \ar[
      rr,
      uphordown,
      "{
        \mathbf{ch}^{\mathcal{A}}_X
      }"{description}
    ]
    &
    \mathbf{Map}\bracket({
      X,
      L^{\mathbb{R}}
      \mathcal{A}
    })
    \ar[
      r,
      "{ \sim }",
      "{
        \text{
     \cref{ParameterizedFundamentalRHTTheorem}
        }
      }"{swap}
    ]
    &
    \shape\mathbf{\Omega}^1_{\mathrm{cl}}
    \bracket({
      X;
      \mathfrak{l}\mathcal{A}
    })
  \end{tikzcd}
\end{equation}
(where on the left we mean that $\mathcal{A}$ is simply connected and of rational finite type).

These maps hence extract the underlying differential form (flux density!) shadow from maps to the classifying space (charges!).

To make this interpretation more tangible,
consider:
\begin{enumerate}
\item
The \emph{nonabelian cohomology} with coefficients in the loop $\infty$-group of $\mathcal{A}$ is
\begin{equation}
\label
{NonabelianCohomologyInDegree1}
  H^1\bracket({
    X;
    \Omega\mathcal{A}
  })
  :=
  \pi_0\,
  \mathbf{Map}\bracket({
    \shape X,
    \mathcal{A}
  })
\end{equation}
and, more generally, if $\mathcal{A}$ is not necessarily connected:
\begin{equation}
\label
{NonabelianCohomologyInDegree0}
  H^0\bracket({
    X;
    \mathcal{A}
  })
  :=
  \pi_0\,
  \mathbf{Map}\bracket({
    \shape X,
    \mathcal{A}
  })
  \mathrlap{\,.}
\end{equation}
(\parencites[Def. 6.0.6]{Toen2002}[Def. 6]{Lurie2014}[\S 2]{FSS23-Char}, more exposition in \parencites[\S 1]{SS25-TEC}[\S 4]{SS26-Orb}).

This subsumes:
\begin{enumerate}
\item 
Ordinary nonabelian cohomology with coefficients in a Lie group $G$, taking $\mathcal{A} \defneq B G$,
\begin{equation}
  H^1\bracket({
    X; G
  })
  \simeq
  \pi_0 \, 
  \mathrm{Map}\bracket({
    X, B G
  })
  \mathrlap{\,,}
\end{equation}
(cf. \cite[Ex. 2.2]{FSS23-Char}).

\item 
Ordinary abelian cohomology with coefficients $A \in \mathrm{AbGrp}$, taking $\mathcal{A} \defneq B^n A$:
\begin{equation}
  \begin{aligned}
  H^n\bracket({
    X;
    A
  })
  &
  \simeq
  H^1\bracket({
    X
    ;
    \Omega
    B^n A
  })
  \end{aligned}
\end{equation}
(cf. \cite[Ex. 2.1]{FSS23-Char}).

\item 
Complex topological K-theory at rank-0, taking $\mathcal{A} \defneq B \mathrm{U}$
\begin{equation}
  K^0_{\mathrm{rk}=0}\bracket({
    X
  })
  \simeq
  H^1\bracket({
    X;
    \Omega
    B \mathrm{U}
  })
  \mathrlap{\,,}
\end{equation}
(cf. \cite[Ex. 2.11]{FSS23-Char}).

\item 
Any Whitehead-generalized abelian cohomology theory $E^n(-)$, taking $\mathcal{A} \defneq E_n$ to be a stage in the corresponding spectrum of classifying spaces:
\begin{equation}
  E^n\bracket({X})
  \simeq
  H^0\bracket({
    X;
    E_n
  })
\end{equation}
(cf. \parencites[\S III.2]{Adams1974}[\S 3.4]{Kochman1996}[\S 12]{AguilarGitlerPrieto2002}{BarnesRoitzheim2020}[Ex. 2.10]{FSS23-Char}).

\item
\emph{$n$-Cohomotopy}, taking $\mathcal{A} \defneq S^n$ (\parencites{Borsuk1936}{Pontrjagin1938}{Spanier1949}, cf. \parencites[\S VII]{STHu59}[Ex. 2.7]{FSS23-Char}):
\begin{equation}
  \pi^n(X)
  :=
  H^1\bracket({
    X; \Omega S^n
  })
  \mathrlap{\,.}
\end{equation}
\end{enumerate}
\item
The \emph{nonabelian de Rham cohomology}
\cite[Def. 6.3]{FSS23-Char} with coefficients in $\mathfrak{a} \in L_\infty \mathrm{Alg}^{\mathrm{ft}}_{\mathbb{R}}$ is the set of \emph{concordance classes} of closed forms:
\begin{equation}
\label
{NonabelianDeRhamCohomology}
  \begin{aligned}
  H^1_{\mathrm{dR}}\bracket({
    X;
    \mathfrak{a}
  })
  &
  :=
  \Omega^1_{\mathrm{cl}}\bracket({
    X;
    \mathfrak{a}
  })
  \big/
  \mathrm{concordance}\,.
  \end{aligned}
\end{equation}
This means that a pair of closed $\mathfrak{a}$-valued forms on $X^d$ is in the same cohomology class iff they are the boundary values of a closed $\mathfrak{a}$-valued form on $[0,1] \times X^d$.

By the discussion in \cref{OnEquationsOfMotion}, this means that the nonabelian de Rham class of on-shell flux densities is preserved by time evolution, hence reflects \emph{conserved total flux} (\cite[\S 3.1]{SS25-Flux} cf. \cref{ConservedTotalFlux}):

\begin{figure}[htb]
\caption{\label{ConservedTotalFlux}
Time evolution of on-shell flux is a
\emph{concordance} of Cauchy data. The concordance class in nonabelian de Rham cohomology \cref{NonabelianDeRhamCohomology} is the \emph{conserved total flux}.
}
\centering
\adjustbox
{rndfbox=4pt, scale=0.9}{
  \includegraphics
   [width=13cm]
   {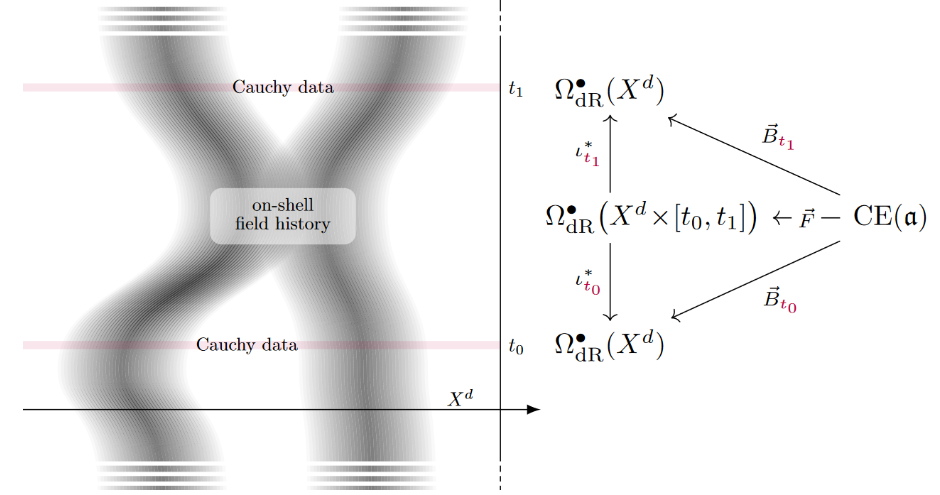}
}
\end{figure}

For $\mathfrak{a} \defneq \mathfrak{l}B^n \mathbb{Z}$ this definition \cref{NonabelianDeRhamCohomology} reproduces \cite[Prop. 6.4]{FSS23-Char} ordinary de Rham cohomology:
\begin{equation}
  H^1_{\mathrm{dR}}\bracket({
    X;
    \mathfrak{l}B^n \mathbb{Z}
  })
  \simeq
  H^n_{\mathrm{dR}}\bracket(X)\,.
\end{equation}

Generally, for $\mathfrak{a} \simeq \mathfrak{l}\mathcal{A}$ \cref{RealWhiteheadBracketLInfinityAlgebra},
the following \emph{nonabelian de Rham theorem} \cite[Thm. 6.5]{FSS23-Char}  naturally identifies nonabelian de Rham cohomology \cref{NonabelianDeRhamCohomology} with \emph{nonabelian real cohomology}:
\begin{equation}
  \begin{alignedat}{2}
    H^1_{\mathrm{dR}}\bracket({
      X;
      \mathfrak{l}\mathcal{A}
    })
    &
    \defneq
    \Omega^1_{\mathrm{cl}}\bracket({
      X;
      \mathfrak{l}\mathcal{A}
    })
    \big/
    \mathrm{concordance}
    &\;\;\;\;&
    \text{\footnotesize by \cref{NonabelianDeRhamCohomology}}
    \\
    & 
    \simeq  
    \pi_0\,
    \shape
    \mathbf{\Omega}
      ^1
      _{\mathrm{cl}}
    \bracket({
      X;
      \mathfrak{l}\mathcal{A}
    })
    &&
    \text{\footnotesize by \cref{ShapeEquivalentToSingularComplex}}
    \\
    & \simeq
    \pi_0\,
    \mathbf{Map}\bracket({
      X,
      L^{\mathbb{R}}
      \mathcal{A}
    })
    &&
    \text{\footnotesize by \cref{ParameterizedFundamentalRHTTheorem}}
    \\
    & \defneq
    H^1\bracket({
      X;
      \Omega L^{\mathbb{R}}
      \mathcal{A}
    })
    &&
    \text{\footnotesize by \cref{NonabelianCohomologyInDegree1}.}
  \end{alignedat}
\end{equation}
\end{enumerate}

Therefore, on $\pi_0$ the differential character map \cref{DifferentialCharacterMap} becomes a map from nonabelian integral to nonabelian real/de Rham cohomology, which as such is the plain \emph{nonabelian character map} \cite[Def. IV.2]{FSS23-Char}:
\begin{equation}
\label
{TheCharacterMap}
  \begin{tikzcd}[
    column sep=-3pt
  ]
    H^1\bracket({
      X;
      \Omega\mathcal{A}
    })
    \ar[
      rrrrrr,
      uphordown,
      "{
        \mathrm{ch}^{\mathcal{A}}_X
      }"
    ]
    &\simeq&
    \pi_0\,
    \mathbf{Map}\bracket({
      X;
      \mathcal{A}
    })
    \ar[
      rr,
      "{
        \scaledbracket({
          \mathbf{ch}
            ^{\mathcal{A}}
            _{X}
        })_\ast
      }"
    ]
    &\phantom{----}&
    \pi_0\,
    \shape 
    \mathbf{\Omega}^1_{\mathrm{cl}}
    \bracket({
      X;
      \mathfrak{l}
      \mathcal{A}
    })
    &\simeq&
    H^1_{\mathrm{dR}}\bracket({
      X;
      \mathfrak{l}
      \mathcal{A}
    })\,.
  \end{tikzcd}
\end{equation}

\noindent For example:
\begin{enumerate}
\item
$\mathrm{ch}^{B^n \mathbb{Z}}$ is the classical \emph{de Rham map},

\item 
$\mathrm{ch}^{B\mathrm{U}}$ is the \emph{Chern character} on rank-0 K-theory,

\item 
$\mathrm{ch}^{E_n}$ is the \emph{Chern-Dold character} on the generalized abelian cohomology $E^n(-)$,

\item 
$\mathrm{ch}^{S^4}$ is a nonabelian character map from 4-Cohomotopy \cref{Cohomotopy} to classes of total C-field flux \cref{CEAlgebraOflS4}. 
\end{enumerate}

\noindent Hence, if we understand:
\begin{enumerate}
\item
nonabelian cohomology as the home of \emph{charge},

\item 
nonabelian de Rham cohomology as the home of \emph{total flux},
\end{enumerate}
then the character map \cref{TheCharacterMap} is understood as mapping charges to the total flux which they source. 

However, the charges in a cohomology theory are generally more constrained than total flux in real cohomology. For example, the image of the de Rham character
\begin{equation}
  \begin{tikzcd}
    H^n\bracket({
      X;
      \mathbb{Z}
    })
    \ar[
      r,
      "{
        \mathrm{ch}^{B^n \mathbb{Z}}_X
      }"
    ]
    &
    H^n_{\mathrm{dR}}(X)
  \end{tikzcd}
\end{equation}
is a \emph{lattice} inside a real vector space.

Therefore, asking that some total flux $\bracket[{\vec B}]$ \emph{lifts} through the character map, hence that it is the \emph{character image} of a charge $\chi$,
\begin{equation}
  \bracket[{
    \vec B
  }]
  =
  \mathrm{ch}^{\mathcal{A}}_{X}({
  \chi
  })
  \mathrlap{\,,}
\end{equation}
is to ask that it is ``quantized'' (in the sense of \emph{discretized}) in some form. This is the coarse form of \emph{flux quantization}:
\begin{equation}
\label
{CoarseFLuxQuantizationDiagram}
  \begin{tikzcd}[row sep=5pt]
    &&
    H^1\bracket({
      X^d;
      \Omega
      \mathcal{A}
    })
    \ar[
      dd,
      "{
        \mathrm{ch}
          ^{\mathcal{A}}
          _{X}
      }"
    ]
    \\
    \ast
    \ar[
      dr,
      shorten >=-3pt,
      "{
        \vec B
      }"{name=flux},
      "{
        \text{flux}
      }"{swap, sloped}
    ]
    \ar[
      urr,
      bend left=20pt,
      dashed,
      "{ \chi }"{
        swap, 
        pos=.4,
        name=charge
      },
      "{
        \text{charge}
      }"{
        sloped, 
        pos=.4
      }
    ]
    \\
    &
    \Omega^1_{\mathrm{cl}}\bracket({
      X^d,
      \mathfrak{a}
    })
    \ar[
      r,
      ->>
    ]
    &
    H^1_{\mathrm{dR}}\bracket({
      X^d;
      \mathfrak{a}
    })
  \end{tikzcd}
\end{equation}
and the choice of 
\begin{equation}
\label
{FluxQuantizationLaw}
  \mathcal{A}
  \in 
  \mathrm{Grpd}_\infty
  \text{ with }
  \mathfrak{l}\mathcal{A}
  \underset{\mathrm{iso}}{\simeq}
  \mathfrak{a}
\end{equation}
is the corresponding \emph{flux quantization law}.

However, the quantization of total flux $\bracket[{\vec B}]$ alone is not sufficient for defining a local higher gauge field: We need to \emph{locally} impose structure on the actual flux $\vec B$ which \emph{implies} global charges quantizing the corresponding global flux as above. 

This is a curious request: Locally the flux may take any value (as long as it satisfies the equations of motion), but globally it must ``integrate'' to certain prescribed discrete values. But this request has a natural solution in cohesive homotopy theory; this is what we turn to now in \cref{OnGlobalCompletion}.

\subsection
{Global Completion}
\label
{OnGlobalCompletion}

We are now ready to put all the pieces together and construct the \emph{phase space} stack of globally completed higher gauge fields 
for a given flux quantization law (\cref{OnFluxQuantization}).

\subsubsection
{The Phase Space}
\label
{OnThePhaseSpace}

\paragraph
{The Phase Space As Such}
\label
{OnThePhaseSpaceAsSuch}

Namely, after a choice $\mathcal{A}$ of flux quantization \cref{FluxQuantizationLaw} for given Gauss laws $\mathfrak{a}$, 
we have now obtained a pair of maps with common target, one from the smooth moduli set of on-shell flux densities, the other from the moduli space of charges:
\begin{equation}
  \begin{tikzcd}[
    column sep=35pt
  ]
    & 
    \mathbf{Map}\bracket({
      X^d,
      \mathcal{A}
    })
    \ar[
      d,
      "{
        \mathbf{ch}
          ^{\mathcal{A}}
      }"{pos=.4, swap},
      "{
        \text{\cref{TheCharacterMap}}
      }"
    ]
    \\
    \mathbf{\Omega}^1_{\mathrm{cl}}
    \bracket({
      X^d;
      \mathfrak{a}
    })
    \ar[
      r,
      "{
        \eta^{\shape}
      }",
      "{
        \text{\cref{ShapeUnit}}
      }"{swap}
    ]
    &
    \shape 
    \mathbf{\Omega}^1_{\mathrm{cl}}
    \bracket({
      X^d;
      \mathfrak{l}
      \mathcal{A}
    })    \mathrlap{\,.}
  \end{tikzcd}
\end{equation}
This is the local refinement of the corresponding coincident maps in \cref{CoarseFLuxQuantizationDiagram}, in that it reproduces the latter on gauge equivalence classes \cref{0truncationUnit}:
\begin{equation}
  \begin{tikzcd}[
    column sep=20pt
  ]
    & 
    \mathbf{Map}\bracket({
      X^d,
      \mathcal{A}
    })
    \ar[
      d,
      "{
        \mathbf{ch}
          ^{\mathcal{A}}
      }"{pos=.4}
    ]
    \\
    \mathbf{\Omega}^1_{\mathrm{cl}}
    \bracket({
      X^d;
      \mathfrak{a}
    })
    \ar[
      r,
      "{
        \eta^{\shape}
      }"
    ]
    &
    \shape 
    \mathbf{\Omega}^1_{\mathrm{cl}}
    \bracket({
      X^d;
      \mathfrak{l}
      \mathcal{A}
    })
  \end{tikzcd}
  \;\;\;\;\;
  \inlinetikzcd{
    {}
    \ar[
      rr,
      |->,
      "{
        [-]_0 
      }"
    ]
    \&\&
    {}
  }
  \;\;\;
  \begin{tikzcd}[
    column sep=20pt
  ]
    & 
    H^1\bracket({
      X^d;
      \Omega \mathcal{A}
    })
    \ar[
      d,
      "{
        \mathrm{ch}
          ^{\mathcal{A}}
      }"{pos=.4}
    ]
    \\
    \mathbf{\Omega}^1_{\mathrm{cl}}
    \bracket({
      X^d;
      \mathfrak{a}
    })
    \ar[
      r,
      ->>
    ]
    &
    H^1_{\mathrm{dR}}\bracket({
      X^d;
      \mathfrak{l}\mathcal{A}
    })
    \mathrlap{\,.}
  \end{tikzcd}
\end{equation}

Therefore we may now ask for \emph{local} flux quantization by asking for a lift as above in \cref{CoarseFLuxQuantizationDiagram} but before passing to gauge equivalence classes, and hence up to a specified homotopy $\hat A$:
\begin{equation}
\label
{LocalFluxQuantization}
  \begin{tikzcd}
    &&
    \mathbf{Map}\bracket({
      X^d,
      \mathcal{A}
    })
    \ar[
      dd,
      "{
        \mathbf{ch}^{\mathcal{A}}
      }"
    ]
    \\
    \ast
    \ar[
      dr,
      dashed,
      "{
        \vec B
      }"{name=fluxes},
      "{
        \text{fluxes}
      }"{
        swap, 
        sloped,
        color=darkblue
      }
    ]
    \ar[
      urr,
      bend left=25,
      dashed,
      "{
        \chi
      }"{
        swap,
        name=charges,
        pos=.4,
      },
      "{
        \text{charges}
      }"{
        sloped,
        pos=.4,
        color=darkblue
      }
    ]
    \ar[
      from=charges,
      to=fluxes,
      Rightarrow,
      dashed,
      "{ 
        \hat A 
      }"{
        pos=.55, swap
      },
      "{
        \text{potentials}
      }"{
        sloped, 
        swap, 
        pos=.4,
        yshift=-2pt,
        color=darkgreen
      }
    ]
    \\
    & 
    \mathbf{\Omega}^1_{\mathrm{cl}}
    \bracket({
      X^d;
      \mathfrak{a}
    })
    \ar[
      r,
      "{
        \eta^{\shape}
      }"
    ]
    &
    \shape
    \mathbf{\Omega}^1_{\mathrm{cl}}
    \bracket({
      X^d;
      \mathfrak{l}
      \mathcal{A}
    })    \mathrlap{\,.}
  \end{tikzcd}
\end{equation}

That homotopy $\hat A$ turns out to encode the \emph{gauge potentials}. In this way, such a diagram is the proper enhancement of the flux densities to a globally defined higher gauge field configuration.

The entirety of choices of such data forms, by the universal property \cref{UniversalPropertyOfHomotopyPullback}, the moduli stack, which is the homotopy fiber product \cref{HomotopyPullback} of our two maps:
\begin{equation}
\label
{PhaseSpaceAsPullback}
  \begin{tikzcd}[
    row sep=22pt
  ]
    &
    \mathbf{Phs}\bracket({
      X^d; 
      \mathcal{A}
    })
    \ar[
      r,
      "{ 
        p^{\mathcal{A}} 
      }"
    ]
    \ar[
      dd,
      "{
        p^{\Omega}
      }"
    ]
    \ar[
      ddr,
      phantom,
      "{ \lrcorner }"{pos=.1}
    ]
    &
    \mathbf{Map}\bracket({
      X^d,
      \mathcal{A}
    })
    \ar[
      dd,
      "{
        \mathbf{ch}^{\mathcal{A}}
      }"
    ]
    \\
    \ast
    \ar[
      dr,
      dashed,
      "{
        \vec B
      }"{name=fluxes}
    ]
    \ar[
      ur,
      bend left=20,
      dashed,
      "{
        \scaledbracket({
          \chi,
          \hat A
        })
      }"{
        sloped
      }
    ]
    \\
    & 
    \mathbf{\Omega}^1_{\mathrm{cl}}
    \bracket({
      X^d;
      \mathfrak{a}
    })
    \ar[
      r,
      "{
        \eta^{\shape}
      }"
    ]
    &
    \shape
    \mathbf{\Omega}^1_{\mathrm{cl}}
    \bracket({
      X^d;
      \mathfrak{l}\mathcal{A}
    }) \mathrlap{\,.}
  \end{tikzcd}
\end{equation}

This 
\begin{equation}
  \mathbf{Phs}\bracket({
    X^d;
    \mathcal{A}
  })
  \in
  \mathrm{SmthGrpd}_\infty
\end{equation}
is the globally completed \emph{phase space} of the global completion of our Maxwell-type higher gauge theory by flux quantization in $\mathcal{A}$-cohomology, in that it is the gauged smooth set of (Cauchy data for) \emph{on-shell} higher gauge fields, hence those solving their equations of motion (here: of Maxwell-type). 

We notice that the above homotopy pullback is the image of a homotopy pullback under $\mathbf{Map}\bracket({X^d; -})$, by \cref{ModuliSetOfClosedFormsOnManifold,ShapeOfModuliOfClosedForms}:
\begin{equation}
\label
{PhaseSpacePullbackAsImageUnderMap}
  \text{\cref{PhaseSpaceAsPullback}}
  \;\simeq\;
  \begin{tikzcd}[
   column sep=55pt,
   row sep=50pt
  ]
    \mathbf{Map}\bracket({
      X^d,
      \mathcal{A}_{\mathrm{dff}}
    })
    \ar[
      r,
      "{
        \mathbf{Map}\scaledbracket({
          X^d,
          \,
          p^{\mathcal{A}}
        })      
      }"
    ]
    \ar[
      d,
      "{
        \mathbf{Map}\scaledbracket({
          X^d,
          p^{\Omega}
        })
      }"{description}
    ]
    \ar[
      dr,
      phantom,
      "{ \lrcorner }"{pos=.1}
    ]
    &
    \mathbf{Map}\bracket({
      X^d, 
      \mathcal{A}
    })
    \ar[
      d,
      "{
        \mathbf{Map}
        \scaledbracket({
          X^d,
          \mathbf{ch}^{\mathcal{A}}
        })
      }"{description}
    ]
    \\
    \mathbf{Map}\bracket({
      X^d,
      \mathbf{\Omega}^1_{\mathrm{cl}}
      (\ast; \mathfrak{a})
    })
    \ar[
      r,
      "{
        \mathbf{Map}\scaledbracket({
          X^d,
          \eta^{\shape}
        })
      }"
    ]
    &
    \mathbf{Map}\bracket({
      X^d,
      \shape
      \mathbf{\Omega}^1_{\mathrm{cl}}
      \bracket({
        \ast;
        \mathfrak{l}
        \mathcal{A}
      })
    })
    \mathrlap{\,.}
  \end{tikzcd}
\end{equation}
Here in the top left we used that $\mathbf{Map}\bracket({X^d,-})$ preserves homotopy limits, and identified the \emph{differential moduli stack} of $\mathcal{A}$:
\begin{equation}
\label
{DifferentialModuliStack}
  \mathcal{A}_{\mathrm{dff}}
  :=
  \mathcal{A}
  \underset
    {
      \shape
      \mathbf{\Omega}^1_{\mathrm{cl}}
      \scaledbracket({
        \ast;
        \mathfrak{l}
        \mathcal{A}
      })
    }
    {\times}
    \mathbf{\Omega}^1_{\mathrm{cl}}
    \bracket({
      \ast;
      \mathfrak{l}\mathcal{A}
    })
    \in
    \mathrm{SmthGrpd}_\infty
    \mathrlap{\,.}
\end{equation}
This object neatly exhibits the combination of plain homotopy types ($\mathcal{A}$, on the left) with differential geometric spaces ($\mathbf{\Omega}^1_{\mathrm{cl}}(\ast;\mathfrak{a})$, on the right) that we achieve among smooth $\infty$-groupoids (cf. \cref{GeomHomotopySchematics}).

\paragraph
{Restriction to Solitonic Fields}

In variation of the construction \cref{PhaseSpaceAsPullback}, we may restrict to solitonic fields whose flux densities $\vec B$ on the Cauchy surface $X^d$ are compactly supported \cref{SolitonicOnShellFluxDensities}:
\begin{equation}
  \inlinetikzcd{
    \vec B
    \in
    \mathbf{\Omega}^1_{
      \mathrm{cl},
      \mathcolor{purple}{\mathrm{c}}
    }
    \bracket({
      X^d;\mathfrak{a}
    })
    \ar[
      r,
      hook
    ]
    \&
    \mathbf{\Omega}^1_{\mathrm{cl}}
    \bracket({
      X^d;\mathfrak{a}
    })
    \mathrlap{\,.}
  }
\end{equation}
By \cref{OnShapeOfSolitonicFluxesIsPointedMaps}, the shape of these compactly supported flux moduli is the pointed mapping space out of the one-point compactification $X^d_{\cpt}$ of phase space. Therefore, the variant of the phase space \cref{PhaseSpaceAsPullback} for solitonic fields is the following pullback:
\begin{equation}
  \begin{tikzcd}[row sep=15pt, 
    column sep=30pt
  ]
    \mathbf{Phs}_{\mathrm{c}}\bracket({
      X^d;
      \mathcal{A}
    })
    \ar[r]
    \ar[d]
    \ar[
      dr,
      phantom,
      "{
        \lrcorner
      }"{pos=.1}
    ]
    &
    \mathbf{Map}^\ast\bracket({
      X^d_{\cpt}, 
      \mathcal{A}
    })
    \ar[
      d,
      "{
        \mathbf{ch}^{\mathcal{A}}
      }"
    ]
    \\
    \mathbf{\Omega}^1_{%
      \mathrm{cl},
      {\mathcolor{purple}{\mathrm{c}}}
    }
    \bracket({
      X^d;
      \mathfrak{a}
    })
    \ar[
      r,
      "{
        \eta^{\shape}
      }"
    ]
    &
    \shape
    \mathbf{\Omega}^1_{%
      \mathrm{cl},
      {\mathcolor{purple}{\mathrm{c}}}
    }
    \bracket({
      X^d;
      \mathfrak{a}
    })
    \mathrlap{\,.}
  \end{tikzcd}
\end{equation}
The shape of this solitonic phase space is thus found, analogously to \cref{ShapeOfPhaseSpace}, to be the space of pointed maps into the classifying space:
\begin{equation}
\label
{ShapeOfSolitonicPhaseSpace}
  \shape 
  \mathbf{Phs}_{\mathrm{c}}\bracket({
    X^d;
    \mathcal{A}
  })
  \simeq
  \mathbf{Map}^\ast\bracket({
    X^d_{\cpt},
    \mathcal{A}
  })
  \mathrlap{\,.}
\end{equation}

\paragraph
{The Gauge Potentials}
Among the most profound aspects of this construction is that the homotopies in \cref{LocalFluxQuantization} are really locally equivalent to higher gauge potentials as known in the traditional literature. This is the case but requires non-trivial analysis:

\begin{example}
\label[example]
{ExamplesOfDifferentialCohomology}
\begin{enumerate}
\item
\textbf{Differential ordinary cohomology}.
For ordinary cohomology coefficients, $\mathcal{A} \defneq B^{n+1} \mathbb{Z}$, \cite[Prop. 9.5]{FSS23-Char} shows that the differential moduli stack \cref{DifferentialModuliStack} is that classifying ordinary differential $(n+1)$-cohomology (Deligne cohomology, higher $\mathrm{U}(1)$-bundle $(n-1)$-gerbes with connection, cf. \parencites[\S 2.5]{FSS13-CupCS}[\S 3.1]{FSS15-Stacky}[\S 2.8]{FRS2016})
\begin{equation}
  \bracket({
    B^{n+1} \mathbb{Z}
  })_{\mathrm{dff}}
  \simeq
  \mathbf{B}^n\mathrm{U}(1)_{\mathrm{conn}}
  \mathrlap{\,.}
\end{equation}
This shows that the above construction reproduces the usual global completions of abelian higher gauge fields, such as the ordinary magnetic field in ordinary differential 2-cohomology (\emph{Dirac charge quantization}, cf. \parencites{WuYang1976}{Alvarez1985}[\S 2]{Freed2002}{LazaroiuShahbazi2022b}), the Kalb-Ramond B-field in ordinary differential 3-cohomology (\parencites{Gawedzki1988}{FreedWitten1999}), and so on; see \cite[Ex. 3.10]{SS25-Flux} for further pointers.

\item 
\textbf{Differential K-theory}.
For $\mathcal{A} \defneq {B\mathrm{U}} \sslash B \mathrm{U}(1)$ flux-quantizing type IIA supergravity with the $H_7$-Bianchi disregarded \cref{CharacteristicOfIIA}, \cite[Exs. 9.2, 11.2]{FSS23-Char} shows that $\mathcal{A}_{\mathrm{dff}}$ \cref{DifferentialModuliStack} is the moduli stack for (twisted) differential rank-0 K-theory on 10D spacetimes, a candidate model for the gauge potentials of the RR-field in 10D supergravity (cf. \parencites{Freed2002}[\S 3]{Szabo2013}{GS22-KTheory}). 

\item \textbf{Differential abelian cohomology.}
Generally, for $\mathcal{A} \defneq E_n$ a stage in a spectrum $E_\bullet$ of spaces, representing a (Whitehead-generalized) abelian cohomology theory, $E^n_{\mathrm{dff}}$ is \cite[Ex. 9.1]{FSS23-Char} the moduli stack for differential $E$-cohomology (\cite[\S 4]{HopkinsSinger2005}, cf. \cite{Bunke2012}) in degree $n$.

But this construction generalizes now further to differential refinement of \emph{nonabelian} cohomology \cref{NonabelianCohomologyInDegree0}, such as:

\item 
\label
{DifferentialCohomotopy}
\textbf{Differential 4-cohomotopy}.
For $\mathcal{A} \defneq S^4$ flux-quantizing 11D SuGra \cref{CEAlgebraOflS4},
the differential cohomology theory classified by $S^4_{\mathrm{dff}}$ \cref{DifferentialModuliStack} is \emph{differential 4-Cohomotopy} \parencites[\S 4]{FSS15-M5WZW}[\S 3.1]{Grady2021}. That this indeed induces the local gauge potentials (and their local gauge transformations) expected in 11D SuGra (cf. \cref{On11D10DSupergravity}) is shown in \parencites[Prop. 1.1]{GSS24-SuGra}[\S 4.1]{GSS25-M5}{Banerjee2025-Potentials}.

\end{enumerate}
\end{example}

\paragraph
{The Shape of Phase Space}

The shape of the phase space is the groupoid of \emph{global symmetries} exhibited by \emph{large gauge transformations}.

The properties \cref{IdempotencyViaShapeUnit,ShapePreservesHomotopyFibersOverDiscrete} immediately imply the following crucial fact:
\begin{proposition}
\label[proposition]
{ShapeOfPhaseSpaceIsMappingSpace}
The shape of the phase space is the moduli space of the topological charges:
\begin{equation}
\label
{ShapeOfPhaseSpace}
  \shape
  \mathbf{Phs}\bracket({
    X^d;
    \mathcal{A}
  })
  \simeq
  \mathbf{Map}\bracket({
    X^d
    ,
    \mathcal{A}
  })
  \mathrlap{\,.}
\end{equation}
\end{proposition}
Hence the gauge equivalence classes \cref{0truncationUnit} of the shape \cref{TheShapeModality} of the phase space \cref{PhaseSpaceAsPullback} constitute the nonabelian cohomology \cref{NonabelianCohomologyInDegree1} of the Cauchy surface with coefficients in the flux quantization law $\mathcal{A}$ \cref{FluxQuantizationLaw}:
\begin{equation}
\label
{0TruncOfShapeOfPhaseSpaceIsCohomology}
  \bracket[{
    \shape
    \mathbf{Phs}\bracket({
      X^d;
      \mathcal{A}
    })
  }]_0
  \simeq
  H^1\bracket({
    X^d;
    \Omega
    \mathcal{A}
  })
  \mathrlap{\,.}
\end{equation}
\begin{proof}
  Since the defining pullback \cref{PhaseSpaceAsPullback} is over the geometrically discrete object $\shape \mathbf{\Omega}^1_{\mathrm{cl}}\bracket({ X^d; \mathfrak{l}\mathcal{A} })$, the shape operation may be taken inside, according to \cref{ShapePreservesHomotopyFibersOverDiscrete}. At the same time, the mapping space is shape-modal, $\shape \mathbf{Map}\bracket({ X^d, \mathcal{A} }) \simeq \mathbf{Map}\bracket({ X^d, \mathcal{A} })$, by \cref{MappingIntoGeometricallyDiscrete}, 
  and the shape of the bottom map is an equivalence, by \cref{IdempotencyViaShapeUnit}.   \begin{equation}
   \shape
   \left(
  \begin{tikzcd}[
    row sep=15pt,
    column sep=15pt
  ]
    \mathbf{Phs}\bracket({
      X^d; 
      \mathcal{A}
    })
    \ar[
      r
    ]
    \ar[
      dd
    ]
    \ar[
      ddr,
      phantom,
      "{ \lrcorner }"{pos=.1}
    ]
    &
    \mathbf{Map}\bracket({
      X^d,
      \mathcal{A}
    })
    \ar[
      dd,
      "{
        \mathbf{ch}^{\mathcal{A}}
      }"
    ]
    \\
    \\
    \mathbf{\Omega}^1_{\mathrm{cl}}
    \bracket({
      X^d;
      \mathfrak{a}
    })
    \ar[
      r,
      "{
        \eta^{\shape}
      }"
    ]
    &
    \shape
    \mathbf{\Omega}^1_{\mathrm{cl}}
    \bracket({
      X^d;
      \mathfrak{l}\mathcal{A}
    }) \mathrlap{\,.}
  \end{tikzcd}
  \right)
  \simeq
  \left(
  \begin{tikzcd}[
    row sep=15pt,
    column sep=15pt
  ]
   \shape
    \mathbf{Phs}\bracket({
      X^d; 
      \mathcal{A}
    })
    \ar[
      r,
      "{ \sim }"
    ]
    \ar[
      dd
    ]
    \ar[
      ddr,
      phantom,
      "{ \lrcorner }"{pos=.1}
    ]
    &
    \mathbf{Map}\bracket({
      X^d,
      \mathcal{A}
    })
    \ar[
      dd,
      "{
        \mathbf{ch}^{\mathcal{A}}
      }"
    ]
    \\
    \\
    \shape \mathbf{\Omega}^1_{\mathrm{cl}}
    \bracket({
      X^d;
      \mathfrak{a}
    })
    \ar[
      r,
      "{
        \shape \eta^{\shape}
      }",
      "{
        \sim
      }"{swap}
    ]
    &
    \shape
    \mathbf{\Omega}^1_{\mathrm{cl}}
    \bracket({
      X^d;
      \mathfrak{l}\mathcal{A}
    }) 
  \end{tikzcd}
  \right)
  \mathrlap{\,.}
  \end{equation}
This exhibits the top map on the right as the pullback of an equivalence, and thus itself as an equivalence, as shown.
\end{proof}

For completeness we mention the following further strengthening of \cref{ShapeOfPhaseSpaceIsMappingSpace}, which---in the special abelian case that the classifying space $\mathcal{A}$ is a stage in a spectrum of spaces---is part of the structure constituting the \emph{differential cohomology hexagon} (\parencites[\S 3]{BunkeNikolausVoelkl2016}, cf. \parencites[\S 6.2.a]{AmabelEtAl2021}): 
\begin{proposition}
  In the defining pullback diagram \cref{PhaseSpaceAsPullback,PhaseSpacePullbackAsImageUnderMap} for the phase space stack:
  \begin{enumerate}
  \item 
  the shape unit on $\mathrm{Phs}\bracket({X^d; \mathcal{A}})$ is equivalently the top map, 
  
  \item the shape of the left map is equivalently the right map:
  \end{enumerate}
  \begin{equation}
    \begin{tikzcd}[
      ampersand replacement=\&
    ]
      \mathcal{A}_{\mathrm{dff}}
      \ar[
        r,
        "{ p^{\mathcal{A}} }"{description}
      ]
      \ar[
        d,
        "{ p^{\Omega} }"{swap}
      ]
      \ar[
        dr,
        phantom,
        "{ \lrcorner }"{pos=.1}
      ]
      \&
      \mathcal{A}
      \ar[
        d,
        "{ \mathbf{ch}^{\mathcal{A}} }"
      ]
      \\
      \mathbf{\Omega}^1_{\mathrm{cl}}({
        \ast;
        \mathfrak{a}
      })
      \ar[
        r,
        "{ \eta^{\shape}_\Omega }"
      ]
      \&
      \shape
      \mathbf{\Omega}^1_{\mathrm{cl}}({
        \ast;
        \mathfrak{a}
      })      
    \end{tikzcd}
    \;\;\simeq\;\;
    \begin{tikzcd}[
      ampersand replacement=\&
    ]
      \mathcal{A}_{\mathrm{dff}}
      \ar[
        r,
        "{
          \eta^{\shape}_{\mathcal{A}_{\mathrm{dff}}}
        }"{description}
      ]
      \ar[
        d,
        "{ p^\Omega }"{swap}
      ]
      \ar[
        dr,
        phantom,
        "{ \lrcorner }"{pos=.1}
      ]
      \&
      \shape \mathcal{A}_{\mathrm{dff}}
      \ar[
        d,
        "{
          \shape
          p^{\Omega}
        }"
      ]
      \\
      \mathbf{\Omega}^1_{\mathrm{cl}}
      ({
        \ast; \mathfrak{a}
      })
      \ar[
        r,
        "{
          \eta^{\shape}
           _{\Omega}
        }"
      ]
      \&
      \shape 
      \mathbf{\Omega}^1_{\mathrm{cl}}
      ({
        \ast; \mathfrak{a}
      })
      \mathrlap{\,.}
    \end{tikzcd}
  \end{equation}
\end{proposition}
\begin{proof}
  First, \cref{IdempotencyViaShapeUnit,ShapePreservesHomotopyFibersOverDiscrete} imply the equivalence: 
  $
    \inlinetikzcd{
    \big({
      \eta^\shape_{\mathcal{A}}
    }\big)^{-1}
    \circ
    \shape p^{\mathcal{A}}
    :
    \shape \mathcal{A}_{\mathrm{dff}}
    \ar[
      r,
      "{ \sim }"
    ]
    \&
    \mathcal{A}
    }
  $.
  Combining this with naturality of the shape unit, 
  $
    \eta^{\shape}_{\mathcal{A}}
    \circ 
    p^{\mathcal{A}}
    \simeq
    \shape p^{\mathcal{A}}
    \circ 
    \eta^{\shape}_{\mathcal{A}_{\mathrm{dff}}}
  $,
   yields
  $
    \begin{tikzcd}[
      ampersand replacement=\&,
      column sep=25pt,
      row sep=-10pt
    ]
      \mathcal{A}_{\mathrm{dff}}
      \ar[
        rr,
        "{
          \eta^{\shape}
        }"
      ]
      \ar[
        dr,
        shorten=-2pt,
        "{
          p^{\mathcal{A}}
        }"{swap}
      ]
      \&\&[-17pt]
      \shape \mathcal{A}_{\mathrm{dff}}
      \ar[
        dl,
        shorten=-2pt,
        "{ \sim }"{
          sloped, 
          swap,
          pos=.3
        }
      ]
      \\
      \&
      \mathcal{A}
    \end{tikzcd}
    .
  $
  From this, using commutativity 
  $
    \shape \eta^\shape_{\Omega}
    \circ
    \shape p^{\Omega}
    \simeq
    \shape \mathbf{ch}^{\mathcal{A}}
    \circ 
    \shape p^{\mathcal{A}}
  $
  and idempotency, $\shape \mathbf{ch}^{\mathcal{A}} \simeq \mathbf{ch}^{\mathcal{A}}$,
  it follows that $\shape p^{\Omega} \simeq \mathbf{ch}^{\mathcal{A}}$.
\end{proof}

\paragraph
{Brane Charges}

Spacetimes of singular (``black'') $p$-branes in $1+d$-dimensional spacetimes are, near the horizon, of the form 
(cf. \parencites[\S 1.2.1]{Marolf2012}): \begin{equation}
\label
{NearHorizonSpacetime}
  \overbrace{
  \mathbb{R}^{1,d}
  \setminus
  \mathbb{R}^{1,p}
  }^{%
    \substack{%
      \text{complement of}
      \\
      \text{singular brane}
    }
  }
  \;\;\simeq\;\;
  \overbrace{
  \mathbb{R}^{1,p}
  }^{%
    \mathclap{%
      \substack{%
        \text{singular brane}
        \\
        \text{worldvolume}
      }
    }
  }
  \;\times\;
  \overbrace{
  \bracket({
    \mathbb{R}^{d-p}
    \setminus \{0\}
  })}^{
    \mathclap{%
      \substack{%
        \text{punctured}
        \\
        \text{transverse space}
      }
    }
  }
  \;\;\simeq\;\;
  \underbrace{
  \overbrace{
  \mathbb{R}^{1,p}
  }^{%
    \mathclap{%
      \substack{%
        \text{singular brane}
        \\
        \text{worldvolume}
      }
      \;
    }
  }
  \;\times\;
  \overbrace{%
    \mathbb{R}^1_{> 0}
  }^{\mathclap{
    \;
    \substack{
      \text{radial}
      \\
      \text{direction}
    }
  }}
  }_{%
    \substack{%
      \text{contractible}
    }
  }
  \;\times\;
  \overbrace{
    S^{d-p-1}
  }^{%
    \substack{%
      \text{linking}
      \\
      \text{sphere}
    }
  }
  \mathrlap{\,.}
\end{equation}

For example, the ``black 0-branes'' of ordinary  4D Einstein--Maxwell theory are the 
\emph{Reissner--Nordstr{\"o}m} spacetimes, topologically the \emph{complement} of the worldline of a point charge (cf. \parencites{Carter1968}{Carter2009}): 
\begin{equation}
  X^{1,3}_{\mathrm{RN}}
  \,\simeq\,
  X^{1,0} 
  \times
  X^3_{\mathrm{RN}}
  \,\simeq\,
  \smash{%
  \overbrace{%
    \mathbb{R}^{1,0}
  }^{%
    \mathclap{%
      \substack{%
        \text{singular worldline}
      }
    }
  }
  }
  \times 
  \mathbb{R}_{> 0}
  \times
  \smash{
  \overbrace{
    S^2
  }^{
    \mathclap{
      \substack{
        \text{linking sphere}
      }
    }
  }
  }
  \,\simeq\,
  \mathbb{R}^{1,3}
  \setminus
  \mathbb{R}^{1,0}
  \mathrlap{\,.}
\end{equation}
For these, the original argument of \emph{Dirac charge quantization}  (\cite{Dirac1931}, cf. \parencites{WuYang1976}{Alvarez1985}[\S 2.1]{SS25-Flux}) is, in modernized paraphrase, the observation that the magnetic charge seen on such spacetimes in integral 2-cohomology is:
\begin{equation}
  H^2\bracket({
    X^{3}_{\mathrm{RN}}
    ;\,
    \mathbb{Z}
  })
  \,\simeq\,
  H^2\bracket({
    \mathbb{R}_{> 0} \times S^2
    ;\,
    \mathbb{Z}
  })
  \,\simeq\,
  H^2\bracket({
    S^2
    ;\,
    \mathbb{Z}
  })
  \,\simeq\,
  \pi_0 \mathbf{Map}\bracket({
    S^2
    ,
    B^2 \mathbb{Z}
  })
  \,\simeq\,
  \mathbb{Z}
  \mathrlap{\,,}
\end{equation}
reflecting an integer number of ``poles'' (or of the unit charges these carry) surrounded by the linking sphere.

Generally then, the brane charges seen near horizon \cref{NearHorizonSpacetime} after flux-quantization in $\Omega\mathcal{A}$-cohomology form the homotopy groups of $\mathcal{A}$ (assuming here that $\mathcal{A}$ is simply connected):
\begin{equation}
\label{GeneralFormulaForBraneCharge}
  \mathllap{
  \substack{
    \text{\color{gray}fundamental}
    \\
    \text{\color{gray}$p$-brane charges}
  }
  \;\;
  }
  H^1\bracket({
    \mathbb{R}^d 
    \setminus
    \mathbb{R}^p
    ;\,
    \Omega
    \,
    \mathcal{A}
  })
  \,\simeq\,
  \pi_0
  \mathbf{Map}\bracket({
    S^{d-p-1}
    ,\,
    \mathcal{A}
  })
  \,\simeq\,
  \pi_{d-p-1}\bracket({
    \mathcal{A}
  })
  \mathrlap{\,.}
\end{equation}
By the fundamental relation \cref{RealWhiteheadBracketLInfinityAlgebra}, the non-torsion charges among these brane charges are spanned by a linear basis for the real Whitehead-bracket $L_\infty$-algebra
\cref{RealWhiteheadBracketLInfinityAlgebra} and hence, under \cref{AdmissibleFluxQuantizationCondition,GaussLawSolutionsAsClosedDiffForms,LocalFluxQuantization}, by the available flux species.

For example, in the case of 11D SuGra \cref{11DSuGraEoM} we have two flux species $G_4$ and $G_7$, which after refinement to differential Cohomotopy (\cref{DifferentialCohomotopy} above) are accompanied by local gauge potentials $C_3$ and $C_6$, witnessing the closure of $G_4$ and of $G_7 - \tfrac{1}{2} C_3 \wedge G_4$, respectively (cf. \parencites[Ex. 4.7]{GSS26-SuGra}), where (with $\mathrm{vol}_{S^n}$ denoting the unit volume form): 
\begin{enumerate}
\item 
near the horizon $\mathbb{R}^{10} \setminus \mathbb{R}^{5} \simeq_{\mathrm{hmtp}} S^4$ of a singular M5-brane, $\int_{S^4} G_4$ is the integral charge carried by this 5-brane:
\begin{equation}
  \begin{tikzcd}
    \mathllap{
    \mathbb{R}^{10}
    \setminus
    \mathbb{R}^5
    \,\simeq\,
    \mathbb{R}^5
    \times
    \mathbb{R}^1_{> 0}
    \times
    \,
    }
    S^4
    \ar[
      r,
      "{
        n
        \in
        \mathbb{Z}
      }",
      "{\ }"{swap, name=charges}
    ]
    \ar[
      d,
      "{ 
        \scaledbracket({
          n \cdot
          \mathrm{vol}_{S^4}
          ,\,
          0
        })
      }"{swap},
      "{\ }"{name=fluxes}
    ]
    \ar[
      from=charges,
      to=fluxes,
      Rightarrow
    ]
    &
    S^4
    \ar[
      d,
      "{
        \mathbf{ch}^{S^4}
      }"
    ]
    \\
    \mathbf{\Omega}^1_{\mathrm{cl}}
    \bracket({
      \ast; 
      \mathfrak{l}S^4
    })
    \ar[
      r,
      "{
        \eta^{\shape}
      }"
    ]
    &
    \shape
    \mathbf{\Omega}^1_{\mathrm{cl}}
    \bracket({
      \ast; 
      \mathfrak{l}S^4
    })
    \mathrlap{\,,}
  \end{tikzcd}
\end{equation}

\item
near the horizon $\mathbb{R}^{10} \setminus \mathbb{R}^2 \simeq_{\mathrm{hmtp}} S^7$ of a singular M2-brane, $-2\int_{S^7} \bracket({ G_7 - \tfrac{1}{2}C_3 \wedge G_4 })$ is the integral (non-torsion) part of the charge carried by this M2-brane (the ``\emph{Page charge}'' \parencites[(8)]{Page1983}[(43)]{DuffStelle1990}, cf. \parencites{Moore2004}[p. 21]{BaggerLambertMukhiPapageorgakis2013}):
\begin{equation}
  \begin{tikzcd}
    \mathllap{
    \mathbb{R}^{10}
    \setminus
    \mathbb{R}^2
    \,\simeq\,
    \mathbb{R}^2
    \times
    \mathbb{R}^1_{> 0}
    \times
    \,
    }
    S^7
    \ar[
      r,
      "{
        (n, [k])
        \in
        \mathbb{Z}
        \times
        \mathbb{Z}_{/12}
      }",
      "{\ }"{swap, name=charges}
    ]
    \ar[
      d,
      "{ 
        (0
        ,\, 
        -
        \tfrac{n}{2}
        \cdot \mathrm{vol}_{S^7})
      }"{swap},
      "{\ }"{name=fluxes}
    ]
    \ar[
      from=charges,
      to=fluxes,
      Rightarrow
    ]
    &
    S^4
    \ar[
      d,
      "{
        \mathbf{ch}^{S^4}
      }"
    ]
    \\
    \mathbf{\Omega}^1_{\mathrm{cl}}
    \bracket({
      \ast; 
      \mathfrak{l}S^4
    })
    \ar[
      r,
      "{
        \eta^{\shape}
      }"
    ]
    &
    \shape
    \mathbf{\Omega}^1_{\mathrm{cl}}
    \bracket({
      \ast; 
      \mathfrak{l}S^4
    })
    \mathrlap{\,.}
  \end{tikzcd}
\end{equation}

\end{enumerate}

\begin{remark}[Physical meaning of the admissibility condition]
\label[remark]
{UnderstandingTheAdmissibilityCondition}
These considerations reveal the physical meaning of the admissibility condition $\mathfrak{l}\mathcal{A} \simeq \mathfrak{a}$ \cref{AdmissibleFluxQuantizationCondition} on flux quantization laws $\mathcal{A}$ for higher gauge theories characterized by $\mathfrak{a}$:
\begin{enumerate}
\item
  Recall that $\mathrm{CE}\bracket({\mathfrak{l}\mathcal{A}})$, by \cref{RealWhiteheadBracketLInfinityAlgebra}, is a \emph{minimal} Sullivan algebra for $\mathcal{A}$
  (as such uniquely determined, up to isomorphism, by the homotopy type of $\mathcal{A}$). This implies that the flux species $\vec F$ \cref{FluxDensitiesOnX} form a linear basis for the rationalized homotopy groups of $\Omega \mathcal{A}$, cf. \cref{RealWhiteheadBracketLInfinityAlgebra}.

\item 
 A non-minimal Sullivan algebra is one that, up to isomorphism, in addition contains \emph{contractible pairs}, corresponding to Gauss laws of the form 
 \begin{equation}
   \label{ContractiblePairs}
   \begin{aligned}
     \mathrm{d}\, F_p & = F_{p+1}
     \\
     \mathrm{d}\, F_{p + 1} & = 0
     \mathrlap{,.}
   \end{aligned}
 \end{equation}
 where neither of the flux species $F_p$ nor $F_{p+1}$ is associated with a homotopy group of $\mathcal{A}$.
  
\item
  Hence by demanding an isomorphism $\mathfrak{a} \simeq \mathfrak{l}\mathcal{A}$ for admissible flux quantization \cref{FluxQuantizationLaw}, we are implicitly saying that higher Maxwell-type theories which involve contractible pairs \cref{ContractiblePairs} of flux densities
  do not have any admissible flux quantization, in the present sense. 

\item
  With the above \cref{GeneralFormulaForBraneCharge}, we see that after IR-completion of the theory, flux species participating in contractible pairs \cref{ContractiblePairs} are \emph{not sourced by integral branes} and hence they are \emph{not sourced by any branes}, in that they do not represent the de Rham image of any topological charges. 

\item Therefore the admissibility condition \cref{AdmissibleFluxQuantizationCondition} for flux-quantization laws ensures that 
\begin{standout}
\emph{All flux species that appear are sourced by some brane species}. 
\end{standout}
This is indeed the expected situation in completions of higher-dimensional supergravity (cf. Polchinski's \emph{completeness hypothesis} \parencites[\S 2 (2), p. 3]{Polchinski2004}[p. 1 \& \S 4.3]{BanksSeiberg2011}).

\item
In the abelian case, where $\mathcal{A}$ is a stage in a spectrum, this condition translates in \cref{ExamplesOfDifferentialCohomology} to what \cite[Def. 4.46]{Bunke2012} calls ``\emph{canonical}'' differential data.

\item
  Still, one may consider relaxing this condition. We briefly touch on this in \cref{OnChernSimonsType} below.
\end{enumerate}
\end{remark}

\paragraph
{The Points of Phase Space}
If we forget the smooth structure on phase space, by applying the flat modality $\flat(-)$ \cref{FlatModalityOnGaugedSmoothSets}, and then pass to gauge equivalence classes \cref{0truncationUnit}, we get the discrete set of gauge equivalence classes of on-shell field configurations:
\begin{equation}
  \bracket[{
    \flat
    \mathbf{Phs}\bracket({
      X^d;
      \mathcal{A}
    })
  }]_0
  \in
  \mathrm{Set}\,.
\end{equation}
This is equivalently the cohomology of the Cauchy surface with coefficients in the differential moduli stack $\mathcal{A}_{\mathrm{dff}}$ \cref{DifferentialModuliStack}
\begin{equation}
\label
{DifferentialNonabelianCohomology}
  \begin{alignedat}{2}
  \bracket[{
    \flat
    \mathbf{Phs}\bracket({
      X^d;
      \mathcal{A}
    })
  }]_0
  & 
  \simeq
  \bracket[{
    \flat
    \mathbf{Map}\bracket({
      X^d,
      \mathcal{A}_{\mathrm{dff}}
    })
  }]_0
  &\;\;\;\;&
  \text{\footnotesize by \cref{PhaseSpacePullbackAsImageUnderMap}}
  \\
  & \simeq
  \bracket[{
    \mathrm{Hom}\bracket({
      X^d
      ,
      \mathcal{A}_{\mathrm{dff}}
    })
  }]_0
  \\
  &
  =:
  H^1_{\mathrm{dff}}\bracket({
    X^d;
    \Omega\mathcal{A}
  })
  \mathrlap{\,,}
  \end{alignedat}
\end{equation}
which is the \emph{differential nonabelian cohomology} \parencites[Def. 9.3]{FSS23-Char} of the Cauchy surface, with coefficients in $\mathcal{A}$.

For example:

\begin{enumerate}
\item 
  For \colorbox{lightgray}{{$\mathcal{A} \defneq B^2 \mathbb{Z}^2$}} flux-quantizing 4D vacuum Maxwell theory \cref{CharacteristicOfVacuumElectromagnetism}, the completed gauge fields have classes in (two copies of) \emph{ordinary differential cohomology in degree 2},
  \begin{equation}
    H^1_{\mathrm{dff}}\bracket({
      X^d;
      \Omega B^2 \mathbb{Z}^2
    })
    \simeq
    H^2_{\mathrm{dff}}\bracket({
      X^d;
      \mathbb{Z}^2
    })
  \end{equation}
  (cf. \cite[Prop. 9.5]{FSS23-Char}). 

  This is equivalently the usual flux quantization that one finds in modern textbooks (going back to \cite{Dirac1931,Schwinger1966,Zwanziger1968}), equivalently given by isomorphism classes of $\mathrm{U}(1)$-principal connections (cf. \parencites{WuYang1976}{Alvarez1985}[\S 5]{WenWitten1985}[\S 7.1]{Brylinski1993}[\S 16.4e]{Frankel2011}[\S 2]{Freed2002}). 

  The consideration of quantization not just of the magnetic but also of the electric flux, hence of differential coefficients in $B^2 \mathbb{Z}^2$, is less widely appreciated but has been considered in \parencites{FreedMooreSegal2007}{FreedMooreSegal2007b}[Rem. 2.3]{BeckerBeniniSchenkelSzabo2017}{LazaroiuShahbazi2022a}{LazaroiuShahbazi2022b}[(21)]{SS24-Phase}.

  Of course, one may also choose not to integrally quantize the electric fluxes at all, hence to take the flux quantization law to be $B^2\mathbb{Z} \times B^2 \mathbb{Q}$. If one means to model real-world electromagnetism, then the ``right'' choice of flux quantization is a matter of comparison to experiment. 
  
  While integral magnetic flux quantization has been securely observed experimentally (in the form of flux quanta attached to \emph{Abrikosov vortices} in type II superconductors, cf. \cite[\S 5.3]{Timm2023}), the phenomenology of electric flux quantization may yet have to receive attention.

\item
 For \colorbox{lightgray}{{$\mathcal{A} \defneq B^3 \mathbb{Z}$}} flux-quantizing the NS field in 10D supergravity \cref{CharacteristicOfIIA}, the completed gauge fields have classes in \emph{ordinary differential cohomology in degree 3} (again by \cite[Prop. 9.5]{FSS23-Char}):
 \begin{equation}
   H^1_{\mathrm{dff}}\bracket({
     X^d;
     \Omega B^3 \mathbb{Z}
   })
   \simeq
   H^3_{\mathrm{dff}}\bracket({
     X^d;
     \mathbb{Z}
   })
   \mathrlap{\,.}
 \end{equation}
 This is the flux quantization traditionally considered for this \emph{NS B-field}, equivalently given by \emph{Deligne cohomology in degree 3} and by \emph{bundle gerbes with connection and curving} \parencites{Gawedzki1988}{FreedWitten1999}{Gawedzki2002}{CareyJohnsonMurray2004}{BonoraFerrariSavelli2008}{FerrariRuffino2012}{Grady2019}.

\item 
  For \colorbox{lightgray}{{$\mathcal{A} \defneq B\mathrm{U}$}} flux-quantizing 10D type IIA supergravity with trivial NS-field background \cref{CharacteristicOfIIA}, the completed gauge fields have classes \cite[Ex. 9.2]{FSS23-Char} in \emph{differential K-theory}:
  \begin{equation}
    H^1_{\mathrm{dff}}\bracket({
      X^d;
      \Omega B\mathrm{U}
    })
    \simeq
    K^0_{\mathrm{dff}}\bracket({
      X^d
    })_{\mathrm{rk}=0}
  \end{equation}
  (the twisting by the background NS-field is readily incorporated, too; we come to this in \cref{OnSourcesAndBranes,On11D10DSupergravity}.)
  
  This is the RR-flux quantization traditionally considered in the string theory literature \parencites[\S 3.7]{Freed2001}{Freed2002}{KahleValentino2014}{GS22-KTheory}, going back to suggestions in \cite{GreenHarveyMoore1996,MinasianMoore1997,MooreWitten2000}, cf. \cite[\S 4.1]{SS25-Flux}.

\item 
  For \colorbox{lightgray}{{$\mathcal{A} \defneq S^4$}} flux-quantizing 11D supergravity \cref{CEAlgebraOflS4}, the completed gauge fields have classes in \emph{differential 4-Cohomotopy} \parencites[\S 4]{FSS15-M5WZW}[Ex. 9.3]{FSS23-Char}[\S 3.1]{Grady2021}.

\end{enumerate}

\paragraph
{The Reduced Phase Space}

The \emph{reduced phase space} (cf. \cite{HenneauxTeitelboim1992}) is the (concrete, \cref{OnConcretification}) quotient of the phase space by gauge symmetries \cref{0truncationUnit}:
\begin{equation}
\label
{TheReducedPhaseSpace}
  \sharp_1 \bracket[{
    \mathbf{Phs}\bracket({
      X^d;
      \mathcal{A}
    })
  }]_0
  \in 
  \mathrm{SmthSet}
  \mathrlap{\,.}
\end{equation}

We now have the following fundamental result:
\begin{standout}
The points of the reduced phase space are the differential nonabelian cohomology classes of the Cauchy surface, with coefficients in the flux quantization law $\mathcal{A}$.
\end{standout}
This holds by these natural isomorphisms:
\begin{equation}
  \begin{alignedat}{2}
    \flat 
    \bracket({
    \sharp_1
    \bracket[{
      \mathbf{Phs}\bracket({
        X^d;
        \mathcal{A}
      })
    }]_0
    })
    & 
    \simeq
    \flat 
    \bracket[{
      \mathbf{Phs}\bracket({
        X^d;
        \mathcal{A}
      })
    }]_0
    &\;\;&
    \text{\footnotesize by \cref{PointsOfConcretification}
    }
    \\
    & \simeq
    \bracket[{
      \flat 
      \mathbf{Phs}\bracket({
        X^d;
        \mathcal{A}
      })
    }]_0
    &&
    \text{ \footnotesize
      by \cref{FlatCommutesWith0Truncation}
    }
    \\
    & \simeq
    H^1_{\mathrm{dff}}\bracket({
      X^d;
      \Omega
      \mathcal{A}
    })
    &&
    \text{\footnotesize 
      by \cref{DifferentialNonabelianCohomology}.
    }
  \end{alignedat}
\end{equation}

\paragraph
{The Pre-Symplectic Structure}

Reduced phase spaces like \cref{TheReducedPhaseSpace} are meant to carry a \emph{pre-symplectic structure}. 
We do not currently have a general result on its existence for our globally/IR-completed phase spaces,
but in examples it exists as expected, which we briefly indicate as \cref{SymplecticStructureOf5DMCSTheory,SymplecticStructureOf11DMCSTheory} below (details in \cite{SS26-Symplectic}).

First note:
\begin{remark}[Stacky incarnation of the variational bicomplex]
  Recall that plots   
  \[
      \mathrm{Map}\bracket({
        X^d,
        \Omega^1_{\mathrm{cl}}(
          \ast,
          \mathfrak{a}
        )
      })(U)
      \subset
      \Omega^\bullet_{\mathrm{dR}}\bracket({
        U \times X^d
      })^{\mathrm{dim}(\mathfrak{a})}
  \] 
  ($U \in \mathrm{CartSp}$)
  are tuples $\bracket\{{ F^{(i)}_{\mathrm{deg}(i)} }\}_{i =1}^{\mathrm{dim}(\mathfrak{a})}$ of differential forms on the product manifold $U \times X^d$. As such they decompose
  with respect to the bicomplex for the pair of differentials 
  \begin{enumerate}
  \item
  $\mathrm{d}_X$ -- the \emph{horizontal differential},
  \item $\delta := \mathrm{d}_U$ -- the \emph{variational differential}
  \end{enumerate}
  as
  \begin{equation}
    \label
    {BicomplexDecompositionOfDifferentialForms}
    F^{(i)}_{\mathrm{deg}(i)}
    =
    \sum_{
      \substack{
        n_U, n_X \in \mathbb{N}
        \\
        n_U + n_X = \mathrm{deg}(i)
      }
    } 
    F^{(i)}_{(n_U, n_X)}
    \mathrlap{\,.}
  \end{equation}
\end{remark}

\begin{example}
[Pre-Symplectic structure of 5D MCS theory]
\label[example]
{SymplecticStructureOf5DMCSTheory}
Consider 5D Maxwell-Chern-Simons theory \cref{EoMOf5DMCSTheory} IR-completed by flux quantization classified by any $\mathcal{A}$ admitting a rational equivalence 
\begin{equation}
  \begin{tikzcd}
    \mathcal{A}
    \ar[
      r,
      "{ \sim_{\mathbb{R}} }"
    ]
    &
    \mathrm{fib}\big(
      B^2 \mathbb{Z}
      \ar[
        r,
        "{ \cup^2 }"
      ]
      &[-10pt]
      B^4 \mathbb{Z}
    \big)
  \end{tikzcd}
\end{equation}
to the homotopy fiber of the universal cup square on integral 2-cohomology (such as $\mathcal{A} \defneq \shape S^2$).

Remarkably, computation shows that there exists a \emph{unique} (up to scale) closed 2-form $\omega_2 \in \Omega^2_{\mathrm{cl}}\bracket({\mathbf{Phs}\bracket({X^4; \mathcal{A}})})$ with the property that:
\begin{enumerate}
  \item
  it is plot-wise the fiber integral $\int_{X^4}$ over the Cauchy surface $X^4$ of a polynomial in bicomplex components \cref{BicomplexDecompositionOfDifferentialForms} of the flux densities and their variational derivatives,

  \item
  it descends to a closed 2-form on the reduced phase space \cref{TheReducedPhaseSpace},
\end{enumerate}
namely, the following form, which is actually exact (and the shown potential also descends to the reduced phase space):
\begin{equation}
\label
{TheUpstairsSymplecticFormOf5DMCSTheory}
\begin{aligned}
  \omega_2
  & \propto
  \int_{X^4}
  \bracket({
    F_{(1,1)}
    \wedge
    \delta G_{(0,3)}
  })
  -
  \tfrac{1}{2}
  \int_{X^4}
  F_{(1,1)}
  \wedge
  F_{(1,1)}
  \wedge
  F_{(0,2)}
  \\
  & =
  \delta
  \left(
    \tfrac{2}{3}
    \int_{X^4}
    F_{(1,1)}
    \wedge
    G_{(0,3)}
    -
    \tfrac{1}{3}
    \int_{X^4}
    F_{(0,2)}
      \wedge
    G_{(1,2)}
  \right)
  \mathrlap{\,,}
\end{aligned}
\end{equation}
where we now assume that $X^4$ is closed.

In order to compare this to expressions obtained by traditional Lagrangian means (cf. \cite{CrnkovicWitten1987}), one needs to restrict to topologically trivial fields. In this case one finds:
\begin{equation}
\label
{TheSymplecticFormOf5DMCSTheoryForToptrivialCase}
\begin{aligned}
  F_2 = \mathrm{d}A_1
  \;\;\;
  \Rightarrow
  \;\;\;
  \omega_2
  & \propto
  \int_{X^4}
  \bracket({
    \delta A_{(0,1)}
    \wedge
    \delta G_{(0,3)}
  })
  -
  \tfrac{1}{2}
  \int_{X^4}
  \delta A_{(0,1)}
  \wedge
  \delta A_{(0,1)}
  \wedge
  F_{(0,2)}
\end{aligned}
\end{equation}
(where the terms involving $A_{(1,0)}$ cancel out).

This coincides with the symplectic form obtained by Lagrangian means 
(cf. \cite[(33),(35)]{LikoBooth2008}\footnote{Beware that these authors have a computational error in the relative prefactor.}).
Conversely, this means that our form $\omega_2$ \cref{TheUpstairsSymplecticFormOf5DMCSTheory}, identified by non-Lagrangian methods, lifts the traditional Lagrangian expression \cref{TheSymplecticFormOf5DMCSTheoryForToptrivialCase} to the completed phase space of flux-quantized fields.
\end{example}

\begin{example}
[Pre-Symplectic structure of 11D MCS theory]
\label[example]
{SymplecticStructureOf11DMCSTheory}

The directly analogous statement to \cref{SymplecticStructureOf5DMCSTheory} holds for 11D Maxwell-Chern-Simons theory \cref{11DSuGraEoM} (with $X^4$ replaced by $X^{10}$, $F_2$ replaced by $G_4$, and $G_3$ replaced by $G_7$), for any flux-quantization $\mathcal{A}$ equipped with a rational homotopy equivalence of the form
\begin{equation}
  \begin{tikzcd}
    \mathcal{A}
    \ar[
      r,
      "{ \sim_{\mathbb{R}} }"
    ]
    &
    \mathrm{fib}\big(
      B^4 \mathbb{Z}
      \ar[
        r,
        "{ \cup^2 }"
      ]
      &[-10pt]
      B^8 \mathbb{Z}
    \big)
    \mathrlap{\,,}
  \end{tikzcd}
\end{equation}
such as $\mathcal{A} \defneq \shape S^4$ (the choice expanded on in \cref{On11D10DSupergravity}).
\end{example}

\subsubsection
{The Observables}

With the globally completed phase space in hand, we get access to the actual globally defined observables. 

\paragraph
{Ordinary Observables}
An ordinary \emph{observable} is (cf. \parencites[\S 4]{SimmsWoodhouse1976}[\S 1.1]{Kirillov1990}) a smooth map from phase space to the complex numbers, sending on-shell field configurations $\Phi$ to the given value $\mathscr{O}\bracket({\Phi})$ that is observed/measured:
\begin{equation}
\label
{OrdinaryObservable}
  \begin{tikzcd}[row sep=-3pt, 
    column sep=0pt
  ]
    \mathbf{Phs}\bracket({
      X^d;
      \mathcal{A}
    })
    \ar[
      rr,
      "{
        \mathscr{O}
      }"
    ]
    &&
    \mathbb{C}
    \\
    \Phi 
      &\longmapsto&
    \mathscr{O}\bracket({\Phi})\,.
  \end{tikzcd}
\end{equation}
For definiteness: By
\begin{equation}
  \mathbb{C}
  \in
  \inlinetikzcd{
    \mathrm{Alg}\bracket({
    \mathrm{SmthMfd}
    })
    \ar[r, hook]
    \&
    \mathrm{Alg}\bracket({
      \mathrm{SmthGrpd}_\infty
    })
  },
\end{equation}
we mean the complex numbers with their canonical smooth structure as the Euclidean space $\mathbb{R}^2$.
Being a smooth manifold, as a gauged smooth set, the complex numbers are therefore both concrete \cref{Concreteness} and 0-truncated \cref{ZeroTruncationModality}: 
\begin{equation}
  \begin{aligned}
    \mathbb{C}
    & 
    \simeq 
    \sharp_1 \mathbb{C}
    \\
    \mathbb{C}
    & \simeq
    \bracket[{\mathbb{C}}]_0
    \mathrlap{\,.}
  \end{aligned}
\end{equation}
But this immediately implies, by the universal properties \cref{UniversalPropertyOfConcretification,UniversalPropertyOfZeroTruncation}, that ordinary observables \cref{OrdinaryObservable} factor through the reduced phase space \cref{TheReducedPhaseSpace}:
\begin{equation}
\label
{ObservableOnPhaseSpaceFactor}
  \begin{tikzcd}
    \mathbf{Phs}\bracket({
      X^d;
      \mathcal{A}
    })
    \ar[
      r
    ]
    \ar[
      rr,
      uphordown,
      "{
        \mathscr{O}
      }"{description}
    ]
    &
    \sharp_1
    \bracket[{
      \mathbf{Phs}\bracket({
        X^d;
        \mathcal{A}
      })    
    }]_0
    \ar[
      r,
      dashed
    ]
    &
    \mathbb{C}
  \end{tikzcd}
\end{equation}
This is exactly the expected/desired physical property of ordinary observables: They are gauge invariant (and send smooth families of field configurations to smooth families of observed numbers).

\paragraph
{Topological Observables}
An observable \cref{OrdinaryObservable} is \emph{topological} (physics jargon for: \emph{homotopical}) if it factors through the shape of the phase space. 
\begin{equation}
  \begin{tikzcd}
    \mathbf{Phs}\bracket({
      X^d;
      \mathcal{A}
    })
    \ar[
      r,
      "{
        \eta^{\shape}
      }"
    ]
    \ar[
      rr,
      uphordown,
      "{
        \mathscr{O}_{\mathrm{top}}
      }"{description}
    ]
    &
    \shape
    \mathbf{Phs}\bracket({
      X^d;
      \mathcal{A}
    })
    \ar[
      r,
      dashed
    ]
    &
    \mathbb{C}
    \mathrlap{\,.}
  \end{tikzcd}
\end{equation}
Combined with the implied gauge invariance \cref{ObservableOnPhaseSpaceFactor} 
and with the characterization \cref{0TruncOfShapeOfPhaseSpaceIsCohomology} of the shape of phase space,
this means that ordinary topological observables are functions of the nonabelian cohomology of the Cauchy surface with coefficients in $\Omega\mathcal{A}$, hence of the set of topological charges:
\begin{equation}
  \begin{tikzcd}
    \mathbf{Phs}\bracket({
      X^d;
      \mathcal{A}
    })
    \ar[
      r,
      "{
        \eta^{\shape}
      }"
    ]
    \ar[
      rrr,
      uphordown,
      "{
        \mathscr{O}_{\mathrm{top}}
      }"{description}
    ]
    &
    \shape
    \mathbf{Phs}\bracket({
      X^d;
      \mathcal{A}
    })
    \ar[
      r,
      "{
        \eta^{[-]_0}
      }"
    ]
    &
    H^1\bracket({
      X^d;
      \Omega\mathcal{A}
    })
    \ar[
      r,
      dashed
    ]
    &
    \mathbb{C}
    \mathrlap{\,.}
  \end{tikzcd}
\end{equation}

Moreover, \emph{realistic} observables (those measurable in experiment) will only detect a finite number of connected components of phase space, whence realistic topological observables are actually identified with the 0-\emph{homology} of phase space:
\begin{equation}
\label{RealisticTopologicalObservables}
  \begin{aligned}
  \mathrm{TopObs}_0\bracket({
    X^d; \mathcal{A}
  })
  & 
  \simeq
  H_0\bracket({
    \shape \mathbf{Phs}\bracket({
      X^d; \mathcal{A}
    })
    ;\mathbb{C}
  })
  \\
  & 
  \simeq
  H_0\bracket({
    \mathrm{Map}\bracket({
      X^d, \mathcal{A}
    })
    ;\mathbb{C}
  })
  \mathrlap{\,.}
  \end{aligned}
\end{equation}

Or rather, restricting to solitonic field configurations \cref{ShapeOfSolitonicPhaseSpace} whose flux densities vanish at infinity \cref{SolitonicOnShellFluxDensities} along the Cauchy surface, the realistic topological observables \cref{RealisticTopologicalObservables} constitute the homology of the pointed mapping space \cref{ShapeOfSolitonicPhaseSpace} out of the one-point compactification of the Cauchy surface:
\begin{equation}
\label
{RealisticSolitonicTopologicalObservables}
  \mathrm{SolTopObs}_0\bracket({
    X^d;
    \mathcal{A}
  })
  =
  H_0\bracket({
    \mathrm{Map}^\ast\bracket({
      X^d_{\cpt},
      \mathcal{A}
    });
    \mathbb{C}
  })
  \mathrlap{\,.}
\end{equation}

These vector spaces of observables carry a commutative pointwise product that makes them a complex algebra, which may be understood as the algebra of \emph{classical observables}. Next we pass to a non-commutative quantum algebra structure.

\paragraph
{Topological Light-Cone Quantization}
Consider then the case that spacetime is not just globally hyperbolic but of the form
\begin{equation}
\label
{LightFrontSpacetime}
  X^{1,d}
  \simeq
  \mathbb{R}^{1,1}
  \times
  X^{d-1}
  \mathrlap{\,,}
\end{equation}
whence the Cauchy surfaces are of the form
\begin{equation}
  X^d 
  \simeq
  \mathbb{R}^1 
  \times
  X^{d-1}
  \mathrlap{\,.}
\end{equation}
On such spacetimes one may naturally consider \emph{light-front evolution} (\cite[\S 5]{Dirac1949}, cf. \cite{WanPowis1994}).

Remarkably, the realistic solitonic topological observables \cref{RealisticSolitonicTopologicalObservables} on such spacetimes \cref{LightFrontSpacetime} form the homology of a based loop space (cf. \cite[(18)]{SS24-Obs})
\begin{equation}
\label
{SolitonicTopologicalObservables}
  \begin{aligned}
  \mathrm{SolTopObs}_0\bracket({
    X^d;
    \mathcal{A}
  })
  &
  \simeq
  H_0\bracket({
    \mathrm{Map}^\ast\bracket({
      \bracket({
        \mathbb{R}^1 
        \times
        X^{d-1}
      })_{\cpt}
      ,
      \mathcal{A}
    });
    \mathbb{C}
  })
  \\
  & \simeq
  H_0\bracket({
    \Omega\, 
    \mathrm{Map}^\ast\bracket({
      X^{d-1}_{\cpt}
      ,
      \mathcal{A}
    });
    \mathbb{C}
  })
  \end{aligned}
\end{equation}
and as such inherit \cite[(20)]{SS24-Obs} a (generally) non-commutative (star-)algebra structure induced by loop concatenation (and reversal), the \emph{Pontrjagin algebra} structure:
\begin{equation}
  \begin{tikzcd}
    \mathrm{SolTopObs}_0
    \otimes_{{}_{\mathbb{C}}}
    \mathrm{SolTopObs}_0  
    \ar[
      r,
      "{
        \mathrm{conc}_\ast
      }"
    ]
    &
    \mathrm{SolTopObs}_0 
    \mathrlap{\,.}
  \end{tikzcd}
\end{equation}
As such, $\mathrm{SolTopObs}_0$, now equipped with the Pontrjagin product, has the form of an algebra of quantum observables%
\footnote{
Beware that the usual ``canonical'' deformation quantization of symplectic/Poisson structures cannot see a nontrivial quantum algebra structure of topological observables, as these are locally constant functions on phase space and as such have vanishing Poisson brackets. 
}, which one may understand as the result of \emph{light-front quantization} \cite[\S 2.2]{SS25-Complete} in the topological sector.

Concretely, with the restriction to homology in degree 0, this is the group algebra of the fundamental group of the pointed maps from $X^{d-1}_{\cpt}$ to $\mathcal{A}$:
\begin{equation}
\label
{SolitonicTopologicalObservablesAsGroupAlgebra}
  \mathrm{SolTopObs}_0
  \bracket({
    X^d; \mathcal{A}
  })
  \simeq
  \mathbb{C}\bracket[{
    \pi_1
      \mathrm{Map}^\ast\bracket({
        X^{d-1}_{\cpt}
        ,
        \mathcal{A}
      })
  }]
  \mathrlap{\,.}
\end{equation}

\section
{Enhancements}

We briefly indicate some natural further enhancements of the above constructions. 

For brevity we omit the important topic of higher gauge fields on \emph{orbifolds}, flux-quantized in twisted \emph{equivariant} nonabelian cohomology theory. For more on that see \cite{SS25-TEC,SS26-Orb}.

\subsection
{Sources and Branes}
\label
{OnSourcesAndBranes}

\paragraph
{Maxwell Theory with Sources}
The classical Maxwell equations of course generally have a source term $J_3$, the electric current density 3-form, and are as such traditionally written as:
\begin{equation}
\label
{TraditionalMaxwellEquationsWithSource}
  \mbox{``$
    \mathrm{d}\, F_2
    = 
    0
    \,,\;\;
    \mathrm{d}\, G_2
    = J_3
    \mathrlap{\,,}
    \;\;
    \star\, F_2 = G_2
  $''}
\end{equation}
Here we have put this traditional expression in quotation marks because it glosses over the important subtlety that the topological class of $J_3$ is on a different footing compared to that of $F_2$ and $G_2$: 
\begin{enumerate}
\item
The electric current density $J_3$ is meant to be integrable over the Cauchy surface, its integral being the net number of electrons.

\item Its class is not in plain cohomology but subject to coboundaries that preserve this integral, and hence the net number of electrons.

\item The Maxwell equation $\mathrm{d}\, G_2 = J_3$ is meant to apply after forgetting these integrability conditions, so that it trivializes only the plain de Rham class of $J_3$, without forcing the net number of electrons to vanish.
\end{enumerate}

We may mathematically implement these desiderata by considering the \emph{one-point compactification} $X^d_{\cpt}$ of the Cauchy surface, assuming for simplicity that this still has the structure of a smooth manifold (for instance $\mathbb{R}^3_{\cpt} \simeq S^3$).  
Then we have a canonical smooth map
\begin{equation}
  \begin{tikzcd}
    X^d 
    \ar[
      r,
      hook,
      "{ \phi }"
    ]
    &
    X^d_{\cpt}
  \end{tikzcd}
\end{equation}
with $F_2, G_2 \in \Omega^2_{\mathrm{dR}}\bracket({X^d})$ while $J_3 \in \Omega^3_{\mathrm{dR}}\bracket({\mathbb{R}^{1,0} \times X^d_{\cpt}})$, and with the would-be Maxwell equations \cref{TraditionalMaxwellEquationsWithSource} really reading as follows:
\begin{equation}
\label
{ProperMaxwellEquationsWithSources}
  \begin{aligned}
    \mathrm{d}\, F_2
    & = 
    0,
    \\
    \mathrm{d}\, G_2
    & = 
    \mathcolor{purple}{\phi^\ast} 
    J_3,
    \\
    \mathrm{d}\, J_3 & = 0
    \mathrlap{\,.}
  \end{aligned}
\end{equation}
Such a system of EoMs no longer has a characteristic $L_\infty$-algebra, but a \emph{fibration} of such:
\begin{equation}
  \begin{tikzcd}[
    ampersand replacement=\&,
    column sep=-2pt
  ]
    \mathfrak{a}
    \ar[
      dd,
      ->>,
      "{ 
        \mathfrak{l}\wp 
      }"
    ]
    \&[10pt]
    \mathrm{CE}({
      \mathfrak{a}
    })
    \& := \&
    \mathbb{R}_{_{\mathrm{d}}}
    \left[
      \begin{aligned}
        & f_2
        \\
        & g_2
        \\
        & j_3
      \end{aligned}
    \right]
    \big/
    \left(
      \begin{aligned}
        \mathrm{d}\, f_2 & = 0
        \\
        \mathrm{d}\, g_2 & = j_3
        \\
        \mathrm{d}\, j_3 & = 0
      \end{aligned}
    \right)
    \\
    \\
    \mathfrak{b}
    \&
    \mathrm{CE}({
      \mathfrak{b}
    })
    \ar[
      uu,
      hook,
      "{ 
        \mathfrak{l}\wp^\ast 
      }"{swap},
    ]
    \& := \&
    \mathbb{R}_{_{\mathrm{d}}}
    [
      j_3
    ]
    \big/
    (
      \mathrm{d}\, j_3 = 0
    )
    \mathrlap{\,,}
    \;\,
    \ar[
      uu,
      shift left=28pt,
      hook,
      "{
        \substack{
          j_3
          \\
          \rotatebox[origin=c]{90}{
            $\mapsto$
          }
          \\
          j_3
        }
      }"{swap}
    ]
  \end{tikzcd}
\end{equation}
in that germs of solutions of Maxwell's equations with sources \cref{ProperMaxwellEquationsWithSources} are equivalent to commuting squares of this form:
\begin{equation}
  \begin{tikzcd}[row sep=20pt]
    X^d
    \ar[
      d,
      "{ 
        (0,\iota) 
      }",
      hook
    ]
    \ar[
      r,
      dashed,
      "{
        (F_2, G_2)
      }"
    ]
    &
    \mathbf{\Omega}^1_{\mathrm{cl}}
    (\ast; \mathfrak{a})
    \ar[
      d,
      ->>,
      "{
        \wp_\ast
      }"
    ]
    \\
    \mathbb{R}^{1,0}
    \times
    X^d_{\cpt}
    \ar[
      r,
      "{
        J_3
      }"
    ]
    &
    \mathbf{\Omega}^1_{\mathrm{cl}}
    (\ast; \mathfrak{b})
    \mathrlap{\,,}
  \end{tikzcd}
\end{equation}
where in the bottom left we have included the time-axis in order that the bottom map specifies the full temporal history of the current density (for which plain Maxwell theory provides no equations of motion).

\paragraph
{Background Fluxes}

Generally, given 
\begin{enumerate}
\item 
  an embedding of smooth manifolds
  \begin{equation}
    \begin{tikzcd}
      \Sigma^p
      \ar[
        r, 
        "{ \phi }"
      ]
      &
      X^d\,,
    \end{tikzcd}
  \end{equation}

\item 
  a fibration of finite-type $L_\infty$-algebras
  \begin{equation}
    \begin{tikzcd}
      \mathfrak{a}
      \ar[
        r,
        ->>,
        "{ \mathfrak{l}\wp }"
      ]
      &
      \mathfrak{b}\,,
    \end{tikzcd}
  \end{equation}
\end{enumerate}
a dashed map completing a commuting diagram of the form
\begin{equation}
\label
{RelativeTwistedBianchiIdentitiesAsDiagram}
  \begin{tikzcd}
    \Sigma^p
    \ar[
      rr,
      dashed
    ]
    \ar[
      d,
      hook,
      "{ \phi }"
    ]
    &&
    \mathbf{\Omega}^1_{\mathrm{cl}}({
      \ast;
      \mathfrak{a}
    })
    \ar[
      d,
      ->>,
      "{
        \scaledbracket({
          \mathfrak{l}\wp
        })_\ast
      }"
    ]
    \\
    X^d
    \ar[rr]
    &&
    \mathbf{\Omega}^1_{\mathrm{cl}}({
      \ast;
      \mathfrak{b}
    })
  \end{tikzcd}
\end{equation}
encodes on-shell flux densities on $\Sigma^p$ whose Bianchi identities may be \emph{twisted} by \emph{background fluxes} on $X^d$.

Examples:
\begin{enumerate}

\item \emph{RR-Fluxes Twisted by NS5-Fluxes:}
\begin{equation}
  \begin{tikzcd}[
    column sep=small
  ]
    X^{9}
    \ar[
      d,
      hook
    ]
    \ar[
      rr,
      dashed,
      "{
        F_{2\bullet}
      }"
    ]
    &&
    \mathbf{\Omega}^1_{\mathrm{cl}}\bracket({
      \ast;
      \mathfrak{l}\bracket({
        \truncation{9}{B\mathrm{U}} 
          \sslash 
        B\mathrm{U}(1)
      })
    })
    \ar[
      d,
      ->>
    ]
    \\
    \mathllap{
      \mathbb{R}^{1,0}
      \times
    }
    X^9_{\cpt}
    \ar[
      rr,
      "{
        H_3
      }"
    ]
    &&
    \mathbf{\Omega}^1_{\mathrm{cl}}\bracket({
      \ast; 
      \mathfrak{l}B^2 \mathrm{U}(1)
    })
  \end{tikzcd}
  \;\;\;\;\;
  \begin{aligned}
    \mathrm{d}\,
    F_{2\bullet}
    & =
    F_{2\bullet - 2} \wedge 
    \phi^\ast H_3
    \\[28pt]
    \mathrm{d}\, H_3 & = 0\,.
  \end{aligned}
\end{equation}

\item
\emph{Self-Dual Flux on M5 Twisted by C-Field Flux:}
\begin{equation}
\label
{FluxOnM5Branes}
  \begin{tikzcd}
    \Sigma^5
    \ar[
      d,
      hook,
      "{ \phi }"
    ]
    \ar[
      rr,
      "{
        H_3
      }",
      dashed
    ]
    &&
    \mathbf{\Omega}^1_{\mathrm{cl}}\bracket({
      \ast;
      \mathfrak{l}_{_{S^4}}
      S^7
    })
    \ar[
      d,
      ->>,
      "{
        \scaledbracket({
          \mathfrak{l}h_{\mathbb{H}}
        })_\ast
      }"
    ]
    \\
    X^{10}
    \ar[
      rr,
      "{ (G_4, G_7) }",
      dashed
    ]
    &&
    \mathbf{\Omega}^1_{\mathrm{cl}}\bracket({
      \ast;
      \mathfrak{l}S^4
    })
  \end{tikzcd}
  \;\;\;\;\;\;\;
  \begin{aligned}
    \mathrm{d}\, H_3
    & =
    \phi^\ast G_4
    \\
    \\
    \mathrm{d}\, G_4 & = 0
    \\
    \mathrm{d}\, G_7 
      & = 
    \tfrac{1}{2} G_4 \wedge G_4
    \mathrlap{\,.}
  \end{aligned}
\end{equation}

\end{enumerate}

\paragraph
{Twisted Relative Gauge Fields}

Our construction above of globally completed higher gauge fields carries over straightforwardly to this case of relative Bianchi identities twisted by background fluxes, simply by passing to the \emph{arrow category} of $\mathrm{SmthGrpd}_\infty$: 

The choice of classifying space for (now: twisted relative) flux quantization is enhanced to a choice of fibration 
\begin{equation}
  \begin{tikzcd}
    \mathcal{A}
    \ar[
      r,
      "{ \wp }"
    ]
    &
    \mathcal{B}
    \mathrlap{\,,}
  \end{tikzcd}
\end{equation}
whose \emph{relative Whitehead $L_\infty$-algebra}
\begin{equation}
  \begin{tikzcd}
    \mathfrak{l}_{_{\mathcal{B}}}
    \mathcal{A}
    \ar[
      r,
      "{
        \mathfrak{l}\wp
      }"
    ]
    &
    \mathfrak{l}\mathcal{B}
  \end{tikzcd}
\end{equation}
characterizes the twisted Bianchi identities as above \cref{RelativeTwistedBianchiIdentitiesAsDiagram}.

Now given such a pair of maps $\phi$, $\wp$, consider their \emph{twisted relative mapping space}
\begin{equation}
  \mathbf{Map}(\phi, \wp)
  :=
  \mathbf{Map}\bracket({
    \Sigma
    ,
    \mathcal{A}
  })
  \underset
    {
      \mathbf{Map}\scaledbracket({
        \Sigma
        , 
        \mathcal{B}
      })
    }
    {\times}
  \mathbf{Map}\bracket({
    X
    ,
    \mathcal{B}
  })
  =
  \left\{
  \begin{tikzcd}[row sep=small]
    \Sigma
    \ar[
      r,
      dashed
    ]
    \ar[
      d,
      "{ \phi }"
    ]
    &
    \mathcal{A}
    \ar[
      d,
      "{ \wp }"
    ]
    \\
    X
    \ar[
      r,
      dashed
    ]
    &
    \mathcal{B}
  \end{tikzcd}
  \right\}
  \mathrlap{,}
\end{equation}
and similarly the \emph{twisted relative flux moduli}
\begin{equation}
\label
{TwistedRelativeOnShellFluxDensities}
  \mathbf{\Omega}^1_{\mathrm{cl}}\bracket({
    \phi;
    \mathfrak{l}\wp
  })
  :=
  \mathbf{\Omega}^1_{\mathrm{cl}}
  \bracket({
    \Sigma;
    \mathfrak{l}_{_{\mathcal{B}}}
    \mathcal{A}
  })
  \underset
    {
      \mathbf{\Omega}^1_{\mathrm{cl}}\scaledbracket({
        \Sigma; 
        \mathfrak{l}\mathcal{B}
      })
    }
    {\times}
  \mathbf{\Omega}^1_{\mathrm{cl}}
  \bracket({
    X;
    \mathfrak{l}\mathcal{B}
  })
  =
  \left\{
  \begin{tikzcd}
    \Sigma
    \ar[
      r,
      dashed
    ]
    \ar[
      d,
      "{ \phi }"
    ]
    &
    \mathbf{\Omega}^1_{\mathrm{cl}}\bracket({
      \ast;
      \mathfrak{l}_{_{\mathcal{B}}}
      \mathcal{A}
    })
    \ar[
      d,
      "{
        \scaledbracket({\mathfrak{l}\wp})_\ast
      }"
    ]
    \\
    X
    \ar[r]
    &
    \mathbf{\Omega}^1_{\mathrm{cl}}\bracket({
      \ast;
      \mathfrak{l}\mathcal{B}
    })    
  \end{tikzcd}
  \right\}
  \mathrlap{.}
\end{equation}

Applying rationalization componentwise, we obtain the twisted relative version of the differential character map and from this the twisted relative version of the phase space,  as the homotopy pullback:%
\footnote{
  We deduce \cref{RelativePhaseSpacePullback}, hence that shape preserves the homotopy fiber product \cref{TwistedRelativeOnShellFluxDensities}, by appeal to \cref{ShapePreservesFiberProductsWithSingFibrantMaps}:
  To show that
  $\mathrm{Sing}\big( (\mathfrak{l}\wp)_\ast :  \inlinetikzcd{ \mathbf{\Omega}^1_{\mathrm{cl}}\bracket({\Sigma; \mathfrak{l}_{_{\mathcal{B}}}\mathcal{A} }) \ar[r] \& \mathbf{\Omega}^1_{\mathrm{cl}}\bracket({\Sigma; \mathfrak{l}\mathcal{B} }) }\big)$
  is a Kan fibration, we need to exhibit dashed lifts in these diagrams of dgc-algebras:
  $
    \begin{tikzcd}[
      ampersand replacement=\&,
      sep=small
    ]
      \mathrm{CE}(\mathfrak{l}\mathcal{B})
      \ar[
        d,
        "{
          (\mathfrak{l}\wp)^\ast
        }"{swap}
      ]
      \ar[
        r
      ]
      \&
      \Omega^\bullet_{\mathrm{dR}}\bracket({
        \Sigma
          \times
        \boldsymbol{\Delta}^n
      })
      \ar[
        d
      ]
      \\
      \mathrm{CE}(\mathfrak{l}_{_{\mathcal{B}}}
        \mathcal{A}
      )
      \ar[
        r
      ]
      \ar[
        ur,
        dashed
      ]
      \&
      \Omega^\bullet_{\mathrm{dR}}\bracket({
        \Sigma \times
        \boldsymbol{\Lambda}^n_k
      }).
    \end{tikzcd}
  $
  But here the left map is a cofibration by assumption (formal dual of $\mathfrak{l}\wp$ being a fibration) and the right map is an acyclic fibration by the same kind of argument as previously in \cref{AcyclicDGCAFibrationArgument}. Therefore these lifts exist by the Bousfield-Gugenheim model structure.
}
\begin{equation}
\label
{RelativePhaseSpacePullback}
  \begin{tikzcd}[column sep=large]
    \mathbf{Phs}(\phi;\wp)
    \ar[r]
    \ar[d]
    \ar[
      dr,
      phantom,
      "{ \lrcorner }"{pos=.1}
    ]
    &
    \mathbf{Map}(\phi; \wp)
    \ar[
      d,
      "{
        \mathbf{ch}^{\wp}
      }"
    ]
    \\
    \mathbf{\Omega}^1_{\mathrm{cl}}
    (\phi; \mathfrak{l}\wp)
    \ar[
      r,
      "{
        \eta^{\shape}
      }"
    ]
    &
    \shape
    \mathbf{\Omega}^1_{\mathrm{cl}}
    (\phi; \mathfrak{l}\wp)
  \end{tikzcd}
\end{equation}

For more details on such twisted (differential nonabelian) cohomology see \cite[\S 3 \& \S V]{FSS23-Char}. For early discussion of plain relative cohomology applied to string/M-theory see \cite{Sati2011}.

\subsection
{CS/BF-Like Fields}
\label{OnChernSimonsType}

Above we discussed equations of motion of higher Maxwell-type. Here we generalize to the situation where some of the flux densities $F^{(i)}$ in addition satisfy the simple but consequential equation of motion of the form:
\[
  F^{(i)} = 0
  \mathrlap{\,.}
\]
Such equations of motion famously hold in  abelian Chern-Simons theories of $p$-form fluxes in $D=2p-1$, and higher BF-theories.

It is a truism that such EoM imply a \emph{topological field theory}, with the entire field content consisting of the gauge potentials and their topological charges. 
But we are now positioned to accurately say what this means globally, particularly when there is coupling to further non-topological fields.

The relevant generalization of the above discussion is straightforward: In \cref{OnThePhaseSpace} we bluntly took the map from the moduli of on-shell flux densities to the rationalized classifying space to be the canonical one, the shape unit
\begin{equation}
  \begin{tikzcd}
    \mathbf{\Omega}^1_{\mathrm{cl}}
    \bracket({
      X^d;
      \mathfrak{a}
    })
    \ar[
      r,
      "{ \eta^{\shape} }"
    ]
    &
    \shape
    \mathbf{\Omega}^1_{\mathrm{cl}}
    \bracket({
      X^d;
      \mathfrak{a}
    })   
    \mathrlap{\,.}
  \end{tikzcd}
\end{equation}
But in order to construct a phase space 
along the lines of \cref{LocalFluxQuantization}, all it takes is any map of this form.

In particular, we may consider that the characteristic $L_\infty$-algebra of the Maxwell-type higher gauge fields is not actually isomorphic to $\mathfrak{l}\mathcal{A}$, but just equipped with a morphism
\begin{equation}
\label
{MorphismFromCharacteristicTOLa}
  \begin{tikzcd}
    \mathfrak{a}
    \ar[
      r,
      "{ \prime }"
    ]
    &
    \mathfrak{l}\mathcal{A}\,.
  \end{tikzcd}
\end{equation}
Forming the corresponding homotopy pullback
immediately entails that generators annihilated by $\mathrm{CE}(\prime)$ correspond to gauge fields whose gauge potentials do appear in the phase space, but whose flux densities vanish.

In this case the conclusion \cref{ShapeOfPhaseSpace} may no longer hold, and the 0-truncated phase space may coincide with plain $\mathcal{A}$-cohomology only up to rational corrections. 

A physical application of this situation is discussed in \parencites[\S 2.4]{SS26-BBC}[\S 2]{SS25-Srni}.

\section
{Examples and Applications}
\label
{ExamplesAndApplications}

We very briefly indicate a couple of examples and applications, with pointers to further reading.
\footnote{
  \notinjournal{%
  A popular exposition of applications to grand open questions in contemporary physics is given in \cite{Schreiber2026}.
  }
  Potential relations to currently fashionable topics in theoretical physics are explored in \cite{AlfonsiKimLuciani2026}.
}

\subsection
{11D/10D Supergravity}
\label
{On11D10DSupergravity}

For (the higher gauge sectors of) supergravity (SuGra) theories descending from 11D SuGra, there is a ``covariant'' form of IR-completion (not requiring a choice of Cauchy surface). We briefly indicate how this works; dedicated lecture notes on this topic are in \cite{GSS26-SuGra}. 

\paragraph
{M-Flux Quantization in 4-Cohomotopy}

We have seen that the characteristic $L_\infty$-algebra of 11D supergravity is $\mathfrak{a} \defneq \mathfrak{l}S^4$. This means that one admissible flux quantization law in 11D (in fact the minimal one, in number of cells) is \parencites[\S 2.5]{Sati2018}{GSS24-SuGra} the actual 4-sphere
\begin{equation}
  \mathcal{A} \defneq \shape S^4
  \mathrlap{\,,}
\end{equation}
which is the classifying space for \emph{4-Cohomotopy} cohomology theory (cf. \parencites[\S VII]{STHu59}[Ex. 2.7]{FSS23-Char}):
\begin{equation}
  \pi^4(X)
  \simeq
  H^1\bracket({
    X;
    \Omega S^4
  })
  \mathrlap{\,.}
\end{equation}

This may naturally be twisted by an $\mathrm{Sp}(2)$-tangential structure, in which guise it implies \cite{FSS20-H} several subtle topological conditions expected in M-theory (cf. \cite{Witten1997Flux,Witten1997Fivebrane}), such as:
\begin{itemize}
  \item 
  \cite[Prop. 3.13]{FSS20-H}:
  the $\tfrac{1}{4}p_1$-shifted half-integral flux quantization of $G_4$,

  \item 
  \cite{FSS21-Hopf,SS21-M5Anomaly}:
  the cancellation of M5-brane anomalies.
\end{itemize}

\paragraph
{Double Dimensional Reduction}

The operation of \emph{double dimensional reduction} along principal circle fiber bundles is implemented on flux quantized higher gauge fields by \emph{cyclification} of their classifying spaces (\parencites[\S 2.2]{BMSS2019}[\S 2.2]{SS24-Cyc}):
\begin{equation}
  \mathrm{Cyc}\bracket({
    \mathcal{A}
  })
  :=
  \mathbf{Map}\bracket({
    S^1, 
    \mathcal{A}
  })
  \sslash 
  S^1
  \mathrlap{\,,}
\end{equation}
in that for $S^1$-principal bundles
\begin{equation}
  \begin{tikzcd}[
    row sep=small, 
    column sep=large
  ]
    X^{10}
    \ar[
      d
    ]
    \ar[
      dr,
      phantom,
      "{ \lrcorner }"{pos=.1}
    ]
    \ar[r]
    &
    \ast
    \ar[d]
    \\
    X^9
    \ar[
      r,
      "{
        c_1
      }"
    ]
    &
    B S^1
  \end{tikzcd}
\end{equation}
there is a natural equivalence of classifying maps
\begin{equation}
  \big\{
  \begin{tikzcd}
    X^{10}
    \ar[
      r,
      dashed
    ]
    &
    \mathcal{A}
  \end{tikzcd}
  \big\}
  \;\;\simeq\;\;
  \left\{
  \begin{tikzcd}[
    column sep=15pt,
    row sep=-1pt
  ]
    X^{9}
    \ar[
      rr,
      dashed
    ]
    \ar[
      dr,
      "{ c_1 }"{swap}
    ]
    &&
    \mathrm{Cyc}\bracket({\mathcal{A}})
    \ar[dl]
    \\
    &
    B S^1
  \end{tikzcd}
  \right\}
\end{equation}
with compatible rationalization.
In particular, the characteristic $L_\infty$-algebra (\cref{CharacteristicLInfinityAlgebra})  of IIA fluxes obtained by plain reduction from 11D \cref{RestrictedEoMOfIIA} is \parencites[\S 3]{FSS17-Sphere}[Ex. 2.5]{SatiVoronov2024}[Ex. 3.10]{GSS25-TD} the Whitehead bracket $L_\infty$-algebra of the cyclic loop space of the 4-sphere:
\begin{equation}
\label
{lCycS4}
  \mathfrak{a}
  \defneq
  \mathfrak{l}\mathrm{Cyc}\bracket({S^4})
  \mathrlap{\,,}
  \;\;\;\;
  \mathrm{CE}\bracket({
    \mathfrak{l}
    \mathrm{Cyc}\bracket({S^4})
  })
  \simeq
  \mathbb{R}_{_{\mathrm{d}}}
  \left[
  \begin{aligned}
    & f_2
    \\
    & f_4
    \\
    & f_6
    \\
    & h_3
    \\
    & h_7
  \end{aligned}
  \right]
  \Big/
  \left(
  \begin{aligned}
    \mathrm{d}\, f_2
    & = 0
    \\
    \mathrm{d}\, f_4
    & = h_3 \wedge f_2
    \\
    \mathrm{d}\, f_6
    & = h_3 \wedge f_4
    \\
    \mathrm{d}\, h_3 & = 0
    \\
    \mathrm{d}\, h_7 & = 
    \tfrac{1}{2}
    f_4 f_4 - f_2 f_6
  \end{aligned}
  \right)
  \mathrlap{.}
\end{equation}

\paragraph
{IIA Flux Quantization in Twisted Unstable K-Theory}
Another admissible flux quantization law in 11D turns out to be $\mathcal{A} \defneq B \widehat{\mathrm{SU}(2)}$, the homotopy fiber of the square of the second Chern class $c_2$ on $B \mathrm{SU}(2)$:
\begin{equation}
  \begin{tikzcd}[row sep=small, column sep=large]
    B \widehat{\mathrm{SU}(2)}
    \ar[r]
    \ar[d]
    \ar[
      dr,
      phantom,
      "{ \lrcorner }"{pos=.1}
    ]
    &
    \ast
    \ar[d]
    \\
    B \mathrm{SU}(2)
    \ar[
      r,
      "{ (c_2)^2 }"
    ]
    & 
    B^8 \mathbb{Z}
    \mathrlap{\,.}
  \end{tikzcd}
\end{equation}
The canonical comparison map
induces isomorphisms on $\pi_{\leq 7}$:
\begin{equation}
  \begin{tikzcd}
    S^4 
    \ar[
      r, 
      "{
        \sim_{\leq 7}
      }"
    ]
    &
    B \widehat{\mathrm{SU}(2)}
    \mathrlap{\,.}
  \end{tikzcd}
\end{equation}

Under double dimensional reduction, this entails a flux quantization of type IIA supergravity in $\mathrm{Cyc}\bracket({ B \widehat{\mathrm{SU}(2)} })$, which receives a canonical map
\begin{equation}
  \begin{tikzcd}
    B \widehat{\mathrm{U}(2)}
    \ar[r]
    &
    \mathrm{Cyc}\bracket({
      B \widehat{\mathrm{SU}(2)}
    })
    \mathrlap{\,,}
  \end{tikzcd}
\end{equation}
exhibiting it as a twisted deformation of ``unstable K-theory'' \cite{BaSS26-UnstableK}. 

This yields (cf. \cite{GS26-IIA}) a consistent IR-completion of 10D supergravity of type IIA which retains an (unstable, deformed) semblance of the famously conjectured flux quantization in twisted K-theory, but is compatible with oxidation to globally completed 11D supergravity.

\paragraph
{Superspace formulation}

The development of the higher gauge theory in \cref{OnTheGlobalPhaseSpace} capitalized on the fact (\cref{GaussLawAsLInfinityClosure}) that (germs of) solutions to higher Maxwell-type equations of motion are equivalent to closed $L_\infty$-valued forms when restricted to a Cauchy surface (solutions to the electromagnetic Gauss law constraints). Here it is the implicit time-evolution \cref{MapToCauchyData} away from the Cauchy surface which absorbs the Hodge duality constraint from the higher Maxwell-type equations and leaves the purely cohomological $L_\infty$-closure condition that allows for cohomological flux quantization of the theory (\cref{OnNonabelianCharacter}).

Remarkably, for some supergravity theories there is an alternative \emph{covariant} form --- meaning: not requiring existence or choice of Cauchy surfaces ---  of this ``cohomologization by metric decoupling''. 

This requires formulating the theory on super-spacetime and constraining the fermionic form of the super-flux densities (cf. \cite[\S 2]{GSS24-SuGra} for relevant background on supergeometry):

Specifically, consider a super-spacetime $X^{({1,10\vert \mathbf{32}})}$ with super-torsion free super-coframe 
\begin{equation}
\label
{SuperCoframe}
  \left(
    \begin{aligned}
    &
    \bracket({
      E^a
      \in
      \Omega^{1,\mathrm{evn}}\bracket({
        X^{1,10\vert \mathbf{32}}
      })
    })_{a=0}^{10}, 
    \\
    &
    \bracket({
      \Psi^\alpha
      \in
      \Omega^{1,\mathrm{odd}}\bracket({
        X^{1,10\vert \mathbf{32}}
      })
    })_{\alpha=1}^{32}
    \end{aligned}
  \right)
\end{equation}
and promote the ordinary duality-symmetric flux density forms $G_4$ and $G_7$ \cref{11DSuGraEoM} to differential superforms $G^s_4$ and $G^s_7$ with these specific components:
\begin{equation}
\label
{11DSuperFluxDensities}
  \begin{aligned}
    G^s_4
    & 
    :=
    \tfrac{1}{4!}
    \bracket({
      G_4
    })_{a_1 \cdots a_4}
    E^{a_1}
    \cdots
    E^{a_4}
    +
    \tfrac{1}{2}\bracket({
      \overline{\Psi}
      \Gamma_{a_1 a_2}
      \Psi
    })
    E^{a_1}
    E^{a_2}
    \\
    G^s_7
    & 
    :=
    \tfrac{1}{7!}
    \bracket({
      G_7
    })_{a_1 \cdots a_7}
    E^{a_1}
    \cdots
    E^{a_7}
    +
    \tfrac{1}{5!}\bracket({
      \overline{\Psi}
      \Gamma_{a_1 \cdots a_5}
      \Psi
    })
    E^{a_1}
    \cdots
    E^{a_5}
    \mathrlap{\,.}
  \end{aligned}
\end{equation}
Then we have the following remarkable result, following \parencites{CremmerFerrara1980}{BrinkHowe1980}[\S 6]{CandielloLechner1994}[\S III.8.5]{CDF1991}:
\begin{proposition}[{\cite[Thm. 3.1]{GSS24-SuGra} reviewed in more detail as \cite[Thm. 3.16]{GSS26-SuGra}}]
Fields $E$ \textup{(graviton)}, $\Psi$ \textup{(gravitino)} and $G_4, G_7$ \textup{(duality-symmetric C-field)} satisfy the equations of motion of 11D Supergravity iff the corresponding super-flux densities \cref{11DSuperFluxDensities} are $\mathfrak{l}S^4$-closed, hence iff
\begin{equation}
  \bracket({
    G^s_4, G^s_7
  })
  \in
  \Omega^1_{\mathrm{cl}}\bracket({
    X^{1,10\vert \mathbf{32}};
    \mathfrak{l}S^4
  })
  \mathrlap{\,.}
\end{equation}
\end{proposition}
In particular, the Hodge-duality relation \cref{11DSuGraEoM}
is implied \parencites[p. 878]{CDF1991}[Lem. 3.3(iii)]{GSS24-SuGra}:
\begin{equation}
  (G^s_4, G^s_7) \in \Omega^1_{\mathrm{cl}}\bracket({X^{1,10\vert \mathbf{32}}; \mathfrak{l}S^4})
  \;\;\Rightarrow\;\;
  G_7 = \star G_4 
  \mathrlap{\,.}
\end{equation}

This entails that 11D SuGra admits \emph{covariant} flux quantization (not requiring existence or choice of Cauchy surfaces) along the same lines as \cref{OnTheGlobalPhaseSpace} but replacing the Cauchy surface $X^{10}$ by super-spacetime $X^{1,10\vert \mathbf{32}}$ and using only the sub-sheaf of $\mathfrak{l}S^4$-closed super-forms whose components are of the form \cref{11DSuperFluxDensities}:
\begin{equation}
\begin{aligned}
  \Omega^1_{%
    \mathrm{cl},%
    \mathcolor{purple}{s}%
  }
  \bracket({
    X^{1,10\vert \mathbf{32}};
    \mathfrak{l}S^4
  })
  &
  :=
  \bracketmid\{{
    (G^s_4, G^s_7)
    \in
    \Omega^1_{\mathrm{cl}}\bracket({
    X^{1,10\vert \mathbf{32}};
    \mathfrak{l}S^4
  })
  }{
    \text{of the form \cref{11DSuperFluxDensities}}
  }\}
  \mathrlap{\,,}
  \\
  \mathbf{\Omega}^1_{%
    \mathrm{cl},%
    \mathcolor{purple}{s}%
  }
  \bracket({
    X^{1,10\vert \mathbf{32}}
    ;
    \mathfrak{l}S^4
  })
  & 
  :=
  \mathbf{\Omega}^1_{\mathrm{cl}}
  \bracket({
    X^{1,10\vert \mathbf{32}};
    \mathfrak{l}S^4
  })
  \underset
    {\;\;\;\;\mathclap{
      \sharp_1
      \mathbf{\Omega}^1_{\mathrm{cl}}
      \scaledbracket({
       X^{1,10\vert \mathbf{32}};
        \mathfrak{l}S^4
      })
    }\;\;\;\;}
    {\times}
  C^\infty\bracket({
    -,
    \Omega^1_{%
      \mathrm{cl},%
      \mathcolor{purple}{s}%
    }
    \bracket({
      X^{1,10\vert \mathbf{32}};
      \mathfrak{l}S^4
    })
  })
  \mathrlap{\,.}
\end{aligned}
\end{equation}

Beware that this smooth moduli set of constrained super-flux densities depends on the choice of super-coframe \cref{SuperCoframe} which encodes the super-gravitational background field. But in this form it is straightforward to generalize to the smooth moduli set of compatible coframe fields and flux fields.
This leads to flux quantization of 11D SuGra beyond just its gauge sector. 

Such super-space flux quantization is not restricted to 11D supergravity:
\begin{enumerate}
\item
  After promoting the M5-brane embedding (more on which in \cref{OnM5BraneProbes}) to a $\sfrac{1}{2}$-BPS \emph{super-embedding}
  \begin{equation}
    \inlinetikzcd{
      \Sigma^{1,5\vert 2\cdot\mathbf{8}}
      \ar[
        rr,
        hook,
        "{ \phi_s }"
      ]
      \&\&
      X^{1,10\vert \mathbf{32}}
    }
  \end{equation}
  of a super-worldvolume with coframe field $\bracket({ \bracket({e^a})_{a=0}^5, \bracket({ \psi^\alpha })_{\alpha=1}^{16} })$,
  the notorious non-linear self-duality (cf. \cite[Rem. 3.19]{FSS21-StrStruc}) of the 3-form flux $H_3$ on the 5-brane is \emph{implied} (\parencites[Prop. 3.18(iii), Rem. 3.20]{FSS21-StrStruc} following \cite{HoweSezgin1997,HoweSezginWest1997}) by imposing on its super-flux
  \begin{equation}
    H^s_3 
    :=
    \bracket({
      H_3
    })_{a_1 a_2 a_3}
    e^{a_1 a_2 a_3}
  \end{equation}
  the $\mathfrak{l}_{S^4}S^7$-valued Bianchi identity
  \begin{equation}
    \mathrm{d}
    H^s_3
    =
    \phi_s^\ast
    G^s_4
    \mathrlap{\,,}
  \end{equation}
  hence by the condition that
  \begin{equation}
    \bracket({
      \bracket({G_4^s, G_7^s}),
      H_3^s
    })
    \in
    \Omega^1_{\mathrm{cl}}\bracket({
      \phi_s;
      \mathfrak{l}h_{\mathbb{H}}
    })
  \end{equation}
  constitute twisted relative closed differential forms \cref{TwistedRelativeOnShellFluxDensities} with coefficients in the quaternionic Hopf fibration $\inlinetikzcd{ S^7 \ar[r, "{ h_{\mathbb{H}} }"] \& S^4 }$ (cf. \parencites[\S3.7]{FSS20-H}{FSS21-Hopf}).

\item
After dimensional reduction on superspace, an analogous statement \cite{GS26-IIA} characterizes the actual equations of motion \cref{EoMOfIIA} of 10D IIA supergravity as equivalent to the constraint that the corresponding super-flux densities constitute a closed $\mathfrak{l}\mathrm{Cyc}\bracket({S^4})$-valued \cref{lCycS4} differential form on super-spacetime $X^{1,9\vert \mathbf{16} \oplus \overline{\mathbf{16}}}$:
\begin{equation}
  \bracket({
    F^s_2,
    F^s_4,
    F^s_6,
    H^s_3,
    H^s_7
  })
  \in
  \Omega^1_{\mathrm{cl}}\bracket({
    X^{1,9\vert \mathbf{16}\oplus \overline{\mathbf{16}}};
    \mathfrak{l}\mathrm{Cyc}\bracket({
      S^4
    })
  })
  \mathrlap{\,.}
\end{equation}
In this manner, the phenomenon will propagate to at least all those supergravity theories that arise as reductions from 11D, rendering them all amenable to global completion by covariant electromagnetic flux quantization.
\end{enumerate}

\subsection
{M5-Brane Probes}
\label
{OnM5BraneProbes}

We have seen in \cref{FluxOnM5Branes} that the characteristic $L_\infty$-fibration for the (non-linearly) self-dual flux density on M5-probes of the 11D bulk is the relative real Whitehead $L_\infty$-algebra of the quaternionic Hopf fibration. This means that the initial flux quantization of such M5-probes is classified by the actual quaternionic Hopf fibration. We briefly indicate some interesting consequences; for more see the dedicated lecture notes in \parencites[\S 3.3 \& \S 4.3]{GSS26-SuGra}.

Putting the M5, thus flux-quantized, on an A-type orbifold turns out \cite{SS25-Seifert,BaSS26-MString} to yield the geometric engineering of topological order which reproduces in fine detail properties of anyons in fractional quantum Hall (FQH) systems (cf. \cite{SS25-AbelianAnyons,SS25-FQH,SS26-BBC}, exposition in \cite{SS25-ISQS29,SS25-Srni,SS25-Rickles}). 
Briefly: 
\begin{enumerate}
\item
For an M5 worldvolume embedding of the form (cf. \cite[p. 5]{CouzensLuscherSparks2025})
\begin{equation}
  \begin{tikzcd}[
    column sep=-4pt
  ]
  \Sigma^{1,5}
  \ar[
    d,
    hook,
    "{ \phi }"
  ]
  &\defneq&
  \mathbb{R}^{1,1}
  &\times&
  \Sigma^2
  &\times&
  \mathbb{C}
  \ar[
    in=60,
    out=180-60,
    looseness=4,
    shift right=4pt,
    "{ 
      \,\mathclap{G}\,
    }"{description}
  ]
  \\
  X^{1,10}
  &\defneq&
  \mathbb{R}^{1,1}
  &\times&
  \mathbb{R}^{5}
  &\times&
  \mathbb{C}
  \ar[
    in=60,
    out=180-60,
    looseness=4,
    shift right=4pt,
    "{ 
      \,\mathclap{G}\,
    }"{description}
  ]
  &\times&
  \mathbb{C}
  \ar[
    in=60,
    out=180-60,
    looseness=4,
    shift right=4pt,
    "{ 
      \,\mathclap{\overline{G}}\,
    }"{description}
  ]
  \end{tikzcd}
\end{equation}
(where $G \defneq \mathbb{Z}_{/n} \subset \mathrm{SU}(2)$, $n \geq 2$, and $\overline{G}$ its opposite).
\item 
For flux quantization by the \emph{twistor fibration} with its canonical $G$-action on one of the two $\mathbb{C}^2$-factors:
\begin{equation}
  \begin{tikzcd}[
    column sep=-5pt,
    row sep=6pt
  ]
    &&
    \mathbb{C}P^3
    \ar[
      out=60,
      in=180-60,
      shift left=1pt,
      looseness=4,
      "{ 
        \,\mathclap{G}\,
      }"{description}
    ]
    \ar[rrrr, ->>, "{ t_{\mathbb{H}} }"]
    \ar[d, equals]
    &&[27pt]&&
    S^4
    \ar[
      out=60,
      in=180-60,
      shift left=1pt,
      looseness=4,
      "{ 
        \,\mathclap{G}\,
      }"{description}
    ]
    \ar[d, equals]
    \\
    \big(
    \mathbb{C}^2 
    \ar[
      out=60-180,
      in=-60,
      shift left=1pt,
      looseness=4,
      "{ 
        \,\mathclap{G}\,
      }"{description}
    ]
    &\times& \!\!
    \mathbb{C}^2
    &\!\!\!\!\setminus \{0\}
    \big)\big/\mathbb{C}^\times
    \ar[r, ->>]
    &
    \big(
    \mathbb{C}^2 
    \ar[
      out=60-180,
      in=-60,
      shift left=1pt,
      looseness=4,
      "{ 
        \,\mathclap{G}\,
      }"{description}
    ]
    &\times&
    \mathbb{C}^2
    &\!\setminus \{0\}
    \big)\big/\mathbb{H}^\times
    \mathrlap{\,.}
  \end{tikzcd}
\end{equation}
\end{enumerate}
For trivial background C-field charge,
the resulting classifying space on the orbi-singularity $\bracket({\Sigma^{1,5}})^G \simeq \mathbb{R}^{1,1} \times \Sigma^2$ turns out to be essentially the 2-sphere $S^2 \simeq \mathbb{C}P^1$:
\begin{equation}
  \begin{aligned}
  \bracket({\mathbb{C}P^3})^{
    \mathbb{Z}_{/n}
  }
  &
  \simeq
  \begin{cases}
    \begin{aligned}
    &
    \bracketmid\{{
      \bracket[{
        0 : 0 : w_1 : w_2
      }]
    }{
      (w_1,w_2) \in 
      \mathbb{C}^2 \setminus \bracket\{{(0,0)}\}
    }\}
    \\
    \sqcup
    \;
    &
    \bracketmid\{{
      \bracket[{
        z_1\, : z_2\, : 0 : 0
      }]
    }{
      (z_1\,,z_2\,) \in 
      \mathbb{C}^2 \setminus \bracket\{{(0,0)}\}
    }\}
    \end{aligned}
    &
    \text{ if } n = 2
    \\[20pt]
    \begin{aligned}
    &
    \bracketmid\{{
      \bracket[{
        0 : 0 : w_1 : w_2
      }]
    }{
      (w_1,w_2) \in 
      \mathbb{C}^2 \setminus \bracket\{{(0,0)}\}
    }\}
    \\
    \sqcup
    \;
    &
    \bracket[{
      1 : 0 : 0 : 0
    }]
    \\
    \sqcup
    \;
    &
    \bracket[{
      0 : 1 : 0 : 0
    }]
    \end{aligned}
    &
    \text{ if } n \geq 3
  \end{cases}
  \\
  & 
  \simeq
  \begin{cases}
    \mathbb{C}P^1 \sqcup \mathbb{C}P^1
    &
    \text{ if } n = 2
    \\
    \mathbb{C}P^1 \sqcup \ast \sqcup \ast
    &
    \text{ if } n \geq 3
    \mathrlap{\,.}
  \end{cases}
  \end{aligned}
\end{equation}

Since the basepoint $[0:0:1:0]$ is inside one of these fixed $\mathbb{C}P^1$, the solitonic charges over a connected domain in this situation are equivalently classified by pointed maps:
\begin{equation}
  \begin{tikzcd}
    \mathbb{R}^{1}_{\cpt}
    \wedge 
    \Sigma^2_{\cpt}
    \ar[
      r,
      dashed
    ]
    &
    S^2
    \mathrlap{\,.}
  \end{tikzcd}
\end{equation}
Therefore, the light-cone quantum algebra of the realistic solitonic topological observables \cref{SolitonicTopologicalObservables,SolitonicTopologicalObservablesAsGroupAlgebra} on such M5-brane probes is the group algebra 
\begin{equation}
  \mathrm{SolTopObs}_0\bracket({
    \mathbb{R}^1 \times \Sigma^2_{\cpt};
    S^2
  })
  \simeq
  \mathbb{C}\bracket[{
    \pi_1
    \,
    \mathrm{Map}^\ast\bracket({
      \Sigma^2_{\cpt}
      ,
      S^2
    })
  }]
  \mathrlap{\,.}
\end{equation}
of the fundamental group of the pointed mapping space from the given surface, with the point at infinity adjoined, to the 2-sphere.

This turns out to reproduce in fine detail \cite{SS25-FQAH,SS25-FQH,SS25-Crys} the quantum observables of abelian Chern-Simons theory, hence of the topological quantum order of fractional quantum Hall (FQH) systems on $\Sigma^2$ (cf. string-theoretic conjectures along these lines in \cite{ChoGangKim2020} and see review in \cite{SS25-ISQS29}).

Notably: \onlyinjournal{\footnote{
  These applications of the IR-completed M5-brane to topological quantum order of fractional quantum (anomalous) Hall systems are novel. On this topic, at the time of this writing, the articles 
  \parencites
    {SS25-ISQS29}
    {SS25-Srni}
    {SS25-Rickles}
    {SS25-AbelianAnyons}
    {SS25-FQAH}
    {SS26-Islands}
  have been accepted for publication after peer-review, while the articles 
  \parencites
    {SS25-FQH}
    {SS25-Crys}
    {SS26-BBC}
    {KSS26-HigherDimAnyons}
  are still under review. 
}
}
\begin{enumerate}
\item 
For $\Sigma^2 \defneq \mathbb{R}^2$ the plane, one finds (\parencites{SS25-WilsonLoops}{SS25-AbelianAnyons}) that:
\begin{equation}
  \mathrm{SolTopObs}_0\bracket({
    \mathbb{R}^1
    \times
    \mathbb{R}^2;
    S^2
  })
  \simeq
  \big\{
  \text{
    renormalized Wilson
    loop observables
  }
  \big\}
  \mathrlap{\,.}
\end{equation}
is exactly the algebra of observables whose pure states detect the \emph{braiding phases} $\zeta = e^{\pi \mathrm{i} n/K}$ of anyon worldlines (framed links) with net crossing number (writhe) $n \in \mathbb{Z}$ in the FQH state with filling fraction $\nu = 1/K$.

\item
For $\Sigma^2 \defneq T^2$ the torus, one finds (\cite{Kallel2025,LarmoreThomas1980,Hansen1974}, cf. \cite[Thm. 3.5]{KSS26-HigherDimAnyons}) that 
\begin{equation}
  \mathrm{SolTopObs}_0\bracket({
    \mathbb{R}^1
    \times
    T^2;
    S^2
  })
  \simeq
  \mathbb{C}\bracket[{%
    \text{%
      integer Heisenberg group at level 2%
    }%
  }]
\end{equation}
is the group algebra of the \emph{integer Heisenberg group} at level=2. This is exactly (\cite{SS25-FQAH}) the algebra of quantum observables for FQH anyons on a torus (\parencites[(4.9)]{WenNiu1990}[(4.14)]{Iengo1992}, cf. \parencites[(4.21)]{Fradkin2013}[(5.28)]{Tong2016}), such as on the Brillouin torus of fractional Chern insulators.

\item
\label
{SuperconductingIslands}
For $\Sigma^2 \defneq S^2 \setminus \{x_0, \cdots, x_n\}$ the $n$-punctured open disk---modeling an FQH liquid with flux-expelling (for instance superconducting) islands---the covariantized topological observables turn out to allow \emph{nonabelian} anyons attached to these islands \cite{SS26-Islands}.  This is a novel prediction that potentially opens a new experimental pathway to realizing nonabelian anyonic states in otherwise abelian FQH quantum materials, which is the anyon variant thought to be necessary for universal topological quantum computing.  
\end{enumerate}

\medskip

We intentionally do not go into further detail here, so as not to overburden these lecture notes. Extensive discussion may be found in the references cited above.

\appendix

\section
{Basic Background}

For reference in the main text, here we briefly recall some relevant basic notions from higher homotopy and higher algebra.

\subsection
{Higher Lie Algebras}
\label{OnLInfinityAlgebras}

We briefly recall $L_\infty$-algebras (\cite{LadaMarkl1995}, cf. \parencites{Reinhold2019}{KraftSchnitzer2024}{CastellaniDAuria2025}) in the special case of \emph{finite type}, where they are dually incarnated in the guise of their Chevalley-Eilenberg dgc-algebras (cf. \parencites{LadaStasheff1992}[\S 6.1]{SatiSchreiberStasheff2009}[\S 4]{FSS23-Char}). This dual formulation is more concise and transparent than the ``direct'' one. In particular, it makes immediate the generalization to super-$L_\infty$-algebras \parencites{FSS15-WZW}[\S 3]{FSS2019}[\S 2.1]{GSS25-TD}, which in this dual guise are (mis-)named ``FDAs'' in the supergravity literature  (cf. \parencites{vanNieuwenhuizen1983}{CDF1991}{Castellani2018}{FSS2019}{CastellaniDAuria2025}).

Our ground field is the real numbers, $\mathbb{R}$. By \emph{graded} we mean $\mathbb{Z}$-graded (though all constructions considered here happen to be just $\mathbb{N}$-graded). We say \emph{gc-algebras} for graded-commutative algebras and \emph{dgc-algebras} for differential graded-commutative algebras with differential of degree +1.

\paragraph
{The base case of ordinary Lie algebras}
For $\bracket({\mathfrak{g},[-,-]})$ a finite-dimensional Lie algebra, recall the construction of its \emph{Chevalley-Eilenberg algebra} $\mathrm{CE}(\mathfrak{g})$, cf. \parencites[\S 6.7]{deAzcarragaIzquierdo1995}[\S 3.2]{deAzcarragaIzquierdoPerezBuone1998}:

Write $\mathfrak{g}^\vee := \mathfrak{g}^\ast$ for the linear dual of the underlying vector space and 
\begin{equation}
\label
{GrassmannAlgebraOnDualLieAlgebra}
  \wedge^\bullet \mathfrak{g}^\vee
  :=
  \mathrm{Sym}\bracket({
    \mathfrak{g}^\vee[1]
  })
\end{equation}
for its \emph{Grassmann algebra}, the free graded-commutative algebra for which $\wedge^n \mathfrak{g}^\vee$ is the subspace of degree $n$.

A graded derivation $\mathrm{d}$ on $\wedge^\bullet \mathfrak{g}^\vee$ is thus determined by its restriction $\inlinetikzcd{ \wedge^1 \mathfrak{g}^\vee \ar[r] \& \wedge^\bullet \mathfrak{g}^\vee }$. Taking this restriction to be the linear dual of the Lie bracket,
$
  \mathrm{d}
    _{\vert \wedge^1 \mathfrak{g}^\vee}
  :=
  [-,-]^\ast
$,
yields a graded derivation of degree 1. One checks that its nilpotency is equivalent to the Jacobi identity of the bracket:
\begin{equation}
  \mathrm{d}^2 = 0
  \;\;
  \Leftrightarrow
  \;\;
  \text{
   Jacobi identity on $[-,-]$}
  \mathrlap{\,.}
\end{equation}
The resulting dgc-algebra
is the \emph{Chevalley-Eilenberg algebra} of the Lie algebra:
\begin{equation}
\label
{CEAlgebraOfLieAlgebra}
  \mathrm{CE}\bracket({
    \mathfrak{g},
    [-,-]
  })
  :=
  \bracket({
    \wedge^\bullet
    \mathfrak{g}^\vee,
    \mathrm{d} 
    :=
    [-,-]^\ast
  })
  \mathrlap{\,.}
\end{equation}
The key observation now is that this construction is a \emph{fully faithful} functor on the category $\mathrm{LieAlg}^{\mathrm{fd}}_{_{\mathbb{R}}}$ of finite-dimensional Lie algebras: 
\begin{equation}
  \begin{tikzcd}
    \mathrm{LieAlg}
      ^{\mathrm{fd}}
      _{_{\mathbb{R}}}
    \ar[r, hook, "{ \mathrm{CE} }"]
    &
    \mathrm{dgcAlg}^{\mathrm{op}}
     _{_{\mathbb{R}}}
    \mathrlap{\,.}
  \end{tikzcd}
\end{equation}

This means that we may work (op-)equivalently with the dgc-algebras $\mathrm{CE}(\mathfrak{g})$. But for these there is an evident generalization, simply by allowing the vector space $\mathfrak{g}$ to be $\mathbb{N}$-graded:

\paragraph
{The case of $L_\infty$-algebras}

For $\mathfrak{g}$ an $\mathbb{N}$-graded vector space of finite type (meaning that each degree $\mathfrak{g}_n$ is a finite-dimensional vector space), write $\mathfrak{g}^\vee$ for its degreewise dual vector space, hence with 
\begin{equation}
  \bracket({
    \mathfrak{g}^\vee
  })_n
  :=
  \bracket({\mathfrak{g}_n})^\ast
\end{equation}
and write
\begin{equation}
\label
{GradedGrassmannAlgebra}
  \wedge^\bullet \mathfrak{g}^\vee
  :=
  \mathrm{Sym}\bracket({
    \mathfrak{g}^\vee[1]
  })
\end{equation}
for the graded Grassmann gc-algebra, which in low degrees is
\begin{equation}
  \bracket({
    \wedge^\bullet \mathfrak{g}^\vee
  })_1
  =
  \wedge^1 \mathfrak{g}^\vee_0
  \,,
  \;\;\;
  \bracket({
    \wedge^\bullet \mathfrak{g}^\vee
  })_2
  =
  \wedge^2 \mathfrak{g}^\vee_0
  \oplus
  \wedge^1 \mathfrak{g}^\vee_1
  \,,
  \;\;
  \text{etc.}
\end{equation}

Then the category of (connective) finite-type $L_\infty$-algebras is the full subcategory of $\mathrm{dgcAlg}^{\mathrm{op}}_{_{\mathbb{R}}}$ on those whose underlying gc-algebra is of the form \cref{GradedGrassmannAlgebra}:
\begin{equation}
\label
{LInfinityFaithfulInDGCAlgOp}
  \begin{tikzcd}
    L_\infty \mathrm{Alg}^{\mathrm{ft}}
      _{_{\mathbb{R}}}
    \ar[
      r,
      hook,
      "{ \mathrm{CE} }"
    ]
    &
    \mathrm{dgcAlg}^{\mathrm{op}}
      _{_{\mathbb{R}}}
    \mathrlap{\,.}
  \end{tikzcd}
\end{equation}
Unwinding this: Since a differential $\mathrm{d}$  on \cref{GradedGrassmannAlgebra} is again determined by its restriction to $\wedge^1 \mathfrak{g}^\vee$, it is given by a sequence of linear duals of $n$-ary graded-skew symmetric linear maps for $n \in \mathbb{N}$:
\begin{equation}
  \mathrm{d}_{\vert
    \wedge^1\mathfrak{g}^\vee
  }
  =
  [-]^\ast
  +
  [-,-]^\ast
  +
  [-,-,-]^\ast
  +
  \cdots
  :
  \inlinetikzcd{
    \wedge^1 \mathfrak{g}^\vee
    \ar[
      r
    ]
    \&
    \wedge^1 \mathfrak{g}^\vee
    \oplus
    \wedge^2 \mathfrak{g}^\vee
    \oplus
    \wedge^3 \mathfrak{g}^\vee
    \oplus 
    \cdots
    ,
  }  
\end{equation}
and the nilpotency condition $\mathrm{d}^2 = 0$ is hence equivalently a (somewhat complicated-looking) generalization of the Jacobi identity on this system of brackets. Such a system of brackets on graded vector spaces $\mathfrak{g}$ that  satisfies these conditions is exactly what makes the structure $\bracket({ \mathfrak{g}, \bracket({[-], [-,-], [-,-,-], \cdots  }) })$ an $L_\infty$-algebra.

\paragraph
{Closed $L_\infty$-Algebra Valued Forms}
Consider $X \in \mathrm{SmthMfd}$.
Given again an ordinary Lie algebra $\mathfrak{g}$ of finite dimension, it is immediate that the $\mathfrak{g}$-valued differential 1-forms on $X$ are equivalently gc-algebra homomorphisms out of the dual Grassmann algebra \cref{GrassmannAlgebraOnDualLieAlgebra}, like this:
\begin{equation}
  \Omega^1_{\mathrm{dR}}\bracket({
    X;
    \mathfrak{g}
  })
  \simeq
  \mathrm{Hom}\bracket({
    \wedge^\bullet \mathfrak{g}^\vee,
    \Omega^\bullet_{\mathrm{dR}}\bracket({
      X
    })
  })
  \mathrlap{\,.}
\end{equation}
Moreover, a moment of reflection shows that the \emph{flat} \parencites[(261)]{SatiSchreiberStasheff2009} or \emph{Maurer-Cartan} (\parencites[proof of Prop. 1.1]{Getzler2009}[Prop. 6.1]{Berglund2015}, cf. also \parencites[Def. 5.1]{Lazarev2013}[\S 2.1]{Yalin2016}, here we will say: \emph{closed}) $\mathfrak{g}$-valued forms among these correspond exactly to those gc-homomorphisms which are in fact dgc-homomorphisms out of the Chevalley-Eilenberg algebra \cref{CEAlgebraOfLieAlgebra} : 
\begin{equation}
  \begin{aligned}
  \Omega^1_{\mathrm{cl}}\bracket({
    X;
    \mathfrak{g}
  })
  &
  :=
  \bracketmid\{{
    A 
      \in 
    \Omega^1_{\mathrm{dR}}\bracket({
      X;
      \mathfrak{g}
    })
  }{
    \mathrm{d}\, A 
    + 
    \tfrac{1}{2}
    [A \wedge A]
    = 0
  }\}
  \\
  & 
  \simeq
  \mathrm{Hom}\bracket({
    \mathrm{CE}(\mathfrak{g}),
    \Omega^\bullet_{\mathrm{dR}}\bracket({
      X
    })
  })
  \mathrlap{\,.}
  \end{aligned}
\end{equation}

Taking this characterization as the definition, it immediately generalizes to the case where $\mathfrak{g}$ is a finite-type $L_\infty$-algebra. We say that dgc-homomorphisms out of its Chevalley-Eilenberg algebra \cref{LInfinityFaithfulInDGCAlgOp} into the de Rham dgc-algebra are the closed (flat) $\mathfrak{g}$-valued differential forms:
\begin{equation}
\label
{ClosedLInfinityValuedDifferentialForms}
  \Omega^1_{\mathrm{cl}}\bracket({
    X;
    \mathfrak{g}
  })
  :=
  \mathrm{Hom}\bracket({
    \mathrm{CE}(\mathfrak{g}),
    \Omega^\bullet_{\mathrm{dR}}(X)
  })
  \mathrlap{\,.}
\end{equation}

\subsection
{Classifying Spaces}

We briefly recall the notion of \emph{classifying space}, beginning with its origin in ordinary abelian cohomology and nonabelian cohomology, and then passing to classifying spaces for generalized abelian and nonabelian cohomology.

\paragraph
{Looping and Delooping}

For $\mathcal{X}, \mathcal{A} \in \mathrm{Grpd}_\infty$ equipped with base points $\infty \in \mathcal{X}$ and $0 \in \mathcal{A}$, their \emph{pointed mapping space} is the homotopy fiber \cref{HomotopyPullback} of the point evaluation map on the plain mapping space:
\begin{equation}
  \begin{tikzcd}[row sep=small]
  \mathbf{Map}^\ast\bracket({
    \mathcal{X},
    \mathcal{A}
  })
  \ar[r]
  \ar[d]
  \ar[
    dr,
    phantom,
    "{ \lrcorner }"{pos=.1}
  ]
  &
  \mathbf{Map}\bracket({
    \mathcal{X},
    \mathcal{A}
  })
  \ar[
    d,
    "{
      \mathrm{ev}_\infty
    }"
  ]
  \\
  \{0\}
  \ar[r]
  &
  \mathcal{A} \mathrlap{\,.}
  \end{tikzcd}
\end{equation}

In particular, the \emph{based loop space} of $(\mathcal{A},0)$ is
\begin{equation}
  \Omega\mathcal{A}
  \simeq
  \mathbf{Map}^\ast\bracket({
    S^1,
    \mathcal{A}
  })
  \mathrlap{\,.}
\end{equation}
Under concatenation and reversal of loops, the based loop space is an $\infty$-group (an $A_\infty$-monoid whose connected components form a group), and on pointed connected 
$\infty$-groupoids this construction is actually an equivalence (the \emph{May recognition theorem} \cite{May1972}):
\begin{equation}
\label
{PlainLoopingDeloopingEquivalence}
  \begin{tikzcd}
    \mathrm{Grpd}^\ast_{[1,\infty]}
    \ar[
      rr,
      shift left=5pt,
      "{ \Omega }"
    ]
    \ar[
      rr,
      phantom,
      "{ \sim }"
    ]
    \ar[
      rr,
      <-,
      shift right=5pt,
      "{ B }"{swap}
    ]
    &&
    \mathrm{Grp}\bracket({
      \mathrm{Grpd}_\infty
    })
    \mathrlap{\,.}
  \end{tikzcd}
\end{equation}
Here the inverse equivalence $B(-)$ is hence also called the \emph{delooping} operation.

The analogous equivalence holds in every $\infty$-topos (\parencites[Lem. 7.2.2.1]{Lurie2009}[\S 5.2.6]{Lurie2017}, cf. \cite[Thm. 2.19]{NSS2015a}), hence in particular among smooth $\infty$-groupoids (\cref{OnSmoothInfinityGroupoids}), where we denote the delooping in boldface:
\begin{equation}
\label
{SmoothLoopingDeloopingEquivalence}
  \begin{tikzcd}
    \mathrm{SmthGrpd}
      ^\ast
      _{[1,\infty]}
    \ar[
      rr,
      shift left=5pt,
      "{ \Omega }"
    ]
    \ar[
      rr,
      phantom,
      "{ \sim }"
    ]
    \ar[
      rr,
      <-,
      shift right=5pt,
      "{ \mathbf{B} }"{swap}
    ]
    &&
    \mathrm{Grp}\bracket({
      \mathrm{SmthGrpd}_\infty
    })
    \mathrlap{\,.}
  \end{tikzcd}
\end{equation}

The shape operation \cref{TheShapeModality} preserves $\infty$-group structure and deloopings, due to \cref{AdjointQuadrupleOfInftyFunctors,ShapePreservesHomotopyFibersOverDiscrete}, cf. \cite[Prop. 9.1.4]{SS26-Orb}.

\paragraph
{Traditional Classifying Spaces}

For $G$ a Hausdorff, Lie or Fr{\'e}chet-Lie group, the traditional \emph{classifying space} for $G$-principal bundles (\parencites{Milnor1956}{Milgram1967}{Steenrod1968}, cf. \parencites[\S 11]{Switzer1975}[\S 4.12--4.13]{Husemoller1994}{Mitchell2011}[Thm. 3.5.1]{RudolphSchmidt2017}) is, as a homotopy type, the shape \cref{TheShapeModality} of the smooth groupoid $\mathbf{B}G$ with a single object and $G$ worth of automorphisms (\cite[Thm. 5.2.13]{SS26-Bun}):
\begin{equation}
\label
{ClassicalClassifyingSpace}
  B G
  \defneq
  \shape \mathbf{B} G
  \simeq
  B \shape G
  \mathrlap{\,.}
\end{equation}
This means that for $X \in \mathrm{SmthMfd}$, the space \cref{ClassicalClassifyingSpace} classifies isomorphism classes of $G$-principal bundles over $X$ and hence classifies ordinary nonabelian cohomology:
\begin{equation}
\label
{OrdinaryNonabelianCohomology}
  H^1\bracket({
    X;
    G
  })
  \simeq
  \pi_0
  \,
  \mathbf{Map}\bracket({
    X,
    B G
  })
  \simeq
  G \mathrm{PncplBndl}_X \big/\mathrm{iso}
  \mathrlap{\,.}
\end{equation}

If here $G \defneq A \in \mathrm{AbGrp}$ is an abelian group, then $B A$ itself inherits $\infty$-group structure (2-group structure) and so on recursively, whence for $n \in \mathbb{N}$ we have the \emph{Eilenberg-MacLane space} (cf. \parencites[\S 19]{Steenrod1972}[\S 6]{AguilarGitlerPrieto2002})
\begin{equation}
\label
{TheEMSpaces}
  B^{n} A
  :=
  B B^{n-1}A
  \mathrlap{\,,}
\end{equation}
traditionally denoted $K(A,n)$, which classifies ordinary abelian cohomology (\parencites[p. 243]{Eilenberg1940}[p. 520--521]{EilenbergMacLane1954III}, cf. \parencites[\S 7.1]{AguilarGitlerPrieto2002}[Ex. 2.1]{FSS23-Char}):
\begin{equation}
  H^n\bracket({
    X;
    A
  })
  \simeq
  \pi_0\,
  \mathbf{Map}\bracket({
    X,
    B^n A
  })
  \mathrlap{\,.}
\end{equation}

\paragraph
{Spectra of Classifying Spaces}

By construction, the sequence of EM-spaces $\bracket({ B^n A })_{n \in \mathbb{N}}$ \cref{TheEMSpaces} forms a \emph{spectrum of spaces} (cf. \parencites[\S III.2]{Adams1974})
\begin{equation}
\label
{SpectrumOfSpaces}
  E_\bullet
  :=
  \bracket({
  \bracket({
    E_n \in \mathrm{Grpd}^\ast_\infty
  })_{n \in \mathbb{N}}
  \,,
  \text{ with }
  \big(
    \inlinetikzcd{
      E_n
      \ar[
        r,
        "{ \sim }"
      ]
      \&
      \Omega E_{n+1}
    }
  \big)_{n \in \mathbb{N}}
  })
  \mathrlap{\,.}
\end{equation}
The \emph{Brown representability theorem} says (\parencites{Brown1962}{Adams1971}, cf. \parencites[\S 3.4]{Kochman1996}[Thm. 12.2.18]{AguilarGitlerPrieto2002}) that such spectra are exactly the classifying spaces of Whitehead-generalized abelian cohomology theories $E^\bullet(-)$ (cf. \cite[Ex. 2.10]{FSS23-Char}):
\begin{equation}
  E^n\bracket(X)
  \simeq
  \pi_0
  \,
  \mathbf{Map}\bracket({
    X,
    E_n
  })
  \mathrlap{\,.}
\end{equation}

As a key example beyond ordinary cohomology, the classifying space for \emph{rank-0 complex topological K-theory}, $K^0(-)_{\mathrm{rk}=0}$ (\parencites{Atiyah1967}) is (cf. \parencites[\S 11]{Switzer1975}[Def. 9.2.8]{AguilarGitlerPrieto2002})
\begin{equation}
\label{KU0}
\begin{aligned}
  &
  B \mathrm{U}
  :=
  \varinjlim_n 
  B \mathrm{U}(n)
  \\
  &
  K^0(X)_{\mathrm{rk=0}}
  \simeq
  \pi_0 \mathbf{Map}\bracket({
    X
    ,
    B \mathrm{U}
  })
  \mathrlap{\,.}
\end{aligned}
\end{equation}
Remarkably, while the finite stages $B \mathrm{U}(n)$ fail to have further delooping for $n \geq 2$, the colimiting space does deloop, by the homotopy-theoretic version of \emph{Bott periodicity} (cf. \parencites[\S B.2]{AguilarGitlerPrieto2002}[\S 15]{HusemollerEtAl2008}):
\begin{equation}
  \begin{tikzcd}
    B \mathrm{U}
    \ar[
      r,
      "{ \sim }"
    ]
    &
    \Omega
    \mathrm{SU}
    \ar[
      r,
      "{ \sim }"
    ]
    &
    \Omega^2
    \cdots
  \end{tikzcd}
\end{equation}
making a spectrum of spaces, $\Sigma^2 \mathrm{ku}$ (the 2-connective K-theory spectrum), which exhibits rank-zero K-theory as a Whitehead-generalized abelian cohomology theory.

A further variant relevant below (but just for bookkeeping) is cohomology with coefficients in the \emph{Postnikov truncation} (cf. \parencites[\S IV]{Whitehead1978}[\S Cor. 3.7]{GoerssJardine2009}[Prop. 1.9]{FSS23-Char}) $\truncation{9}{B \mathrm{U}}$ of $B \mathrm{U}$, which agrees with rank-0 K-theory \cref{KU0} on 9-manifolds $X^9$ (and hence on globally hyperbolic 10-spacetimes, but generally not beyond that case):
\begin{equation}
  \label
  {9TruncationOfBU}
  K^0\bracket({
    X^9
  })_{\mathrm{rk}=0}
  \simeq
  \pi_0 \mathbf{Map}\bracket({
    X^9
    ,
    \truncation{9}{B \mathrm{U}}
  })
  \mathrlap{\,.}
\end{equation}
(This $\truncation{9}{B \mathrm{U}}$ is the 0th space in the spectrum $\Sigma^2 \truncation{7}{\mathrm{ku}}$.)

\paragraph
{Unstable Classifying Spaces}

In this vein, one may understand the finite/unstable stages $B \mathrm{U}(n)$ as classifying forms of \emph{unstable K-theory}.

In fact, since every connected space $\mathcal{A}$ is the classifying space of its loop $\infty$-group \cref{PlainLoopingDeloopingEquivalence}
\begin{equation}
  \mathcal{A}
  \simeq
  B\bracket({
    \Omega \mathcal{A}
  })
  \mathrlap{\,,}
\end{equation}
every connected space is the classifying space of the evident generalization of ordinary nonabelian cohomology \cref{OrdinaryNonabelianCohomology}:
\begin{equation}
  H^1\bracket({
    X;
    \Omega\mathcal{A}
  })
  :=
  \pi_0
  \,
  \mathbf{Map}\bracket({
    X,
    \mathcal{A}
  })
  \mathrlap{\,.}
\end{equation}
Moreover, here the connectedness assumption only serves to give the cohomology in degree 1. More generally we may say that \emph{nonabelian cohomology} is what is classified by \emph{any} (pointed) $\mathcal{A} \in \mathrm{Grpd}_\infty$:
\begin{equation}
  H^0\bracket({
    X;
    \mathcal{A}
  })
  :=
  \pi_0\,
  \mathbf{Map}\bracket({
    X,
    \mathcal{A}
  })
\end{equation}
(\parencites[Def. 6.0.6]{Toen2002}[Def. 6]{Lurie2014}[\S 2]{FSS23-Char}, more exposition in \parencites[\S 1]{SS25-TEC}[\S 4]{SS26-Orb}, under suitable conditions this is the \emph{non-abelian Poincar{\'e} dual} of \emph{factorization homology} \parencites{Salvatore2001}{Lurie2014}[\S 4]{AyalaFrancis2015}[Thm. 1.1]{Zou2023}).

In this pattern, yet more generally we have nonabelian cohomology in every negative degree  $n \in \mathbb{N}$
\begin{equation}
  H^{-n}\bracket({
    X;
    \mathcal{A}
  })
  :=
  \pi_0\,
  \mathbf{Map}\bracket({
    X,
    \Omega^n\mathcal{A}
  })
  \mathrlap{\,,}
\end{equation}
while the existence of nonabelian cohomology in positive degrees $n \in \mathbb{N}$ is conditional on higher deloopings
\begin{equation}
  H^n\bracket({
    X;
    \Omega \mathcal{A}
  })
  :=
  \pi_0\,
  \mathbf{Map}\bracket({
    X
    ,
    B^n \Omega \mathcal{A}
  })
\end{equation}
which need not exist --- and in examples rarely exist for finite $n$ without already existing for all $n$ and thus reducing to the stable case \cref{SpectrumOfSpaces}.

The most basic example of an unstable classifying space is the $n$-sphere for $n \geq 2$. The corresponding nonabelian cohomology theory is \emph{$n$-Cohomotopy} (\parencites{Borsuk1936}{Pontrjagin1938}{Spanier1949}, cf. \parencites[\S VII]{STHu59}[Ex. 2.7]{FSS23-Char}):
\begin{equation}
\label
{Cohomotopy}
  \pi^n\bracket({
    X
  })
  :=
  \pi_0\,
  \mathbf{Map}\bracket({
    X, \shape S^n
  })
  \mathrlap{\,.}
\end{equation}

\paragraph
{Fibrations of Classifying Spaces}

Unstable classifying spaces appear already when \emph{twisting} abelian cohomology theories.

For instance, the classifying space $B \mathrm{U}$ \cref{KU0} for rank-zero K-theory carries an $\infty$-action by $B \mathrm{U}(1)$, whose homotopy quotient is a $B\mathrm{U}$-fibration over $B^2 \mathrm{U}(1)$:
\begin{equation}
\label
{ClassifyingFibrationForTwistedKTheory}
  \begin{tikzcd}
    B \mathrm{U}
    \ar[r]
    \ar[
      d,
    ]
    \ar[
      dr,
      phantom,
      "{ \lrcorner }"{pos=.1}
    ]
    &
    B \mathrm{U} \sslash B \mathrm{U}(1)
    \ar[
      d,
      "{ p }"
    ]
    \\
    \ast
    \ar[r]
    & 
    B^2 \mathrm{U}(1) \mathrlap{\,.}
  \end{tikzcd}
\end{equation}
Hence, while the homotopy quotient is not itself a stage in a spectrum, it is so \emph{fiberwise} over $B^2 \mathrm{U}(1)$. Accordingly, the \emph{fiberwise} maps to this quotient, over a fixed \emph{twisting} map $\tau : \inlinetikzcd{ X \ar[r] \& B^2 \mathrm{U}(1) }$, classify \emph{$\tau$-twisted rank-0 K-theory} (\cite{AtiyahSegal2004}, cf. \parencites[Ex. 3.4]{FSS23-Char}[Def. 5.1.22]{SS26-Orb}[Ex. 6.2.5]{SS26-Bun}):
\begin{equation}
 K^\tau_{\mathrm{rk=0}}
 \bracket({X})
 :=
 \pi_0\,
 \mathbf{Map}\bracket({
   (X,\tau)
   ,
   \bracket({
     B\mathrm{U}
     \sslash
     B\mathrm{U}(1),
     p
   })
 })_{B^2 \mathrm{U}(1)}
  =
  \pi_0
  \left\{
  \begin{tikzcd}[row sep=small]
    & 
    B \mathrm{U} 
      \sslash 
    B\mathrm{U}(1)
    \ar[
      d,
      "{ p }"
    ]
    \\
    X
    \ar[
      r,
      "{\tau}"
    ]
    \ar[
      ur,
      dashed
    ]
    &
    B^2 \mathrm{U}(1)
  \end{tikzcd}
  \right\}
 \mathrlap{.}
\end{equation}

\paragraph
{Cyclic Loop Spaces}

A particularly important operation for obtaining new unstable classifying spaces from given spaces $\mathcal{A}$ is \emph{cyclification}.

Given $\mathcal{G} \in \mathrm{Grp}\bracket({\mathrm{Grpd}_\infty})$, there is an induced $\mathcal{G}$-action on 
$\mathbf{Map}\bracket({\mathcal{G},\mathcal{A}})$. The homotopy quotient by that action turns out to be the classifying space for $\mathcal{A}$-cohomology on $\mathcal{G}$-principal fibrations
\begin{equation}
\label
{ExtCycAdjunctionEquivalence}
  \begin{tikzcd}[row sep=small, column sep=large]
    \mathcal{P}
    \ar[d]
    \ar[r]
    \ar[
      dr,
      phantom,
      "{ \lrcorner }"{pos=.1}
    ]
    &
    \ast
    \ar[d]
    \\
    X
    \ar[
      r,
      "{ c }"
    ]
    &
    B \mathcal{G}
  \end{tikzcd}
\end{equation}
in that we have a natural equivalence
(\parencites[\S 2.2]{BMSS2019}[\S 2.2]{SS24-Cyc})
\begin{equation}
\label{DimReductionViaCyc}
  \big\{
  \begin{tikzcd}
     \mathcal{P}
     \ar[
       r,
       dashed
      ]
     &
     \mathcal{A}
  \end{tikzcd}
  \big\}
  \;
  \begin{tikzcd}[
    column sep=40pt
  ]
    {}
    \ar[
      r,
      shift left=5pt,
      "{ \mathrm{reduction} }"
    ]
    \ar[
      r,
      phantom,
      "{ \sim }"
    ]
    \ar[
      r,
      <-,
      shift right=5pt,
      "{
        \mathrm{oxidation}
      }"{swap}
    ]
    &
    {}
  \end{tikzcd}
  \;
  \left\{
  \begin{tikzcd}[row sep=-2pt, 
    column sep=25pt
  ]
    X
    \ar[
      rr,
      dashed
    ]
    \ar[
      dr,
      "{ c }"{swap}
    ]
    &
    &[-10pt]
    \mathbf{Map}\bracket({
      \mathcal{G},
      \mathcal{A}
    })
    \sslash
    \mathcal{G}
    \ar[
      dl
    ]
    \\
    &
    B \mathcal{G}
  \end{tikzcd}
  \right\}
  \mathrlap{\,.}
\end{equation}
In particular, when $\mathcal{G} \defneq \shape S^1$ is the circle, then $\mathcal{P}$ is (the homotopy type of) a circle-principal bundle and this homotopy quotient is the \emph{cyclic loop space}
\begin{equation}
\label
{Cyclification}
  \mathrm{Cyc}\bracket({
    \mathcal{A}
  })
  :=
  \mathbf{Map}\bracket({
    S^1,
    \mathcal{A}
  })
  \sslash \shape S^1
  \mathrlap{\,.}
\end{equation}

But if we think of $\mathcal{A}$ here as classifying charges of a field theory on $\mathcal{P}$, then this equivalence \cref{ExtCycAdjunctionEquivalence} says that $\mathrm{Cyc}\bracket({\mathcal{A}})$ classifies the same information, but now transmuted to a larger set of charges on the base space $X$. 
This is the global topological guise of what in physics is known as \emph{double dimensional reduction} of fluxes (cf. \parencites[\S 3]{FSS18-TD}[\S 2.2]{GSS25-TD}).

\paragraph
{Comparison to the Literature}
This equivalence
\cref{DimReductionViaCyc} is due to \parencites[\S 4]{Schreiber2017}[\S 2.2]{BMSS2019}[\S 2.2]{SS24-Cyc}, following the analogous rational observation in \parencites[\S 3]{FSS17-Sphere}[\S 3]{FSS18-TD}[\S 2]{FSS18-TDR}[\S 8]{FSS2019}, for which see also \parencites[\S 2.4]{SatiVoronov2024}[\S 2.2]{GSS25-TD}. 
The special case when $\mathcal{A} \simeq B G$ is the classifying space of a connected Lie group $G$ was essentially considered in \cite{BergmanVaradarajan2004}, with focus on $G \defneq E_8$; before that, the further specialization to  $c_1 = 0$ was considered in \parencites{MathaiSati04-E8}. Much earlier, the special case $\mathcal{A} = B G$, $c_1 = 0$ of \cref{DimReductionViaCyc} was called the \emph{caloron correspondence} in \parencites{GarlandMurray1988}, cf. \parencites{MurrayVozzo2009}{HekmatiMurrayVozzo2011}.


\newpage

\paragraph
{Declarations}

\begin{description}
  \item[] The authors declare that they have no conflict of interest.

  \item[] No data were generated during this purely theoretical study.
\end{description}

\printbibliography

@article{BanksSeiberg2011,
    author        = "Banks, Tom and Seiberg, Nathan",
    title         = "{Symmetries and Strings in Field Theory and Gravity}",
    eprint        = "1011.5120",
    archivePrefix = "arXiv",
    primaryClass  = "hep-th",
    doi           = "10.1103/PhysRevD.83.084019",
    journal       = "Phys. Rev. D",
    volume        = "83",
    number        = "8",
    pages         = "084019",
    year          = "2011"
}

@article{Polchinski2004,
    author        = "Polchinski, Joseph",
    title         = "{Monopoles, Duality, and String Theory}",
    eprint        = "hep-th/0304042",
    archivePrefix = "arXiv",
    primaryClass  = "hep-th",
    doi           = "10.1142/S0217751X0401866X",
    journal       = "Int. J. Mod. Phys. A",
    volume        = "19",
    number        = "supp01",
    pages         = "145--156",
    year          = "2004"
}

@article{BaggerLambertMukhiPapageorgakis2013,
    author        = {Bagger, Jonathan and Lambert, Neil and Mukhi, Sunil and Papageorgakis, Constantinos},
    title         = {{Multiple Membranes in M-theory}},
    journal       = {Physics Reports},
    volume        = {527},
    number        = {1},
    pages         = {1--100},
    year          = {2013},
    month         = {6},
    doi           = {10.1016/j.physrep.2013.01.006},
    eprint        = {1203.3546},
    archivePrefix = {arXiv},
    primaryClass  = {hep-th}
}

@article{DuffStelle1990,
    author = "Duff, M. J. and Stelle, K. S.",
    title = "{{Multi-membrane solutions of $D = 11$ supergravity}}",
    journal = "Phys. Lett. B",
    volume = "253",
    number = "1-2",
    pages = "113--118",
    year = "1991",
    doi = "10.1016/0370-2693(91)91371-2"
}

@article{Page1983,
    author = "Page, Don N.",
    title = "{{Classical stability of round and squashed seven-spheres in eleven-dimensional supergravity}}",
    doi = "10.1103/PhysRevD.28.2976",
    journal = "Phys. Rev. D",
    volume = "28",
    number = "12",
    pages = "2976--2982",
    year = "1983"
}

@article{Moore2004,
    author = "Moore, Gregory W.",
    title = "{{Anomalies, Gauss laws, and Page charges in M-theory}}",
    eprint = "hep-th/0409158",
    archivePrefix = "arXiv",
    doi = "10.1016/j.crhy.2004.12.005",
    journal = "Comptes Rendus Physique",
    volume = "6",
    pages = "251--259",
    year = "2005"
}

@article{Carter2009,
    author = "Carter, Brandon",
    title = "{{Republication of: Black hole equilibrium states --- Part I Analytic and geometric properties of the Kerr solutions}}",
    doi = "10.1007/s10714-009-0888-5",
    journal = "Gen. Rel. Grav.",
    volume = "41",
    pages = "2873--2938",
    year = "2009"
}

@article{Carter1968,
  title = {{Global Structure of the Kerr Family of Gravitational Fields}},
  author = {Carter, Brandon},
  journal = {Physical Review},
  volume = {174},
  number = {5},
  pages = {1559--1571},
  year = {1968},
  publisher = {American Physical Society},
  doi = {10.1103/PhysRev.174.1559}
}

@incollection{Marolf2012,
  author        = {Marolf, Donald},
  title         = {{Black holes and branes in supergravity}},
  booktitle     = {Black Holes in Higher Dimensions},
  editor        = {Horowitz, Gary T.},
  publisher     = {Cambridge University Press},
  year          = {2012},
  pages         = {273--324},
  doi           = {10.1017/CBO9781139004176.012},
  eprint        = {1107.1022},
  archivePrefix = {arXiv},
  primaryClass  = {hep-th}
}

@article{Hamilton1982,
  author  = {Hamilton, Richard S.},
  title   = {{The inverse function theorem of Nash and Moser}},
  journal = {Bulletin of the American Mathematical Society},
  volume  = {7},
  number  = {1},
  pages   = {65--222},
  year    = {1982},
  doi     = {10.1090/S0273-0979-1982-15004-2},
  note    = {Section~I.4}
}

@article{Waldorf2012,
  author        = {Waldorf, Konrad},
  title         = {{Transgression to Loop Spaces and its Inverse I: Diffeological Bundles and Fusion Maps}},
  journal       = {Cahiers de Topologie et G{\'e}om{\'e}trie Diff{\'e}rentielle Cat{\'e}goriques},
  volume        = {LIII},
  number        = {3},
  pages         = {162--210},
  year          = {2012},
  eprint        = {0911.3212},
  archiveprefix = {arXiv},
  primaryclass  = {math.DG},
  url           = {http://cahierstgdc.com/wp-content/uploads/2017/03/Waldorf.pdf},
  mrnumber      = {3053549},
  zbl           = {1270.58004}
}

@article{Losik1992,
  author   = {Losik, Mark V.},
  title    = {{Fr{\'e}chet Manifolds as Diffeologic Spaces}},
  journal  = {Russian Mathematics (Izvestiya VUZ. Matematika)},
  volume   = {36},
  number   = {5},
  pages    = {36--42},
  year     = {1992},
  mrnumber = {1213569},
  zbl      = {0774.58002},
  url      = {https://ncatlab.org/nlab/files/Losik_FrechetAsDiffeological.pdf}
}

@article{Losik1994,
  author    = {Losik, Mark V.},
  title     = {{Categorical Differential Geometry}},
  journal   = {Cahiers de Topologie et G{\'e}om{\'e}trie Diff{\'e}rentielle Cat{\'e}goriques},
  volume    = {35},
  number    = {4},
  pages     = {274--290},
  year      = {1994},
  url       = {http://www.numdam.org/item/CTGDC_1994__35_4_274_0/},
  mrnumber  = {1307227},
  zbl       = {0812.58002}
}

@incollection{Souriau1984,
  author    = {Souriau, Jean-Marie},
  title     = {{Groupes diff{\'e}rentiels et physique math{\'e}matique}},
  booktitle = {Group Theoretical Methods in Physics},
  editor    = {Denardo, G. and Ghirardi, G. and Weber, T.},
  series    = {Lecture Notes in Physics},
  volume    = {201},
  pages     = {511--513},
  publisher = {Springer},
  year      = {1984},
  doi       = {10.1007/BFb0016198}
}

@article{WuYang1975,
  author  = {Wu, Tai Tsun and Yang, Chen Ning},
  title   = {{Concept of Nonintegrable Phase Factors and Global Formulation of Gauge Fields}},
  journal = {Physical Review D},
  volume  = {12},
  number  = {12},
  pages   = {3845--3857},
  year    = {1975},
  doi     = {10.1103/PhysRevD.12.3845}
}

@article{WenWitten1985,
  author  = {Wen, Xiao-Gang and Witten, Edward},
  title   = {{Electric and Magnetic Charges in Superstring Models}},
  journal = {Nuclear Physics B},
  volume  = {261},
  pages   = {651--677},
  year    = {1985},
  doi     = {10.1016/0550-3213(85)90592-9}
}

@article{WuYang1976,
  author  = {Wu, Tai Tsun and Yang, Chen Ning},
  title   = {{Dirac's Monopole Without Strings: Classical Lagrangian Theory}},
  journal = {Physical Review D},
  volume  = {14},
  number  = {2},
  pages   = {437--445},
  year    = {1976},
  doi     = {10.1103/PhysRevD.14.437}
}

@book{Tu2017,
  author    = {Tu, Loring W.},
  title     = {{Differential Geometry: Connections, Curvature, and Characteristic Classes}},
  series    = {Graduate Texts in Mathematics},
  volume    = {275},
  publisher = {Springer},
  address   = {Cham},
  year      = {2017},
  doi       = {10.1007/978-3-319-55084-8},
  isbn      = {978-3-319-55082-4}
}

@book{Jardine2015,
  author    = {Jardine, John F.},
  title     = {{Local Homotopy Theory}},
  series    = {Springer Monographs in Mathematics},
  publisher = {Springer},
  address   = {New York},
  year      = {2015},
  doi       = {10.1007/978-1-4939-2300-7},
  isbn      = {978-1-4939-2299-4}
}

@article{Brown1973,
  author  = {Brown, Kenneth S.},
  title   = {{Abstract homotopy theory and generalized sheaf cohomology}},
  journal = {Transactions of the American Mathematical Society},
  year    = {1973},
  volume  = {186},
  pages   = {419--458},
  doi     = {10.2307/1996573},
  jstor   = {1996573},
  publisher = {American Mathematical Society}
}

@misc{Rogers2024,
  author        = {Rogers, Christopher L. and Wolfson, Jesse},
  title         = {{Lie's Third Theorem for Lie $\infty$-Algebras}},
  year          = {2024},
  archiveprefix = {arXiv},
  eprint        = {2409.08957},
  primaryclass  = {math.RA},
  doi           = {10.48550/arXiv.2409.08957}
}

@misc{Cabrera2026,
  author        = {Cabrera, Alejandro and del Hoyo, Matias},
  title         = {{Geometric differentiation of simplicial manifolds}},
  year          = {2026},
  archiveprefix = {arXiv},
  eprint        = {2602.09885},
  primaryclass  = {math.DG},
  doi           = {10.48550/arXiv.2602.09885}
}

@article{Pridham2013,
  author        = {Pridham, Jonathan},
  title         = {{Presenting higher stacks as simplicial schemes}},
  journal       = {Advances in Mathematics},
  volume        = {238},
  pages         = {184--245},
  year          = {2013},
  doi           = {10.1016/j.aim.2013.01.009},
  archiveprefix = {arXiv},
  eprint        = {0905.4044},
  primaryclass  = {math.AG}
}

@article{Zhu2009b,
  author        = {Zhu, Chenchang},
  title         = {{Kan replacement of simplicial manifolds}},
  journal       = {Letters in Mathematical Physics},
  volume        = {90},
  number        = {1},
  pages         = {383--405},
  year          = {2009},
  doi           = {10.1007/s11005-009-0353-0},
  archiveprefix = {arXiv},
  eprint        = {0812.4150},
  primaryclass  = {math.DG}
}

@article{Zhu2009,
  author        = {Zhu, Chenchang},
  title         = {{$n$-Groupoids and Stacky Groupoids}},
  journal       = {International Mathematics Research Notices},
  volume        = {2009},
  number        = {21},
  pages         = {4087--4141},
  year          = {2009},
  doi           = {10.1093/imrn/rnp080},
  archiveprefix = {arXiv},
  eprint        = {0801.2057},
  primaryclass  = {math.DG}
}

@misc{Zhu2006,
  author        = {Zhu, Chenchang},
  title         = {{Lie $n$-groupoids and stacky Lie groupoids}},
  year          = {2006},
  eprint        = {math/0609420},
  archiveprefix = {arXiv},
  primaryclass  = {math.DG}
}

@article{Henriques2008,
  author        = {Henriques, Andr{\'e}},
  title         = {{Integrating $L_\infty$-algebras}},
  journal       = {Compositio Mathematica},
  volume        = {144},
  number        = {4},
  pages         = {1017--1045},
  year          = {2008},
  doi           = {10.1112/S0010437X07003405},
  archiveprefix = {arXiv},
  eprint        = {math/0603563}
}

@unpublished{GSS26-KfromM,
  author = {Golub, Nikita and Sati, Hisham and Schreiber, Urs},
  title  = {{On the Derivation of Twisted K-Theory from M-Theory}},
  year   = {2026},
  note   = {In preparation},
  url    = {https://ncatlab.org/schreiber/show/K+from+M}
}

@misc{MooreSaxena2025,
    author = "Moore, Gregory W. and Saxena, Vivek",
    title = "{{TASI Lectures On Topological Field Theories And Differential Cohomology}}",
    eprint = "2510.07408",
    archivePrefix = "arXiv",
    primaryClass = "hep-th",
    month = "10",
    year = "2025"
}

@book{BryantGriffithsGrossman2003,
  author    = {Bryant, Robert L. and Griffiths, Phillip A. and Grossman, Daniel A.},
  title     = {{Exterior Differential Systems and Euler-Lagrange Partial Differential Equations}},
  series    = {Chicago Lectures in Mathematics},
  publisher = {University of Chicago Press},
  address   = {Chicago},
  year      = {2003},
  month     = {8},
  isbn      = {978-0-226-07794-9},
  url       = {https://press.uchicago.edu/ucp/books/book/chicago/E/bo3637779.html}
}

@article{Frankel1974,
  author  = {Frankel, Theodore},
  title   = {{Maxwell's Equations}},
  journal = {The American Mathematical Monthly},
  volume  = {81},
  number  = {4},
  pages   = {343--349},
  year    = {1974},
  month   = {4},
  doi     = {10.1080/00029890.1974.11993557},
  jstor   = {2318995}
}

@book{MisnerThorneWheeler1973,
  author    = {Misner, Charles W. and Thorne, Kip S. and Wheeler, John Archibald},
  title     = {{Gravitation}},
  publisher = {W. H. Freeman},
  address   = {San Francisco},
  year      = {1973},
  isbn      = {978-0-7167-0344-0},
  note      = {\S 3.4 and \S 4.3}
}

@inbook{Polchinski1998II,
  author    = {Polchinski, Joseph},
  title     = {{Low Energy Supergravity}},
  booktitle = {{String Theory. Vol. 2: Superstring Theory and Beyond}},
  series    = {Cambridge Monographs on Mathematical Physics},
  chapter   = {12.1},
  pages     = {84--94},
  publisher = {Cambridge University Press},
  year      = {1998},
  month     = {10},
  doi       = {10.1017/CBO9780511618123},
  isbn      = {978-0-521-63304-8}
}

@article{HuqNamazie1985,
  author  = {Huq, M. and Namazie, M. A.},
  title   = {{Kaluza-Klein supergravity in ten dimensions}},
  journal = {Classical and Quantum Gravity},
  volume  = {2},
  number  = {3},
  pages   = {293--308},
  year    = {1985},
  month   = {5},
  doi     = {10.1088/0264-9381/2/3/007}
}

@article{CampbellWest1984,
  author  = {Campbell, I. C. G. and West, Peter C.},
  title   = {{$N=2$, $D=10$ non-chiral supergravity and its spontaneous compactification}},
  journal = {Nuclear Physics B},
  volume  = {243},
  number  = {1},
  pages   = {112--124},
  year    = {1984},
  month   = {8},
  doi     = {10.1016/0550-3213(84)90388-2}
}

@article{GianiPernici1984,
  author  = {Giani, F. and Pernici, Mario},
  title   = {{$N=2$ supergravity in ten dimensions}},
  journal = {Physical Review D},
  volume  = {30},
  number  = {2},
  pages   = {325--333},
  year    = {1984},
  month   = {7},
  doi     = {10.1103/PhysRevD.30.325}
}

@inproceedings{Cremmer1980,
  author    = {Cremmer, Eug{\`e}ne},
  title     = {{Supergravities in 5 Dimensions}},
  booktitle = {Superspace and Supergravity},
  editor    = {Hawking, Stephen and Ro{\v{c}}ek, Martin},
  publisher = {Cambridge University Press},
  year      = {1981},
  pages     = {267--282},
  isbn      = {978-0-521-23908-0},
  url       = {https://cds.cern.ch/record/100081},
  note      = {inSPIRE: \href{https://inspirehep.net/literature/172773}{172773}}
}

@phdthesis{Compere2007,
  author        = {Comp{\`e}re, Geoffrey},
  title         = {{Symmetries and conservation laws in Lagrangian gauge theories with applications to the mechanics of black holes and to gravity in three dimensions}},
  school        = {Universit{\'e} Libre de Bruxelles},
  year          = {2007},
  month         = {8},
  eprint        = {0708.3153},
  archiveprefix = {arXiv},
  primaryclass  = {hep-th}
}

@misc{BelovMoore2006,
  author        = {Belov, Dmitriy and Moore, Gregory W.},
  title         = {{Holographic Action for the Self-Dual Field}},
  year          = {2006},
  month         = {5},
  eprint        = {hep-th/0605038},
  archiveprefix = {arXiv},
  doi           = {10.48550/arXiv.hep-th/0605038}
}

@article{PastiSorokinTonin1997,
  author        = {Pasti, Paolo and Sorokin, Dmitri and Tonin, Mario},
  title         = {{On Lorentz Invariant Actions for Chiral $p$-Forms}},
  journal       = {Physical Review D},
  volume        = {55},
  number        = {10},
  pages         = {6292--6298},
  year          = {1997},
  month         = {5},
  doi           = {10.1103/PhysRevD.55.6292},
  eprint        = {hep-th/9611100},
  archiveprefix = {arXiv}
}

@article{HenneauxTeitelboim1988,
  author  = {Henneaux, Marc and Teitelboim, Claudio},
  title   = {{Dynamics of chiral (self-dual) $p$-forms}},
  journal = {Physics Letters B},
  volume  = {206},
  number  = {4},
  pages   = {650--654},
  year    = {1988},
  month   = {6},
  doi     = {10.1016/0370-2693(88)90712-5}
}

@book{Hovey1999,
  author    = {Hovey, Mark},
  title     = {{Model Categories}},
  series    = {Mathematical Surveys and Monographs},
  volume    = {63},
  publisher = {American Mathematical Society},
  year      = {1999},
  doi       = {10.1090/surv/063},
  isbn      = {978-0-8218-4361-1}
}

@book{MayPonto2012,
  author    = {May, J. Peter and Ponto, Kate},
  title     = {{More Concise Algebraic Topology: Localization, Completion, and Model Categories}},
  series    = {Chicago Lectures in Mathematics},
  publisher = {University of Chicago Press},
  year      = {2012},
  isbn      = {978-0-226-51179-5},
  doi       = {10.7208/chicago/9780226511795.001.0001},
  url       = {https://www.math.uchicago.edu/~may/BOOKS/KateBookFinal.pdf}
}

@book{Arkowitz2011,
  author    = {Arkowitz, Martin},
  title     = {{Introduction to Homotopy Theory}},
  series    = {Universitext},
  publisher = {Springer},
  year      = {2011},
  doi       = {10.1007/978-1-4419-7329-0},
  isbn      = {978-1-4419-7328-3}
}

@book{Strom2011,
  author    = {Strom, Jeffrey},
  title     = {{Modern Classical Homotopy Theory}},
  series    = {Graduate Studies in Mathematics},
  volume    = {127},
  publisher = {American Mathematical Society},
  year      = {2011},
  doi       = {10.1090/gsm/127},
  isbn      = {978-0-8218-5286-6}
}

@book{Quillen1967,
  author    = {Quillen, Daniel G.},
  title     = {{Homotopical Algebra}},
  series    = {Lecture Notes in Mathematics},
  volume    = {43},
  publisher = {Springer},
  year      = {1967},
  doi       = {10.1007/BFb0097438},
  isbn      = {978-3-540-03914-3}
}

@book{Kelly1982,
  author    = {Kelly, Max},
  title     = {{Basic Concepts of Enriched Category Theory}},
  series    = {London Mathematical Society Lecture Note Series},
  volume    = {64},
  publisher = {Cambridge University Press},
  year      = {1982},
  isbn      = {978-0-521-28702-9},
  note      = {Reprinted in: \emph{Reprints in Theory and Applications of Categories}, No.~10 (2005), pp.~1--136},
  url       = {http://www.tac.mta.ca/tac/reprints/articles/10/tr10abs.html}
}

@book{RiehlVerity2022,
  author    = {Riehl, Emily and Verity, Dominic},
  title     = {{Elements of $\infty$-Category Theory}},
  series    = {Cambridge Studies in Advanced Mathematics},
  volume    = {194},
  publisher = {Cambridge University Press},
  year      = {2022},
  month     = {1},
  doi       = {10.1017/9781108936880},
  isbn      = {978-1-108-83798-9}
}

@incollection{Bergner2010,
  author        = {Bergner, Julia E.},
  title         = {{A Survey of $(\infty,1)$-Categories}},
  booktitle     = {{Towards Higher Categories}},
  editor        = {Baez, John C. and May, J. Peter},
  series        = {The IMA Volumes in Mathematics and its Applications},
  volume        = {152},
  pages         = {69--83},
  publisher     = {Springer},
  year          = {2010},
  doi           = {10.1007/978-1-4419-1524-5_2},
  eprint        = {math/0610239},
  archiveprefix = {arXiv},
  primaryclass  = {math.AT}
}

@article{Kan1957css,
  author  = {Kan, Daniel M.},
  title   = {{On c.s.s. Complexes}},
  journal = {American Journal of Mathematics},
  volume  = {79},
  number  = {3},
  pages   = {449--476},
  year    = {1957},
  month   = {7},
  doi     = {10.2307/2372558}
}

@article{Kan1956III,
  author  = {Kan, Daniel M.},
  title   = {{Abstract Homotopy III}},
  journal = {Proceedings of the National Academy of Sciences},
  volume  = {42},
  number  = {7},
  pages   = {419--421},
  year    = {1956},
  month   = {7},
  doi     = {10.1073/pnas.42.7.419}
}

@article{LadaMarkl1995,
  author        = {Lada, Tom and Markl, Martin},
  title         = {{Strongly homotopy Lie algebras}},
  journal       = {Communications in Algebra},
  volume        = {23},
  number        = {6},
  pages         = {2147--2161},
  year          = {1995},
  doi           = {10.1080/00927879508825335},
  eprint        = {hep-th/9406095},
  archiveprefix = {arXiv}
}

@article{Berglund2015,
  author        = {Berglund, Alexander},
  title         = {{Rational homotopy theory of mapping spaces via Lie theory for $L_\infty$-algebras}},
  journal       = {Homology, Homotopy and Applications},
  volume        = {17},
  number        = {2},
  pages         = {343--369},
  year          = {2015},
  doi           = {10.4310/HHA.2015.v17.n2.a16},
  eprint        = {1110.6145},
  archiveprefix = {arXiv},
  primaryclass  = {math.AT}
}

@article{Getzler2009,
  author        = {Getzler, Ezra},
  title         = {{Lie theory for nilpotent $L_\infty$-algebras}},
  journal       = {Annals of Mathematics},
  series        = {Second Series},
  volume        = {170},
  number        = {1},
  pages         = {271--301},
  year          = {2009},
  month         = {7},
  doi           = {10.4007/annals.2009.170.271},
  eprint        = {math/0404003},
  archiveprefix = {arXiv},
  primaryclass  = {math.AT}
}

@article{Zou2023,
  author        = {Zou, Foling},
  title         = {{A geometric approach to equivariant factorization homology and nonabelian Poincar{\'e} duality}},
  journal       = {Mathematische Zeitschrift},
  volume        = {303},
  number        = {4},
  pages         = {98},
  year          = {2023},
  month         = {4},
  doi           = {10.1007/s00209-023-03253-2},
  eprint        = {2008.08234},
  archiveprefix = {arXiv},
  primaryclass  = {math.AT}
}

@incollection{Salvatore2001,
  author        = {Salvatore, Paolo},
  title         = {{Configuration spaces with summable labels}},
  booktitle     = {{Cohomological Methods in Homotopy Theory}},
  editor        = {Aguad{\'e}, Jaume and Broto, Carles and Casacuberta, Carles},
  series        = {Progress in Mathematics},
  volume        = {196},
  pages         = {375--395},
  publisher     = {Birkh{\"a}user},
  year          = {2001},
  doi           = {10.1007/978-3-0348-8312-2_23},
  eprint        = {math/9907073},
  archiveprefix = {arXiv},
  primaryclass  = {math.AT}
}

@book{May1972,
  author    = {May, J. Peter},
  title     = {{The Geometry of Iterated Loop Spaces}},
  series    = {Lecture Notes in Mathematics},
  volume    = {271},
  publisher = {Springer},
  year      = {1972},
  doi       = {10.1007/BFb0067491},
  isbn      = {978-3-540-05904-2}
}

@book{deAzcarragaIzquierdo1995,
  author    = {de Azc{\'a}rraga, Jos{\'e} A. and Izquierdo, Jos{\'e} M.},
  title     = {{Lie Groups, Lie Algebras, Cohomology and Some Applications in Physics}},
  series    = {Cambridge Monographs on Mathematical Physics},
  publisher = {Cambridge University Press},
  year      = {1995},
  month     = {11},
  doi       = {10.1017/CBO9780511599897},
  isbn      = {978-0-521-46501-4}
}

@article{deAzcarragaIzquierdoPerezBuone1998,
  author        = {de Azc{\'a}rraga, Jos{\'e} A. and Izquierdo, Jos{\'e} M. and P{\'e}rez Bueno, J. C.},
  title         = {{An Introduction to Some Novel Applications of Lie Algebra Cohomology in Physics}},
  journal       = {Rev. R. Acad. Cien. Exactas F{\'{\i}}s. Nat. Ser. A Mat.},
  volume        = {95},
  number        = {2},
  pages         = {225--248},
  year          = {2001},
  eprint        = {physics/9803046},
  archiveprefix = {arXiv}
}

@article{Yalin2016,
  author        = {Yalin, Sinan},
  title         = {{Maurer-Cartan spaces of filtered $L_\infty$-algebras}},
  journal       = {Journal of Homotopy and Related Structures},
  volume        = {11},
  number        = {3},
  pages         = {375--407},
  year          = {2016},
  month         = {9},
  doi           = {10.1007/s40062-015-0108-9},
  eprint        = {1409.4741},
  archiveprefix = {arXiv},
  primaryclass  = {math.AT}
}

@article{MathaiSati04-E8,
  author        = {Mathai, Varghese and Sati, Hisham},
  title         = {{Some Relations between Twisted K-theory and E8 Gauge Theory}},
  journal       = {Journal of High Energy Physics},
  volume        = {2004},
  number        = {03},
  pages         = {016},
  year          = {2004},
  month         = {03},
  doi           = {10.1088/1126-6708/2004/03/016},
  eprint        = {hep-th/0312033},
  archiveprefix = {arXiv},
  primaryclass  = {hep-th}
}

@article{MurrayVozzo2009,
  author        = {Murray, Michael K. and Vozzo, Raymond F.},
  title         = {{The caloron correspondence and higher string classes for loop groups}},
  journal       = {Journal of Geometry and Physics},
  volume        = {60},
  number        = {9},
  pages         = {1235--1250},
  year          = {2010},
  eprint        = {0911.3464},
  archiveprefix = {arXiv},
  primaryclass  = {math.DG},
  doi           = {10.1016/j.geomphys.2010.05.003}
}

@article{HekmatiMurrayVozzo2011,
  author        = {Hekmati, Pedram and Murray, Michael K. and Vozzo, Raymond F.},
  title         = {{The general caloron correspondence}},
  journal       = {Journal of Geometry and Physics},
  volume        = {62},
  number        = {2},
  pages         = {224--241},
  year          = {2012},
  eprint        = {1105.0805},
  archiveprefix = {arXiv},
  primaryclass  = {math.DG},
  doi           = {10.1016/j.geomphys.2011.10.015}
}

@article{GarlandMurray1988,
  author  = {Garland, H. and Murray, Michael K.},
  title   = {{Kac-Moody monopoles and periodic instantons}},
  journal = {Communications in Mathematical Physics},
  volume  = {120},
  number  = {2},
  pages   = {335--351},
  year    = {1988},
  doi     = {10.1007/BF01217968}
}

@article{BergmanVaradarajan2004,
  author        = {Bergman, Aaron and Varadarajan, Uday},
  title         = {{Loop groups, Kaluza--Klein reduction and $M$-theory}},
  journal       = {Journal of High Energy Physics},
  volume        = {2005},
  number        = {06},
  pages         = {043},
  year          = {2005},
  eprint        = {hep-th/0406218},
  archiveprefix = {arXiv},
  doi           = {10.1088/1126-6708/2005/06/043}
}

@unpublished{Schreiber2017,
  author = {Schreiber, Urs},
  title  = {{Duality in String/M-Theory from Cyclic Cohomology of Super Lie $n$-Algebras}},
  note   = {talk at Topological Quantum Field Theory Seminar, Lisbon},
  month  = {1},
  year   = {2017},
  url    = {https://ncatlab.org/schreiber/show/Lisbon+2017}
}

@article{AyalaFrancis2015,
  author        = {Ayala, David and Francis, John},
  title         = {{Factorization homology of topological manifolds}},
  journal       = {Journal of Topology},
  volume        = {8},
  number        = {4},
  pages         = {1045--1084},
  year          = {2015},
  note          = {{\S}4},
  doi           = {10.1112/jtopol/jtv028},
  archiveprefix = {arXiv},
  eprint        = {1206.5522},
  primaryclass  = {math.AT}
}

@article{Spanier1949,
  author    = {Spanier, Edwin},
  title     = {{Borsuk's Cohomotopy Groups}},
  journal   = {Annals of Mathematics},
  series    = {Second Series},
  volume    = {50},
  number    = {1},
  pages     = {203--245},
  year      = {1949},
  publisher = {Annals of Mathematics},
  doi       = {10.2307/1969362}
}

@article{Pontrjagin1938,
  author    = {Pontrjagin, Lev S.},
  title     = {{Classification of continuous maps of a complex into a sphere, Communication I}},
  journal   = {Comptes Rendus (Doklady) de l'Acad{\'e}mie des Sciences de l'URSS},
  series    = {Nouvelle S{\'e}rie},
  volume    = {19},
  number    = {3},
  pages     = {147--149},
  year      = {1938},
  note      = {JFM 64.0673.01}
}

@article{Borsuk1936,
  author  = {Borsuk, Karol},
  title   = {{Sur les groupes des classes de transformations continues}},
  journal = {Comptes Rendus Hebdomadaires des S{\'e}ances de l'Acad{\'e}mie des Sciences},
  volume  = {202},
  number  = {2--3},
  pages   = {1400--1403},
  year    = {1936},
  url     = {https://gallica.bnf.fr/ark:/12148/bpt6k3154f}
}

@book{Switzer1975,
  author    = {Switzer, Robert M.},
  title     = {{Algebraic Topology --- Homotopy and Homology}},
  series    = {Grundlehren der mathematischen Wissenschaften},
  volume    = {212},
  publisher = {Springer},
  year      = {1975},
  doi       = {10.1007/978-3-642-61923-6}
}

@article{Adams1971,
  author  = {Adams, John Frank},
  title   = {{A variant of E. H. Brown's representability theorem}},
  journal = {Topology},
  volume  = {10},
  number  = {3},
  pages   = {185--198},
  year    = {1971},
  month   = {8},
  doi     = {10.1016/0040-9383(71)90003-6}
}

@book{BarnesRoitzheim2020,
  author    = {Barnes, David and Roitzheim, Constanze},
  title     = {{Foundations of Stable Homotopy Theory}},
  series    = {Cambridge Studies in Advanced Mathematics},
  volume    = {185},
  publisher = {Cambridge University Press},
  year      = {2020},
  month     = {3},
  doi       = {10.1017/9781108636575},
  isbn      = {978-1-108-48278-3}
}

@book{Adams1974,
  author    = {Adams, John Frank},
  title     = {{Stable Homotopy and Generalized Homology}},
  series    = {Chicago Lectures in Mathematics},
  publisher = {University of Chicago Press},
  year      = {1974},
  isbn      = {978-0-226-00524-9},
  url       = {https://press.uchicago.edu/ucp/books/book/chicago/S/bo21302708.html}
}

@article{EilenbergMacLane1954III,
  author  = {Eilenberg, Samuel and Mac Lane, Saunders},
  title   = {{On the Groups $H(\Pi,n)$, III: Operations and Obstructions}},
  journal = {Annals of Mathematics},
  series  = {Second Series},
  volume  = {60},
  number  = {3},
  pages   = {513--557},
  year    = {1954},
  month   = {11},
  doi     = {10.2307/1969849}
}

@article{Eilenberg1940,
  author  = {Eilenberg, Samuel},
  title   = {{Cohomology and Continuous Mappings}},
  journal = {Annals of Mathematics},
  series  = {Second Series},
  volume  = {41},
  number  = {1},
  pages   = {231--251},
  year    = {1940},
  month   = {1},
  doi     = {10.2307/1968828}
}

@article{Steenrod1972,
  author  = {Steenrod, Norman},
  title   = {{Cohomology operations, and obstructions to extending continuous functions}},
  journal = {Advances in Mathematics},
  volume  = {8},
  number  = {3},
  pages   = {371--416},
  year    = {1972},
  month   = {6},
  doi     = {10.1016/0001-8708(72)90004-7}
}

@misc{Mitchell2011,
  author       = {Mitchell, Stephen A.},
  title        = {{Notes on principal bundles and classifying spaces}},
  howpublished = {Lecture notes, University of Washington},
  year         = {2011},
  month        = {6},
  url          = {https://sites.math.washington.edu/~mitchell/Notes/prin.pdf}
}

@article{Milnor1956,
  author  = {Milnor, John},
  title   = {{Construction of Universal Bundles, II}},
  journal = {Annals of Mathematics},
  series  = {Second Series},
  volume  = {63},
  number  = {3},
  pages   = {430--436},
  year    = {1956},
  month   = {5},
  doi     = {10.2307/1970012}
}

@article{Steenrod1968,
  author  = {Steenrod, Norman},
  title   = {{Milgram's classifying space of a topological group}},
  journal = {Topology},
  volume  = {7},
  number  = {4},
  pages   = {349--368},
  year    = {1968},
  month   = {11},
  doi     = {10.1016/0040-9383(68)90012-8}
}

@article{Castellani2018,
  author        = {Castellani, Leonardo},
  title         = {{Supergravity in the group-geometric framework: a primer}},
  journal       = {Fortschr. Phys.},
  volume        = {66},
  number        = {4},
  pages         = {1800014},
  year          = {2018},
  month         = {4},
  doi           = {10.1002/prop.201800014},
  eprint        = {1802.03407},
  archiveprefix = {arXiv},
  primaryclass  = {hep-th}
}

@article{Lazarev2013,
  author        = {Lazarev, Andrey},
  title         = {{Maurer--Cartan moduli and models for function spaces}},
  journal       = {Advances in Mathematics},
  volume        = {235},
  pages         = {296--320},
  year          = {2013},
  doi           = {10.1016/j.aim.2012.11.009},
  eprint        = {1109.3715},
  archiveprefix = {arXiv},
  primaryclass  = {math.AT}
}

@article{Reinhold2019,
  author  = {Reinhold, Ben},
  title   = {{$L_\infty$-algebras and their cohomology}},
  journal = {Emergent Scientist},
  year    = {2019},
  volume  = {3},
  pages   = {4},
  doi     = {10.1051/emsci/2019003}
}

@article{KraftSchnitzer2024,
  author        = {Kraft, Andreas and Schnitzer, Jonas},
  title         = {{An Introduction to {$L_\infty$}-Algebras and Their Homotopy Theory}},
  journal       = {Reviews in Mathematical Physics},
  volume        = {36},
  number        = {01},
  pages         = {2330006},
  year          = {2024},
  doi           = {10.1142/S0129055X23300066},
  eprint        = {2207.01861},
  archiveprefix = {arXiv},
  primaryclass  = {math.QA}
}

@book{Mardesic2000,
  author    = {Marde{\v{s}}i{\'c}, Sibe},
  title     = {{Strong Shape and Homology}},
  series    = {Springer Monographs in Mathematics},
  publisher = {Springer},
  address   = {Berlin, Heidelberg},
  year      = {2000},
  isbn      = {978-3-540-66198-6},
  doi       = {10.1007/978-3-662-13064-3},
  mrnumber  = {1774366}
}

@book{MardesicSegal1982,
  author    = {Marde{\v{s}}i{\'c}, Sibe and Segal, Jack},
  title     = {{Shape Theory: The Inverse System Approach}},
  series    = {North-Holland Mathematical Library},
  volume    = {26},
  publisher = {North-Holland Publishing Company},
  address   = {Amsterdam},
  year      = {1982},
  isbn      = {978-0-444-86286-0},
  mrnumber  = {0676974},
  url       = {https://www.sciencedirect.com/bookseries/north-holland-mathematical-library/vol/26/suppl/C}
}

@misc{Dugger1998,
  author       = {Dugger, Daniel},
  title        = {{Sheaves and Homotopy Theory}},
  year         = {1998},
  howpublished = {Unpublished manuscript},
  url          = {https://ncatlab.org/nlab/files/DuggerSheavesAndHomotopyTheory.pdf}
}

@book{DydakSegal1978,
  author    = {Dydak, Jerzy and Segal, Jack},
  title     = {{Shape Theory: An Introduction}},
  series    = {Lecture Notes in Mathematics},
  volume    = {688},
  publisher = {Springer-Verlag},
  address   = {Berlin, Heidelberg},
  year      = {1978},
  isbn      = {978-3-540-08955-1},
  doi       = {10.1007/BFb0067572}
}

@article{BorsukDydak1980,
  author  = {Borsuk, Karol and Dydak, Jerzy},
  title   = {{What is the theory of shape?}},
  journal = {Bulletin of the Australian Mathematical Society},
  year    = {1980},
  volume  = {22},
  number  = {2},
  pages   = {161--198},
  doi     = {10.1017/S000497270000647X}
}

@book{Borsuk1975,
  author    = {Borsuk, Karol},
  title     = {{Theory of Shape}},
  series    = {Monografie Matematyczne},
  volume    = {59},
  publisher = {PWN---Polish Scientific Publishers},
  address   = {Warszawa},
  year      = {1975},
  mrnumber  = {0418088}
}

@incollection{Lawvere1991,
  author    = {Lawvere, F. William},
  title     = {{Some Thoughts on the Future of Category Theory}},
  booktitle = {Category Theory: Proceedings of the International Conference held in Como, Italy, July 22--28, 1990},
  editor    = {Carboni, Aurelio and Pedicchio, Maria Cristina and Rosolini, Giuseppe},
  series    = {Lecture Notes in Mathematics},
  volume    = {1488},
  publisher = {Springer},
  address   = {Berlin, Heidelberg},
  year      = {1991},
  pages     = {1--13},
  doi       = {10.1007/BFb0084208}
}

@book{LawvereRosebrugh2003,
  author    = {Lawvere, F. William and Rosebrugh, Robert},
  title     = {{Sets for Mathematics}},
  publisher = {Cambridge University Press},
  address   = {Cambridge},
  year      = {2003},
  isbn      = {978-0-521-01060-3},
  doi       = {10.1017/CBO9780511755460}
}

@article{Lawvere1994,
  author  = {Lawvere, F. William},
  title   = {{Cohesive Toposes and Cantor's ``lauter Einsen''}},
  journal = {Philosophia Mathematica},
  year    = {1994},
  volume  = {2},
  number  = {1},
  pages   = {5--15},
  doi     = {10.1093/philmat/2.1.5}
}

@incollection{GoerssSchemmerhorn2007,
  author        = {Goerss, Paul and Schemmerhorn, Kirsten},
  title         = {{Model Categories and Simplicial Methods}},
  booktitle     = {Interactions between Homotopy Theory and Algebra},
  series        = {Contemporary Mathematics},
  volume        = {436},
  publisher     = {American Mathematical Society},
  address       = {Providence, RI},
  year          = {2007},
  pages         = {3--49},
  doi           = {10.1090/conm/436/08403},
  eprint        = {math/0609537},
  archiveprefix = {arXiv},
  primaryclass  = {math.AT}
}

@incollection{DwyerSpalinski1995,
  author    = {Dwyer, William G. and Spali{\'n}ski, Jan},
  title     = {{Homotopy Theories and Model Categories}},
  booktitle = {Handbook of Algebraic Topology},
  editor    = {James, I. M.},
  publisher = {North-Holland},
  address   = {Amsterdam},
  year      = {1995},
  chapter   = {2},
  pages     = {73--126},
  isbn      = {978-0-444-81779-2},
  doi       = {10.1016/B978-044481779-2/50003-1}
}

@article{Noohi2007,
  author  = {Noohi, Behrang},
  title   = {{Notes on 2-groupoids, 2-groups and crossed modules}},
  journal = {Homology, Homotopy and Applications},
  year    = {2007},
  volume  = {9},
  number  = {1},
  pages   = {75--106},
  doi     = {10.4310/hha.2007.v9.n1.a3}
}

@book{Makeenko2002,
  author    = {Makeenko, Yuri},
  title     = {{Methods of Contemporary Gauge Theory}},
  series    = {Cambridge Monographs on Mathematical Physics},
  publisher = {Cambridge University Press},
  year      = {2002},
  isbn      = {978-0-521-80911-5},
  doi       = {10.1017/CBO9780511535147}
}

@book{Rubakov2002,
  author     = {Rubakov, Valery},
  translator = {Wilson, Stephen S.},
  title      = {{Classical Theory of Gauge Fields}},
  publisher  = {Princeton University Press},
  address    = {Princeton, NJ},
  year       = {2002},
  isbn       = {978-0-691-05927-3},
  doi        = {10.1515/9781400825097}
}

@incollection{Geroch1996,
  author        = {Geroch, Robert},
  title         = {{Partial Differential Equations of Physics}},
  booktitle     = {General Relativity: Proceedings of the Forty-Sixth Scottish Universities Summer School in Physics},
  editor        = {Hall, G. S. and Pulham, J. R.},
  publisher     = {IOP Publishing},
  year          = {1996},
  pages         = {19--94},
  doi           = {10.1201/9780203753804-2},
  eprint        = {gr-qc/9602055},
  archiveprefix = {arXiv}
}

@article{Khavkine2012,
  author  = {Khavkine, Igor},
  title   = {{Characteristics, Conal Geometry and Causality in Locally Covariant Field Theory}},
  journal = {arXiv},
  year    = {2012},
  doi     = {10.48550/arxiv.1211.1914}
}

@article{BandosBerkovitsSorokin1997,
    author = "Bandos, Igor and Berkovits, Nathan and Sorokin, Dmitri",
    title = "{{Duality-symmetric eleven-dimensional supergravity and its coupling to M-branes}}",
    eprint = "hep-th/9711055",
    archivePrefix = "arXiv",
    doi = "10.1016/S0550-3213(98)00102-3",
    journal = "Nucl. Phys. B",
    volume = "522",
    pages = "214--233",
    year = "1998"
}

@article{MkrtchyanValach2022,
    author = "Mkrtchyan, Karapet and Valach, Fridrich",
    title = "{{Democratic actions for type II supergravities}}",
    journal = "Phys. Rev. D",
    volume = "107",
    number = "6",
    pages = "066027",
    year = "2023",
    doi = "10.1103/PhysRevD.107.066027",
    eprint = "2207.00626",
    archivePrefix = "arXiv",
    primaryClass = "hep-th"
}

@article{BandosNurmagambetovSorokin2003,
    author = "Bandos, Igor and Nurmagambetov, Alexei and Sorokin, Dmitri",
    title = "{{Various Faces of Type IIA Supergravity}}",
    journal = "Nucl. Phys. B",
    volume = "676",
    pages = "189--228",
    year = "2004",
    doi = "10.1016/j.nuclphysb.2003.10.036",
    eprint = "hep-th/0307153",
    archivePrefix = "arXiv",
    primaryClass = "hep-th"
}

@article{BergshoeffKalloshOrtinRoestVanProeyen2001,
    author = "Bergshoeff, Eric and Kallosh, Renata and Ort{\'i}n, Tom{\'a}s and Roest, Diederik and Van Proeyen, Antoine",
    title = "{{New Formulations of D=10 Supersymmetry and D8-O8 Domain Walls}}",
    eprint = "hep-th/0103233",
    archivePrefix = "arXiv",
    doi = "10.1088/0264-9381/18/17/303",
    journal = "Class. Quant. Grav.",
    volume = "18",
    number = "17",
    pages = "3359--3382",
    year = "2001"
}

@article{DallAgataLechnerTonin1998,
    author = "Dall'Agata, Gianguido and Lechner, Kurt and Tonin, Mario",
    title = "{{$D=10$, $N=\mathrm{IIB}$ Supergravity: Lorentz-invariant actions and duality}}",
    eprint = "hep-th/9806140",
    archivePrefix = "arXiv",
    doi = "10.1088/1126-6708/1998/07/017",
    journal = "JHEP",
    volume = "07",
    pages = "017",
    year = "1998"
}

@article{HehlItinObukhov2016,
  author        = "Hehl, Friedrich W. and Itin, Yakov and Obukhov, Yuri N.",
  title         = "{{On Kottler's path: origin and evolution of the premetric program in gravity and in electrodynamics}}",
  journal       = "International Journal of Modern Physics D",
  volume        = "25",
  number        = "11",
  pages         = "1640016",
  year          = "2016",
  doi           = "10.1142/S0218271816400162",
  eprint        = "1607.06159",
  archivePrefix = "arXiv",
  primaryClass  = "gr-qc"
}

@book{HehlObukhov2003,
  author    = "Hehl, Friedrich W. and Obukhov, Yuri N.",
  title     = "{{Foundations of Classical Electrodynamics: Charge, Flux, and Metric}}",
  series    = "Progress in Mathematical Physics",
  volume    = "33",
  publisher = "Birkh{\"a}user",
  address   = "Boston, MA",
  year      = "2003",
  doi       = "10.1007/978-1-4612-0051-2",
  isbn      = "978-0-8176-4222-8"
}

@incollection{Delphenich2015,
  author        = "Delphenich, David H.",
  title         = "{{Pre-metric electromagnetism as a path to unification}}",
  booktitle     = "{Unified Field Mechanics: Natural Science Beyond the Veil of Spacetime}",
  editor        = "Amoroso, Richard L. and Kauffman, Louis H. and Rowlands, Peter",
  publisher     = "World Scientific",
  pages         = "215--220",
  year          = "2015",
  doi           = "10.1142/9789814719063_0023",
  eprint        = "1512.05183",
  archivePrefix = "arXiv",
  primaryClass  = "gr-qc"
}

@article{Kottler1922,
  author  = "Kottler, Friedrich",
  title   = "{{Maxwell'sche Gleichungen und Metrik}}",
  journal = "Sitzungsber. Akad. Wiss. Wien, Abt. IIa",
  volume  = "131",
  pages   = "119--146",
  year    = "1922",
  note    = "English translation by D. H. Delphenich: \url{https://ncatlab.org/nlab/files/Kottler-MaxwellEquationsAnadMetric.pdf}"
}

@article{Post1972,
  author  = "Post, E. Jan",
  title   = "{{The constitutive map and some of its ramifications}}",
  journal = "Annals of Physics",
  volume  = "71",
  number  = "2",
  pages   = "497--518",
  year    = "1972",
  doi     = "10.1016/0003-4916(72)90129-7"
}

@article{vanDantzig1934,
  author  = "van Dantzig, David",
  title   = "{{The fundamental equations of electromagnetism, independent of metrical geometry}}",
  journal = "Mathematical Proceedings of the Cambridge Philosophical Society",
  volume  = "30",
  number  = "4",
  pages   = "421--427",
  year    = "1934",
  doi     = "10.1017/S0305004100012664"
}

@article{Cartan1924,
  author    = "Cartan, {\'E}lie",
  title     = "{Sur les vari{\'e}t{\'e}s {\`a} connexion affine, et la th{\'e}orie de la relativit{\'e} g{\'e}n{\'e}ralis{\'e}e (premi{\`e}re partie) (Suite)}",
  journal   = "Annales scientifiques de l'{\'E}cole Normale Sup{\'e}rieure",
  series    = "3",
  volume    = "41",
  pages     = "1--25",
  year      = "1924",
  doi       = "10.24033/asens.753",
  url       = "http://www.numdam.org/item/ASENS_1924_3_41__1_0/"
}

@book{BlumenhagenLuestTheisen2013,
  author    = "Blumenhagen, Ralph and L{\"u}st, Dieter and Theisen, Stefan",
  title     = "{{Basic Concepts of String Theory}}",
  series    = "Theoretical and Mathematical Physics",
  publisher = "Springer",
  address   = "Heidelberg",
  year      = "2013",
  doi       = "10.1007/978-3-642-29497-6",
  isbn      = "978-3-642-29496-9"
}

@article{BunsterHenneaux2011,
  author        = "Bunster, Claudio and Henneaux, Marc",
  title         = "{{The Action for Twisted Self-Duality}}",
  journal       = "Phys. Rev. D",
  volume        = "83",
  number        = "12",
  pages         = "125015",
  year          = "2011",
  doi           = "10.1103/PhysRevD.83.125015",
  eprint        = "1103.3621",
  archivePrefix = "arXiv",
  primaryClass  = "hep-th"
}

@book{KrieglMichor1997,
  author    = {Kriegl, Andreas and Michor, Peter W.},
  title     = {{The Convenient Setting of Global Analysis}},
  series    = {Mathematical Surveys and Monographs},
  volume    = {53},
  publisher = {American Mathematical Society},
  address   = {Providence, RI},
  year      = {1997},
  isbn      = {978-0-8218-0780-4},
  doi       = {10.1090/surv/053},
  url       = {https://doi.org/10.1090/surv/053}
}

@incollection{Michor2020,
  author        = {Michor, Peter W.},
  title         = {{Manifolds of Mappings for Continuum Mechanics}},
  booktitle     = {Geometric Continuum Mechanics},
  editor        = {Segev, Reuven and Epstein, Marcelo},
  series        = {Advances in Continuum Mechanics},
  volume        = {42},
  pages         = {3--75},
  year          = {2020},
  publisher     = {Birkh{\"a}user},
  address       = {Cham},
  doi           = {10.1007/978-3-030-42683-5_1},
  eprint        = {1909.00445},
  archiveprefix = {arXiv},
  primaryclass  = {math.DG}
}

@book{Michor1980,
  author    = {Michor, Peter W.},
  title     = {{Manifolds of Differentiable Mappings}},
  series    = {Shiva Mathematics Series},
  volume    = {3},
  publisher = {Shiva Publishing},
  address   = {Orpington, UK},
  year      = {1980},
  isbn      = {0-906812-03-8},
  url       = {https://www.mat.univie.ac.at/~michor/manifolds_of_differentiable_mappings.pdf}
}

@inproceedings{Townsend1996,
  author        = {Townsend, Paul K.},
  title         = {{Four Lectures on M-Theory}},
  booktitle     = {1996 Summer School in High Energy Physics and Cosmology},
  editor        = {Gava, E. and Masiero, A. and Narain, K. S. and Randjbar-Daemi, S. and Senjanovi{\'c}, G.},
  series        = {The ICTP Series in Theoretical Physics},
  volume        = {13},
  pages         = {385--438},
  publisher     = {World Scientific},
  year          = {1997},
  eprint        = {hep-th/9612121},
  archiveprefix = {arXiv},
  primaryclass  = {hep-th}
}

@article{CremmerJuliaScherk1978,
  author  = {Cremmer, Eug\`ene and Julia, Bernard and Scherk, Jo\"el},
  title   = {{Supergravity theory in 11 dimensions}},
  journal = {Phys. Lett. B},
  year    = {1978},
  volume  = {76},
  pages   = {409--412},
  doi     = {10.1016/0370-2693(78)90894-8}
}

@book{BuchbinderKuzenko1995,
  author    = {Buchbinder, Ioseph L. and Kuzenko, Sergei M.},
  title     = {{Ideas and Methods of Supersymmetry and Supergravity -- or: A Walk Through Superspace}},
  publisher = {Institute of Physics Publishing},
  year      = {1995},
  doi       = {10.1201/9780367802530}
}

@book{TraubenbergEtAl2020,
  author    = {Rausch de Traubenberg, Michel and Valenzuela, Mauricio},
  title     = {{A Supergravity Primer --- From Geometrical Principles to the Final Lagrangian}},
  publisher = {World Scientific},
  year      = {2020},
  doi       = {10.1142/11557}
}

@incollection{DomingoTraubenberg2023,
  title={{Supergravity: Application in Particle Physics}},
  author={Domingo, Florian and de Traubenberg, Michel Rausch},
  booktitle={Handbook of Quantum Gravity},
  publisher={Springer},
  year={2023},
  doi={10.1007/978-981-19-3079-9},
  eprint={2209.12541},
  archivePrefix={arXiv},
  primaryClass={hep-th}
}

@article{Witten1998,
  author  = {Witten, Edward},
  title   = {{Magic, Mystery, and Matrix}},
  journal = {Notices of the American Mathematical Society},
  year    = {1998},
  volume  = {45},
  number  = {9},
  pages   = {1124--1129},
  url     = {https://www.ams.org/notices/199809/witten.pdf}
}

@article{Witten1996,
  author  = {Witten, Edward},
  title   = {{Five-branes And M-Theory On An Orbifold}},
  journal = {Nuclear Physics B},
  year    = {1996},
  volume  = {463},
  pages   = {383--397},
  doi     = {10.1016/0550-3213(96)00032-6}
}

@incollection{Duff1999,
  author    = {Duff, M. J.},
  title     = {{A Layman's Guide to M-theory}},
  booktitle = {Trends in Mathematical Physics},
  series    = {AMS/IP Studies in Advanced Mathematics},
  year      = {1999},
  volume    = {13},
  pages     = {184--213},
  doi       = {10.1090/amsip/013}
}

@article{Duff1996,
  author  = {Duff, M. J.},
  title   = {{M-Theory (the Theory Formerly Known as Strings)}},
  journal = {International Journal of Modern Physics A},
  year    = {1996},
  volume  = {11},
  pages   = {5623--5641},
  doi     = {10.1142/S0217751X96002583}
}

@article{Horava1996,
  author  = {Hořava, Petr and Witten, Edward},
  title   = {{Heterotic and Type I string dynamics from eleven dimensions}},
  journal = {Nuclear Physics B},
  year    = {1996},
  volume  = {460},
  pages   = {506--524},
  doi     = {10.1016/0550-3213(95)00621-4}
}

@article{HenneauxTeitelboim1986,
  author  = {Henneaux, Marc and Teitelboim, Claudio},
  title   = {{$p$-Form electrodynamics}},
  journal = {Foundations of Physics},
  year    = {1986},
  volume  = {16},
  number  = {7},
  pages   = {593--617},
  doi     = {10.1007/BF01889624}
}

@inbook{Basdevant2007,
  author    = {Basdevant, Jean-Louis},
  title     = {{Variational Principles in Physics}},
  publisher = {Springer},
  year      = {2007},
  chapter   = {6},
  note      = {Section 6.2.1},
  doi       = {10.1007/978-3-031-21692-3}
}

@incollection{Anderson1992,
  author    = {Anderson, Ian M.},
  title     = {{Introduction to the variational bicomplex}},
  booktitle = {Mathematical Aspects of Classical Field Theory},
  editor    = {Gotay, Mark J. and Marsden, Jerrold E. and Moncrief, Vincent},
  series    = {Contemporary Mathematics},
  volume    = {132},
  pages     = {51--73},
  year      = {1992},
  publisher = {American Mathematical Society},
  address   = {Providence, RI},
  doi       = {10.1090/conm/132/1188434}
}

@misc{Anderson1989,
  author       = {Anderson, Ian M.},
  title        = {{The Variational Bicomplex}},
  howpublished = {Department of Mathematics, Utah State University},
  year         = {1989},
  url          = {https://ncatlab.org/nlab/files/AndersonVariationalBicomplex.pdf}
}

@incollection{Zuckerman1987,
  author    = {Zuckerman, Gregg J.},
  title     = {{Action Principles and Global Geometry}},
  booktitle = {Mathematical Aspects of String Theory},
  editor    = {Yau, Shing-Tung},
  series    = {Advanced Series in Mathematical Physics},
  volume    = {1},
  publisher = {World Scientific},
  year      = {1987},
  chapter   = {1},
  pages     = {259--284},
  doi       = {10.1142/9789812798411_0013},
  url       = {https://ncatlab.org/nlab/files/ZuckermanVariation.pdf}
}

@inproceedings{Claessens2005,
  author    = {Claessens, Laurent},
  title     = {{Field theory from a bundle point of view}},
  booktitle = {Proceedings of the First Modave Summer School in Mathematical Physics},
  editor    = {Bouchard, Vincent and Wijns, Alexander},
  series    = {Modave Schools and Workshops},
  volume    = {1},
  publisher = {International Solvay Institutes},
  year      = {2005},
  pages     = {1--47},
  url       = {https://laurent.claessens-donadello.eu/pdf/lectures.pdf}
}

@article{BernalSanchez2005,
  author        = {Bernal, Antonio N. and S{\'a}nchez, Miguel},
  title         = {{Smoothness of time functions and the metric splitting of globally hyperbolic spacetimes}},
  journal       = {Communications in Mathematical Physics},
  volume        = {257},
  number        = {1},
  pages         = {43--50},
  year          = {2005},
  month         = {6},
  doi           = {10.1007/s00220-005-1346-1},
  eprint        = {gr-qc/0401112},
  archiveprefix = {arXiv},
  primaryclass  = {gr-qc}
}

@incollection{MinguzziSanchez2008,
  author        = {Minguzzi, Ettore and S{\'a}nchez, Miguel},
  title         = {{The causal hierarchy of spacetimes}},
  booktitle     = {Recent Developments in Pseudo-Riemannian Geometry},
  editor        = {Alekseevsky, Dmitry V. and Baum, Helga},
  series        = {ESI Lectures in Mathematics and Physics},
  volume        = {4},
  publisher     = {European Mathematical Society},
  pages         = {299--358},
  year          = {2008},
  isbn          = {978-3-03719-051-7},
  doi           = {10.4171/051-1/9},
  eprint        = {gr-qc/0609119},
  archiveprefix = {arXiv},
  primaryclass  = {gr-qc},
}

@article{DragoGinouxMurro2024,
  author  = {Drago, Nicoló and Ginoux, Nicolas and Murro, Simone},
  title   = {{On the Cauchy problem for the Faraday tensor on globally hyperbolic manifolds with timelike boundary}},
  journal = {Rendiconti Lincei, Matematica e Applicazioni},
  year    = {2024},
  volume  = {34},
  pages   = {809--829},
  doi     = {10.4171/rlm/1029}
}

@book{Fre2013,
  title={{Black Holes, Cosmology and Introduction to Supergravity}},
  author={Fr{\'e}, Pietro},
  series={Gravity, a Geometrical Course},
  volume={2},
  year={2013},
  publisher={Springer},
  doi={10.1007/978-94-007-5443-0}
}

@inproceedings{Cremmer1982,
  author        = {Cremmer, E.},
  title         = {{Dimensional Reduction in Field Theory and Hidden Symmetries in Extended Supergravity}},
  booktitle     = {Supergravity '81},
  editor        = {Ferrara, S. and Taylor, J. G.},
  publisher     = {Cambridge University Press},
  address       = {Cambridge},
  year          = {1982},
  pages         = {313--368},
  note          = {Lectures given at the 1981 Spring School on Supergravity, ICTP, Trieste, Italy. Report No.: LPTENS-81-18},
  url           = {https://inspirehep.net/literature/168547}
}

@book{WessBagger1992,
  author    = {Wess, Julius and Bagger, Jonathan},
  title     = {{Supersymmetry and Supergravity}},
  edition   = {Second},
  publisher = {Princeton University Press},
  year      = {1992},
  isbn      = {9780691025308}
}

@article{Brown1962,
  author    = {Brown, Jr., Edgar H.},
  title     = {{Cohomology Theories}},
  journal   = {Annals of Mathematics},
  series    = {Second Series},
  volume    = {75},
  number    = {3},
  pages     = {467--484},
  year      = {1962},
  publisher = {Annals of Mathematics},
  doi       = {10.2307/1970209}
}

@article{vanNieuwenhuizen1981,
  author  = {van Nieuwenhuizen, Peter},
  title   = {{Supergravity}},
  journal = {Physics Reports},
  volume  = {68},
  number  = {4},
  pages   = {189--398},
  year    = {1981},
  month   = {2},
  doi     = {10.1016/0370-1573(81)90157-5}
}

@incollection{Sezgin2023,
  title={{Survey of supergravities}},
  author={Sezgin, Ergin},
  booktitle={Handbook of Quantum Gravity},
  publisher={Springer},
  year={2023},
  doi={10.1007/978-981-19-3079-9},
  eprint={2312.06754},
  archivePrefix={arXiv},
  primaryClass={hep-th}
}

@book{Ortin2015,
  author    = {Ortín, Tomás},
  title     = {{Gravity and Strings}},
  edition   = {2nd},
  publisher = {Cambridge University Press},
  year      = {2015},
  doi       = {10.1017/CBO9781139019750}
}

@book{FreedmanVanProeyen2012,
    author    = {Freedman, Daniel Z. and Van Proeyen, Antoine},
    title     = {{Supergravity}},
    publisher = {Cambridge University Press},
    address   = {Cambridge},
    year      = {2012},
    doi       = {10.1017/CBO9781139026833},
    isbn      = {978-0-521-19401-3}
}

@book{SalamSezgin1989,
  editor    = {Salam, Abdus and Sezgin, Ergin},
  title     = {{Supergravities in Diverse Dimensions}},
  publisher = {World Scientific / North-Holland},
  address   = {Singapore / Amsterdam},
  year      = {1989},
  note      = {Commentary and Reprints (2 Volumes)},
  doi       = {10.1142/0277}
}

@article{DuffNilssonPope1986,
    author   = {Duff, Mike and Nilsson, Bengt and Pope, Christopher},
    title    = {{Kaluza-Klein Supergravity}},
    journal  = {Physics Reports},
    volume   = {130},
    number   = {1--2},
    pages    = {1--142},
    year     = {1986},
    doi      = {10.1016/0370-1573(86)90163-8}
}

@article{SS26-Islands,
  author        = {Sati, Hisham and Schreiber, Urs},
  title         = {{FQH Liquids with Flux-Expulsion Islands allow Nonabelian Anyons}},
  journal       = {Europhysics Letters},
  year          = {2026},
  doi           = {10.1209/0295-5075/aea022},
  eprint        = {2602.18748},
  archiveprefix = {arXiv},
  primaryclass  = {cond-mat.mes-hall}
}

@article{BunkeNikolausVoelkl2016,
  author        = {Bunke, Ulrich and Nikolaus, Thomas and V{\"o}lkl, Michael},
  title         = {{Differential cohomology theories as sheaves of spectra}},
  journal       = {Journal of Homotopy and Related Structures},
  volume        = {11},
  number        = {1},
  pages         = {1--66},
  year          = {2016},
  doi           = {10.1007/s40062-014-0092-5},
  eprint        = {1311.3188},
  archiveprefix = {arXiv},
  primaryclass  = {math.AT}
}

@unpublished{SS26-Symplectic,
  author = {Sati, Hisham and Schreiber, Urs},
  title  = {{Prequantizing the C-field in IR-Completed 11D Supergravity}},
  year = {2026},
  note   = {In preparation}
}

@article{LikoBooth2008,
  author        = {Liko, Tomas and Booth, Ivan},
  title         = {{Supersymmetric isolated horizons}},
  journal       = {Classical and Quantum Gravity},
  year          = {2008},
  volume        = {25},
  number        = {10},
  pages         = {105020},
  doi           = {10.1088/0264-9381/25/10/105020},
  archiveprefix = {arXiv},
  eprint        = {0712.3308},
  primaryclass  = {gr-qc}
}

@incollection{CrnkovicWitten1987,
  author    = {Crnkovi{\'c}, {\v{C}}edomir and Witten, Edward},
  title     = {{Covariant Description of Canonical Formalism in Geometrical Theories}},
  booktitle = {Three Hundred Years of Gravitation},
  editor    = {Hawking, Stephen W. and Israel, Werner},
  chapter   = {16},
  pages     = {676--684},
  publisher = {Cambridge University Press},
  address   = {Cambridge},
  year      = {1987},
  isbn      = {978-0-521-37976-2},
  url       = {https://www.ias.edu/sites/default/files/sns/files/CovariantPaper-1987.pdf}
}

@misc{Schreiber2026,
  author       = {Schreiber, Urs},
  title        = {{The \$1 Million Puzzle: Hadrons and Quantum Computers --- Why Modern Physics Needs a New Global Mathematics}},
  howpublished = {Public evening lecture at Alfried Krupp Wissenschaftskolleg, Greifswald},
  month        = may,
  day          = {5},
  year         = {2026},
  url          = {https://ncatlab.org/schreiber/show/WIKO+2026}
}

@misc{AlfonsiKimLuciani2026,
  author        = {Alfonsi, Luigi and Kim, Hyungrok and Luciani, William G. A.},
  title         = {{Generalised Symmetries and Swampland-Type Constraints from Charge Quantisation via Rational Homotopy Theory}},
  eprint        = {2604.22656},
  archivePrefix = {arXiv},
  primaryClass  = {hep-th},
  doi           = {10.48550/arXiv.2604.22656},
  year          = {2026}
}

@misc{GSS26-SuGra,
    author        = "Giotopoulos, Grigorios and Sati, Hisham and Schreiber, Urs",
    title         = "{{Higher Superspace Supergravity and its IR-Completions}}",
    eprint        = "2607.20312",
    archivePrefix = "arXiv",
    primaryClass  = "hep-th",
    year          = "2026"
}

@misc{Clough2023,
  title         = {{The homotopy theory of differentiable sheaves}}, 
  author        = {Clough, Adrian},
  year          = {2023},
  eprint        = {2309.01757},
  archivePrefix = {arXiv},
  primaryClass  = {math.AT}
}

@misc{BunkShahbazi2023,
  author        = {Bunk, Severin  and  Shahbazi, Carlos S.},
  title         = {{Higher Geometric Structures on Manifolds and the Gauge Theory of Deligne Cohomology}}, 
  year          = {2023},
  eprint        = {2304.06633},
  archivePrefix = {arXiv},
  primaryClass  = {math.DG}
}

@article{Pavlov2024,
  author        = {Pavlov, Dmitri},
  title         = {{Projective model structures on diffeological spaces and smooth sets and the smooth Oka principle}},
  journal       = {Homology, Homotopy, and Applications},
  volume        = {26},
  number        = {2},
  pages         = {375--408},
  year          = {2024},
  doi           = {10.4310/HHA.2024.v26.n2.a18},
  eprint        = {2210.12845},
  archivePrefix = {arXiv},
  primaryClass  = {math.AT}
}

@misc{AmabelEtAl2021,
  title         = {{Differential Cohomology: Categories, Characteristic Classes, and Connections}},
  editor        = {Amabel, Araminta and Debray, Arun and Haine, Peter J.},
  author        = {Amabel, Araminta and Chua, Dexter and Debray, Arun and Devalapurkar, Sanath and Freed, Daniel and Haine, Peter J. and Hopkins, Michael and Parker, George and Reid, Cameron and Zhang, Allen},
  year          = {2021},
  eprint        = {2109.12250},
  archiveprefix = {arXiv},
  primaryclass  = {math.AT}
}

@article{Bunk2022,
  author        = {Bunk, Severin},
  title         = {{The $\mathbb{R}$-Local Homotopy Theory of Smooth Spaces}},
  journal       = {Journal of Homotopy and Related Structures},
  volume        = {17},
  pages         = {593--650},
  year          = {2022},
  doi           = {10.1007/s40062-022-00318-7},
  eprint        = {2007.06039},
  archivePrefix = {arXiv},
  primaryClass  = {math.AT}
}

@article{Szabo2013,
  author  = {Szabo, Richard J.},
  title   = {{Quantization of Higher Abelian Gauge Theory in Generalized Differential Cohomology}},
  journal = {Proceedings of 7th International Conference on Mathematical Methods in Physics — PoS(ICMP 2012)},
  year    = {2013},
  pages   = {009},
  doi     = {10.22323/1.175.0009}
}

@misc{GS26-IIA,
  author        = {Giotopoulos, Grigorios and Sati, Hisham},
  title         = {{Flux Quantization on 10D Type IIA Superspace via Cyclification from 11D}},
  year          = {2026},
  eprint        = {2606.16857},
  archivePrefix = {arXiv},
  primaryClass  = {hep-th}
}

@article{HoweSezginWest1997,
  author        = {Howe, Paul S. and Sezgin, Ergin and West, Peter C.},
  title         = {{The six-dimensional self-dual tensor}},
  journal       = {Physics Letters B},
  volume        = {400},
  number        = {3--4},
  pages         = {255--259},
  year          = {1997},
  doi           = {10.1016/S0370-2693(97)00365-1},
  eprint        = {hep-th/9702111},
  archivePrefix = {arXiv},
  primaryClass  = {hep-th}
}

@article{CandielloLechner1994,
  author        = {Candiello, Antonio and Lechner, Kurt},
  title         = {{Duality in supergravity theories}},
  journal       = {Nuclear Physics B},
  volume        = {412},
  pages         = {479--501},
  year          = {1994},
  doi           = {10.1016/0550-3213(94)90389-1},
  eprint        = {hep-th/9309143},
  archivePrefix = {arXiv},
  primaryClass  = {hep-th}
}

@misc{BunkEtAl2026,
  author        = {Bunk, Severin and M{\"u}ller, Lukas and Nuiten, Joost and Szabo, Richard J.},
  title         = {{Symmetries and Higher-Form Connections in Derived Differential Geometry}},
  year          = {2026},
  eprint        = {2602.03441},
  archivePrefix = {arXiv},
  primaryClass  = {math.DG},
}

@phdthesis{Schreiber2005,
  author        = {Schreiber, Urs},
  title         = {{From Loop Space Mechanics to Nonabelian Strings}},
  school        = {Universit{\"{a}}t Duisburg-Essen},
  year          = {2005},
  eprint        = {hep-th/0509163},
  archivePrefix = {arXiv},
  primaryClass  = {hep-th}
}

@article{FRS2016,
  author        = {Fiorenza, Domenico and Rogers, Chris and Schreiber, Urs},
  title         = {Higher $U(1)$-gerbe connections in geometric prequantization},
  journal       = {Reviews in Mathematical Physics},
  volume        = {28},
  number        = {06},
  pages         = {1650012},
  year          = {2016},
  doi           = {10.1142/S0129055X16500124},
  eprint        = {1304.0236},
  archivePrefix = {arXiv},
  primaryClass  = {math.DG}
}

@incollection{BorstenEtAl2025,
    author        = "Borsten, Leron and Jalali Farahani, Mehran and Jur{\v{c}}o, Branislav and Kim, Hyungrok and N{\'{a}}ro{\v{z}}n{\'{y}}, Ji{\v{r}}{\'{i}} and Rist, Dominik and Saemann, Christian and Wolf, Martin",
    title         = "{{Higher Gauge Theory}}",
    booktitle     = "{Encyclopedia of Mathematical Physics}",
    edition       = "2nd",
    volume        = "4",
    pages         = "159--185",
    year          = "2025",
    doi           = "10.1016/B978-0-323-95703-8.00217-2",
    eprint        = "2401.05275",
    archivePrefix = "arXiv",
    primaryClass  = "hep-th"
}

@article{FSS13-CupCS,
  author    = {Fiorenza, Domenico and Sati, Hisham and Schreiber, Urs},
  title     = {{Extended higher cup-product Chern-Simons theories}},
  journal   = {Journal of Geometry and Physics},
  volume    = {74},
  pages     = {130--163},
  year      = {2013},
  doi       = {10.1016/j.geomphys.2013.07.011},
  eprint    = {1207.5449},
  archivePrefix = {arXiv},
  primaryClass  = {hep-th}
}

@article{CouzensLuscherSparks2025,
    author        = "Couzens, Christopher and L{\"u}scher, Alice and Sparks, James",
    title         = "{{Localizing punctures in M-theory}}",
    journal       = "J. High Energ. Phys.",
    volume        = "2026",
    number        = "2",
    pages         = "180",
    year          = "2026",
    doi           = "10.1007/JHEP02(2026)180",
    eprint        = "2511.03397",
    archivePrefix = "arXiv",
    primaryClass  = "hep-th"
}

@article{WanPowis1994,
  author  = {Wan, K. K. and Powis, J. J.},
  title   = {{Quantum mechanics in Dirac’s front form}},
  journal = {International Journal of Theoretical Physics},
  year    = {1994},
  volume  = {33},
  pages   = {553--574},
  doi     = {10.1007/BF00670516}
}

@misc{AyalaFrancis2025,
  title         = {{A parametrized Pontryagin--Thom theorem}},
  author        = {Ayala, David and Francis, John},
  year          = {2025},
  eprint        = {2512.10274},
  archivePrefix = {arXiv},
  primaryClass  = {math.AT}
}

@book{GoerssJardine2009,
  author    = {Goerss, Paul and Jardine, J. F.},
  title     = {{Simplicial homotopy theory}},
  publisher = {Birkhäuser},
  series    = {Modern Birkhäuser Classics},
  year      = {2009},
  note      = {Reprint of the 1999 edition (Progress in Mathematics, Vol. 174)},
  doi       = {10.1007/978-3-0346-0189-4}
}

@article{DonosGauntlettPantelidou2013,
  author        = {Donos, Aristomenis and Gauntlett, Jerome P. and Pantelidou, Christiana},
  title         = {{Semi-local quantum criticality in string/M-theory}},
  journal       = {J. High Energ. Phys.},
  volume        = {2013},
  number        = {3},
  pages         = {103},
  year          = {2013},
  doi           = {10.1007/JHEP03(2013)103},
  eprint        = {1212.1462},
  archivePrefix = {arXiv},
  primaryClass  = {hep-th}
}

@article{DonosetalEtAl2013,
  author        = {Donos, Aristomenis and Gauntlett, Jerome P. and Sonner, Julian and Withers, Benjamin},
  title         = {{Competing orders in M-theory: superfluids, stripes and metamagnetism}},
  journal       = {J. High Energ. Phys.},
  volume        = {2013},
  number        = {3},
  pages         = {108},
  year          = {2013},
  doi           = {10.1007/JHEP03(2013)108},
  eprint        = {1212.0871},
  archivePrefix = {arXiv},
  primaryClass  = {hep-th}
}

@article{GauntlettetalSonnerWiseman2010,
  author        = {Gauntlett, Jerome and Sonner, Julian and Wiseman, Toby},
  title         = {{Quantum Criticality and Holographic Superconductors in M-theory}},
  journal       = {J. High Energ. Phys.},
  volume        = {2010},
  number        = {2},
  pages         = {060},
  year          = {2010},
  doi           = {10.1007/JHEP02(2010)060},
  eprint        = {0912.0512},
  archivePrefix = {arXiv},
  primaryClass  = {hep-th}
}

@article{GubserPufuRocha2010,
  author  = {Gubser, Steven S. and Pufu, Silviu S. and Rocha, Fabio D.},
  title   = {{Quantum critical superconductors in string theory and M-theory}},
  journal = {Phys. Lett. B},
  year    = {2010},
  volume  = {683},
  pages   = {201--204},
  doi     = {10.1016/j.physletb.2009.12.017}
}

@article{GauntlettSonnerWiseman2009,
  author  = {Gauntlett, Jerome P. and Sonner, Julian and Wiseman, Toby},
  title   = {{Holographic superconductivity in M-Theory}},
  journal = {Phys. Rev. Lett.},
  year    = {2009},
  volume  = {103},
  pages   = {151601},
  doi     = {10.1103/PhysRevLett.103.151601}
}

@book{Gaggioli2019,
  author    = {Baggioli, Matteo},
  title     = {{Applied Holography -- A Practical Mini-Course}},
  series    = {SpringerBriefs in Physics},
  publisher = {Springer},
  year      = {2019},
  doi       = {10.1007/978-3-030-35184-7}
}

@article{BuijsFelixMurillo2011,
  author        = {Buijs, Urtzi and F{\'e}lix, Yves and Murillo, Aniceto},
  title         = {{$L_\infty$-models of based mapping spaces}},
  journal       = {J. Math. Soc. Japan},
  volume        = {63},
  number        = {2},
  pages         = {503--524},
  year          = {2011},
  doi           = {10.2969/jmsj/06320503},
  eprint        = {1209.4756},
  archivePrefix = {arXiv},
  primaryClass  = {math.AT}
}

@article{FerrariRuffino2012,
    author        = "Ferrari Ruffino, Fabio",
    title         = "{{Classifying A-field and B-field configurations in the presence of D-branes - Part II: Stacks of D-branes}}",
    journal       = "Nucl. Phys. B",
    volume        = "858",
    pages         = "377--404",
    year          = "2012",
    doi           = "10.1016/j.nuclphysb.2012.01.013",
    eprint        = "1104.2798",
    archivePrefix = "arXiv",
    primaryClass  = "hep-th"
}

@article{BonoraFerrariSavelli2008,
    author = "Bonora, Loriano and Ferrari Ruffino, Fabio and Savelli, Raffaele",
    title = "{{Classifying A-field and B-field configurations in the presence of D-branes}}",
    journal = "JHEP",
    volume = "12",
    pages = "078",
    year = "2008",
    doi = "10.1088/1126-6708/2008/12/078",
    eprint = "0810.4291",
    archivePrefix = "arXiv",
    primaryClass = "hep-th"
}

@article{CareyJohnsonMurray2004,
    author = "Carey, Alan and Johnson, Stuart and Murray, Michael",
    title = "{{Holonomy on D-Branes}}",
    journal = "Journal of Geometry and Physics",
    volume = "52",
    number = "2",
    pages = "186--216",
    year = "2004",
    doi = "10.1016/j.geomphys.2004.02.008",
    eprint = "hep-th/0204199",
    archivePrefix = "arXiv"
}

@article{Gawedzki2002,
    author = "Gawedzki, Krzysztof and Reis, Nuno",
    title = "{{WZW Branes and Gerbes}}",
    journal = "Rev. Math. Phys.",
    volume = "14",
    number = "12",
    pages = "1281--1334",
    year = "2002",
    doi = "10.1142/S0129055X02001557",
    eprint = "hep-th/0205233",
    archivePrefix = "arXiv"
}

@article{FSS15-M5WZW,
    author = "Fiorenza, Domenico and Sati, Hisham and Schreiber, Urs",
    title = "{{The WZW term of the M5-brane and differential cohomotopy}}",
    journal = "J. Math. Phys.",
    volume = "56",
    number = "10",
    pages = "102301",
    year = "2015",
    doi = "10.1063/1.4932618",
    eprint = "1506.07557",
    archivePrefix = "arXiv",
    primaryClass = "math-ph"
}

@article{GreenHarveyMoore1996,
    author        = "Green, Michael B. and Harvey, Jeffrey A. and Moore, Gregory W.",
    title         = "{{I-brane inflow and anomalous couplings on D-branes}}",
    journal       = "Class. Quant. Grav.",
    volume        = "14",
    pages         = "47--52",
    year          = "1997",
    doi           = "10.1088/0264-9381/14/1/008",
    eprint        = "hep-th/9605033",
    archivePrefix = "arXiv",
    primaryClass  = "hep-th",
    reportNumber  = "EFI-96-15, RU-96-31"
}

@article{MinasianMoore1997,
  author        = {Minasian, Ruben and Moore, Gregory},
  title         = {{K-theory and Ramond-Ramond charge}},
  journal       = {Journal of High Energy Physics},
  volume        = {1997},
  number        = {11},
  pages         = {002},
  year          = {1997},
  doi           = {10.1088/1126-6708/1997/11/002},
  eprint        = {hep-th/9710230},
  archivePrefix = {arXiv},
  eprinttype    = {arxiv},
  primaryClass  = {hep-th}
}

@inproceedings{Freed2001,
  author    = {Freed, Daniel S.},
  title     = {{K-Theory in Quantum Field Theory}},
  booktitle = {Current Developments in Mathematics, 2001},
  series    = {Current Developments in Mathematics},
  volume    = {2001},
  number    = {1},
  pages     = {41--87},
  year      = {2001},
  publisher = {International Press of Boston},
  doi       = {10.4310/cdm.2001.v2001.n1.a2},
  eprint    = {math-ph/0206031},
  archivePrefix = {arXiv}
}

@article{KahleValentino2014,
    author = "Kahle, Alexander and Valentino, Alessandro",
    title = "{{T-Duality and Differential K-Theory}}",
    journal = "Communications in Contemporary Mathematics",
    volume = "16",
    number = "02",
    pages = "1350014",
    year = "2014",
    doi = "10.1142/S0219199713500144",
    eprint = "0912.2516",
    archivePrefix = "arXiv",
    primaryClass = "math.KT"
}

@misc{Timm2023,
  author       = {Timm, Carsten},
  title        = {{Theory of Superconductivity}},
  howpublished = {Lecture notes, TU Dresden},
  year         = {2023},
  url          = {https://ncatlab.org/nlab/files/Timm-Superconductivity.pdf}
}

@article{Schwinger1966,
    author = "Schwinger, Julian",
    title = "{{Magnetic Charge and Quantum Field Theory}}",
    doi = "10.1103/PhysRev.144.1087",
    journal = "Phys. Rev.",
    volume = "144",
    pages = "1087--1093",
    year = "1966"
}

@article{Zwanziger1968,
    author = "Zwanziger, Daniel",
    title = "{{Quantum Field Theory of Particles with Both Electric and Magnetic Charges}}",
    doi = "10.1103/PhysRev.176.1489",
    journal = "Phys. Rev.",
    volume = "176",
    pages = "1489--1495",
    year = "1968"
}

@misc{BaSS26-UnstableK,
  title={{Flux Quantization of Type IIA in Unstable K-Theory}},
  author={Banerjee, Pinak and Sati, Hisham and Schreiber, Urs},
  journal={arXiv preprint arXiv:2605.25035},
  year={2026},
  month={5},
  eprint={2605.25035},
  archivePrefix={arXiv},
  primaryClass={hep-th}
}

@misc{SS26-BBC,
  title={{Bulk-Edge Correspondence via Higher Gauge Theory}}, 
  author={Hisham Sati and Urs Schreiber},
  year={2026},
  eprint={2605.10232},
  archivePrefix={arXiv},
  primaryClass={hep-th},
}

@incollection{Kirillov1990,
  author    = {Kirillov, Alexandre},
  title     = {{Geometric Quantization}},
  booktitle = {{Dynamical Systems IV}},
  series    = {Encyclopedia of Mathematical Sciences},
  volume    = {4},
  publisher = {Springer},
  year      = {1990},
  pages     = {137--172},
  doi       = {10.1007/978-3-662-06793-2_2}
}

@book{HusemollerEtAl2008,
  title={{Basic Bundle Theory and K-Cohomology Invariants}},
  author={Husem{\"o}ller, Dale and Joachim, Michael and Jur{\v{c}}o, Branislav and Schottenloher, Martin},
  series={Lecture Notes in Physics},
  volume={726},
  year={2008},
  publisher={Springer},
  doi={10.1007/978-3-540-74956-1}
}

@unpublished{Lurie2017,
  author = {Lurie, Jacob},
  title  = {{Higher Algebra}},
  year   = {2017},
  month  = {9},
  url    = {https://www.math.ias.edu/~lurie/papers/HA.pdf}
}

@article{SS24-Cyc,
  author        = {Sati, Hisham and Schreiber, Urs},
  title         = {{Cyclification of Orbifolds}},
  journal       = {Communications in Mathematical Physics},
  volume        = {405},
  pages         = {67},
  year          = {2024},
  doi           = {10.1007/s00220-023-04929-w},
  eprint        = {2212.13836},
  archiveprefix = {arXiv},
  primaryclass  = {math.AT}
}

@article{Maxwell1865,
  author  = {James Clerk Maxwell},
  title   = {{A Dynamical Theory of the Electromagnetic Field}},
  journal = {Philosophical Transactions of the Royal Society of London},
  year    = {1865},
  volume  = {155},
  pages   = {459--512},
  doi     = {10.1098/rstl.1865.0008},
  url     = {https://www.jstor.org/stable/108892}
}

@article{Waldorf2026,
  author        = {Waldorf, Konrad},
  title         = {{Adjusted connections on non-abelian bundle gerbes}},
  year          = {2026},
  eprint        = {2604.24513},
  archivePrefix = {arXiv},
  primaryClass  = {math.DG}
}

@article{RistSaemannWolf2026,
  author        = {Rist, Dominik and Saemann, Christian and Wolf, Martin},
  title         = {{Explicit Non-Abelian Gerbes with Connections}},
  journal       = {Journal of Physics A: Mathematical and Theoretical},
  year          = {2026},
  volume        = {59},
  number        = {3},
  pages         = {035201},
  doi           = {10.1088/1751-8121/ae2e60},
  eprint        = {2203.00092},
  archivePrefix = {arXiv},
  primaryClass  = {math.AT}
}

@article{SchreiberWaldorf2013,
  author        = {Urs Schreiber and Konrad Waldorf},
  title         = {{Connections on non-abelian gerbes and their holonomy}},
  journal       = {Theory and Applications of Categories},
  year          = {2013},
  volume        = {28},
  number        = {17},
  pages         = {476--540},
  url           = {http://www.tac.mta.ca/tac/volumes/28/17/28-17abs.html},
  eprint        = {0808.1923},
  archivePrefix = {arXiv},
  primaryClass  = {math.DG}
}

@incollection{BaezSchreiber2007,
  author    = {John C. Baez and Urs Schreiber},
  title     = {{Higher Gauge Theory}},
  booktitle = {Categories in Algebra, Geometry and Mathematical Physics},
  series    = {Contemporary Mathematics},
  volume    = {431},
  pages     = {7--30},
  publisher = {American Mathematical Society},
  year      = {2007},
  doi       = {10.1090/conm/431/08270},
  eprint    = {math/0511710},
  archivePrefix = {arXiv},
  primaryClass  = {math.DG}
}

@techreport{Joyal2008,
  author      = {Joyal, Andr\'{e}},
  title       = {{The theory of quasi-categories and its applications}},
  institution = {Centre de Recerca Matem\`{a}tica},
  address     = {Barcelona},
  type        = {Quaderns},
  number      = {45},
  month       = 2,
  year        = {2008},
  note        = {Lectures at Advanced Course on Simplicial Methods in Higher Categories, Vol. II},
  url         = {http://mat.uab.cat/~kock/crm/hocat/advanced-course/Quadern45-2.pdf}
}

@book{Borceux1994-I,
  author    = {Borceux, Francis},
  title     = {{Handbook of Categorical Algebra: Volume 1, Basic Category Theory}},
  series    = {Encyclopedia of Mathematics and its Applications},
  volume    = {50},
  publisher = {Cambridge University Press},
  year      = {1994},
  doi       = {10.1017/CBO9780511525858},
  address   = {Cambridge}
}

@book{Richter2020,
  author    = {Richter, Birgit},
  title     = {{From Categories to Homotopy Theory}},
  series    = {Cambridge Studies in Advanced Mathematics},
  volume    = {188},
  publisher = {Cambridge University Press},
  year      = {2020},
  doi       = {10.1017/9781108855891},
  address   = {Cambridge}
}

@book{CattaneoGiaquintoXu2011,
  editor    = {Cattaneo, Alberto S. and Giaquinto, Anthony and Xu, Ping},
  title     = {{Higher Structures in Geometry and Physics: In Honor of Murray Gerstenhaber and Jim Stasheff}},
  series    = {Progress in Mathematics},
  volume    = {287},
  publisher = {Birkh{\"a}user},
  address   = {Boston, MA},
  year      = {2011},
  doi       = {10.1007/978-0-8176-4735-3}
}

@article{FSS2019,
  author  = {Fiorenza, Domenico and Sati, Hisham and Schreiber, Urs},
  title   = {{The Rational Higher Structure of M-theory}},
  journal = {Fortschritte der Physik},
  year    = {2019},
  volume  = {67},
  number  = {8--9},
  pages   = {1910017},
  doi     = {10.1002/prop.201910017},
  eprint  = {1903.02834},
  archivePrefix = {arXiv},
  primaryClass  = {hep-th}
}

@article{JurcoEtAl2019,
  author  = {Jur{\v{c}}o, Branislav and Saemann, Christian and Schreiber, Urs and Wolf, Martin},
  title   = {{Higher Structures in M-Theory}},
  journal = {Fortschritte der Physik},
  year    = {2019},
  volume  = {67},
  number  = {8--9},
  pages   = {1910001},
  doi     = {10.1002/prop.201910001},
  eprint  = {1903.02807},
  archivePrefix = {arXiv},
  primaryClass  = {hep-th}
}

@misc{GagliardoEtAl2025,
  title={{Principal 3-Bundles with Adjusted Connections}}, 
  author={Gianni Gagliardo and Christian Saemann and Roberto Tellez-Dominguez},
  year={2025},
  eprint={2505.13368},
  archivePrefix={arXiv},
  primaryClass={math-ph}
}

@article{Lawvere2007,
  author  = {Lawvere, F. William},
  title   = {{Axiomatic cohesion}},
  journal = {Theory and Applications of Categories},
  year    = {2007},
  volume  = {19},
  number  = {3},
  pages   = {41--49},
  url     = {http://www.tac.mta.ca/tac/volumes/19/3/19-03abs.html}
}

@article{BaezHoffnung2011,
  author        = {Baez, John C. and Hoffnung, Alexander E.},
  title         = {{Convenient Categories of Smooth Spaces}},
  journal       = {Transactions of the American Mathematical Society},
  year          = {2011},
  volume        = {363},
  number        = {11},
  pages         = {5789--5825},
  doi           = {10.1090/S0002-9947-2011-05107-X},
  eprint        = {0807.1704},
  archivePrefix = {arXiv},
  primaryClass  = {math.DG}
}

@article{Steenrod1967,
  author  = {Steenrod, Norman E.},
  title   = {{A convenient category of topological spaces}},
  journal = {Michigan Mathematical Journal},
  year    = {1967},
  volume  = {14},
  number  = {2},
  pages   = {133--152},
  doi     = {10.1307/mmj/1028999711}
}

@book{Tu2011,
  author    = {Tu, Loring W.},
  title     = {An Introduction to Manifolds},
  edition   = {2nd},
  series    = {Universitext},
  publisher = {Springer},
  address   = {New York, NY},
  year      = {2011},
  doi       = {10.1007/978-1-4419-7400-6},
  isbn      = {978-1-4419-7399-3}
}

@article{Jardine1987,
  author  = {Jardine, John F.},
  title   = {{Simplicial presheaves}},
  journal = {Journal of Pure and Applied Algebra},
  year    = {1987},
  volume  = {47},
  number  = {1},
  pages   = {35--87},
  doi     = {10.1016/0022-4049(87)90100-9},
  issn    = {0022-4049}
}

@article{McKay2019,
  author        = {McKay, Benjamin},
  title         = {{Introduction to exterior differential systems}},
  journal       = {Banach Center Publications},
  year          = {2019},
  volume        = {117},
  pages         = {45--55},
  doi           = {10.4064/bc117-2},
  eprint        = {1706.09697},
  archivePrefix = {arXiv},
  primaryClass  = {math.DG},
  publisher     = {Institute of Mathematics, Polish Academy of Sciences}
}

@book{BryantEtAl1991,
  title     = {{Exterior Differential Systems}},
  author    = {Bryant, Robert L. and Chern, Shiing-Shen and Gardner, Robert B. and Goldschmidt, Hubert L. and Griffiths, Phillip A.},
  series    = {Mathematical Sciences Research Institute Publications},
  volume    = {18},
  year      = {1991},
  publisher = {Springer-Verlag},
  address   = {New York, NY},
  doi       = {10.1007/978-1-4613-9714-4},
  isbn      = {978-1-4613-9716-8}
}

@misc{Schreiber2014,
  author       = {Schreiber, Urs},
  title        = {{Higher field bundles for Gauge fields}},
  howpublished = {Talk given at \textit{Operator and geometric analysis on quantum theory}, Levico Terme (Trento), Italy},
  month        = sep,
  year         = {2014},
  url          = {https://ncatlab.org/schreiber/show/Higher+Field+Bundles},
  note         = {September 15--19, 2014}
}

@incollection{BursztyndelHoyo2025,
  author        = {Bursztyn, Henrique and del Hoyo, Matias},
  title         = {{Lie Groupoids}},
  booktitle     = {Encyclopedia of Mathematical Physics},
  edition       = {2nd},
  publisher     = {Elsevier},
  year          = {2025},
  pages         = {469},
  doi           = {10.1016/B978-0-323-95703-8.00024-0},
  eprint        = {2309.14105},
  archivePrefix = {arXiv},
  primaryClass  = {math.DG}
}

@book{Brylinski1993,
  author    = {Brylinski, Jean-Luc},
  title     = {{Loop Spaces, Characteristic Classes and Geometric Quantization}},
  series    = {Progress in Mathematics},
  volume    = {107},
  publisher = {Birkhäuser},
  address   = {Boston, MA},
  year      = {1993},
  doi       = {10.1007/978-0-8176-4731-5},
  isbn      = {978-0-8176-3644-9}
}

@book{GiachettaMangiarottiSardanashvily2009,
  title     = {{Advanced Classical Field Theory}},
  author    = {Giachetta, Giovanni and Mangiarotti, Luigi and Sardanashvily, Gennadi},
  year      = {2009},
  publisher = {World Scientific},
  address   = {Singapore},
  doi       = {10.1142/7189},
  isbn      = {978-981-283-914-5}
}

@inproceedings{Giotopoulos2025,
  author        = {Giotopoulos, Grigorios},
  title         = {{Sheaf Topos Theory: A powerful setting for Lagrangian Field Theory}},
  booktitle     = {Proceedings of Science},
  volume        = {490},
  year          = {2025},
  pages         = {300},
  doi           = {10.22323/1.490.0300},
  eprint        = {2504.08095},
  archivePrefix = {arXiv},
  primaryClass  = {math-ph}
}

@book{GriffithsMorgan2013,
  title     = {{Rational Homotopy Theory and Differential Forms}},
  author    = {Griffiths, Phillip and Morgan, John},
  series    = {Progress in Mathematics},
  volume    = {16},
  year      = {2013},
  publisher = {Birkh{\"a}user},
  address   = {Boston, MA},
  doi       = {10.1007/978-1-4614-8468-4},
  isbn      = {978-1-4614-8467-7}
}

@article{LadaStasheff1992,
    author        = {Lada, Tom and Stasheff, Jim},
    title         = {{Introduction to sh Lie algebras for physicists}},
    doi           = {10.1007/BF00671791},
    journal       = {Int. J. Theor. Phys.},
    volume        = {32},
    pages         = {1087--1103},
    year          = {1993},
    eprint        = {hep-th/9209099},
    archivePrefix = {arXiv},
    primaryClass  = {hep-th}
}

@article{Friedman2012,
  author  = {Friedman, Greg},
  title   = {{An elementary illustrated introduction to simplicial sets}},
  journal = {Rocky Mountain Journal of Mathematics},
  volume  = {42},
  number  = {2},
  pages   = {353--423},
  year    = {2012},
  doi     = {10.1216/RMJ-2012-42-2-353},
  eprint  = {0809.4221},
  archivePrefix = {arXiv},
  primaryClass  = {math.AT}
}

@incollection{AntolinCamarena2016,
  author    = {Antolín-Camarena, Omar},
  title     = {{A whirlwind tour of the world of $(\infty,1)$-categories}},
  booktitle = {Mexican Mathematicians Abroad: Recent Contributions},
  editor    = {Aguilar, Noé and Hernandez-Torres, Gerardo and Martinez-Garcia, Ivan and Perez-Sandoval, Juan Pablo},
  series    = {Contemporary Mathematics},
  volume    = {657},
  publisher = {American Mathematical Society},
  address   = {Providence, RI},
  year      = {2016},
  pages     = {1--18},
  doi       = {10.1090/conm/657/13106},
  eprint    = {1303.4669},
  archivePrefix = {arXiv},
  primaryClass  = {math.CT}
}

@incollection{Hess2007,
  author    = {Hess, Kathryn},
  title     = {{Rational homotopy theory: a brief introduction}},
  booktitle = {Interactions between Homotopy Theory and Algebra},
  editor    = {Avramov, Luchezar L. and Christensen, J. Daniel and Dwyer, William G. and Mandell, Michael A. and Shipley, Brooke},
  series    = {Contemporary Mathematics},
  volume    = {436},
  publisher = {American Mathematical Society},
  address   = {Providence, RI},
  year      = {2007},
  pages     = {175--202},
  doi       = {10.1090/conm/436/08409},
  eprint    = {math/0604626},
  archivePrefix = {arXiv},
  primaryClass  = {math.AT}
}

@misc{KSS26-HigherDimAnyons,
  title = {{Higher-Dimensional Anyons via Higher Cohomotopy}},
  author = {Kallel, Sadok and Sati, Hisham and Schreiber, Urs},
  year = {2026},
  eprint = {2601.03150},
  archivePrefix = {arXiv},
  primaryClass = {cond-mat.str-el},
  month = {01},
}

@book{Awodey2010,
  author    = {Awodey, Steve},
  title     = {{Category Theory}},
  series    = {Oxford Logic Guides},
  volume    = {52},
  edition   = {2nd},
  publisher = {Oxford University Press},
  address   = {Oxford},
  year      = {2010},
  doi       = {10.1093/acprof:oso/9780198568612.001.0001},
  isbn      = {978-0-19-923718-0}
}

@article{BaSS26-MString,
  author        = {Banerjee, Pinak and Sati, Hisham and Schreiber, Urs},
  title         = {{Flux Quantization on M-Strings}},
  journal       = {Journal of Physics A: Math. Theor.},
  volume        = {59},
  number        = {26},
  year          = {2026},
  month         = {3},
  doi           = {10.1088/1751-8121/ae808a},
  eprint        = {2603.14440},
  archivePrefix = {arXiv},
  primaryClass  = {hep-th}
}

@article{Banerjee2025-Potentials,
    author        = {Banerjee, Pinak},
    title         = {{Gauge potentials on the M5 brane in twisted equivariant cohomotopy}},
    journal       = {Journal of High Energy Physics},
    volume        = {2025},
    number        = {12},
    pages         = {162},
    year          = {2025},
    doi           = {10.1007/JHEP12(2025)162},
    eprint        = {2507.07049},
    archivePrefix = {arXiv},
    primaryClass  = {hep-th}
}

@article{BMSS2019,
  author        = {Braunack-Mayer, Vincent and Sati, Hisham and Schreiber, Urs},
  title         = {{Gauge enhancement of Super M-Branes via rational parameterized stable homotopy theory}},
  journal       = {Communications in Mathematical Physics},
  volume        = {371},
  pages         = {197--265},
  year          = {2019},
  doi           = {10.1007/s00220-019-03441-4},
  eprint        = {1806.01115},
  archivePrefix = {arXiv},
  primaryClass  = {hep-th}
}

@article{Bunke2012,
      title={Differential cohomology}, 
      author={Ulrich Bunke},
      year={2012},
          eprint={1208.3961},
      archivePrefix={arXiv},
      primaryClass={math.AT},
      note={Course notes}
}

@article{GSS24-SuGra,
  author        = {Giotopoulos, Grigorios and Sati, Hisham and Schreiber, Urs},
  title         = {{Flux quantization on 11D superspace}},
  journal       = {Journal of High Energ. Physics},
  volume        = {2024},
  number        = {7},
  pages         = {82},
  year          = {2024},
  month         = {7},
  doi           = {10.1007/JHEP07(2024)082},
  eprint        = {2403.16456},
  archiveprefix = {arXiv},
  primaryclass  = {hep-th}
}

@article{Herzog2007,
    author = "Herzog, Christopher P. and Kovtun, Pavel and Sachdev, Subir and Son, Dam Thanh",
    title = "{{Quantum critical transport, duality, and M-theory}}",
    journal = "Phys. Rev. D",
    volume = "75",
    year = "2007",
    pages = "085020",
    doi = "10.1103/PhysRevD.75.085020",
    eprint = "hep-th/0701036",
    archivePrefix = "arXiv",
    primaryClass = "hep-th"
}

@article{SS25-Complete,
  author        = {Sati, Hisham and Schreiber, Urs},
  title         = {{Complete Topological Quantization of Higher Gauge Fields}},
  journal       = {SciPost Physics Lecture Notes},
  volume        = {127},
  year          = {2026},
  doi           = {10.21468/SciPostPhysLectNotes.127},
  eprint        = {2512.12431},
  archivePrefix = {arXiv},
  primaryClass  = {hep-th}
}

@book{James1984,
  author    = {James, Ioan Mackenzie},
  title     = {{General Topology and Homotopy Theory}},
  publisher = {Springer},
  year      = {1984},
  doi       = {10.1007/978-1-4613-8283-6}
}

@article{CremmerJuliaLuPope1998,
  author        = {Cremmer, Eugene and Julia, Bernard and L{\"u}, H. and Pope, Christopher N.},
  title         = {{Dualisation of Dualities, II: Twisted Self-Duality of Doubled Fields and Superdualities}},
  journal       = {Nuclear Physics B},
  volume        = {535},
  number        = {1--2},
  pages         = {242--292},
  year          = {1998},
  month         = {11},
  doi           = {10.1016/S0550-3213(98)00552-5},
  eprint        = {hep-th/9806106},
  archivePrefix = {arXiv},
  primaryClass  = {hep-th}
}

@article{MiemiecSchnakenburg2006,
  author        = {Miemiec, Andr{\'e} and Schnakenburg, Igor},
  title         = {{Basics of M-Theory}},
  journal       = {Fortschritte der Physik},
  volume        = {54},
  number        = {1},
  year          = {2006},
  month         = {1},
  pages         = {5--72},
  doi           = {10.1002/prop.200510256},
  eprint        = {hep-th/0509137},
  archivePrefix = {arXiv},
  primaryClass  = {hep-th}
}

@article{ChoGangKim2020,
    author        = {Cho, Gil Young and Gang, Dongmin and Kim, Hee-Cheol},
    title         = {{M-theoretic Genesis of Topological Phases}},
    journal       = {Journal of High Energy Physics},
    volume        = {2020},
    number        = {11},
    pages         = {115},
    year          = {2020},
    doi           = {10.1007/JHEP11(2020)115},
    eprint        = {2007.01532},
    archivePrefix = {arXiv},
    primaryClass  = {hep-th}
}

@article{FSS18-TD,
    author = {Fiorenza, Domenico and Sati, Hisham and Schreiber, Urs},
    title = {{T-Duality from super Lie $n$-algebra cocycles for super $p$-branes}},
    journal = {Adv. Theor. Math. Phys.},
    volume = {22},
    number = {5},
    year = {2018},
    pages = {1209--1270},
    doi = {10.4310/ATMP.2018.v22.n5.a3},
    eprint = {1611.06536},
    archivePrefix = {arXiv},
    primaryClass = {hep-th}
}

@article{FSS18-TDR,
  author        = {Fiorenza, Domenico and Sati, Hisham and Schreiber, Urs},
  title         = {{T-duality in rational homotopy theory via $L_\infty$-algebras}},
  journal       = {Geometry, Topology and Maths. Physics Journal},
  volume        = {1},
  year          = {2018},
  month         = {12},
  note          = {Special volume in tribute of Jim Stasheff and Dennis Sullivan},
  eprint        = {1712.00758},
  archiveprefix = {arXiv},
  primaryclass  = {math-ph}
}

@article{GSS25-TD,
    author = {Giotopoulos, Grigorios and Sati, Hisham and Schreiber, Urs},
    title = {{Super-Lie$_\infty$ T-Duality and M-Theory}},
    journal = {Rev. Math. Phys.},
    year = {2025},
    doi = {10.1142/S0129055X25500187},
    eprint = {2411.10260},
    archivePrefix = {arXiv},
    primaryClass = {hep-th}
}

@article{FSS15-WZW,
    author = {Fiorenza, Domenico and Sati, Hisham and Schreiber, Urs},
    title = {{Super Lie $n$-algebra extensions, higher WZW models, and super $p$-branes with tensor multiplet fields}},
    journal = {International Journal of Geometric Methods in Modern Physics},
    volume = {12},
    number = {02},
    pages = {1550018},
    year = {2015},
    doi = {10.1142/S0219887815500188},
    eprint = {1308.5264},
    archivePrefix = {arXiv},
    primaryClass = {hep-th}
}

@inproceedings{vanNieuwenhuizen1983,
    author = {van Nieuwenhuizen, Peter},
    title = {{Free Graded Differential Superalgebras}},
    booktitle = {{Group Theoretical Methods in Physics}},
    series = {Lecture Notes in Physics},
    volume = {180},
    publisher = {Springer},
    year = {1983},
    pages = {228--247},
    doi = {10.1007/3-540-12291-5_29}
}

@article{CastellaniDAuria2025,
    author = {Castellani, Leonardo and D'Auria, Riccardo},
    title = {{The $L_\infty$-structure of Free Differential Algebras and Supergravity}},
    journal = {International Journal of Modern Physics A},
    year = {2025},
    doi = {10.1142/S0217751X25501787},
    eprint = {2507.20344},
    archivePrefix = {arXiv},
    primaryClass = {hep-th}
}

@article{HoweSezgin1997,
  author        = {Howe, P. S. and Sezgin, E.},
  title         = {{$D = 11, p = 5$}},
  journal       = {Phys. Lett. B},
  volume        = {394},
  year          = {1997},
  pages         = {62--66},
  doi           = {10.1016/S0370-2693(96)01672-3},
  eprint        = {hep-th/9611008},
  archivePrefix = {arXiv},
  primaryClass  = {hep-th}
}

@article{BrinkHowe1980,
  author  = {Brink, L. and Howe, P.},
  title   = {{Eleven-Dimensional Supergravity on the Mass-Shell in Superspace}},
  journal = {Phys. Lett. B},
  volume  = {91},
  year    = {1980},
  pages   = {384--386},
  doi     = {10.1016/0370-2693(80)91002-3}
}

@article{CremmerFerrara1980,
  author  = {Cremmer, E. and Ferrara, S.},
  title   = {{Formulation of Eleven-Dimensional Supergravity in Superspace}},
  journal = {Phys. Lett. B},
  volume  = {91},
  year    = {1980},
  pages   = {61--66},
  doi     = {10.1016/0370-2693(80)90662-0}
}

@book{CDF1991,
  author    = {Castellani, L. and D'Auria, R. and Fr{\'e}, P.},
  title     = {{Supergravity and Superstrings -- A Geometric Perspective}},
  publisher = {World Scientific},
  year      = {1991},
  doi       = {10.1142/0224}
}

@article{WenNiu1990,
  title     = {{Ground-state degeneracy of the FQH states in presence of random potential and on high genus Riemann surfaces}},
  author    = {Wen, X.-G. and Niu, Q.},
  journal   = {Phys. Rev. B},
  volume    = {41},
  year      = {1990},
  pages     = {9377--9396},
  doi       = {10.1103/PhysRevB.41.9377}
}

@article{MooreWitten2000,
    author = {Moore, Gregory and Witten, Edward},
    title = {{Self-Duality, Ramond-Ramond Fields, and K-Theory}},
    journal = {JHEP},
    volume = {05},
    year = {2000},
    pages = {032},
    doi = {10.1088/1126-6708/2000/05/032},
    eprint = {hep-th/9912279},
    archivePrefix = {arXiv},
    primaryClass = {hep-th}
}

@article{HopkinsSinger2005,
  title = {{Quadratic Functions in Geometry, Topology, and M-Theory}},
  author = {Hopkins, Michael J. and Singer, Isadore M.},
  journal = {Journal of Differential Geometry},
  volume = {70},
  number = {3},
  pages = {329--452},
  year = {2005},
  month = {7},
  doi = {10.4310/jdg/1143642908},
  eprint = {math/0211216},
  archivePrefix = {arXiv},
  primaryClass = {math.AT}
}

@article{Witten1997Flux,
  title = {{On Flux Quantization in M-Theory and the Effective Action}},
  author = {Witten, Edward},
  journal = {Journal of Geometry and Physics},
  volume = {22},
  number = {1},
  pages = {1--13},
  year = {1997},
  month = {4},
  doi = {10.1016/S0393-0440(96)00042-3},
  eprint = {hep-th/9609122},
  archivePrefix = {arXiv},
  primaryClass = {hep-th}
}

@article{Witten1997Fivebrane,
  title = {{Five-Brane Effective Action in M-Theory}},
  author = {Witten, Edward},
  journal = {Journal of Geometry and Physics},
  volume = {22},
  number = {2},
  pages = {103--133},
  year = {1997},
  month = {5},
  doi = {10.1016/S0393-0440(97)80160-X},
  eprint = {hep-th/9610234},
  archivePrefix = {arXiv},
  primaryClass = {hep-th}
}

@article{FreedWitten1999,
  title = {{Anomalies in String Theory with D-Branes}},
  author = {Freed, Daniel S. and Witten, Edward},
  journal = {Asian Journal of Mathematics},
  volume = {3},
  number = {4},
  pages = {819--852},
  year = {1999},
  month = {12},
  doi = {10.4310/AJM.1999.v3.n4.a6},
  eprint = {hep-th/9907189},
  archivePrefix = {arXiv},
  primaryClass = {hep-th}
}

@incollection{Gawedzki1988,
    author    = "Gawedzki, Krzysztof",
    editor    = "'t Hooft, Gerard and Jaffe, Arthur and Mack, Gerhard and Mitter, Peter K. and Stora, Raymond",
    title     = "{{Topological Actions in Two-Dimensional Quantum Field Theories}}",
    booktitle = "{Nonperturbative Quantum Field Theory}",
    series    = "NATO ASI Series",
    volume    = "185",
    pages     = "101--141",
    publisher = "Springer",
    address   = "New York, NY",
    year      = "1988",
    doi       = "10.1007/978-1-4613-0729-7_5"
}

@article{Freed2002,
  title = {{Dirac Charge Quantization and Generalized Differential Cohomology}},
  author = {Freed, Daniel S.},
  journal = {Surveys in Differential Geometry},
  volume = {7},
  number = {1},
  pages = {129--194},
  year = {2002},
  doi = {10.4310/SDG.2002.v7.n1.a6},
  eprint = {hep-th/0011220},
  archivePrefix = {arXiv},
  primaryClass = {hep-th}
}

@book{Whitehead1978,
  author    = {Whitehead, George W.},
  title     = {{Elements of Homotopy Theory}},
  year      = {1978},
  publisher = {Springer},
  address   = {New York},
  series    = {Graduate Texts in Mathematics},
  volume    = {61},
  isbn      = {978-1-4612-6318-0},
  doi       = {10.1007/978-1-4612-6318-0}
}

@article{LazaroiuShahbazi2022b,
  title = {{The geometry and DSZ quantization of four-dimensional supergravity}},
  author = {Lazaroiu, Calin and Shahbazi, Carlos S.},
  journal = {Letters in Mathematical Physics},
  volume = {113},
  number = {1},
  pages = {9},
  year = {2022},
  month = {12},
  ISSN = {1573-0530},
  DOI = {10.1007/s11005-022-01626-y},
  publisher = {Springer},
  eprint = {2101.07778},
  archivePrefix = {arXiv},
  primaryClass = {math.DG}
}

@article{LazaroiuShahbazi2022a,
  title = {{The duality covariant geometry and DSZ quantization of abelian gauge theory}},
  author = {Lazaroiu, Calin and Shahbazi, Carlos S.},
  journal = {Advances in Theoretical and Mathematical Physics},
  volume = {26},
  number = {7},
  pages = {2213--2312},
  year = {2022},
  ISSN = {1095-0753},
  DOI = {10.4310/atmp.2022.v26.n7.a5},
  publisher = {International Press of Boston},
  eprint = {2101.07236},
  archivePrefix = {arXiv},
  primaryClass = {math.DG}
}

@article{BeckerBeniniSchenkelSzabo2017,
  title = {{Abelian Duality on Globally Hyperbolic Spacetimes}},
  author = {Becker, Christian and Benini, Marco and Schenkel, Alexander and Szabo, Richard J.},
  journal = {Communications in Mathematical Physics},
  volume = {349},
  number = {1},
  pages = {361--392},
  year = {2017},
  month = {5},
  ISSN = {1432-0916},
  DOI = {10.1007/s00220-016-2669-9},
  publisher = {Springer Science and Business Media LLC},
  eprint = {1511.00316},
  archivePrefix = {arXiv},
  primaryClass = {hep-th}
}

@article{FreedMooreSegal2007,
  sortkey = {FMS2007a},
  title = {{The Uncertainty of Fluxes}},
  author = {Freed, Daniel S. and Moore, Gregory W. and Segal, Graeme},
  journal = {Communications in Mathematical Physics},
  volume = {271},
  number = {1},
  pages = {247--274},
  year = {2007},
  month = {1},
  ISSN = {1432-0916},
  DOI = {10.1007/s00220-006-0181-3},
  publisher = {Springer Science and Business Media LLC},
  eprint = {hep-th/0605198},
  archivePrefix = {arXiv},
  primaryClass = {hep-th}
}

@article{FreedMooreSegal2007b,
  sortkey = {FMS2007b},
  author        = {Freed, Daniel S. and Moore, Gregory W. and Segal, Graeme},
  title         = {Heisenberg Groups and Noncommutative Fluxes},
  journal       = {Annals of Physics},
  volume        = {322},
  number        = {1},
  pages         = {236--285},
  month         = {1},
  year          = {2007},
  doi           = {10.1016/j.aop.2006.07.014},
  eprint        = {hep-th/0605200},
  archiveprefix = {arXiv}
}

@article{DuffHoweInamiStelle1987,
  title = {{Superstrings in $D=10$ from supermembranes in $D=11$}},
  author = {Duff, M. J. and Howe, P. S. and Inami, T. and Stelle, K. S.},
  journal = {Physics Letters B},
  volume = {191},
  number = {1–2},
  pages = {70–74},
  year = {1987},
  month = {6},
  DOI = {10.1016/0370-2693(87)91323-2},
  publisher = {Elsevier BV}
}

@book{HenneauxTeitelboim1992,
  title     = {{Quantization of Gauge Systems}},
  author    = {Henneaux, Marc and Teitelboim, Claudio},
  year      = {1992},
  publisher = {Princeton University Press},
  address   = {Princeton, NJ},
  isbn      = {978-0691037691},
  doi       = {10.2307/j.ctv10crg0r}
}

@article{Dirac1949,
  title = {{Forms of Relativistic Dynamics}},
  author = {Dirac, P. A. M.},
  journal = {Reviews of Modern Physics},
  volume = {21},
  number = {3},
  pages = {392--399},
  year = {1949},
  month = {7},
  doi = {10.1103/RevModPhys.21.392},
  url = {https://doi.org/10.1103/RevModPhys.21.392},
  publisher = {American Physical Society}
}

@article{Dirac1931,
  title = {{Quantised Singularities in the Electromagnetic Field}},
  author = {Dirac, P. A. M.},
  journal = {Proceedings of the Royal Society of London. Series A, Containing Papers of a Mathematical and Physical Character},
  volume = {133},
  number = {821},
  pages = {60--72},
  year = {1931},
  month = {9},
  doi = {10.1098/rspa.1931.0130}
}

@article{SS25-ISQS29,
  author        = {Sati, H. and Schreiber, U.},
  title         = {{Non-Lagrangian Construction of Anyons via Flux Quantization in Cohomotopy}},
  journal       = {Journal of Physics: Conference Series},
  volume        = {3152},
  number        = {1},
  pages         = {012024},
  year          = {2025},
  month         = {7},
  publisher     = {IOP Publishing},
  doi           = {10.1088/1742-6596/3152/1/012024},
  note          = {The XXIX International Conference on Integrable Systems and Quantum Symmetries},
  eprint        = {2509.02577},
  archivePrefix = {arXiv},
  primaryClass  = {math-ph}
}

@inbook{SimmsWoodhouse1976,
  author    = {David J. Simms and Nicholas M. J. Woodhouse},
  title     = {{Observables}},
  booktitle = {{Lectures on Geometric Quantization}},
  series    = {Lecture Notes in Physics},
  volume    = {53},
  chapter   = {4},
  publisher = {Springer},
  year      = {1976},
  doi       = {10.1007/3-540-07860-6_4},
  isbn      = {978-3-540-07860-9}
}

@article{KhavkineSchreiber2026,
  author        = {Igor Khavkine and Urs Schreiber},
  title         = {{Synthetic Geometry of Differential Equations: Jets and Comonad Structure}},
  journal       = {Journal of Geometry and Physics},
  year          = {2026},
  note          = {to appear},
  eprint        = {1701.06238},
  archivePrefix = {arXiv},
  primaryClass  = {math.DG}
}

@article{Sc18-ToposLectures,
  author       = {Schreiber, Urs},
  title        = {{Categories and Toposes -- Differential Cohesive Higher Toposes}},
  howpublished = {Lecture series at "Modern Mathematics Methods in Physics: Diffeology, Categories and Toposes and Non-commutative Geometry" Summer School},
  month        = {6},
  year         = {2018},
  address      = {Nesin Mathematics Village -- Şirince, Izmir, Turkey},
  url          = {https://ncatlab.org/schreiber/show/Categories+and+Toposes}
}

@article{BeniniSchenkelSchreiber2018,
  author        = {Benini, Marco and Schenkel, Alexander and Schreiber, Urs},
  title         = {{The Stack of Yang--Mills Fields on Lorentzian Manifolds}},
  journal       = {Communications in Mathematical Physics},
  year          = {2018},
  volume        = {359},
  number        = {2},
  pages         = {765--820},
  month         = {3},
  doi           = {10.1007/s00220-018-3120-1},
  archivePrefix = {arXiv},
  eprint        = {1704.01378},
  primaryClass  = {math-ph}
}

@article{IbortMas2025,
  author        = {Ibort, Alberto and Mas, Arnau},
  title         = {{Smooth sets of fields: A pedagogical introduction}},
  journal       = {Geometric Mechanics},
  year          = {2025},
  volume        = {2},
  number        = {3},
  pages         = {251--274},
  month         = {8},
  doi           = {10.1142/s2972458925400052},
  archivePrefix = {arXiv},
  eprint        = {2510.20422},
  primaryClass  = {math-ph}
}

@book{IglesiasZemmour2013,
  author    = {Iglesias-Zemmour, Patrick},
  title     = {{Diffeology}},
  publisher = {American Mathematical Society (AMS)},
  year      = {2013},
  series    = {Mathematical Surveys and Monographs},
  volume    = {185},
  address   = {Providence, RI},
  isbn      = {978-0-8218-9131-5},
  url       = {https://bookstore.ams.org/surv-185}
}

@article{Milgram1967,
  author  = {Milgram, R. James},
  title   = {{The bar construction and abelian $H$-spaces}},
  journal = {Illinois Journal of Mathematics},
  year    = {1967},
  volume  = {11},
  number  = {2},
  pages   = {242--250},
  month   = jun,
  DOI     = {10.1215/ijm/1256054662}
}

@article{NSS2015a,
  author       = {Nikolaus, Thomas and Schreiber, Urs and Stevenson, Danny},
  title        = {{Principal $\infty$-bundles: general theory}},
  journal      = {Journal of Homotopy and Related Structures},
  year         = {2015},
  volume       = {10},
  number       = {4},
  pages        = {749--801},
  month        = {6},
  doi          = {10.1007/s40062-014-0083-6},
  archivePrefix = {arXiv},
  eprint       = {1207.0248},
  primaryClass = {math.AT}
}

@article{NSS2015b,
  author        = {Nikolaus, Thomas and Schreiber, Urs and Stevenson, Danny},
  title         = {{Principal $\infty$-bundles -- Presentations}},
  journal       = {Journal of Homotopy and Related Structures},
  volume        = {10},
  number        = {3},
  pages         = {565--622},
  year          = {2015},
  doi           = {10.1007/s40062-014-0077-4},
  archiveprefix = {arXiv},
  eprint        = {1207.0249},
  primaryclass  = {math.AT}
}

@book{RudolphSchmidt2017,
  author    = {Rudolph, Gerd and Schmidt, Matthias},
  title     = {{Differential Geometry and Mathematical Physics: Part II. Fibre Bundles, Topology and Gauge Fields}},
  publisher = {Springer Netherlands},
  year      = {2017},
  series    = {Theoretical and Mathematical Physics},
  address   = {Dordrecht},
  isbn      = {9789402409598},
  
  doi       = {10.1007/978-94-024-0959-8}
}

@book{Nakahara2018,
  author    = {Nakahara, Mikio},
  title     = {{Geometry, Topology and Physics}},
  publisher = {CRC Press},
  year      = {2018},
  edition   = {2},
  address   = {Boca Raton, FL},
  isbn      = {9781315275826},
  doi       = {10.1201/9781315275826},
  month     = {10}
}

@book{Husemoller1994,
  author    = {Husemoller, Dale},
  title     = {{Fibre Bundles}},
  publisher = {Springer},
  year      = {1994},
  edition   = {3},
  series    = {Graduate Texts in Mathematics},
  volume    = {20},
  address   = {New York},
  isbn      = {9781475722611},
  doi       = {10.1007/978-1-4757-2261-1}
}

@article{Weinstein1996Groupoids,
  author  = {Weinstein, Alan},
  title   = {{Groupoids: Unifying Internal and External Symmetry -- A Tour through some Examples}},
  journal = {Notices of the American Mathematical Society},
  year    = {1996},
  volume  = {43},
  number  = {7},
  pages   = {744--752},
  month   = {7},
  url     = {http://www.ams.org/notices/199607/weinstein.pdf}
}

@book{AguilarGitlerPrieto2002,
  author    = {Aguilar, Marcelo and Gitler, Samuel and Prieto, Carlos},
  title     = {{Algebraic Topology from a Homotopical Viewpoint}},
  publisher = {Springer},
  year      = {2002},
  series    = {Universitext},
  doi       = {10.1007/b97586},
  isbn      = {9780387224893}
}

@book{Lee2012,
  author    = {Lee, John M.},
  title     = {{Introduction to Smooth Manifolds}},
  edition   = {Second},
  series    = {Graduate Texts in Mathematics},
  volume    = {218},
  publisher = {Springer},
  address   = {New York},
  year      = {2012},
  isbn      = {9781441999825},
  doi       = {10.1007/978-1-4419-9982-5}
}

@book{Moerdijk2003,
  author    = {Moerdijk, Ieke and Mr{\v{c}}un, Janez},
  title     = {{Introduction to Foliations and Lie Groupoids}},
  publisher = {Cambridge University Press},
  year      = {2003},
  address   = {Cambridge},
  doi       = {10.1017/cbo9780511615450},
  isbn      = {9780511615450}
}

@book{IbortRodriguez2021,
  author    = {Ibort, Alberto and Rodriguez, Miguel A.},
  title     = {{An Introduction to Groups, Groupoids and Their Representations}},
  publisher = {CRC Press},
  year      = {2021},
  doi       = {10.1201/b22019},
  isbn      = {9781032086767},
}

@book{Atiyah1967,
  author    = {Atiyah, Michael},
  title     = {{K-Theory}},
  publisher = {W. A. Benjamin},
  year      = {1967},
  address   = {New York},
  note      = {Lecture notes by D. W. Anderson, Harvard University, 1964},
  url       = {https://ncatlab.org/nlab/files/AtiyahKTheory.pdf},
}

@article{AtiyahSegal2004,
  author       = {Atiyah, Michael and Segal, Graeme},
  title        = {{Twisted K-theory}},
  journal      = {Ukrainian Mathematical Bulletin},
  year         = {2004},
  volume       = {1},
  number       = {3},
  pages        = {291--330},
  eprint       = {math/0407054},
  primaryClass = {math.KT},
  url          = {http://iamm.su/en/journals/j879/?VID=10}
}

@article{SS21-M5Anomaly,
  author        = {Sati, Hisham and Schreiber, Urs},
  title         = {{Twisted cohomotopy implies M5-brane anomaly cancellation}},
  journal       = {Letters in Mathematical Physics},
  year          = {2021},
  month         = {9},
  volume        = {111},
  number        = {5},
  doi           = {10.1007/s11005-021-01452-8},

  archivePrefix = {arXiv},
  eprint        = {2002.07737},
  primaryClass  = {hep-th}
}

@article{FSS17-Sphere,
  author        = {Fiorenza, Domenico and Sati, Hisham and Schreiber, Urs},
  title         = {{Rational sphere valued supercocycles in M-theory and type IIA string theory}},
  journal       = {Journal of Geometry and Physics},
  year          = {2017},
  volume        = {114},
  pages         = {91--108},
  month         = {4},
  doi           = {10.1016/j.geomphys.2016.11.024},
  eprint        = {1606.03206},
  archiveprefix = {arXiv},
  primaryclass  = {hep-th}
}

@article{SS25-WilsonLoops,
  title = {{Renormalization of Chern-Simons Wilson Loops via Flux Quantization in Cohomotopy}},
  author = {Sati, Hisham and Schreiber, Urs},
  journal = {Reviews in Mathematical Physics},
  year = {2025},
  doi = {https://doi.org/10.1142/S0129055X25500382},
  eprint = {2509.25336},
  archivePrefix = {arXiv},
  primaryClass = {hep-th}
}

@article{Toen2002,
  author       = {To{\"e}n, Bertrand},
  title        = {{Stacks and Non-abelian cohomology}},
  howpublished = {Lecture at the Introductory Workshop on Algebraic Stacks, Intersection Theory, and Non-Abelian Hodge Theory, MSRI},
  year         = {2002},
  url          = {https://perso.math.univ-toulouse.fr/btoen/files/2015/02/msri2002.pdf}
}

@article{Lurie2014,
  author       = {Lurie, Jacob},
  title        = {Nonabelian {P}oincar{\'e} Duality},
  howpublished = {Lecture 8 in: Tamagawa Numbers via Nonabelian Poincare Duality (282y)},
  year         = {2014},
  url          = {http://www.math.harvard.edu/~lurie/282ynotes/LectureVIII-Poincare.pdf}
}

@book{Kochman1996,
  author    = {Kochman, Stanley O.},
  title     = {{Bordism, Stable Homotopy and Adams Spectral Sequences}},
  publisher = {American Mathematical Society},
  year      = {1996},
  series    = {Fields Institute Monographs},
  volume    = {7},
  address   = {Providence, RI},
  isbn      = {978-1-4704-3134-1},
  doi       = {10.1090/fim/007},
  url       = {https://bookstore.ams.org/fim-7}
}

@incollection{Alfonsi2025HigherGeometry,
  author        = {Alfonsi, Luigi},
  title         = {{Higher Geometry in Physics}},
  booktitle     = {Encyclopedia of Mathematical Physics},
  editor        = {Szabo, Richard and Bojowald, Martin},
  edition       = {2},
  publisher     = {Academic Press},
  year          = {2025},
  pages         = {39--61},
  isbn          = {978-0-323-95706-9},
  doi           = {10.1016/b978-0-323-95703-8.00209-3},
  eprint        = {2312.07308},
  archivePrefix = {arXiv},
  primaryClass  = {hep-th}
}

@book{Duff1999World,
  author    = {Duff, Michael J.},
  title     = {{The World in Eleven Dimensions: Supergravity, Supermembranes and M-Theory}},
  publisher = {IOP Publishing},
  year      = {1999},
  isbn      = {978-0-750-30672-0},
doi ={10.1201/9781482268737}
}

@inbook{SatiSchreiberStasheff2009,
  author        = {Sati, Hisham and Schreiber, Urs and Stasheff, Jim},
  editor        = {Battle, Guy and Brydges, David and Kupiainen, Antti and Magnen, Jacques},
  title         = {{$L_\infty$-Algebra Connections and Applications to String- and Chern-Simons $n$-Transport}},
  booktitle     = {{Q}uantum {F}ield {T}heory},
  series        = {Progress in Mathematics},
  volume        = {272},
  pages         = {303--424},
  publisher     = {Birkh{\"a}user Basel},
  year          = {2009},
  isbn          = {978-3-7643-8736-5},
  doi           = {10.1007/978-3-7643-8736-5_17},
  eprint        = {0801.3480},
  archivePrefix = {arXiv},
  primaryClass  = {math.DG}
}

@article{ToenVezzosi2005,
  author        = {To{\"e}n, Bertrand and Vezzosi, Gabriele},
  title         = {{Homotopical Algebraic Geometry I: Topos Theory}},
  journal       = {Advances in Mathematics},
  year          = {2005},
  volume        = {193},
  number        = {2},
  pages         = {257--372},
  doi           = {10.1016/j.aim.2004.05.004},
  archiveprefix = {arXiv},
  eprint        = {math/0207028},
  primaryclass  = {math.AG}
}

@book{Lurie2009,
  author    = {Jacob Lurie},
  title     = {{Higher Topos Theory}},
  series    = {Annals of Mathematics Studies},
  volume    = {170},
  publisher = {Princeton University Press},
  year      = {2009},
  isbn      = {978-0691140490},
  url       = {https://press.princeton.edu/titles/8957.html}
}

@book{Frankel2011,
  author    = {Frankel, Theodore},
  title     = {{The Geometry of Physics: An Introduction}},
  publisher = {Cambridge University Press},
  year      = {2011},
  doi       = {10.1017/CBO9781139061377},
  isbn      = {9781139061377}
}

@article{Alvarez1985,
  author  = {Alvarez, Orlando},
  title   = {{Topological Quantization and Cohomology}},
  journal = {Commun. Math. Phys.},
  year    = {1985},
  volume  = {100},
  number  = {2},
  pages   = {279--309},
  month   = jun,
  doi     = {10.1007/bf01212452}
}

\end{document}